\documentclass[10pt]{article}
\usepackage{amsfonts}
\usepackage{amsmath}
\usepackage{lipsum}
\usepackage{amssymb}
\usepackage{amsmath}
\usepackage{amsthm}
\usepackage{geometry}
\usepackage{multirow}
\usepackage{xr}

\usepackage[authoryear]{natbib}
\renewcommand{\theequation}{\thesection.\arabic{equation}}
\newcommand{\myref}[2]{\hyperref[#1]{#2}}
\numberwithin{equation}{section}

\usepackage{latexsym}
\usepackage{graphicx} % include figures
\usepackage{enumerate}
\usepackage{setspace}
\usepackage[toc,title,titletoc,header]{appendix}
\usepackage[bottom]{footmisc} % forces the footnotes to go at the bottom of page.
\usepackage[pdfborder={0 0 0}]{hyperref}
\hypersetup{
 colorlinks=true,linkcolor=magenta, citecolor=blue, filecolor=magenta, 
 linkbordercolor={0 1 1}, citebordercolor={1 0 0},
}

\usepackage{lineno}
\setpagewiselinenumbers
\newtheorem{theorem}{Theorem}[section]
\newtheorem{lemma}{Lemma}[section]

\newtheorem{assumption}{Assumption}[section]

\newtheorem{example}{Example}[section]

\newcounter{assumptionM}
\newcounter{assumptionA}
\def\theassumptionM{M.\arabic{assumptionM}}
\def\theassumptionA{A.\arabic{assumptionA}}
\begin{document}
	\relax
	\hypersetup{pageanchor=false}
	
	\hypersetup{pageanchor=true}

\author{
Federico A. Bugni\\
Department of Economics\\
Northwestern University\\
\url{federico.bugni@northwestern.edu}
\and
Mengsi Gao\\
Department of Economics\\
UC Berkeley\\
\url{mengsi.gao@berkeley.edu}
}

\bigskip
\title{Inference under Covariate-Adaptive Randomization with Imperfect Compliance
\thanks{We thank the Coeditor and four anonymous referees for comments and suggestions that have greatly improved the manuscript. We also thank Ivan Canay, Azeem Shaikh, Max Tabord-Meehan, and Diego Ubfal for helpful comments and discussion. This research was supported by the National Science Foundation Grant SES-1729280.}}

\maketitle

\vspace{-0.3in}
\thispagestyle{empty}

\begin{spacing}{1.3}
\begin{abstract}

This paper studies inference in a randomized controlled trial (RCT) with covariate-adaptive randomization (CAR) and imperfect compliance of a binary treatment. In this context, we study inference on the local average treatment effect (LATE), i.e., the average treatment effect conditional on individuals that always comply with the assigned treatment. As in \cite{bugni/canay/shaikh:2018,bugni/canay/shaikh:2019}, CAR refers to randomization schemes that first stratify according to baseline covariates and then assign treatment status so as to achieve ``balance'' within each stratum. In contrast to these papers, however, we allow participants of the RCT to endogenously decide to comply or not with the assigned treatment status.

We study the properties of an estimator of the LATE derived from a ``fully saturated'' instrumental variable (IV) linear regression, i.e., a linear regression of the outcome on all indicators for all strata and their interaction with the treatment decision, with the latter instrumented with the treatment assignment. We show that the proposed LATE estimator is asymptotically normal, and we characterize its asymptotic variance in terms of primitives of the problem. We provide consistent estimators of the standard errors and asymptotically exact hypothesis tests. In the special case when the target proportion of units assigned to each treatment does not vary across strata, we can also consider two other estimators of the LATE, including the one based on the ``strata fixed effects'' IV linear regression, i.e., a linear regression of the outcome on indicators for all strata and the treatment decision, with the latter instrumented with the treatment assignment. 

Our characterization of the asymptotic variance of the LATE estimators in terms of the primitives of the problem allows us to understand the influence of the parameters of the RCT. We use this to propose strategies to minimize their asymptotic variance in a hypothetical RCT based on data from a pilot study. We illustrate the practical relevance of these results using a simulation study and an empirical application based on \cite{dupas/karlan/robinson/ubfal:2018}.
\end{abstract}
\end{spacing}

\medskip
\noindent KEYWORDS: Covariate-adaptive randomization, stratified block randomization, treatment assignment, randomized controlled trial, strata fixed effects, saturated regression, imperfect compliance.

\noindent JEL classification codes: C12, C14

\thispagestyle{empty} 
\newpage

%\tableofcontents

\newpage
\hypersetup{pageanchor=true}
\setcounter{page}{1}

\section{Introduction}
This paper studies inference in a randomized controlled trial (RCT) with covariate-adaptive randomization (CAR) with a binary treatment. As in \cite{bugni/canay/shaikh:2018,bugni/canay/shaikh:2019}, CAR refers to randomization schemes that first stratify according to baseline covariates and then assign treatment status so as to achieve ``balance'' within each stratum. As these references explain, CAR is commonly used to assign treatment status in RCTs in all parts of the sciences.\footnote{See \cite{rosenberger/lachin:2016} for a textbook treatment focused on clinical trials and \cite{duflo/glennerster/kremer:2007} and \cite{bruhn/mckenzie:2008} for reviews focused on development economics.}

In contrast to \cite{bugni/canay/shaikh:2018,bugni/canay/shaikh:2019}, this paper allows for imperfect compliance. That is, the participants of the RCT can endogenously decide their treatment, denoted by $D$, which may or may not coincide with the assigned treatment status, denoted by $A$. This constitutes an empirically relevant contribution to this growing literature, as imperfect compliance is a common occurrence in many RCTs. For recent examples of RCTs that use CAR and have imperfect compliance, see \cite{angrist/lavy:2009}, \cite{attanasio/kugler/meghir:2011}, 
\cite{dupas/karlan/robinson/ubfal:2018}, \cite{mcintosh/alegria/ordonez/zenteno:2018}, 
\cite{somville/vandewalle:2018}, among many others.

Our goal is to study the effect of the treatment on an outcome of interest, denoted by $Y$. We consider the potential outcome framework with $Y=Y(1)D + Y(0)(1-D)$, where $Y(1)$ denotes the outcome with treatment and $Y(0)$ denotes the outcome without treatment. We also consider a potential decision framework with $D = D(1)A + D(0)(1-A)$, where $D(1)$ denotes the decision when assigned to treatment and $D(0)$ denotes the decision when not assigned to treatment. In the context of imperfect treatment compliance, a key causal parameter of interest is the so-called local average treatment effect (LATE) introduced in \cite{angrist/imbens:1994}, and given by
\begin{equation}
 \beta ~\equiv~E[~Y(1)-Y(0)~|~D(0)=0,~D(1)=1~].
 \label{eq:late}
\end{equation}
In words, the LATE is the average treatment effect for those individuals who decide to comply with their assigned treatment status, i.e., the ``compliers''.\footnote{The definition of LATE presumes the presence of compliers in the population of interest. This will be formalized in Assumption \ref{ass:1}. Under imperfect compliance, the literature also considers other parameters, such as the intention to treat, defined as $\text{ITT}=E[Y|A=1]-E[Y|A=0]$, or the average treatment effect on the treated, defined as $\text{TOT}=E[ Y( 1) -Y( 0) |D(1)=D(0)=1]$. In this paper, we prefer the LATE over these alternative parameters. First, it can be shown that $\text{ITT} = \text{LATE} \times P(D(0)=0,D(1)=1)$. Unlike the ITT, the LATE represents the average treatment effect for a subset of the population (i.e., the compliers), which makes it preferable. Second, it can be shown that $\text{TOT} = \text{LATE}$ if there are no ``always takers'', i.e., individuals who adopt the treatment regardless of the assignment. However, if always takers are present, the TOT is not identified under our assumptions.}
As in \cite{bugni/canay/shaikh:2019}, we consider inference on the LATE based on simple linear regressions. In Section \ref{sec:SAT}, we study the properties of an estimator of the LATE derived from a ``fully saturated'' instrumental variable (IV) linear regression, i.e., a linear regression of the outcome on all indicators for all strata and their interaction with the treatment decision, where the latter is instrumented with the treatment assignment. We show that its coefficients can be used to consistently estimate the LATE under very general conditions. We show that the proposed LATE estimator is asymptotically normal, and we characterize its asymptotic variance in terms of primitives of the problem. As expected, we show that the asymptotic variance is different from the one under perfect compliance derived in \cite{bugni/canay/shaikh:2018,bugni/canay/shaikh:2019}. We provide consistent estimators of these new standard errors and asymptotically exact hypothesis tests. In addition, we show that the results of the fully saturated regression can be used to estimate all of the primitive parameters of the problem.

In the special case when the target proportion of units being assigned to each of the treatments does not vary across strata, we also consider two other regression-based estimators of the LATE. Section \ref{sec:SFE} proposes an estimator of the LATE based on the ``strata fixed effects'' IV linear regression, i.e., a linear regression of the outcome on indicators for all strata and the treatment decision, where the latter is instrumented with the treatment assignment. In turn, Section \ref{sec:2STT} proposes an estimator of the LATE based on a ``two sample regression'' IV linear regression, i.e., a linear regression of the outcome on a constant and the treatment decision, where the latter is instrumented by the treatment assignment. We show that the LATE estimators produced by both of these regressions are asymptotically normal, and we characterize their asymptotic variances in terms of primitives of the problem. We show how to estimate the corresponding standard errors and generate asymptotically exact hypothesis tests by relying on results from the fully saturated regression.

Sections \ref{sec:SAT}, \ref{sec:SFE}, and \ref{sec:2STT} characterize the asymptotic variance of the three regression-based IV estimators of the LATE in terms of the primitive parameters of the problem. This allows us to understand how the parameters of the RCT with CAR affect the standard errors. In principle, this information can be used to propose strategies to minimize the asymptotic variance of the LATE estimator of a hypothetical RCT, possibly with the aid of data from a pilot RCT. We consider this RCT design problem in Section \ref{sec:optim}, and we establish several interesting results. First, we show that it is optimal to use a CAR method that imposes the highest possible level of ``balance'' on the treatment assignment within each stratum, such as stratified block randomization.\footnote{This finding extends the results obtained by \cite{bugni/canay/shaikh:2018,bugni/canay/shaikh:2019} to the case of imperfect compliance.} Within this class of CAR methods, our second result in Section \ref{sec:optim} establishes that the asymptotic variance of the estimators of the LATE cannot increase when the collection of strata becomes finer.\footnote{See \cite{bai:2022} for analogous results in the context of perfect compliance.} In addition, we show how to use the data from a pilot RCT to estimate the asymptotic variance that would result from using a finer set of strata in a hypothetical RCT. Our third and final result in Section \ref{sec:optim} provides an expression for the optimal treatment propensity in a hypothetical RCT in terms of its primitive parameters. To exploit this result in practice, we provide a consistent estimator of the optimal treatment propensity based on data from a pilot version of the RCT.

In recent work, \cite{ansel/hong/li:2018} also consider inference for the LATE in RCTs with a binary treatment and imperfect compliance. While most of their paper focuses on the case in which treatment assignment is done via simple random sampling, they consider inference on RCTs with CAR in Section 4. In contrast, our paper is entirely focused on RCTs with CAR. This allows us to tailor our assumptions to the problem under consideration, and it enables us to give more detailed formal arguments.\footnote{An example of this is the proof of Lemma \ref{lem:AsyDist}, where we modify the arguments in \citet[Lemma B.2]{bugni/canay/shaikh:2018} to allow for the presence of imperfect compliance.} Second, \cite{ansel/hong/li:2018} consider IV regressions without fully specifying the set of covariates in these regressions. Consequently, they derive the asymptotic variance of their LATE estimators in terms of high-level expressions. In contrast, we fully specify the covariates in our IV regressions according to the specification typically used by practitioners. This allows us to obtain explicit expressions for the asymptotic variance of our LATE estimators, detailing which of these depend on the underlying population and which are chosen by the researcher implementing the RCT. In this sense, our expressions reveal the underlying forces determining the asymptotic variance and enable researchers to choose the RCT parameters to improve the efficiency of their estimators. We consider the topic of RCT design in Section \ref{sec:optim}.

The remainder of the paper is organized as follows. In Section \ref{sec:setup}, we describe the setup of the inference problem and we specify our assumptions. The next three sections consider the problem of inference on the LATE based on different IV regression models. Section \ref{sec:SAT} considers the ``fully saturated'' IV linear regression, Section \ref{sec:SFE} considers the ``strata fixed effects'' IV linear regression, and Section \ref{sec:2STT} considers the ``two-sample regression'' IV linear regression. In each one of these sections, we propose a consistent estimator of the LATE, we characterize its asymptotic distribution, and we propose a consistent estimator of their standard errors and asymptotically valid hypothesis tests. In Section \ref{sec:optim}, we consider the problem of designing a hypothetical RCT with CAR based on data from a pilot RCT with CAR. In Section \ref{sec:MC}, we study the finite sample behavior of our hypothesis tests based via Monte Carlo simulations. Section \ref{sec:application} illustrates the practical relevance of our results by an empirical application based on the RCT in \cite{dupas/karlan/robinson/ubfal:2018}. Section \ref{sec:conclusions} provides concluding remarks. All proofs and several intermediate results are collected in the appendix.

\section{Setup and notation} \label{sec:setup}

We consider an RCT with $n$ participants. For each participant $i=1,\dots,n$, $Y_i \in \mathbb{R}$ denotes the observed outcome of interest, $Z_i \in \mathcal{Z}$ denotes a vector of observed baseline covariates, $A_i \in \{0,1\}$ indicates the treatment assignment, and $D_i \in \{0,1\}$ indicates the treatment decision. Relative to the setup in \cite{bugni/canay/shaikh:2018,bugni/canay/shaikh:2019}, we allow for imperfect compliance, i.e., for $D_i \neq A_i$. 

We consider potential outcome models for both outcomes and treatment decisions. For each participant $i=1,\dots,n$, we use $Y_i(D)$ to denote the potential outcome of participant $i$ if he/she makes treatment decision $D$, and we use $D_i(A)$ to denote the potential treatment decision of participant $i$ if he/she has assigned treatment $A$. These are related to their observed counterparts in the usual manner:
\begin{align}
 D_i &~=~ D_i(1) A_i + D_i(0) (1-A_i),\notag\\
 Y_i &~=~ Y_i(1) D_i + Y_i(0) (1-D_i) .\label{eq:outcome}
\end{align}

Following the usual classification in the LATE framework in \cite{angrist/imbens:1994}, each participant in the RCT can only be one of four types: complier, always taker, never taker, or a defier. An individual $i$ is said to be a complier if $\{D_i(0)=0,D_i(1)=1\}$, an always taker if $\{D_i(0)=D_i(1)=1\}$, a never taker if $\{D_i(0)=D_i(1)=0\}$, and a defier if $\{D_i(0)=1,D_i(1)=0\}$. As usual in the literature, we later impose that there are no defiers in our population of participants in order to identify the LATE. It is convenient to use $C$ to denote a complier, $AT$ to denote an always taker, $NT$ to denote a never taker, and $DEF$ to denote a defier. Our goal in this paper is to consistently estimate the LATE $ \beta \equiv E[Y(1)-Y(0)|C]$ and to test hypotheses about it. In particular, for a prespecified choice of $\beta_0 \in \mathbb{R}$, we are interested in the following hypothesis testing problem
\begin{equation}
 H_0:\beta = \beta_0 ~~~~\text{versus}~~~~ H_1:\beta \neq \beta_0
 \label{eq:HT}
\end{equation}
at a significance level $\alpha \in (0,1)$.

Following \cite{bugni/canay/shaikh:2018,bugni/canay/shaikh:2019}, we use $P_n$ to denote the distribution of the observed data
\begin{align*}
 X^{(n)} ~=~((Y_i,D_i,A_i,Z_i)~:i=1,\dots,n)
\end{align*}
and denote by $Q_n$ the distribution of the underlying random variables, given by
\begin{align*}
 W^{(n)} ~=~ ((Y_i(1),Y_i(0),D_i(0),D_i(1),Z_i)~:i=1,\dots,n).
\end{align*}
Note that $P_n$ is jointly determined by \eqref{eq:outcome}, $Q_n$, and the treatment assignment mechanism. We therefore state our assumptions below in terms of the restrictions on $Q_n$ and the treatment assignment mechanism. In fact, we will not make reference to $P_n$ for the remainder of the paper, and all the operations are understood to be under $Q_n$ and the treatment assignment mechanism.

Strata are constructed from the observed, baseline covariates $Z_i$ using a prespecified function $S:\mathcal{Z} \to \mathcal{S}$, where $\mathcal{S}$ is a finite set. For each participant $i=1,\dots,n$, let $S_i \equiv S(Z_i)$ and let $S^{(n)} = (S_i:i=1,\dots,n)$. By definition, we note that $S^{(n)}$ is completely determined by the covariates in $W^{(n)}$. 

We begin by describing our assumptions on the underlying data generating process (DGP) of $W^{(n)}$.

\begin{assumption}\label{ass:1}
	$W^{(n)}$ is an i.i.d.\ sample that satisfies
	\begin{enumerate}[{(a)}]
	 \item $ E[ Y_{i}( d) ^{2}]<\infty $ for all $d\in \{0,1\}$,
	 \item $p(s) ~\equiv~ P(S_i=s)>0$ for all $s \in \mathcal{S}$,
	 \item $ P( D_{i}( 0) =1,D_{i}(1) =0) =0$ or, equivalently, $ P( D_{i}(1)\geq D_{i}(0)) =1$.
	 \item $\pi _{D( 1) }( s) - \pi _{D( 0) }( s)>0$ for all $s \in \mathcal{S}$, where $\pi_{D(a)}(s) =P(D_i(a)=1|S_i=s) $ for $(a,s) \in \{0,1\} \times \mathcal{S}$.
	\end{enumerate}
\end{assumption}

Assumption \ref{ass:1} requires the underlying data distribution to be i.i.d., i.e., $Q_{n} = Q^n$, where $Q$ denotes the common marginal distribution of $(Y_i(1),Y_i(0),D_i(0),D_i(1),Z_i)$. In addition, the assumption imposes several requirements on $Q$. Assumption \ref{ass:1}(a) demands that the potential outcomes have finite second moments, which is important to develop our asymptotic analysis. Assumption \ref{ass:1}(b) requires all strata to be relevant. Assumption \ref{ass:1}(c) corresponds to \citet[Condition 2]{angrist/imbens:1994}, and imposes the standard ``no defiers'' or ``monotonicity'' condition that is essential to identify the LATE. In other words, this condition implies that there are no participants who will decide to defy the treatment assignment, i.e., decide to both adopt the treatment when assigned to the control and decide to adopt the control when assigned to treatment. To interpret Assumption \ref{ass:1}(d), we note that Lemma \ref{lem:P_type_representation} provides the following expression of the probability of each type of participant conditional on the stratum $s \in \mathcal{S}$:
\begin{align}
P( AT|S=s) &~=~\pi _{D( 0) }( s), \notag\\
P( NT|S=s) &~=~1-\pi _{D( 1) }( s), \notag \\
P( C|S=s) &~=~\pi _{D( 1) }( s) -\pi _{D( 0) }( s). \label{eq:pi_defns} 
\end{align}
By \eqref{eq:pi_defns}, Assumption \ref{ass:1}(d) imposes that every stratum has a non-trivial amount of participants who will decide to comply with the assigned treatment status. This corresponds to a strata-specific version of the so-called ``relevance condition'' imposed in \citet[Condition 1]{angrist/imbens:1994}.

Next, we describe our assumptions on the treatment assignment mechanism. As explained earlier, we focus our analysis on CAR, i.e., on randomization schemes that first stratify according to baseline covariates and then assign treatment status to as to achieve ``balance'' within each stratum. To describe our assumption more formally, we require some further notation. Let $A^{(n)} = (A_i:i=1,\dots,n)$ denote the vector of treatment assignments. For any $s \in \mathcal{S}$, let $\pi_A(s) \in (0,1)$ denote the ``target'' proportion of participants to assign to treatment in stratum $s$, determined by the researcher implementing the RCT. Also, let
\begin{equation*}
 n_{A}( s) ~\equiv~ \sum_{i=1}^{n}1[ A_{i}=1,S_{i}=s ]
\end{equation*}
denote the number of participants assigned to treatment in stratum $s$, and let 
\begin{equation*}
 n( s) ~\equiv~ \sum_{i=1}^{n}1[ S_{i}=s ]
\end{equation*}
denote the number of participants in stratum $s$. With this notation in place, we now specify our assumption regarding the treatment assignment mechanism.

\begin{assumption}\label{ass:2}
	The treatment assignment mechanism satisfies
\begin{enumerate}[{(a)}]
\item $W^{(n)}\perp A^{(n)}~|~S^{(n)}$,
\item $n_{A}( s) /n( s) ~\overset{p}{\to}~ \pi _{A}( s) \in ( 0,1) $ for all $s\in \mathcal{S}$,
\end{enumerate}
\end{assumption}

In principle, it is possible that $n(s)=0$ for some $s\in \mathcal{S}$, and so the term $n_{A}(s)/n(s)$ may not be properly defined. In any case, we note that Assumption \ref{ass:1} implies that $n(s)/n \overset{p}{\to} p(s)>0$ for all $s\in \mathcal{S}$, and so the zero denominator issue is only a small sample problem and does not affect our asymptotic analysis.\footnote{To avoid defining objects that have a zero denominator, we abuse the notation and redefine the ratio $a/b$ as an arbitrary number (say, zero) whenever $b=0$. We adopt this rule in Assumption \ref{ass:2} and throughout the rest of this paper. For the reasons explained earlier, this alternative ratio definition only serves the purpose of avoiding a zero denominator in small samples and does not affect our asymptotic analysis.}

Assumption \ref{ass:2}(a) essentially requires that the treatment assignment vector $A^{(n)}$ is a function of the strata vector $S^{(n)}$ and an exogenous randomization device. Assumption \ref{ass:2}(b) imposes that the fraction of units assigned to treatment in the stratum $s$ converges in probability to a target proportion $\pi_A(s)$ as the sample size diverges. As we show in Section \ref{sec:SAT}, Assumption \ref{ass:2} imposes sufficient structure of the CAR mechanism to analyze the asymptotic distribution of the LATE estimator in the ``fully saturated'' (SAT) IV regression. However, as we show in Sections \ref{sec:SFE} and \ref{sec:2STT}, Assumption \ref{ass:2} will not be enough to guarantee the consistency of the LATE estimator in the ``strata fixed effects'' (SFE) and ``two sample'' (2S) IV regressions. To analyse the asymptotic properties of these estimators, we replace Assumption \ref{ass:2} with the following condition, which mildly strengthens it.

\begin{assumption}\label{ass:3}
The treatment assignment mechanism satisfies
\begin{enumerate}[{(a)}]
\item $W^{(n)}\perp A^{(n)}~|~S^{(n)}$,
\item $\{( \sqrt{n}( n_{A}( s) /n( s) -\pi _{A}( s) ) : s\in \mathcal{S} )|S^{(n)}\}~\overset{d}{\to}~N( \mathbf{0},\Sigma _{A}) $ w.p.a.1, and, for some $\tau(s) \in [0,1]$, $$\Sigma _{A}~\equiv~ diag( (\tau ( s) \pi _{A}( s)( 1-\pi _{A}( s) ) /p( s):s\in \mathcal{S})) .$$
\item $\pi _{A}( s)~=~\pi _{A} \in (0,1)$ for all $s\in \mathcal{S}$.
\end{enumerate}
\end{assumption}

Note that Assumption \ref{ass:2}(a) and Assumption \ref{ass:3}(a) coincide. Assumption \ref{ass:3}(b) strengthens the convergence Assumption \ref{ass:2}(b), and it requires the fraction of units assigned to the treatment in the stratum $s$ to be asymptotically normal, conditional on the vector of strata $S^{(n)}$. For each stratum $s \in \mathcal{S}$, the parameter $\tau(s) \in [0,1]$ determines the amount of dispersion that the CAR mechanism allows on the fraction of units assigned to the treatment in that stratum. A lower value of $\tau(s)$ implies that the CAR mechanism imposes a higher degree of ``balance'' or ``control'' of the treatment assignment proportion relative to its desired target value.\footnote{Assumption \ref{ass:3}(b) is slightly weaker than \citet[Assumption 4.1(c)]{bugni/canay/shaikh:2019} in that we require the condition holding w.p.a.1 instead of a.s. We establish that the w.p.a.1 version of the assumption is sufficient to establish all of our formal results.} Finally, Assumption \ref{ass:3}(c) imposes that the target value for the treatment assignment does not vary by strata. As we show in Sections \ref{sec:SFE} and \ref{sec:2STT}, this condition is key to the consistency of the LATE estimators produced by the SFE and 2SR regressions.

As explained by \cite{bugni/canay/shaikh:2018,bugni/canay/shaikh:2019} and \citet[Sections 3.10 and 3.11]{rosenberger/lachin:2016}, Assumptions \ref{ass:2} and \ref{ass:3} are satisfied by a wide array of CAR schemes. We briefly consider three popular schemes that can easily be seen to satisfy this assumption.

\begin{example}[Simple Random Sampling (SRS)]\label{ex:SRS}
This refers to a treatment assignment mechanism in which $A^{(n)}$ satisfy
\begin{equation}
 P(A^{(n)}=(a_i: i=1,\dots,n)|S^{(n)}=(s_i: i=1,\dots,n),W^{(n)})~=~\prod_{i=1}^{n} \pi_A(s_i)^{a_i} (1-\pi_A(s_i))^{1-a_i}.
 \label{eq:SRS}
\end{equation}
In other words, SRS assigns each participant in the stratum $s$ to treatment with probability $\pi_A(s)$ and to control with probability $(1-\pi_A(s))$, independent of anything else in sample. 

Note that Assumption \ref{ass:3}(a) follows immediately from \eqref{eq:SRS}. Also, by combining \eqref{eq:SRS}, Assumption \ref{ass:1}, and the Central Limit Theorem (CLT), it is possible to show Assumption \ref{ass:3}(b) holds with $\tau(s) =1$ for all $s \in \mathcal{S}$. In terms of the range of values of $\tau(s)$ allowed by Assumption \ref{ass:3}(b), SRS imposes the least amount of ``balance'' of the treatment assignment proportion relative to its desired target value. Finally, Assumption \ref{ass:3}(c) can be satisfied by setting $\pi_A(s)$ to be constant across strata.
\end{example}

\begin{example}[Stratified Block Randomization (SBR)]\label{ex:SBR}
This is sometimes also referred to as block randomization or permuted blocks within strata. In SBR, the assignments across strata are independent, and independently of the rest of the information in the sample. Within every stratum $s$, SBR assigns exactly $\lfloor n(s)\pi_A(s)\rfloor$ of the $n(s)$ participants in stratum $s$ to treatment and the remaining $n(s) -\lfloor n(s)\pi_A(s)\rfloor $ to control, where all possible $$\binom{n(s)}{\lfloor n(s)\pi_A(s)\rfloor}$$ assignments are equally likely.

As explained by \cite{bugni/canay/shaikh:2018,bugni/canay/shaikh:2019}, this mechanism satisfies Assumptions \ref{ass:3}(a)-(b) with $\tau(s)=0$ for all $s\in \mathcal{S}$. In terms of the range of values of $\tau(s)$ allowed by Assumption \ref{ass:3}(b), SBR imposes the most amount of ``balance'' of the treatment assignment proportion relative to its desired target value. Finally, Assumption \ref{ass:3}(c) can be satisfied by setting $\pi_A(s)$ to be constant across strata.
\end{example}

\begin{example}[Minimization methods] \label{ex:minimization}
Another popular class of treatment assignment mechanisms are the so-called minimization methods. These were originally proposed by \cite{pocock/simon:1975}, and extended and further investigated by \cite{hu/hu:2012}. In these methods, the treatment is assigned recursively for $k=1,\dots,n$ according to
\begin{equation}
P(A_{k}=1 \mid S^{(k)}, A^{(k-1)})~=~\left\{\begin{array}{ll}
\pi_A & \text { if }~ \operatorname{Imb}_{k}(S^{(k)}, A^{(k-1)})=0 \\
\lambda & \text { if } ~\operatorname{Imb}_{k}(S^{(k)}, A^{(k-1)})<0 \\
1-\lambda & \text { if }~ \operatorname{Imb}_{k}(S^{(k)}, A^{(k-1)})>0
\end{array}\right.,    
\label{eq:minimization}
\end{equation}
where $A^{(0)}= \emptyset$, $\lambda \in (\pi_A,1]$, and $\operatorname{Imb}_{k}(S^{(k)}, A^{(k-1)})$ is a weighted measure of imbalance relative to the desired treatment assignment distribution. 

We follow \citet[Page 1797]{hu/hu:2012} and specify $\operatorname{Imb}_{k}(S^{(k)}, A^{(k-1)})$ as the weighted sum of three sources of imbalance: (i) overall imbalance, with associated weight $w_o$, (ii) imbalance within the marginal distribution of the $L$ relevant covariates, each with associated weights $(w_{m,\ell} :\ell =1,\dots,L)$, and (iii) within strata imbalance, with associated weight $w_{s}$. As explained in \citet[Remark 2.1]{hu/hu:2012}, their minimization method includes \cite{pocock/simon:1975} as a special case when we set $w_o=w_{s}=0$. \citet[Theorem 3.2]{hu/hu:2012} provide conditions on $(w_o, w_s, (w_{m,\ell} :\ell =1,\dots,L)')$ under which the strata-specific imbalance is a positive recurrent Markov chain, implying that
they converge to zero at a fast rate. As a corollary, under these conditions, Assumption \ref{ass:3} is satisfied with $\tau(s)=0$ for all $s\in \mathcal{S}$.
\end{example}

\section{``Fully saturated'' (SAT) IV regression}\label{sec:SAT}

In this section, we study the asymptotic properties of an IV estimator of the LATE based on a linear regression model of the outcome of interest on the full set of indicators for all strata and their interaction with the treatment decision, where the latter is instrumented with the treatment assignment. Following the nomenclature in \cite{bugni/canay/shaikh:2019}, we refer to this as the SAT IV regression. We show in this section that this SAT IV regression can consistently estimate the LATE for each stratum, and we derive their joint asymptotic distribution. These estimators can then be combined to produce a consistent estimator of the LATE. We show that this estimator is asymptotically normal and we characterize its asymptotic variance in terms of the primitives parameters of the RCT. We also show that the coefficients and residuals of the SAT IV regression can be used to consistently estimate these primitive parameters, which allows us to propose a consistent estimator of the standard errors of the LATE estimator. All of this allows us to propose hypothesis tests for the LATE that are asymptotically exact, i.e., their limiting rejection probability under the null hypothesis is equal to the nominal level.

In terms of our notation, the SAT IV regression is the result of regressing $Y_{i}$ on $(1[S_{i}=s]:s\in \mathcal{S})$ and $(D_i 1[S_{i}=s]:s\in \mathcal{S})$. Since the treatment decision $D_{i}$ is endogenously decided by the RCT participant, we instrument it with the exogenous treatment assignment $A_i$. To define these IV estimators precisely, set
\begin{align*}
 \mathbf{Y}_{n}~&\equiv~( Y_{i}:i=1,\dots, n),\\
 \mathbf{X}_{n}^{\text{sat}}~&\equiv~( ( (1[S_{i}=s]:s\in \mathcal{S})',(D_{i} 1[S_{i}=s]:s\in \mathcal{S})') :i=1,\dots, n)',\\
 \mathbf{Z}_{n}^{\text{sat}}~&\equiv~( ( (1[S_{i}=s]:s\in \mathcal{S})',(A_{i} 1[S_{i}=s]:s\in \mathcal{S})') :i=1,\dots, n)'.
\end{align*}
The IV estimators of the coefficients in SAT regression are
\begin{align}
 ( (\hat{\gamma}_{\text{sat}}(s):s\in \mathcal{S})',(\hat{\beta}_{\text{sat}}(s):s\in \mathcal{S})')' ~\equiv~( {\mathbf{Z}_{n}^{\text{sat}}}'\mathbf{X}_{n}^{\text{sat}}) ^{-1}( {\mathbf{Z}_{n}^{\text{sat}}}'\mathbf{Y}
_{n}),
\label{eq:IV_SAT_estimators}
\end{align}
where $\hat{\gamma}_{\text{sat}}(s)$ corresponds to the IV estimator of the coefficient on $1[S_{i}=s]$ and $\hat{\beta}_{\text{sat}}(s)$ corresponds to the IV estimator of the coefficient on $D_i 1[S_{i}=s]$.\footnote{In principle, it is possible that ${\mathbf{Z}_{n}^{\text{sat}}}'\mathbf{X}_{n}^{\text{sat}}$ is singular. In any case, we show in Lemma \ref{lem:MatrixSAT} that ${\mathbf{Z}_{n}^{\text{sat}}}'\mathbf{X}_{n}^{\text{sat}}/n$ converges in probability to a non-singular matrix. Thus, the singularity is only a small sample problem and does not affect our asymptotic analysis. We can avoid this problem by using $M^{-1}$ to denote any generalized inverse of $M$.}

Under Assumptions \ref{ass:1} and \ref{ass:2}, Theorem \ref{thm:plim_SAT} in the appendix shows that, for each stratum $s \in \mathcal{S}$,
\begin{align}
\hat{\gamma}_{\mathrm{sat}}(s) &~\overset{p}{\to }~\gamma ( s)~\equiv~ \left[ 
\begin{array}{c}
\pi _{D( 1) }( s) E[ Y( 0) |C,S=s] -\pi _{D( 0) }( s) E[ Y( 1) |C,S=s ] + \\
\pi _{D( 0) }( s) E[ Y( 1) |AT,S=s ] +( 1-\pi _{D( 1) }( s) ) E[ Y( 0) |NT,S=s]
\end{array}
\right]\notag\\
\hat{\beta}_{\mathrm{sat}}(s) &~\overset{p}{\to }~\beta ( s) ~\equiv~E[ Y( 1)-Y( 0) |C,S=s].
\label{eq:SAT_limit}
\end{align}
This last equation in \eqref{eq:SAT_limit} reveals that $\hat{\beta}_{\mathrm{sat}}$ is a consistent estimator of the vector of strata-specific LATE. To define a consistent estimator of the LATE based on these, all we then need is a consistent estimator of the probability a participant belongs to each strata conditional on being a complier, i.e., $P( S=s|C)$ for $s \in \mathcal{S}$. To this end, for every $s \in \mathcal{S}$, let 
\begin{align}
\hat{P}(S=s,C)& ~\equiv~ \frac{n(s)}{n}\left(\frac{n_{AD}(s)}{n_{A}(s)}-\frac{ n_{D}(s)-n_{AD}(s)}{n(s)-n_{A}(s)}\right) \notag \\
\hat{P}(C)& ~\equiv~ \sum_{s\in \mathcal{S}}\frac{n(s)}{n}\left(\frac{n_{AD}(s)}{ n_{A}(s)}-\frac{n_{D}(s)-n_{AD}(s)}{n(s)-n_{A}(s)}\right)\notag \\
\hat{P}( S=s|C)& ~\equiv~\frac{\hat{P}(S=s,C)}{\hat{P}(C)}~=~\frac{\frac{n(s)}{n}( \frac{n_{AD}(s)}{ n_{A}(s)}-\frac{n_{D}(s)-n_{AD}(s)}{n(s)-n_{A}(s)})}{\sum_{\tilde{s}\in \mathcal{S}}\frac{n(\tilde{s})}{n}( \frac{n_{AD}(\tilde{s})}{n_{A}(\tilde{s})}-\frac{n_{D}(\tilde{s})-n_{AD}(\tilde{s})}{n(\tilde{s})-n_{A}(\tilde{s})})},\label{eq:P_C}
\end{align}
where
\begin{align*}
n_{AD}( s) &~\equiv~ \sum_{i=1}^{n}1[ A_{i}=1,D_{i}=1,S_{i}=s ] \\
n_{D}( s) &~\equiv~ \sum_{i=1}^{n}1[ D_{i}=1,S_{i}=s] .
\end{align*}
Under Assumptions \ref{ass:1} and \ref{ass:2}, Theorem \ref{thm:plim_SAT} also shows that $\hat{P}( S=s|C)$ is a consistent estimator of $P( S=s|C)$ for every $s \in \mathcal{S}$. It is then natural to propose the following estimator of the LATE:
\begin{align}
 \hat{\beta}_{\text{sat}} ~\equiv~\sum_{s\in \mathcal{S}} \hat{P}(S=s|C) \hat{\beta}_{\text{sat}}(s).
 \label{eq:beta_SAT}
\end{align}
It follows from our previous discussion, the continuous mapping theorem, and the law of iterated expectations that $\hat{\beta}_{\text{sat}}$ is a consistent estimator of the LATE. The following theorem confirms this result, and also characterizes its asymptotic distribution in terms of primitive parameters of the RCT.

%%%%%%%% DIVIDER %%%%%%%%%%%%
\begin{theorem}[SAT main result]\label{thm:AsyDist_SAT} 
Suppose that Assumptions \ref{ass:1} and \ref{ass:2} hold. Then,
\begin{align*}
 \sqrt{n}(\hat{\beta}_{\mathrm{sat}}-\beta )~~\overset{d}{\to }~~N(0,V_{\mathrm{sat}}),
\end{align*}
where $\beta \equiv E[Y(1)-Y(0)|C]$ and $V_{\mathrm{sat}} \equiv V_{{Y} ,1}^{\mathrm{sat}}+V_{{Y},0}^{\mathrm{sat}}+V_{D,1}^{\mathrm{sat}}+V_{D,0}^{\mathrm{sat}}+V_{H}^{\mathrm{sat}}$ with
\begin{align}
V_{{Y} ,1}^{\mathrm{sat}} &\equiv\tfrac{1}{P(C)^{2}}\sum_{s\in \mathcal{S}}\frac{ p( s) }{\pi _{A}(s)}\left[
\begin{array}{c}
 V[ Y( 1) |AT,S=s ]\pi _{D( 0) }( s) +V[ Y( 0) |NT,S=s]( 1-\pi _{D( 1) }( s) ) \\
+ V[ Y( 1) |C,S=s] ( \pi _{D( 1) }( s) -\pi _{D( 0) }( s) )\\
+(E[Y(1)|C,S=s]-E[Y(1)|AT,S=s])^{2}\times \\
\pi _{D( 0) }( s) ( \pi _{D( 1) }( s) -\pi _{D( 0) }( s) ) /\pi _{D( 1) }( s)
\end{array}
\right] \notag \\
V_{{Y} ,0}^{\mathrm{sat}} &\equiv\tfrac{1}{P(C)^{2}}\sum_{s\in \mathcal{S}}\frac{ p( s) }{(1-\pi _{A}(s))}\left[
\begin{array}{c}
 V[ Y( 1) |AT,S=s ] \pi _{D( 0) }( s)+ V[ Y( 0) |NT,S=s]( 1-\pi _{D( 1) }( s) ) \\
+ V[ Y( 0) |C,S=s]( \pi _{D( 1) }( s) -\pi _{D( 0) }( s) ) \\
+(E[Y(0)|C,S=s]-E[Y(0)|NT,S=s])^{2}\times \\
( 1-\pi _{D( 1) }( s) ) ( \pi _{D( 1) }( s) -\pi _{D( 0) }( s) ) /( 1-\pi _{D( 0) }( s) )
\end{array}
\right] \notag \\
V_{D ,1}^{\mathrm{sat}} &\equiv\tfrac{1}{P(C)^{2}}\sum_{s\in \mathcal{S}}\frac{p( s) ( 1-\pi _{D( 1) }( s) ) }{\pi _{A}( s) \pi _{D(1)}(s)}\left[
\begin{array}{c}
-\pi _{D(0)}(s)(E[Y(1)|C,S=s]-E[Y(1)|AT,S=s]) \\ 
+\pi _{D(1)}(s)(E[Y(0)|C,S=s]-E[Y(0)|NT,S=s]) \\ 
+\pi _{D(1)}(s)( E[Y(1)-Y(0)|C,S=s]-\beta ) 
\end{array}
\right] ^{2} \notag\\
V_{D ,0}^{\mathrm{sat}} &\equiv\tfrac{1}{P(C)^{2}}\sum_{s\in \mathcal{S}}\frac{p( s) \pi _{D( 0) }( s) }{ (1-\pi _{A}( s))( 1-\pi _{D(0)}(s)) }\left[
\begin{array}{c}
-( 1-\pi _{D(0)}(s)) ( E[Y(1)|C,S=s]-E[Y(1)|AT,S=s]) \\
+( 1-\pi _{D(1)}(s)) ( E[Y(0)|C,S=s]-E[Y(0)|NT,S=s]) \\
+( 1-\pi _{D(0)}(s)) ( E[Y(1)-Y(0)|C,S=s]-\beta )
\end{array}
\right] ^{2}\notag \\
V_{H}^{\mathrm{sat}} &\equiv \tfrac{1}{P(C)^{2}}\sum_{s\in \mathcal{S}}p(s)(\pi _{D(1)}(s)-\pi _{D(0)}(s))^{2}(E[Y(1)-Y(0)|C,S=s]-\beta )^{2}\notag\\
P(C) &\equiv \sum_{s\in \mathcal{S}}p(s)(\pi _{D(1)}(s)-\pi _{D(0)}(s))>0.\label{eq:AsyDist_SAT_Avar}
\end{align}
\end{theorem}
%%%%%%%% DIVIDER %%%%%%%%%%%%

Several remarks about Theorem \ref{thm:AsyDist_SAT} are in order. First, we note that $\hat{\beta}_{\mathrm{sat}}$ is related to the IV estimator $\hat{\beta}_3$ considered by \cite{ansel/hong/li:2018}. In fact, these two estimators coincide if one specifies their covariates as a full vector of strata dummies. By specifying the set of covariates in our regression, we are able to obtain a closed-form expression of the asymptotic variance of $\hat{\beta}_{\mathrm{sat}}$ in terms of the primitive parameters of the RCT. This will become useful in Section \ref{sec:optim}, where we consider the problem of choosing the parameters of the RCT to improve efficiency.

Second, we note that \cite{bugni/canay/shaikh:2019} derive the asymptotic distribution of $\hat{\beta}_{\mathrm{sat}}$ under perfect compliance (i.e., $\pi _{D(1)}(s)=1$, $\pi _{D(0)}(s)=0$). This means that we can understand the consequences of imperfect compliance by comparing Theorem \ref{thm:AsyDist_SAT} and \citet[Theorem 3.1]{bugni/canay/shaikh:2019}. First and foremost, we note that imperfect compliance means that the probability limit of $\hat{\beta}_{\mathrm{sat}}$ is no longer the ATE, but rather the LATE. Second, we note that imperfect compliance introduces significant changes to the asymptotic variance of $\hat{\beta}_{\mathrm{sat}}$. Imperfect compliance not only changes the expressions of $V_{{Y}}^{\mathrm{sat}} = V_{{Y} ,1}^{\mathrm{sat}}+V_{{Y},0}^{\mathrm{sat}}$ and $V_{H}^{\mathrm{sat}}$, but it also adds two new terms, $V_{D,1}^{\mathrm{sat}}$ and $V_{D,0}^{\mathrm{sat}}$. All of this implies that the consistent estimator of $V_{\mathrm{sat}}$ proposed in \citet[Theorem 3.3]{bugni/canay/shaikh:2019} no longer applies, and a new one is required. We do this in Theorem \ref{thm:SAT_se}.

Third, it is notable that Assumption \ref{ass:3} is not required to derive Theorem \ref{thm:AsyDist_SAT}. In other words, the details of the CAR mechanism are not relevant to the asymptotic distribution of $\hat{\beta}_{\mathrm{sat}}$. This was pointed out in the case of perfect compliance by \cite{bugni/canay/shaikh:2019}, and Theorem \ref{thm:AsyDist_SAT} reveals that it also extends to the present setup.
 
Fourth, we note that Theorem \ref{thm:AsyDist_SAT} allows for $\pi _{D( 0) }( s)=P(AT|S=s)=0$ or $1-\pi _{D( 1) }( s)=P(NT|S=s)=0$, but this requires a mild abuse of notation. If $\pi _{D( 0) }( s)=P(AT|S=s)=0$, the mean and the variance of $\{Y(1)|AT,S=s\}$ are not properly defined, but we can set
\begin{equation*}
 V[ Y( 1) |AT,S=s ]\pi _{D( 0) }( s)~=~0~~~\text{and}~~~ E[ Y( 1) |AT,S=s ]\pi _{D( 0) }( s)~=~0.
\end{equation*}
Similarly, if $1-\pi _{D( 1) }( s)=P(NT|S=s)=0$, the mean and the variance of $\{Y(0)|NT,S=s\}$ are not properly defined, but we can set
\begin{equation*}
 V[ Y(0)|NT,S=s ](1-\pi _{D( 1) }( s))~=~0~~~\text{and}~~~E[ Y(0)|NT,S=s ](1-\pi _{D( 1) }( s))~=~0.
\end{equation*}
In particular, in the special case of perfect compliance (i.e., $\pi _{D(1)}(s)=1$, $\pi _{D(0)}(s)=0$), Theorem \ref{thm:AsyDist_SAT} then holds with 
\begin{align*}
\beta &~=~ E[Y(1)-Y(0)]\\
V_{\mathrm{sat}} &~=~ \sum_{s\in \mathcal{S}}{ p( s) }\left(\frac{V[ Y( 1) |S=s]}{\pi _{A}(s)} + \frac{V[ Y( 0) |S=s]}{(1-\pi _{A}(s))}+(E[Y(1)-Y(0)|S=s]-\beta )^{2}\right),
\end{align*}
which can be shown to coincide with the corresponding result in \citet[Section 5]{bugni/canay/shaikh:2019}.
 
As promised earlier, the next result provides a consistent estimator of $V_{\mathrm{sat}}$.

%%%%%%%% DIVIDER %%%%%%%%%%%%
\begin{theorem}[Estimator of SAT asy.\ variance]\label{thm:SAT_se} 
Suppose that Assumptions \ref{ass:1} and \ref{ass:2} hold. Define the following estimators:
\begin{align}
\hat{V}_{1}^{\mathrm{sat}} & ~\equiv~ \frac{1}{\hat{P}( C) ^{2}} \sum_{s\in \mathcal{S}}
\left( \frac{n(s)}{n_{A}(s)}\right) ^{2} \times \notag\\
&~~~~\left[
\begin{array}{l}
\frac{1}{n}\sum_{i=1}^{n} 1[D_{i}=1,A_{i}=1,S_{i}=s][ \hat{u} _{i}+( 1-\frac{n_{AD}( s) }{n_{A}( s) }) ( \hat{\beta}_{\mathrm{sat}}( s) -\hat{\beta}_{\mathrm{sat}}) ] ^{2} +\\
\frac{1}{n}\sum_{i=1}^{n} 1[D_{i}=0,A_{i}=1,S_{i}=s][ \hat{u}_{i}- \frac{n_{AD}( s) }{n_{A}( s) }( \hat{\beta}_{\mathrm{sat}}( s) -\hat{\beta}_{\mathrm{sat}}) ] ^{2}
\end{array}
\right] \notag\\
\hat{V}_{0}^{\mathrm{sat}} & ~\equiv~ \frac{1}{\hat{P}( C) ^{2}} \sum_{s\in \mathcal{S}}
\left( \frac{n(s)}{n(s)-n_{A}(s)}\right) ^{2}\times \notag\\
&~~~~\left[
\begin{array}{l}
\frac{1}{n}\sum_{i=1}^{n} 1[D_{i}=1,A_{i}=0,S_{i}=s][ \hat{u} _{i}+( 1-\frac{n_{D}( s) -n_{AD}( s) }{n( s) -n_{A}( s) }) ( \hat{\beta}_{\mathrm{sat}}( s) - \hat{\beta}_{\mathrm{sat}}) ] ^{2} +\\
\frac{1}{n}\sum_{i=1}^{n} 1[D_{i}=0,A_{i}=0,S_{i}=s][ \hat{u}_{i}- \frac{n_{D}( s) -n_{AD}( s) }{n( s) -n_{A}( s) }( \hat{\beta}_{\mathrm{sat}}( s) -\hat{\beta}_{\mathrm{sat}}) ]^{2}
\end{array}
\right] \notag\\
\hat{V}_{H}^{\mathrm{sat}} & ~\equiv~ \frac{1}{\hat{P}( C) ^{2}}\sum_{s\in \mathcal{S}}\frac{n(s)}{n}\left( \frac{n_{AD}(s)}{n_{A}(s)}-\frac{ n_{D}(s)-n_{AD}(s)}{n( s) -n_{A}(s)}\right) ^{2}(\hat{\beta}_{\mathrm{sat}}(s)- \hat{\beta}_{\mathrm{sat}})^{2},\label{eq:V_defns}
\end{align}
where $(\hat{u}_{i})_{i=1}^{n}$ are the SAT IV-regression residuals, given by 
\begin{equation}
 \hat{u}_{i}~\equiv~ Y_{i}-\sum_{s\in \mathcal{S}}1[S_{i}=s]\hat{ \gamma}_{\mathrm{sat}}(s)-\sum_{s\in \mathcal{S}}D_{i}1[S_{i}=s]\hat{\beta}_{\mathrm{sat}}(s),
 \label{eq:resid_SAT}
\end{equation}
and $(\hat{\gamma}_{\mathrm{sat}}(s),\hat{\beta}_{\mathrm{sat}}(s))$, $\hat{\beta}_{\mathrm{sat}}$, and $\hat{P}(C)$ are as in \eqref{eq:IV_SAT_estimators}, \eqref{eq:beta_SAT}, and \eqref{eq:P_C}, respectively. Then,
\begin{align}
 \hat{V}_{\mathrm{sat}} ~\equiv~\hat{V}_{1}^{\mathrm{sat}} + \hat{V}_{0}^{\mathrm{sat}} + \hat{V}_{H}^{\mathrm{sat}} ~\overset{p}{\to}~ V_{\mathrm{sat}}. 
\end{align}
\end{theorem}

We can propose hypothesis tests for the LATE by combining Theorems \ref{thm:AsyDist_SAT} and \ref{thm:SAT_se}. For completeness, this is recorded in the next result.

\begin{theorem}[SAT test]\label{thm:SAT_test}
Suppose that Assumptions \ref{ass:1} and \ref{ass:2} hold, and that $V_{\mathrm{sat}}>0$. For the problem of testing \eqref{eq:HT} at level $\alpha \in (0,1)$, Consider the following hypothesis testing procedure
\begin{align*}
\phi_{n}^{\mathrm{sat}}(X^{(n)})~ \equiv~1[~|{\sqrt{n}(\hat{\beta}_{\mathrm{sat}}-\beta_0 )}|~>~\sqrt{\hat{V}_{\mathrm{sat}}}z_{1-\alpha/2}~],
\end{align*}
where $z_{1-\alpha/2}$ is the $(1-\alpha/2)$-quantile of $N(0,1)$. Then,
\begin{align*}
\lim_{n\to \infty}E[\phi_{n}^{\mathrm{sat}}(X^{(n)})]~ =~\alpha
\end{align*}
whenever $H_0$ in \eqref{eq:HT} holds, i.e., $\beta = \beta_0$.
\end{theorem}

\section{``Strata fixed effects'' (SFE) IV regression}\label{sec:SFE}

In this section, we consider the asymptotic properties of an IV estimator of the LATE based on a linear regression model of the outcome of interest with a full set of indicators for all strata and the treatment decision, where the latter is instrumented with the treatment assignment. Following the nomenclature in \cite{bugni/canay/shaikh:2018,bugni/canay/shaikh:2019}, we refer to this as the SFE IV regression. Under certain conditions, we show that this SFE IV regression consistently estimates the LATE. We show that this estimator is asymptotically normal and we characterize its asymptotic variance in terms of the primitives parameters of the RCT. We also propose a consistent estimator of this asymptotic variance by using the results of the SAT IV regression in Section \ref{sec:SAT}. This allows us to propose hypothesis tests for the LATE that are asymptotically exact, i.e., their limiting rejection probability under the null hypothesis is equal to the nominal level.

In terms of our notation, the SFE IV regression is the result of regressing $Y_{i}$ on $(1[S_{i}=s]:s\in \mathcal{S})$ and $D_i$. Since the treatment decision $D_{i}$ is endogenously decided by the RCT participant, we instrument it with the exogenous treatment assignment $A_i$. To define this IV estimator precisely, set
\begin{align*}
 \mathbf{Y}_{n}~&\equiv~(Y_{i}:i=1,\dots, n)^{\prime },\\
 \mathbf{X}_{n}^{\text{sfe}}~&\equiv~(((1[S_{i}=s]:s\in \mathcal{S} )^{\prime },D_{i}):i=1,\dots ,n)^{\prime },\\
 \mathbf{Z}_{n}^{\text{sfe}}~&\equiv~(((1[S_{i}=s]:s\in \mathcal{S} )^{\prime },A_{i}):i=1,\dots ,n)^{\prime }.
\end{align*}
The estimators of the coefficients in IV SFE regression are
\begin{align}
((\hat{\gamma}_{\text{sfe}}(s):s\in \mathcal{S})^{\prime },\hat{\beta}_{ \text{sfe}})'~\equiv ~({\mathbf{Z}_{n}^{\text{sfe}}}^{\prime } \mathbf{X}_{n}^{\text{sfe}})^{-1}({\mathbf{Z}_{n}^{\text{sfe}}}^{\prime } \mathbf{Y}_{n}),
\label{eq:IV_SFE_estimators}
\end{align}
where $\hat{\gamma}_{\text{sfe}}(s)$ corresponds to the IV estimator of the coefficient on $1[S_{i}=s]$ and $\hat{\beta}_{ \text{sfe}}$ corresponds to the IV estimator of the coefficient on $D_i$.

Under Assumptions \ref{ass:1} and \ref{ass:2}, Theorem \ref{thm:plim_SFE} in the appendix shows that
\begin{equation*}
\hat{\beta}_{\mathrm{sfe}}~\overset{p}{\to }~\sum_{s\in \mathcal{S} }\omega (s)E[Y(1)-Y(0)|C,S=s],
\end{equation*}
where $(\omega (s):s\in \mathcal{S})$ are non-negative weights defined by
\begin{equation}
\omega (s)~\equiv~ \frac{\pi _{A}(s)(1-\pi _{A}(s))P(C,S=s)}{\sum_{\tilde{s} \in \mathcal{S}}\pi _{A}(\tilde{s})(1-\pi _{A}(\tilde{s}))P(C,S= \tilde{s})}.
\label{eq:weights}
\end{equation}
These equations show that $\hat{\beta}_{ \text{sfe}}$ is not necessarily a consistent estimator of the LATE under Assumptions \ref{ass:1} and \ref{ass:2}. By inspecting \eqref{eq:weights}, it follows that the consistency of the estimator can be restored provided that the treatment propensity does not vary by strata, i.e., Assumption \ref{ass:3}(c). For this reason, we maintain this condition for the remainder of this section.

The following result reveals that $\hat{ \beta}_{\text{sfe}}$ is a consistent and asymptotically normal estimator of the LATE. The result characterizes the asymptotic distribution of this estimator in terms of primitive parameters of the RCT.

%%%%%%%% DIVIDER %%%%%%%%%%%%
\begin{theorem}[SFE main result] \label{thm:AsyDist_SFE}
Suppose that Assumptions \ref{ass:1} and \ref{ass:3} hold. Then,
\begin{equation*}
\sqrt{n}(\hat{\beta}_{\mathrm{sfe}}-\beta )~~\overset{d}{\to } ~~N(0,V_{\mathrm{sfe}}),
\end{equation*}
where $\beta \equiv E[Y(1)-Y(0)|C]$ and $V_{\mathrm{sfe}}\equiv V_{\mathrm{sat}}+ V_{A}^{\mathrm{sfe}}$ with
\begin{align}
 V_{A}^{\mathrm{sfe}} &~\equiv~ \frac{(1-2\pi _{A})^{2}}{P(C)^{2}\pi _{A}(1-\pi _{A})}\sum_{s\in \mathcal{S}}p(s)\tau (s)(\pi _{D(1)}(s)-\pi _{D(0)}(s))^{2}(E[Y(1)-Y(0)|C,S=s]-\beta )^{2},
 \label{eq:VA}
\end{align}
with $P(C)$ and $V_{\mathrm{sat}}$ as defined in \eqref{eq:AsyDist_SAT_Avar}.
\end{theorem}
%%%%%%%% DIVIDER %%%%%%%%%%%%

We now give several remarks about Theorem \ref{thm:AsyDist_SFE}. First, we note that $\hat{\beta}_{\mathrm{sfe}}$ is related to the IV estimator $\hat{\beta}_2$ considered by \cite{ansel/hong/li:2018}. In fact, these two estimators coincide if one specifies their covariates as a full vector of strata dummies. As pointed out in Section \ref{sec:SAT}, specifying the covariates in the regression allows us to obtain a closed-form expression of the asymptotic variance $\hat{\beta}_{\mathrm{sfe}}$ in terms of the primitive parameters of the RCT.

As we have done in Section \ref{sec:SAT} with the SAT IV regression, we can analyze the consequences of imperfect compliance for the SFE IV regression by comparing Theorem \ref{thm:AsyDist_SFE} and \citet[Theorem 4.3]{bugni/canay/shaikh:2018} or \citet[Theorem 4.1]{bugni/canay/shaikh:2019}. First, note that imperfect compliance means that the probability limit of $\hat{\beta}_{\mathrm{sfe}}$ is not the ATE, but rather the LATE. Second, we note that imperfect compliance introduces significant changes to the asymptotic variance of $\hat{\beta}_{\mathrm{sfe}}$. This implies that the consistent estimators of $V_{\mathrm{sfe}}$ proposed in \citet[Section 4.2]{bugni/canay/shaikh:2018} or \citet[Theorem 4.2]{bugni/canay/shaikh:2019} do not apply, and a new one is required. We provide this in Theorem \ref{thm:SFE_se}.

Third, we note that Theorem \ref{thm:AsyDist_SFE} relies on Assumption \ref{ass:3}, which is stronger than Assumption \ref{ass:2} used to derive Theorem \ref{thm:AsyDist_SAT}. First, and as discussed earlier, Assumption \ref{ass:3}(c) is important to guarantee that $\hat{\beta}_{\mathrm{sfe}}$ is a consistent estimator of the LATE. Second, we note that the derivation of the asymptotic distribution of $\hat{\beta}_{\mathrm{sfe}}$ relies on the details about the CAR mechanism provided in Assumption \ref{ass:3}(b). These types of details were not required to derive the asymptotic distribution of $\hat{\beta}_{\mathrm{sat}}$ in Theorem \ref{thm:AsyDist_SAT}. 

Fourth, it is relevant to note that $V_{\mathrm{sfe}}-V_{\mathrm{sat}}=V_{A}^{\mathrm{sfe}} \geq 0$, which reveals that $\hat{\beta}_{\mathrm{sat}}$ is equally or more efficient than $\hat{\beta}_{\mathrm{sfe}}$. In particular, both estimators have the same asymptotic distribution if and only if $V_{A}^{\mathrm{sfe}} =0$. By inspecting \eqref{eq:VA}, this occurs if the RCT is implemented with either $\pi=1/2$ or $\tau(s)=0$ (e.g., by using SBR as described in Example \ref{ex:SBR}).\footnote{These results resemble those obtained by \citet[Corollary 1.2]{lin:2013} in the context of finite population inference, perfect compliance, and assignment using SRS.}

Fifth, we note that Theorem \ref{thm:AsyDist_SFE} allows for $\pi _{D( 0) }( s)=P(AT|S=s)=0$ or $1-\pi _{D( 1) }( s)=P(NT|S=s)=0$, by using the same abuse of notation as in Section \ref{sec:SAT}. 
%If $\pi _{D( 0) }( s)=P(AT|S=s)=0$, the mean and the variance of $\{Y(1)|AT,S=s\}$ are not properly defined, but we can set
%\begin{equation*}
% V[ Y( 1) |AT,S=s ]\pi _{D( 0) }( s)=0~~~\text{and}~~~ E[ Y( 1) |AT,S=s ]\pi _{D( 0) }( s)=0.
%\end{equation*}
%Similarly, if $1-\pi _{D( 1) }( s)=P(NT|S=s)=0$, the mean and the variance of $\{Y(0)|NT,S=s\}$ are not properly defined, but we can set
%\begin{equation*}
% V[ Y(0)|NT,S=s ](1-\pi _{D( 1) }( s))=0~~~\text{and}~~~E[ Y(0)|NT,S=s ](1-\pi _{D( 1) }( s))=0.
%\end{equation*}
In particular, in the special case of perfect compliance (i.e., $\pi _{D(1)}(s)=1$, $\pi _{D(0)}(s)=0$), Theorem \ref{thm:AsyDist_SFE} then holds with 
\begin{align*}
\beta &~=~ E[Y(1)-Y(0)]\\
V_{\mathrm{sfe}} &~=~ \sum_{s\in \mathcal{S}}{ p( s) }\left[\tfrac{V[ Y( 1) |S=s]}{\pi _{A}(s)} + \tfrac{V[ Y( 0) |S=s]}{(1-\pi _{A}(s))}+\left(1+ \tau (s)\tfrac{(1-2\pi _{A})^{2}}{\pi _{A}(1-\pi _{A})} \right)(E[Y(1)-Y(0)|S=s]-\beta )^{2}
\right],
\end{align*}
which coincides with the corresponding results in \citet[Section 4.2]{bugni/canay/shaikh:2018} and \citet[Section 5]{bugni/canay/shaikh:2019}.

As promised earlier, the next result provides a consistent estimator of $V_{\mathrm{sfe}}$.

%%%%%%%% DIVIDER %%%%%%%%%%%%
\begin{theorem}[Estimator of SFE asy.\ variance] \label{thm:SFE_se}
Assume Assumptions \ref{ass:1} and \ref{ass:3}. Define the following estimator:
\begin{equation}
\hat{V}_{A}^{\mathrm{sfe}}~\equiv~ \frac{1}{\hat{P}(C)^{2}}\sum_{s\in \mathcal{ S}}\frac{n(s)}{n}\tau (s)\frac{(1-2\frac{n_{A}(s)}{n(s)})^{2}}{\frac{n_{A}(s) }{n(s)}(1-\frac{n_{A}(s)}{n(s)})}\left[\frac{n_{AD}(s)}{n_{A}(s)}-\frac{ (n_{D}(s)-n_{AD}(s))}{(n(s)-n_{A}(s))}\right]^{2}(\hat{\beta}_{\mathrm{sat}}(s)- \hat{\beta}_{\mathrm{sat}})^{2},
\end{equation}
where $(\hat{\beta}_{\mathrm{sat}}(s):s\in \mathcal{S})$, $\hat{\beta}_{ \mathrm{sat}}$, and $\hat{P}(C)$ are as in \eqref{eq:IV_SAT_estimators}, \eqref{eq:beta_SAT}, and \eqref{eq:P_C}, respectively. Then,
\begin{equation*}
\hat{V}_{\mathrm{sfe}}~=~\hat{V}_{\mathrm{sat}}+\hat{V}_{A}^{\mathrm{sfe}}~\overset{p}{ \to }~V_{\mathrm{sfe}},
\end{equation*}
where $\hat{V}_{\mathrm{sat}}$ is as in \eqref{eq:V_defns}.
\end{theorem}

%%%%%%%% DIVIDER %%%%%%%%%%%%
%\begin{remark}
%In principle, under the assumptions of Theorem \ref{thm:SFE_se}, said result would also hold if one replaced every $\frac{n_{A}(s)}{n(s)}$ in the expressions for $\hat{V}_{1}^{\mathrm{sat}}$, $\hat{V}_{0}^{\mathrm{sat}}$, $ \hat{V}_{H}^{\mathrm{sat}}$, and $\hat{V}_{A}^{\mathrm{sfe}}$ by any weighted average of $\{ \frac{n_{A}(s)}{n(s)}:s\in S\} $, such as $\hat{\pi}_{A}=\sum_{s\in \mathcal{S}}\frac{n_{A}(s)}{n(s)}\frac{n( s) }{n}$. We decided against this to enable making finite samples comparisons among $\hat{V}_{ \mathrm{sat}}$ and $\hat{V}_{\mathrm{sfe}}$. In particular, if we use the definitions in Theorems \ref{thm:SAT_se} and \ref{thm:SFE_se}, $\hat{V}_{ \mathrm{sat}}\leq\hat{V}_{\mathrm{sfe}}$, with equality if and only if $\hat{V}_{A}^{\mathrm{sfe}}=0$.
%\end{remark}
%%%%%%%% DIVIDER %%%%%%%%%%%%

To conclude the section, we can propose hypothesis tests for the LATE by combining Theorems \ref{thm:AsyDist_SFE} and \ref{thm:SFE_se}. For completeness, this is recorded in the next result.

\begin{theorem}[SFE test]\label{thm:SFE_test} 
Suppose that Assumptions \ref{ass:1} and \ref{ass:3} hold, and that $V_{\mathrm{sfe}}>0$. For the problem of testing \eqref{eq:HT} at level $\alpha \in (0,1)$, Consider the following hypothesis testing procedure
\begin{align*}
\phi_{n}^{\mathrm{sfe}}(X^{(n)})~ \equiv~1[~|{\sqrt{n}(\hat{\beta}_{\mathrm{sfe}}-\beta_0 )}|~>~\sqrt{\hat{V}_{\mathrm{sfe}}}z_{1-\alpha/2}~],
\end{align*}
where $z_{1-\alpha/2}$ is the $(1-\alpha/2)$-quantile of $N(0,1)$. Then,
\begin{align*}
\lim_{n\to \infty}E[\phi_{n}^{\mathrm{sfe}}(X^{(n)})]~ =~\alpha
\end{align*}
whenever $H_0$ in \eqref{eq:HT} holds, i.e., $\beta = \beta_0$.
\end{theorem}
%%%%%%%% DIVIDER %%%%%%%%%%%%

\section{``Two sample'' (2S) IV regression}\label{sec:2STT}

We now consider the asymptotic properties of an IV estimator of the LATE based on a linear regression model of the outcome of interest on a constant and the treatment decision, where the latter is instrumented with the treatment assignment. Following the nomenclature in \cite{bugni/canay/shaikh:2018}, we refer to this as the 2S IV regression. Under certain conditions, the 2S IV regression can consistently estimate the LATE. We show that this estimator is asymptotically normal and we characterize its asymptotic variance in terms of the primitives parameters of the RCT. We also propose a consistent estimator of this asymptotic variance by using the results of the SAT IV regression in Section \ref{sec:SAT}. This allows us to propose hypothesis tests for the LATE that are asymptotically exact, i.e., their limiting rejection probability under the null hypothesis is equal to the nominal level.

In terms of our notation, the 2S IV regression is the result of regressing $Y_{i}$ on $1$ and $D_i$. Since the treatment decision $D_{i}$ is endogenously decided by the RCT participant, we instrument it with the exogenous treatment assignment $A_i$. To define this IV estimator precisely, set
\begin{align*}
 \mathbf{Y}_{n}~&\equiv~(Y_{i}:i=1,\dots, n)^{\prime },\\
 \mathbf{X}_{n}^{\mathrm{2s}}~&\equiv~((1,D_{i} ):i=1,\dots ,n)^{\prime },\\
 \mathbf{Z}_{n}^{\mathrm{2s}}~&\equiv~((1,A_{i} ):i=1,\dots ,n)^{\prime }.
\end{align*}
The estimators of the coefficients in IV 2S regression are
\begin{align}
(\hat{\gamma}_{\mathrm{2s}},\hat{\beta}_{ \mathrm{2s}})^{\prime }~\equiv~
({\mathbf{Z}_{n}^{\mathrm{2s}}}^{\prime } \mathbf{X}_{n}^{\mathrm{2s}})^{-1}({\mathbf{Z}_{n}^{\mathrm{2s}}}^{\prime } \mathbf{Y}_{n}).
\label{eq:IV_2S_estimators}
\end{align}
where $\hat{\gamma}_{\mathrm{2s}}$ corresponds to the IV estimator of the coefficient on $1$ and $\hat{\beta}_{ \mathrm{2s}}$ corresponds to the IV estimator of the coefficient on $D_i$.

Under Assumptions \ref{ass:1} and \ref{ass:2}, Theorem \ref{thm:plim_2SR} in the appendix shows that
\begin{equation}
\hat{\beta}_{\mathrm{2s}}~\overset{p}{\to }~\frac{\sum_{s\in \mathcal{S}}p(s)\left[
\begin{array}{c}
[\pi _{A}(s)-\bar{\pi} _{A} ]\pi _{D(0)}(s)E[Y(1)|AT,S=s]+ \\ 
+[\pi _{A}(s)-\bar{\pi} _{A} ](1-\pi _{D(1)}(s))E[Y(0)|NT,S=s] \\ 
+(1-\bar{\pi} _{A} )\pi _{A}(s)(\pi _{D(1)}(s)-\pi _{D(0)}(s))E[Y(1)|C,S=s] \\
-\bar{\pi} _{A} (1-\pi _{A}(s))(\pi _{D(1)}(s)-\pi _{D(0)}(s))E[Y(0)|C,S=s]
\end{array}
\right] }{(1-\bar{\pi} _{A})(\sum_{s\in \mathcal{S}}p(s)\pi _{A}(s)\pi _{D(1)}(s))-\bar{\pi} _{A} (\sum_{s\in \mathcal{S}}p(s)(1-\pi _{A}(s))\pi _{D(0)}(s))},
\label{eq:weights2}
\end{equation}
with $\bar{\pi} _{A} \equiv \sum_{s\in \mathcal{S}}p(s)\pi _{A}(s)$. This equation reveals that $\hat{\beta}_{\mathrm{2s}}$ is not necessarily a consistent estimator of the LATE under Assumptions \ref{ass:1} and \ref{ass:2}. However, if we additionally impose Assumption \ref{ass:3}(c), it follows that $\hat{\beta}_{\mathrm{2s}}$ becomes a consistent estimator of the LATE. For this reason, we maintain this condition for the remainder of this section.

The following result reveals that $\hat{ \beta}_{\mathrm{2s}}$ is a consistent and asymptotically normal estimator of the LATE. The result characterizes the asymptotic distribution of this estimator in terms of primitive parameters of the RCT.

%%%%%%%% DIVIDER %%%%%%%%%%%%
\begin{theorem}[2S main result] \label{thm:AsyDist_2SR}
Suppose that Assumptions \ref{ass:1} and \ref{ass:3} hold. Then,
\begin{equation*}
\sqrt{n}(\hat{\beta}_{\mathrm{2s}}-\beta )~~\overset{d}{\to } ~~N(0,V_{\mathrm{2s}}),
\end{equation*}
where $\beta \equiv E[Y(1)-Y(0)|C]$ and $V_{\mathrm{2s}}\equiv V_{\mathrm{sat}}+ V_{A}^{\mathrm{2s}}$ with
\begin{align}
V_{A}^{\mathrm{2s}}&~=~
\sum_{s\in \mathcal{S}}\tfrac{p( s) \tau (s)}{\pi _{A}(1-\pi _{A})P( C) ^{2}}\left[
\begin{array}{c}
(\pi _{A}\pi _{D( 0) }( s) +( 1-\pi_{A}) \pi _{D( 1) }( s) ) (E[Y(1)-Y(0)|C,S=s]-\beta ) \\
+\pi _{D(1)}(s)E[Y(0)|C,S=s]-\pi _{D(0)}(s)E[Y(1)|C,S=s] \\ 
+\pi _{D(0)}(s)E[Y(1)|AT,S=s]+(1-\pi _{D(1)}(s))E[Y(0)|NT,S=s]\\ 
-\sum_{\tilde{s}\in \mathcal{S}}p(\tilde{s})\times \\
\left[ 
\begin{array}{c}
(\pi _{A}\pi_{D( 1) }( \tilde{s}) +( 1-\pi _{A}) \pi_{D( 0) }( \tilde{s}) ) ( E[Y(1)-Y(0)|C,S=\tilde{s}]-\beta )\\
+\pi _{D(1)}(\tilde{s})E[Y(0)|C,S=\tilde{s}]-\pi _{D(0)}(\tilde{s})E[Y(1)|C,S=\tilde{s}] \\ 
+\pi _{D(0)}(\tilde{s})E[Y(1)|AT,S=\tilde{s}]+(1-\pi _{D(1)}(\tilde{s}))E[Y(0)|NT,S=\tilde{s}]
\end{array}
\right]
\end{array}
\right] ^{2},
\label{eq:VA2}
\end{align}
with $P(C)$ and $V_{\mathrm{sat}}$ as defined in \eqref{eq:AsyDist_SAT_Avar}.
\end{theorem}
%%%%%%%% DIVIDER %%%%%%%%%%%%

Several remarks about Theorem \ref{thm:AsyDist_2SR} are in order. First, we note that $\hat{\beta}_{\mathrm{2s}}$ coincides with IV estimator $\hat{\beta}_1$ considered by \cite{ansel/hong/li:2018}. Relative to their results, we provide a closed-form expression of the asymptotic variance $\hat{\beta}_{\mathrm{2s}}$ in terms of the primitive parameters of the RCT.

As we have done in the previous sections, we can analyze the consequences of imperfect compliance for the 2S IV regression by comparing Theorem \ref{thm:AsyDist_2SR} and \citet[Theorem 4.1]{bugni/canay/shaikh:2018}. First, note that imperfect compliance means that the probability limit of $\hat{\beta}_{\mathrm{2s}}$ is not the ATE, but rather the LATE. Second, we note that imperfect compliance introduces significant changes to the asymptotic variance of $\hat{\beta}_{\mathrm{2s}}$. This implies that the consistent estimator of $V_{\mathrm{2s}}$ proposed in \citet[Section 4.1]{bugni/canay/shaikh:2018} does not apply, and a new one is required. We provide this in Theorem \ref{thm:2SR_se}.

Third, we note that Theorem \ref{thm:AsyDist_2SR} relies on Assumption \ref{ass:3}, which is stronger than Assumption \ref{ass:2} used to derive Theorem \ref{thm:AsyDist_SAT}. The argument here is the same as in Section \ref{sec:SFE}. First, Assumption \ref{ass:3}(c) is important to guarantee that $\hat{\beta}_{\mathrm{2s}}$ is a consistent estimator of the LATE. Second, the derivation of the asymptotic distribution of $\hat{\beta}_{\mathrm{2s}}$ relies on the details about the CAR mechanism provided in Assumption \ref{ass:3}(b).

Fourth, it is relevant to note that $V_{\mathrm{2s}}-V_{\mathrm{sat}}=V_{A}^{\mathrm{2s}} \geq 0$, which reveals that $\hat{\beta}_{\mathrm{sat}}$ is equally or more efficient than $\hat{\beta}_{\mathrm{2s}}$. In particular, both estimators have the same asymptotic distribution if and only if $V_{A}^{\mathrm{2s}} =0$. By inspecting \eqref{eq:VA}, this occurs if the RCT is implemented with $\tau(s)=0$ (e.g., by using SBR as described in Example \ref{ex:SBR}). We also note that $\tau(s)=0$ implies that $V_{\mathrm{2s}} = V_{\mathrm{sfe}} = V_{\mathrm{sat}}$ and $\pi_A=1/2$ implies that $V_{\mathrm{2s}} \geq V_{\mathrm{sfe}} = V_{\mathrm{sat}}$. Other than these special cases, $V_{\mathrm{2s}}$ and $V_{\mathrm{sfe}}$ cannot be ordered unambiguously (See \citet[Remark 4.8]{bugni/canay/shaikh:2018} for a similar point in the context of perfect compliance).\footnote{As in the previous section, these results resemble those obtained by \citet[Corollary 1.1]{lin:2013} in the context of finite population inference, perfect compliance, and assignment using SRS.}
% I can find examples where one dominates or the other.

Fifth, we note that Theorem \ref{thm:AsyDist_2SR} allows for $\pi _{D( 0) }( s)=P(AT|S=s)=0$ or $1-\pi _{D( 1) }( s)=P(NT|S=s)=0$, by using the same abuse of notation as in Section \ref{sec:SAT}. 
%If $\pi _{D( 0) }( s)=P(AT|S=s)=0$, the mean and the variance of $\{Y(1)|AT,S=s\}$ are not properly defined, but we can set
%\begin{equation*}
% V[ Y( 1) |AT,S=s ]\pi _{D( 0) }( s)=0~~~\text{and}~~~ E[ Y( 1) |AT,S=s ]\pi _{D( 0) }( s)=0.
%\end{equation*}
%Similarly, if $1-\pi _{D( 1) }( s)=P(NT|S=s)=0$, the mean and the variance of $\{Y(0)|NT,S=s\}$ are not properly defined, but we can set
%\begin{equation*}
% V[ Y(0)|NT,S=s ](1-\pi _{D( 1) }( s))=0~~~\text{and}~~~E[ Y(0)|NT,S=s ](1-\pi _{D( 1) }( s))=0.
%\end{equation*}
In the special case of perfect compliance (i.e., $\pi _{D(1)}(s)=1$, $\pi _{D(0)}(s)=0$), Theorem \ref{thm:AsyDist_2SR} then holds with 
\begin{align*}
\beta & ~=~E[Y(1)-Y(0)] \\
V_{\mathrm{2s}}& ~=~\sum_{s\in \mathcal{S}}{p(s)}\left[ 
\begin{array}{c}
\frac{V[Y(1)|S=s]}{\pi _{A}(s)}+\frac{V[Y(0)|S=s]}{(1-\pi _{A}(s))}+(E[Y(1)-Y(0)|S=s]-\beta )^{2} +\\ 
\tau (s)\pi _{A}(1-\pi _{A})\left( \frac{(E[Y(1)|S=s]-E[Y(1)])}{\pi _{A}}+\frac{(E[Y(0)|S=s]-E[Y(0)])}{(1-\pi _{A})}\right) ^{2}
\end{array}
\right] , 
\end{align*}
which coincides with the result in \citet[Section 4.1]{bugni/canay/shaikh:2018}.

As promised earlier, the next result provides a consistent estimator of $V_{\mathrm{2s}}$.

%%%%%%%% DIVIDER %%%%%%%%%%%%
\begin{theorem}[Estimator of 2S asy.\ variance] \label{thm:2SR_se}
Assume Assumptions \ref{ass:1} and \ref{ass:3}. Define the following estimator:
\begin{equation}
\hat{V}_{A}^{\mathrm{2s}}~\equiv~ \sum_{s\in \mathcal{S}}\tfrac{\frac{n(s) }{n}\tau (s)}{\frac{n_{A}( s) }{n( s) }(1-\frac{n_{A}( s) }{n( s) })\hat{P}(C)^{2}}\left[ 
\begin{array}{c}
[\frac{n_{A}( s) }{n( s) }\frac{n_{D}(s) -n_{AD}(s) }{n(s)-n_{A}(s) }+(1-\frac{n_{A}( s) }{n(s) })\frac{n_{AD}( s) }{n_{A}( s) }](\hat{\beta}_{\mathrm{sat}}(s)-\hat{\beta}_{\mathrm{sat}}) \\ 
-\sum_{\tilde{s}\in \mathcal{S}}\frac{n_{D}( \tilde{s}) }{n}(\hat{
\beta}_{\mathrm{sat}}(\tilde{s})-\hat{\beta}_{\mathrm{sat}}) +\hat{\gamma}_{\mathrm{sat}}(s)-\sum_{\tilde{s}\in \mathcal{S}}p(\tilde{s})\hat{\gamma}_{\mathrm{sat}}(\tilde{s})
\end{array}
\right] ^{2},
\end{equation}
where $(\hat{\beta}_{\mathrm{sat}}(s):s\in \mathcal{S})$, $\hat{\beta}_{ \mathrm{sat}}$, and $\hat{P}(C)$ are as in \eqref{eq:IV_SAT_estimators}, \eqref{eq:beta_SAT}, and \eqref{eq:P_C}, respectively. Then,
\begin{equation*}
\hat{V}_{\mathrm{2s}}~=~\hat{V}_{\mathrm{sat}}+\hat{V}_{A}^{\mathrm{2s}}~\overset{p}{ \to }~V_{\mathrm{2s}},
\end{equation*}
where $\hat{V}_{\mathrm{sat}}$ is as in \eqref{eq:V_defns}.
\end{theorem}

%%%%%%%% DIVIDER %%%%%%%%%%%%
%\begin{remark}
%In principle, under the assumptions of Theorem \ref{thm:2SR_se}, said result would also hold if one replaced every $\frac{n_{A}(s)}{n(s)}$ in the expressions for $\hat{V}_{1}^{\mathrm{sat}}$, $\hat{V}_{0}^{\mathrm{sat}}$, $ \hat{V}_{H}^{\mathrm{sat}}$, and $\hat{V}_{A}^{\mathrm{sfe}}$ by any weighted average of $\{ \frac{n_{A}(s)}{n(s)}:s\in S\} $, such as $\hat{\pi}_{A}=\sum_{s\in \mathcal{S}}\frac{n_{A}(s)}{n(s)}\frac{n( s) }{n}$. We decided against this to enable making finite samples comparisons among $\hat{V}_{ \mathrm{sat}}$ and $\hat{V}_{\mathrm{sfe}}$. In particular, if we use the definitions in Theorems \ref{thm:SAT_se} and \ref{thm:2SR_se}, $\hat{V}_{ \mathrm{sat}}\leq\hat{V}_{\mathrm{2s}}$, with equality if and only if $\hat{V}_{A}^{\mathrm{2s}}=0$.
%\end{remark}
%%%%%%%% DIVIDER %%%%%%%%%%%%

To conclude the section, we can propose hypothesis tests for the LATE by combining Theorems \ref{thm:AsyDist_2SR} and \ref{thm:2SR_se}. For completeness, this is recorded in the next result.

\begin{theorem}[2S test]\label{thm:2SR_test} 
Suppose that Assumptions \ref{ass:1} and \ref{ass:3} hold, and that $V_{\mathrm{2s}}>0$. For the problem of testing \eqref{eq:HT} at level $\alpha \in (0,1)$, Consider the following hypothesis testing procedure
\begin{align*}
\phi_{n}^{\mathrm{2s}}(X^{(n)})~ \equiv~1[~|{\sqrt{n}(\hat{\beta}_{\mathrm{2s}}-\beta_0 )}|~>~\sqrt{\hat{V}_{\mathrm{2s}}}z_{1-\alpha/2}~],
\end{align*}
where $z_{1-\alpha/2}$ is the $(1-\alpha/2)$-quantile of $N(0,1)$. Then,
\begin{align*}
\lim_{n\to \infty}E[\phi_{n}^{\mathrm{2s}}(X^{(n)})]~ =~\alpha
\end{align*}
whenever $H_0$ in \eqref{eq:HT} holds, i.e., $\beta = \beta_0$.
\end{theorem}
%%%%%%%% DIVIDER %%%%%%%%%%%%

\section{Designing an RCT based on pilot RCT data}\label{sec:optim}

This section considers the situation of a researcher interested in designing a hypothetical RCT with CAR to estimate the LATE. The researcher is in charge of choosing the parameters of the RCT, which entail the randomization scheme, the stratification function, and the treatment probability vector. To make these decisions, the researcher observes the results of a previous pilot RCT (also with CAR) conducted on the same population of interest. Therefore, the researcher can consistently estimate the main features of the population based on the pilot RCT data.

Throughout this section, we assume that both RCTs satisfy the sampling framework in Section \ref{sec:setup}. We use $n_P$ to denote the sample size of the pilot RCT and $n$ to denote the sample size of the hypothetical RCT. We presume that $n_P$ and $n$ are sufficiently large, so our asymptotic analysis in previous sections accurately represents both RCTs. This framework allows us to exploit the pilot RCT information to improve the asymptotic efficiency of the IV LATE estimators in the hypothetical RCT.\footnote{We recognize that, in practice, there may be a wide array considerations in designing the hypothetical RCT. The fact that we can focus on the asymptotic efficiency of the IV LATE estimator is a by-product of our asymptotic framework.}
%%%%%%

\subsection{Choosing the randomization scheme}\label{sec:RCTChoice_1}

In this section, we consider a hypothetical RCT with CAR with a given stratification function $S:\mathcal{Z} \to \mathcal{S}$. Our objective is to consider the effect that the randomization scheme has on the asymptotic distribution of the LATE estimators.

Under Assumptions \ref{ass:1} and \ref{ass:2}, Theorem \ref{thm:AsyDist_SAT} reveals the randomization scheme has no influence on the asymptotic distribution of the LATE estimator based on the SAT IV regression. If we additionally impose Assumption \ref{ass:3}, then Theorems \ref{thm:AsyDist_SFE} and \ref{thm:AsyDist_2SR} show that the randomization scheme affects the asymptotic variance of the LATE estimators based on the SFE IV and 2S IV regressions via the parameter $(\tau(s):s \in \mathcal{S})$, which appears on $V_A^{\mathrm{sfe}}$ and $V_A^{\mathrm{2s}}$, respectively. For both of these estimators, the optimal choice of the randomization scheme is to set $\tau(s)=0$ for all $s \in \mathcal{S}$. This decision is optimal regardless of the information in the pilot RCT. Under this optimal choice, the asymptotic distribution of the LATE estimators in the SFE IV and 2S IV regressions coincide with the LATE estimator in the SAT IV regression, given in Theorem \ref{thm:AsyDist_SAT}. In practice, $\tau(s)=0$ for all $s \in \mathcal{S}$ can be achieved by implementing the CAR using SBR. 

Based on these arguments, it is natural to focus the remainder of this section on the case in which the hypothetical RCT has $\tau(s)=0$ for all $s \in \mathcal{S}$.

\subsection{Choosing the stratification function}\label{sec:strataChoice}

This section considers the effect of the choice of the stratification function of the hypothetical RCT on the asymptotic variance of the LATE estimators. For the sake of simplicity, we impose some restrictions on the hypothetical RCT under consideration. First, we assume that the hypothetical RCT satisfies Assumptions \ref{ass:1} and \ref{ass:3} with $\tau(s) =0$ for all $s\in \mathcal{S}$. This choice is motivated by the discussion in Section \ref{sec:RCTChoice_1}, and has the additional benefit that we can describe the asymptotic distribution of all three LATE estimators in a single statement.\footnote{Note that Assumption \ref{ass:3}(c) imposes a constant treatment assignment probability. In principle, we could generalize our analysis at the expense of substantially complicating the notation.}

The next result establishes that if the strata of the hypothetical RCT becomes coarser and all else remains equal, then the asymptotic variance of the IV estimators will either remain constant or increase.

%%%%%%%%%%%% DIVIDER %%%%%%%%%%%%

\begin{theorem}\label{thm:coarser_moreVar}
Consider two hypothetical RCTs with CAR on the same population and with the same parameters except for the strata function. In particular, both RCTs satisfy Assumptions \ref{ass:1} and \ref{ass:3}, and both have $\tau ( s) =0$ and the same $\pi _{A}$. The first RCT has strata function $S_{1}:\mathcal{Z} \to \mathcal{S}_{1}$, the second RCT has strata function $S_{2}:\mathcal{Z} \to \mathcal{S}_{2}$, and $S_{1}$ is (weakly) finer than $S_{2}$, i.e.,
\begin{equation}
\forall z,z'\in \mathcal{Z}, ~~S_{1}( z) =S_{1}( z^{\prime }) ~~\Longrightarrow~~ S_{2}( z) =S_{2}( z^{\prime }) . \label{eq:coarser_1}
\end{equation}
Let $\hat{\beta}_{1}$ denote the SAT, SFE, or 2S IV LATE estimator from the first RCT with sample size $n$, and let $\hat{\beta}_{2}$ denote the SAT, SFE, or 2S IV LATE estimator from the second RCT with sample size $n$. Then, as $n\to \infty $,
\begin{equation}
\sqrt{n}( \hat{\beta}_1-\beta ) \overset{d}{\to}N( 0,V_{\mathrm{sat},1}),~~~\sqrt{n}( \hat{\beta}_2-\beta ) \overset{d}{\to}N( 0,V_{\mathrm{sat},2}), \label{eq:coarser_2}
\end{equation}
and
\begin{equation}
 V_{\mathrm{sat},1} ~~\leq~~ V_{\mathrm{sat},2}.
 \label{eq:coarser_3}
\end{equation}
\end{theorem}
%%%%%%%%%%%% DIVIDER %%%%%%%%%%%%

One implication of Theorem \ref{thm:coarser_moreVar} is that, all else equal, a finer strata structure is always preferable from the point of view of the asymptotic efficiency of the LATE estimator. Since the strata are defined based on the baseline covariate $Z$, the finest possible strata structure is one in which each point in the support of $Z$, denoted by $\mathcal{Z}$, is assigned to its own stratum. This idea is feasible if $Z$ is discretely distributed and $\mathcal{Z}$ is finite. On the other hand, this conclusion would fail if $\mathcal{Z}$ takes infinitely many values, as our formal analysis is based on the presumption that $\mathcal{S}$ is a finite set and $p(s)>0$ for all $s \in \mathcal{S}$. In this sense, our asymptotic framework limits our ability to make extreme recommendations based on Theorem \ref{thm:coarser_moreVar}.

It is relevant to connect Theorem \ref{thm:coarser_moreVar} with the recent contributions by \cite{tabordmeehan:2020} and \cite{bai:2022}. Unlike our work, both references consider inference for the ATE under perfect compliance. \cite{bai:2022} considers the problem of treatment assignment in an RCT and shows that the optimal stratified randomization scheme (in the sense of minimizing the asymptotic variance of the ATE estimator) is achieved by a certain matched-pair design, i.e., strata formed by pairs of individuals. While our asymptotic framework does not allow matched-pair designs, one could interpret them as the limiting result of repeatedly splitting up our strata. In this sense, Theorem \ref{thm:coarser_moreVar} is compatible with \cite{bai:2022}'s optimality result. In turn, \cite{tabordmeehan:2020} considers treatment assignment in an RCT using a randomization procedure referred to as {\it stratification trees} in the context of perfect compliance. The main result in \cite{tabordmeehan:2020} shows that these stratification trees can be used to find an optimal stratification function (again, in the sense of minimizing the asymptotic variance of the ATE estimator). To derive this result, he restricts attention to strata functions with a fixed level of complexity or {\it tree depth}. Theorem \ref{thm:coarser_moreVar} implies that the asymptotic variance of the ATE estimator cannot increase if we consider the optimal stratification tree in \cite{tabordmeehan:2020} and we further divide any of its branches. Note that this is compatible with the main results in \cite{tabordmeehan:2020}, as the further divided tree would have a level of complexity not allowed in the optimization problem in his paper.

Consider the situation of a researcher who has completed a pilot RCT with CAR. Motivated by Theorem \ref{thm:coarser_moreVar}, the researcher interested in gaining efficiency in the estimation of the LATE would want to run a hypothetical RCT with a finer strata partition than that of the pilot RCT. In this case, the researcher would naturally be interested in predicting the asymptotic variance of the LATE estimator in this hypothetical RCT based on the pilot RCT data (that is, before implementing the hypothetical RCT). This is precisely the problem addressed in Theorem \ref{thm:finer_strata}. This result provides a consistent estimator of the asymptotic variance of the LATE estimator in the hypothetical RCT with finer strata than the pilot RCT, based exclusively on the data from the pilot RCT. 
%The result follows from showing that the data from the pilot RCT can be reinterpreted as if it belonged to the hypothetical RCT.

\begin{theorem}\label{thm:finer_strata}
Let $(( Y_{i},Z_{i},S^P_i,A_{i})) _{i=1}^{n_{P}}$ denote data from a pilot RCT that satisfies Assumptions \ref{ass:1}, \ref{ass:2}, and \ref{ass:3}(c), and uses a strata function given by $S^{P}:\mathcal{Z}\to \mathcal{S}^{P}$. Consider a hypothetical RCT on the same population that satisfies Assumptions \ref{ass:1} and \ref{ass:3}, that uses the same $\pi _{A}$ as the pilot RCT, $\tau ( s) =0$ for all $s\in \mathcal{S}$, and a strata function $S:\mathcal{Z} \to \mathcal{S}$ that is finer than that of the pilot RCT i.e.,
\begin{equation}
\forall z,z^{\prime }\in \mathcal{Z}, ~~~S( z) =S( z^{\prime }) ~~\Longrightarrow~~ S^{P}( z) =S^{P}( z^{\prime }). \label{eq:newS_0}
\end{equation}
We also assume that $n_{A}^{P}(s)/n^{P}(s)\overset{p}{\to }\pi _{A}$ for all $s\in \mathcal{S}$, where $n_{A}^{P}( s) \equiv \sum_{i=1}^{n_{P}}1[ A_{i}=1,S( Z_{i}) =s] $ and $n^{P}( s) \equiv \sum_{i=1}^{n_{P}}1[ S( Z_{i}) =s] $.

Let $\hat{\beta}$ denote the SAT, SFE, or 2S IV LATE estimator from this hypothetical RCT with sample size $n$. Then, as $n\to \infty $,
\begin{equation}
\sqrt{n}( \hat{\beta}-\beta ) ~~\overset{d}{\to }~~ N( 0,V_{\mathrm{sat}}) . \label{eq:newS_1}
\end{equation}
Furthermore, $V_{\mathrm{sat}}$ can be consistently estimated using Theorem \ref{thm:SAT_se} using the pilot data but with updated strata information, given by $(( Y_{i},Z_{i},S(Z_i),A_{i})) _{i=1}^{n_{P}}$.
\end{theorem}
%%%%%%%%%%%% DIVIDER %%%%%%%%%%%%

Theorem \ref{thm:finer_strata} shows how to use the pilot RCT data to estimate the asymptotic variance of the LATE estimators in a hypothetical RCT with finer strata than that of the pilot RCT. In principle, one could also consider a similar exercise in the case when the hypothetical RCT has coarser strata than that of the pilot RCT. We omit this result for the sake of brevity, as Theorem \ref{thm:coarser_moreVar} shows that such a hypothetical RCT would produce less efficient estimators than those from the pilot RCT.

\subsection{Choosing the treatment propensity}\label{sec:pi_opt}

This section considers the effect of the choice of the treatment propensity vector $(\pi_A(s):s\in \mathcal{S})$ on the asymptotic variance of the LATE estimators, for a given the stratification function $S:\mathcal{Z} \to \mathcal{S}$. Under appropriate assumptions, Theorems \ref{thm:AsyDist_SAT}, \ref{thm:AsyDist_SFE}, and \ref{thm:AsyDist_2SR} provide an asymptotic distribution of our LATE IV estimators for any treatment propensity vector $(\pi_A(s):s\in \mathcal{S})$. The following result calculates the optimal treatment propensity in the sense of minimizing their asymptotic variance, and provides a strategy to estimate it based on data from a pilot RCT.

\begin{theorem}\label{thm:opt_pi}
Consider a hypothetical RCT that satisfies Assumptions \ref{ass:1} and \ref{ass:2}. The treatment assignment probability vector that minimizes the asymptotic variance of the SAT IV estimator is $(\pi _{A}^{\ast }(s):s\in \mathcal{S}) $ with
\begin{equation}
\pi _{A}^{\ast }(s)~\equiv~ \left( 1+\sqrt{\frac{\Pi _{2}(s)}{\Pi _{1}(s)}} \right) ^{-1}, \label{eq:Optimal_Pi_s}
\end{equation}
where
\begin{align}
\Pi _{1}( s) &~\equiv~\left[ 
\begin{array}{c}
\left[ 
\begin{array}{c}
V[Y(1)|AT,S=s]\pi _{D(0)}(s)+V[Y(0)|NT,S=s](1-\pi _{D(1)}(s)) \\ 
+V[Y(1)|C,S=s](\pi _{D(1)}(s)-\pi _{D(0)}(s)) \\ 
+(E[Y(1)|C,S=s]-E[Y(1)|AT,S=s])^{2}\times \\ 
\pi _{D(0)}(s)(\pi _{D(1)}(s)-\pi _{D(0)}(s))/\pi _{D(1)}(s)
\end{array}
\right] + \\ 
\frac{(1-\pi _{D(1)}(s))}{\pi _{D(1)}(s)}\left[ 
\begin{array}{c}
-\pi _{D(0)}(s)(E[Y(1)|C,S=s]-E[Y(1)|AT,S=s]) \\ 
+\pi _{D(1)}(s)(E[Y(0)|C,S=s]-E[Y(0)|NT,S=s]) \\ 
+\pi _{D(1)}(s)(E[Y(1)-Y(0)|C,S=s]-\beta )
\end{array}
\right] ^{2}
\end{array}
\right] \notag \\
\Pi _{2}( s) &~\equiv~\left[ 
\begin{array}{c}
\left[ 
\begin{array}{c}
V[Y(1)|AT,S=s]\pi _{D(0)}(s)+V[Y(0)|NT,S=s](1-\pi _{D(1)}(s)) \\ 
+V[Y(0)|C,S=s](\pi _{D(1)}(s)-\pi _{D(0)}(s)) \\ 
+(E[Y(0)|C,S=s]-E[Y(0)|NT,S=s])^{2}\times \\ 
(1-\pi _{D(1)}(s))(\pi _{D(1)}(s)-\pi _{D(0)}(s))/(1-\pi _{D(0)}(s))
\end{array}
\right] + \\ 
\frac{\pi _{D(0)}(s)}{(1-\pi _{D(0)}(s))}\left[ 
\begin{array}{c}
-(1-\pi _{D(0)}(s))(E[Y(1)|C,S=s]-E[Y(1)|AT,S=s]) \\ 
+(1-\pi _{D(1)}(s))(E[Y(0)|C,S=s]-E[Y(0)|NT,S=s]) \\ 
+(1-\pi _{D(0)}(s))(E[Y(1)-Y(0)|C,S=s]-\beta )
\end{array}
\right] ^{2}
\end{array}
\right] .\label{eq:Opt_pi_expression}
\end{align}

Consider a hypothetical RCT that satisfies Assumptions \ref{ass:1} and \ref{ass:3} with $\tau(s)=0$ for all $s \in \mathcal{S}$. The constant treatment assignment probability that minimizes the asymptotic variance of any of the LATE IV estimators is
\begin{equation}
\pi _{A}^{\ast }~\equiv~ \left( 1+\sqrt{\frac{\sum_{s\in \mathcal{S}}p( s) \Pi _{2}(s)}{\sum_{\tilde{s}\in \mathcal{S}}p( \tilde{s} ) \Pi _{1}(\tilde{s})}}\right) ^{-1}. \label{eq:Optimal_Pi}
\end{equation}

Furthermore, \eqref{eq:Optimal_Pi_s} and \eqref{eq:Optimal_Pi} can be consistently estimated from the data in the pilot RCT. To this end, we propose plug-in estimators of the terms on the right-hand side of \eqref{eq:Optimal_Pi_s} or \eqref{eq:Optimal_Pi}, based on Theorems \ref{thm:plim_SAT} and \ref{thm:primitiveEstimation} in the appendix.
\end{theorem}

\section{Monte Carlo simulations}\label{sec:MC}

In this section, we explore the various inference methods proposed in the paper via Monte Carlo simulations. This exercise has multiple goals. First, we hope to show that we can accurately estimate the LATE and the asymptotic variance of the various LATE estimators. Second, we seek to confirm the accuracy of our asymptotic normal approximation by showing that the empirical coverage rate of our proposed confidence intervals is close to our desired coverage level. Third, we will use these simulations to explore how the asymptotic variance of the various estimators change as we vary the parameters of the DGP. In particular, we will numerically explore the theoretical predictions in Section \ref{sec:optim} regarding the optimal RCT design.

We consider four simulation designs, which we describe in Table \ref{tab:designs}. We consider CAR based on either four or eight strata, i.e., $\mathcal{S} \in \{4,8\}$. These were generated by intersecting the support of $L$ binary covariates with $L \in \{2,3\}$.
Design 1 is our baseline design, where we consider an RCT with four strata, and with treatment assignment probabilities and strata-specific LATE constant across strata. Design 2 is similar to our baseline design, but we divide each stratum into two, effectively increasing the number of strata from four to eight.
Design 3 is similar to our baseline design, but we allow the strata-specific LATE to vary across strata. Finally, Design 4 is similar to Design 3, but we also enable treatment assignment probabilities to vary by strata.

\begin{table}[ht]
 \centering
 \begin{tabular}{cccc}
 \hline\hline
 Design & $|\mathcal{S}|$ & $(\pi_A(s):s \in \mathcal{S})$ & $(\beta(s):s \in \mathcal{S})$ \\\hline
 1 & 4 & constant & constant\\ 
 2 & 8 & constant & constant\\ 
 3 & 4 & constant & not constant\\ 
 4 & 4 & not constant & not constant\\ 
 \hline\hline
 \end{tabular}
 \caption{\small Description of the simulation designs.}
 \label{tab:designs}
\end{table}

For each simulation design, we show the average results of computing the estimators in $5,000$ independent replications of the RCT with sample size $n=200$. For each RCT, we consider that the researcher assigns the treatment according to the following CAR mechanisms: SRS (as described in Example \ref{ex:SRS}), SBR (as described in Example \ref{ex:SBR}), and the minimization methods in \cite{pocock/simon:1975} (PSM) and \cite{hu/hu:2012} (HHM) (as described in Example \ref{ex:minimization}).
The implementation of the minimization methods requires several tuning parameters. For Designs 1, 3, and 4, (i.e., designs with two binary covariates), we follow the guidance in \citet[Section 4.1]{hu/hu:2012}, and we set the biasing probability to $\lambda=0.85$, uniform PSM weights (i.e., $w_{m,1}=w_{m,2} = 0.5$, and $w_{o}=w_{s}=0$ by definition), and HHM weights to $(w_{o}, w_{m, 1}, w_{m, 2}, w_{s}) = (0.3, 0.1, 0.1, 0.5)$. For Design 2, we keep the biasing probability at $\lambda=0.85$ and uniform PSM weights (i.e., $w_{m,1}=w_{m,2}=w_{m,3} = 1/3$, and $w_{o}=w_{s}=0$ by definition), and we change the HHM weights to $(w_{o}, w_{m, 1}, w_{m, 2},w_{m,3}, w_{s}) = (0.04, {1}/{60}, {1}/{60}, {1}/{60}, 0.91)$. We note that our HHM weights for Designs 1, 3, and 4 satisfy Conditions (A)-(B) in \citet[Theorem 3.1]{hu/hu:2012} and those for Design 2 satisfy Condition (C) in \citet[Theorem 3.2]{hu/hu:2012}. Thus, our implementation of HHM satisfies Assumption \ref{ass:3} with $\tau(s) = 0$ for all $s \in \mathcal{S}$. In turn, our implementation of PSM satisfies Assumption \ref{ass:2}, but it is not currently known whether it satisfies Assumption \ref{ass:3}(b). For this reason, we restrict attention to the SAT IV regression when presenting results for PSM.

\subsection{Design 1}
% This is called case 1' in Mengsi's excel spreadsheet

We begin our description of Design 1 by specifying the features known to the researcher implementing the RCT. This researcher knows that there are two binary covariates (i.e., $L=2$), resulting in four strata (i.e., $|\mathcal{S}|=4$), and that RCT participants are assigned into treatment or control by using SRS, SBR, PSM, or HHM with constant treatment assignment probabilities equal to $\pi_A(s)=1/2$ for all $s \in \mathcal{S}$.

We next describe those aspects of Design 1 unknown to the researcher implementing the RCT. First, we make all strata to be equally likely, i.e., $p(s)=1/4$ for all $s \in \mathcal{S}$. Second, RCT participants belong to each type according to i.i.d.\ draws of a multinomial distribution with
\begin{align}
 P(C|S=s) = 0.7,~~P(AT|S=s) = 0.15,~\text{and}~~P(NT|S=s) = 0.15.
 \label{eq:P_types_MC}
\end{align}
Third, conditional on their type and their strata, the RCT participants have potential outcomes that are i.i.d.\ drawn according to a normal distribution with
\begin{align}
 (E(Y(0)|C,S=s):s\in \mathcal{S}) &~=~ [0,0,0,0]\notag\\
 (E(Y(1)|C,S=s):s\in \mathcal{S}) &~=~ [1,1,1,1]\notag\\
 (E(Y(0)|NT,S=s):s\in \mathcal{S}) &~=~ [-0.6, -0.4, -0.2, 0]\notag\\
 (E(Y(1)|AT,S=s):s\in \mathcal{S}) &~=~ [2, 2.2, 2.4, 2.6]
 \label{eq:E_types_MC}
\end{align}
and
\begin{align}
 V(Y(0)|C,S=s)=0.5, ~V(Y(1)|C,S=s)=3,~ V(Y(0)|NT,S=s) =V(Y(1)|AT,S=s)=1.
 \label{eq:V_types_MC}
\end{align}
Note that \eqref{eq:E_types_MC} implies that the strata-specific LATE is $\beta(s) =E(Y(1)-Y(0)|C,S=s) = 1$ and, thus, the LATE is $\beta = 1$.\footnote{As discussed in earlier sections, the ITT is equal to the LATE multiplied by $P(C)=0.7$.} Also, while the distribution of treatment effects for compliers are homogeneous across strata, the corresponding distribution for always takers and never takers is not.\footnote{We have also conducted simulations in a case in which there is homogeneity across strata for all types. These results are omitted for brevity and available upon request.}

The simulation results for Design 1 are provided in Table \ref{tab:Design1}. The results reveal that all the proposed estimators are very close to the true LATE (equal to one). Across our simulations, the bias is almost zero and the squared error of estimation is close to the asymptotic variance. We now describe the behavior of the asymptotic variance of these estimators. As shown by our formal results, the asymptotic variance is constant under SBR and HHM (i.e., $\tau(s)=0$), and the asymptotic variance of the SAT IV estimator is the same across CAR mechanisms. The homogeneity of the distribution of treatment effects for the compliers across strata implies that $V_A^{\mathrm{sfe}}=0$ and, thus, the asymptotic variance of the SFE IV estimator is the same across CAR mechanisms. Finally, the heterogeneity of the distribution of potential outcomes for never takers and always takers causes that $V_A^{\mathrm{2s}}>0$.\footnote{This is the main difference with the simulations in which there is homogeneity across strata for all types. In that case, $V_A^{\mathrm{2s}}=0$ and thus all estimators under consideration have the same asymptotic variance under SRS and SBR.} This explains why the asymptotic variance of the 2S IV estimator is higher under SRS than under SBR and HHM. It is relevant to note that, for all estimators, the asymptotic variance under SBR and HHM is smaller or equal than that under SRS, as demonstrated in Section \ref{sec:RCTChoice_1}. In all cases, the average value of our proposed asymptotic variance estimate is very close to the true asymptotic variance. From these results and the fact that the normal asymptotic approximation is accurate, it follows that the empirical coverage rate of the true LATE is very close to the desired coverage rate of 95\%.

The simulation results in Table \ref{tab:Design1} were obtained using a treatment assignment probability of $\pi_A(s)=1/2$ for all $s \in \mathcal{S}$. The most efficient LATE estimator in Table \ref{tab:Design1} is the SAT IV estimator, which has an asymptotic variance equal to 14.5306. Section \ref{sec:pi_opt} shows us to improve on the efficiency of this estimator by optimizing the treatment assignment probabilities. According to Theorem \ref{thm:opt_pi}, the optimal treatment assignment probability vector is $(\pi_A^*(s):s \in\mathcal{S}) = (0.6362, 0.6339, 0.6303, 0.6256)$ and the optimal constant treatment assignment probability is $\pi_A^* = 0.6314$. For the SAT IV estimator, the former yields an asymptotic variance of 13.5913 and the later yields an asymptotic variance of 13.5922. To put this into perspective, we can compute the efficiency gain in terms of ``effective sample size''. This calculation reveals that using $\pi_A(s)=1/2$ for all $s \in \mathcal{S}$ is approximately equivalent to discarding 6.458\% of the sample relative to the optimal constant treatment assignment probability and 6.465\% of the sample relative to the optimal treatment assignment probability vector.
%\footnote{This is because $15.5408/1,059.3\approx 14.6714/1,000$ and $15.5408/1,059.4 \approx 14.6696/1,000$.} 

\begin{table}[ht]
 \centering
 \begin{tabular}{ccccccc}
 \hline \hline
CAR & Estimator & Avg.\ est.& Avg.\ SE &  Avg.\ AVar.\ est. & AVar.  & Coverage \\
\hline
 & $\hat{\beta}_\mathrm{sat}$ 
& 0.9981 & 14.3750 & 14.4206 & 14.5306 & 0.9478 \\
SBR &$\hat{\beta}_\mathrm{sfe}$ 
& 0.9981 & 14.3751 & 14.4206 & 14.5306 & 0.9478 \\
 & $\hat{\beta}_\mathrm{2s}$ 
 & 0.9982 & 14.3771 & 14.4206 & 14.5306	& 0.9472 \\
 \hline
 & $\hat{\beta}_\mathrm{sat}$ 
 & 1.0023 & 14.2152	& 14.6968 & 14.5306	& 0.9552 \\
SRS & $\hat{\beta}_\mathrm{sfe}$ 
& 1.0023 & 14.2020 & 14.7172 & 14.5306 & 0.9562 \\
 & $\hat{\beta}_\mathrm{2s}$ 
 & 1.0020 & 14.0122	& 14.9885 & 14.5673	& 0.9602 \\
 \hline
 & $\hat{\beta}_\mathrm{sat}$ 
 & 0.9947 & 14.1154	& 14.6231 & 14.5306	& 0.9548 \\
 PSM &$\hat{\beta}_\mathrm{sfe}$ 
 & 0.9947 & 14.1169 & U & U & U \\
 & $\hat{\beta}_\mathrm{2s}$ 
 & 0.9950 & 14.0385 & U & U & U \\
 \hline
 & $\hat{\beta}_\mathrm{sat}$ 
 & 0.9990 & 14.9270	& 14.5686 & 14.5306	& 0.9472 \\
 HHM &$\hat{\beta}_\mathrm{sfe}$ 
 & 0.9990 & 14.9275	& 14.5686 & 14.5306	& 0.9472 \\
 & $\hat{\beta}_\mathrm{2s}$ 
 & 0.9990 & 14.9270	& 14.5686 & 14.5306	& 0.9486 \\
 \hline
 \hline
 \end{tabular}
 \caption{\small Simulation results over $5,000$ replications of Design 1 with sample size $n=200$. The columns labels are as follows: ``CAR'' denotes the CAR treatment assignment mechanism, which can be SBR, SRS, PSM or HHM, ``Estimator'' denotes the LATE estimator under consideration, which can be $\hat{\beta}_\mathrm{sat}$, $\hat{\beta}_\mathrm{sfe}$, or $\hat{\beta}_\mathrm{2s}$, ``Avg.\ est.'' denotes average LATE estimate over simulations, ``Avg.\ SE'' denotes the average squared error of estimation over simulations scaled by $n$, ``Avg.\ AVar.\ est.'' denotes the average asymptotic variance estimate over simulations,  ``AVar.'' denotes the asymptotic variance of the LATE estimator, and ``Coverage'' denotes the coverage rate of the true LATE over the simulations with desired coverage rate of $1-\alpha =95\%$. Finally, we use ``U'' to indicate that the asymptotic properties of PSM are unknown except in the case of the SAT regression.}
 \label{tab:Design1}
\end{table}

\subsection{Design 2}
% This is called case 1' with |S|=8 in Mengsi's excel spreadsheet

In Section \ref{sec:strataChoice}, we showed that the asymptotic variance of the proposed LATE estimators does not increase if the strata of the RCT becomes finer and all else remains equal. We explore this prediction in Design 2, where we are splitting each stratum in Design 1 into two equally-sized strata. This produces a DGP with eight strata (i.e., $|\mathcal{S}|=8$). As in Design 1, RCT participants are still assigned into treatment or control by using SRS, SBR, PSM or HHM with constant treatment assignment probabilities equal to $\pi_A(s)=1/2$.

We now describe the features of Design 2 that are unknown to the researcher. All strata remain equally likely, i.e., $p(s)=1/8$ for all $s \in \mathcal{S}$, and RCT participants are assigned into types according to i.i.d.\ draws of a multinomial distribution with probabilities as in \eqref{eq:P_types_MC}. Conditional on their type and their strata, the RCT participants have potential outcomes that are i.i.d.\ drawn according to a normal distribution with
\begin{align}
 (E(Y(0)|C,S=s):s\in \mathcal{S}) &~=~ [-0.5,0.5,-0.5,0.5,-0.5,0.5,-0.5,0.5]\notag\\
 (E(Y(1)|C,S=s):s\in \mathcal{S}) &~=~ [0.5,1.5,0.5,1.5,0.5,1.5,0.5,1.5]\notag\\
 (E(Y(0)|NT,S=s):s\in \mathcal{S}) &~=~ [-1.1, -0.1, -0.9, 0.1, -0.7, 0.3, -0.5, 0.5]\notag\\
 (E(Y(1)|AT,S=s):s\in \mathcal{S}) &~=~ [1.5, 2.5, 1.7, 2.7, 1.9, 2.9, 2.1, 3.1]
 \label{eq:E_types_MC2}
\end{align}
and
\begin{align}
 V(Y(0)|C,S=s)=0.25, ~V(Y(1)|C,S=s)=2.75,~ V(Y(0)|NT,S=s) =V(Y(1)|AT,S=s)=0.75.
 \label{eq:V_types_MC2}
\end{align}
The parameters of the DGP in Design 2 are the result from splitting each stratum in Design 1 into two equally-sized strata. In particular, the law of iterated expectations and the law of total variance imply that \eqref{eq:E_types_MC2} and \eqref{eq:V_types_MC2} are compatible with \eqref{eq:E_types_MC} and \eqref{eq:V_types_MC}. Note also that \eqref{eq:E_types_MC2} implies that the strata-specific LATE is $\beta(s) =E(Y(1)-Y(0)|C,S=s) = 1$ and, thus, the LATE is $\beta = 1$. 

The simulation results for Design 2 are provided in Table \ref{tab:Design2}. These results are qualitatively similar to those in Table \ref{tab:Design1}. The estimators of the LATE and the various asymptotic variances are very accurate, and the empirical coverage rate is very close to the desired coverage level.
The main noticeable difference between the tables is that the asymptotic variance of each estimator in Design 2 is smaller than the corresponding one in Design 1, as expected from Theorem \ref{thm:coarser_moreVar}. For example, the asymptotic variance of the SAT IV estimator is 14.5306 in the RCT with four strata and becomes 12.4898 in the RCT with eight strata. In terms of ``effective sample size'', this means that using four strata instead of eight is approximately equivalent to discarding 14.04\% of the sample relative to using eight strata.

We can also consider optimizing the treatment assignment probability in these simulations. Focusing on the SAT IV estimator, the optimal treatment assignment probability vector yields an asymptotic variance of 11.366 and the optimal constant treatment assignment probability yields an asymptotic variance of 11.3678. In terms of ``effective sample size'', this means that using $\pi_A(s)=1/2$ for all $s \in \mathcal{S}$ is approximately equivalent to discarding 8.984\% of the sample relative to the optimal constant treatment assignment probability and 8.998\% of the sample relative to the optimal treatment assignment probability vector.

\begin{table}[ht]
\centering
\begin{tabular}{ccccccc}
\hline\hline
CAR & Estimator & Avg.\ est.& Avg.\ SE & Avg.\ AVar.\ est. & AVar. & Coverage \\
\hline
& $\hat{\beta}_\mathrm{sat}$ 
& 0.9987 & 12.0676	& 12.0972 & 12.4898	& 0.9498 \\
SBR &$\hat{\beta}_\mathrm{sfe}$ 
& 0.9987 & 12.0660 & 12.0972 & 12.4898 & 0.9498 \\
& $\hat{\beta}_\mathrm{2s}$ 
& 0.9986 & 12.0766	& 12.0972 & 12.4898	& 0.9484 \\
\hline
& $\hat{\beta}_\mathrm{sat}$ 
& 1.0012 & 12.6288 & 12.5982 & 12.4898 & 0.9490 \\
SRS & $\hat{\beta}_\mathrm{sfe}$ 
& 1.0015 & 12.5865 & 12.6958 & 12.4898 & 0.9480 \\
& $\hat{\beta}_\mathrm{2s}$ 
& 0.9994 & 14.2107	& 15.3848 & 14.5673	& 0.9586 \\
\hline
& $\hat{\beta}_\mathrm{sat}$ 
& 1.0002 & 12.5345	& 12.4266 & 12.4898	& 0.9462 \\
PSM &$\hat{\beta}_\mathrm{sfe}$ 
& 1.0002 & 12.5282 & U & U & U \\
& $\hat{\beta}_\mathrm{2s}$ 
& 1.0000 & 12.3183 & U & U & U \\\hline
& $\hat{\beta}_\mathrm{sat}$ 
& 0.9960 & 12.3157	& 12.2443 & 12.4898	& 0.9446 \\
HHM &$\hat{\beta}_\mathrm{sfe}$ 
& 0.9959 & 12.3127 & 12.2443 & 12.4898 & 0.9440 \\
& $\hat{\beta}_\mathrm{2s}$ 
& 0.9963 & 12.3593	& 12.2443 & 12.4898	& 0.9448 \\
\hline \hline
\end{tabular}
\caption{\small Simulation results over $5,000$ replications of Design 2 with sample size $n=200$. The columns labels are as follows: ``CAR'' denotes the CAR treatment assignment mechanism, which can be SBR, SRS, PSM or HHM, ``Estimator'' denotes the LATE estimator under consideration, which can be $\hat{\beta}_\mathrm{sat}$, $\hat{\beta}_\mathrm{sfe}$, or $\hat{\beta}_\mathrm{2s}$, ``Avg.\ est.'' denotes average LATE estimate over simulations, ``Avg.\ SE'' denotes the average squared error of estimation over simulations scaled by $n$, ``Avg.\ AVar.\ est.'' denotes the average asymptotic variance estimate over simulations, ``AVar.'' denotes the asymptotic variance of the LATE estimator, and ``Coverage'' denotes the coverage rate of the true LATE over the simulations with desired coverage rate of $1-\alpha =95\%$. Finally, we use ``U'' to indicate that the asymptotic properties of PSM are unknown except in the case of the SAT regression.}
\label{tab:Design2}
\end{table}

\subsection{Design 3}
% This is called case 2' with |S|=4 in Mengsi's excel spreadsheet

Relative to Design 1, Design 3 introduces heterogeneity of the strata-specific treatment effects. In terms of the observable features to the researcher, however, Design 3 is identical to Design 1. There are four strata (i.e., $|\mathcal{S}|=4$) and RCT participants are assigned into treatment or control by using SRS, SBR, PSM or HHM with constant treatment assignment probabilities equal to $\pi_A(s)=0.7$ for all $s \in \mathcal{S}$. 

We now describe the features of Design 3 that are unobserved to the researcher. As in Design 1, all strata are equally likely, i.e., $p(s)=1/4$ for all $s \in \mathcal{S}$, and RCT participants are assigned into types according to i.i.d.\ draws of a multinomial distribution with probabilities as in \eqref{eq:P_types_MC}. Conditional on their type and their strata, RCT participants have potential outcomes that are i.i.d.\ drawn according to a normal distribution. The conditional variance of this distribution is as in \eqref{eq:V_types_MC}. The conditional mean of this distribution for never takers and always takes is as in \eqref{eq:E_types_MC}, while that for compliers is as follows
\begin{align}
 (E(Y(0)|C,S=s):s\in \mathcal{S}) &~=~ [0, 0.2, 0.4, 0.6]\notag\\
 (E(Y(1)|C,S=s):s\in \mathcal{S}) &~=~ [-1, 1.2, 1.4, 3.6].
 \label{eq:E_types_MC3}
\end{align}
Note also that \eqref{eq:E_types_MC3} implies that the strata-specific LATE is $(\beta(s):s\in \mathcal{S}) =(E(Y(1)-Y(0)|C,S=s) :s\in \mathcal{S}) = [-1,1,1,3]$. When this information is combined with $P(C|S=s)=0.7$ and $p(s)=1/4$ for all $s\in \mathcal{S}$, it follows that the LATE is (still) $\beta = 1$. 

The simulation results for Design 3 are provided in Table \ref{tab:Design3}. Most of the results are qualitatively similar to those in previous simulations. The estimators of the LATE and the various asymptotic variances are very accurate, and the empirical coverage rate is very close to the desired coverage level. As predicted by our results, all LATE IV estimators have the same asymptotic variance when SBR or HHM is used. The main difference with respect to the results in Design 1 appears when SRS is used.  In Design 1, the SAT and SFE IV estimators have the same asymptotic variance under SRS. In Design 3, the SFE IV estimator has larger asymptotic variance than the SAT IV estimator when SRS is used. This is compatible with Theorem \ref{thm:AsyDist_SFE}, as the heterogeneity of the strata-specific treatment effect in Design 3 implies that $V_{A}^{\mathrm{sfe}}>0$.

Finally, we could also explore the effects of optimizing the treatment assignment probability. We decided to avoid this here and in the next design for the sake of brevity.
%We now explore the effect of optimizing the treatment assignment probability in Design 3. We focus on the SAT IV estimator, which has an asymptotic variance of 19.4935 when the treatment assignment probability is $\pi_A(s)=1/2$ for all $s \in \mathcal{S}$. In contrast, the optimal treatment assignment probability vector yields an asymptotic variance of 19.3159 and the optimal constant treatment assignment probability yields an asymptotic variance of 19.3257. In terms of ``effective sample size'', this means that using $\pi_A(s)=1/2$ for all $s \in \mathcal{S}$ is approximately equivalent to discarding 0.92\% of the sample relative to the optimal constant treatment assignment probability and 0.87\% of the sample relative to the optimal treatment assignment probability vector.

\begin{table}[ht]
\centering
\begin{tabular}{ccccccc}
\hline\hline
CAR & Estimator & Avg.\ est.& Avg.\ SE & Avg.\ AVar.\ est. & AVar. & Coverage \\
\hline
& $\hat{\beta}_\mathrm{sat}$ 
& 0.9943 & 16.9138 & 16.7226 & 16.5909 & 0.9482 \\
SBR &$\hat{\beta}_\mathrm{sfe}$ 
& 0.9942 & 16.8798 & 16.7226 & 16.5909 & 0.9488 \\
& $\hat{\beta}_\mathrm{2s}$ 
& 0.9943 & 16.9976	& 16.7226 & 16.5909	& 0.9486 \\
\hline
& $\hat{\beta}_\mathrm{sat}$ 
& 0.9978 & 17.5020	& 17.0201 & 16.5909	& 0.9462 \\
SRS & $\hat{\beta}_\mathrm{sfe}$ 
& 0.9995 & 19.0459 & 19.1864 & 18.1147 & 0.9506 \\
& $\hat{\beta}_\mathrm{2s}$ 
& 0.9951 & 20.1232	& 19.9878 & 19.1584	& 0.9500 \\
\hline
& $\hat{\beta}_\mathrm{sat}$ 
& 0.9961 & 16.6584	& 16.7337 & 16.5909	& 0.9534 \\
PSM & $\hat{\beta}_\mathrm{sfe}$ 
& 0.9964 & 16.7238 & U & U & U \\
& $\hat{\beta}_\mathrm{2s}$ 
& 0.9963 & 16.6057 & U & U & U \\
\hline
& $\hat{\beta}_\mathrm{sat}$ 
& 0.9929 & 16.4258 & 16.7655 & 16.5909 & 0.9530 \\
HHM & $\hat{\beta}_\mathrm{sfe}$ 
& 0.9931 & 16.4938 & 16.7655 & 16.5909 & 0.9530 \\
& $\hat{\beta}_\mathrm{2s}$ 
& 0.9925 & 16.6222 & 16.7655 & 16.5909 & 0.9508 \\
\hline \hline
\end{tabular}
\caption{\small Simulation results over $5,000$ replications of Design 3 with sample size $n=200$. The columns labels are as follows: ``CAR'' denotes the CAR treatment assignment mechanism, which can be SBR, SRS, PSM or HHM, ``Estimator'' denotes the LATE estimator under consideration, which can be $\hat{\beta}_\mathrm{sat}$, $\hat{\beta}_\mathrm{sfe}$, or $\hat{\beta}_\mathrm{2s}$, ``Avg.\ est.'' denotes average LATE estimate over simulations, ``Avg.\ SE'' denotes the average squared error of estimation over simulations scaled by $n$, ``Avg.\ AVar.\ est.'' denotes the average asymptotic variance estimate over simulations,  ``AVar.'' denotes the asymptotic variance of the LATE estimator, and ``Coverage'' denotes the coverage rate of the true LATE over the simulations with desired coverage rate of $1-\alpha =95\%$. Finally, we use ``U'' to indicate that the asymptotic properties of PSM are unknown except in the case of the SAT regression.}
\label{tab:Design3}
\end{table}

\subsection{Design 4}
% This is called case 4 with |S|=4 in Mengsi's excel spreadsheet

Design 4 considers a researcher who implements CAR with an heterogeneous treatment assignment probability. The RCT in Design 4 has four strata, just like in Design 1 and 3. Unlike all previous designs, RCT participants are assigned into treatment or control by using SRS, SBR, PSM or HHM with treatment assignment probability vector equal to
\begin{align}
  (\pi_A(s):s\in \mathcal{S}) &~=~ [0.3,0.7,0.6,0.8].
  \label{eq:pi_vector_MC4}
\end{align}

We now describe the features of Design 4 that are unobserved to the researcher. First, all strata are equally likely, i.e., $p(s)=1/4$ for all $s \in \mathcal{S}$. Second, RCT participants are assigned into types according to i.i.d.\ draws of a multinomial distribution with probabilities 
\begin{align}
  (P(AT|S=s):s\in \mathcal{S}) &~=~ 
  [0.15, 0.15, 0.1, 0.15]\notag\\
  (P(NT|S=s):s\in \mathcal{S}) &~=~ [0.25,0.15,0.2,0.05]\notag\\
  (P(C|S=s):s\in \mathcal{S}) &~=~ [0.6,0.7,0.7,0.8].
  \label{eq:P_types_MC4}
\end{align}
Conditional on their type and strata, RCT participants have potential outcomes that are i.i.d.\ drawn according to a normal distribution with conditional variance as in \eqref{eq:V_types_MC}, and conditional mean given by
\begin{align}
 (E(Y(0)|C,S=s):s\in \mathcal{S}) &~=~ [0, 0.2, 0.4, 0.6]\notag\\
 (E(Y(1)|C,S=s):s\in \mathcal{S}) &~=~ [-5.6, 3,    4.8, 2]\notag\\
 (E(Y(0)|NT,S=s):s\in \mathcal{S}) &~=~ [-0.6, -0.4, -0.2, 0]\notag\\
 (E(Y(1)|AT,S=s):s\in \mathcal{S}) &~=~ [2, 2.2, 2.4, 2.6]
 \label{eq:E_types_MC4}
\end{align}
These parameters in \eqref{eq:E_types_MC4} imply that the strata-specific LATE is $(\beta(s):s\in \mathcal{S}) =(E(Y(1)-Y(0)|C,S=s) :s\in \mathcal{S}) = [-5.6,2.8,4.4,1.4]$. When this information is combined with $(P(C|S=s):s\in \mathcal{S})$ in \eqref{eq:P_types_MC4} and $p(s)=1/4$ for all $s\in \mathcal{S}$, we can verify that the LATE is (still) $\beta = 1$.

The simulation results for Design 4 are provided in Table \ref{tab:Design4}. Given the heterogeneity of the treatment assignment probability in \eqref{eq:pi_vector_MC4}, neither the SFE IV estimator nor the 2S IV estimator are guaranteed to be consistent for the LATE. In fact, under our current conditions, Theorem \ref{thm:plim_SFE} in the appendix implies that $\hat{\beta}_\mathrm{sfe} \overset{p}{\to} 1.0974 \neq 1=\beta $ and Theorem \ref{thm:plim_2SR} in the appendix implies that $\hat{\beta}_\mathrm{2s} \overset{p}{\to} 2.0422 \neq 1=\beta$. Note that this occurs for all CAR mechanisms. Since these estimators are not consistent for the parameter of interest, there is no point in discussing their asymptotic variance or their coverage. In contrast, the SAT IV estimator remains consistent in this scenario. In addition, the estimator of the asymptotic variance of the SAT IV estimator is accurate and the empirical coverage level is close to the desired coverage level. These simulations confirm that the inference based on the SAT IV estimator is valid under more general conditions than that based on the other two IV estimators.

\begin{table}[ht]
\centering
\begin{tabular}{ccccccc}
\hline\hline
CAR & Estimator & Avg.\ est.& Avg.\ SE & Avg.\ AVar.\ est. & AVar. & Coverage \\\hline
& $\hat{\beta}_\mathrm{sat}$    
& 0.9999 & 47.6372	& 46.4695 & 47.1206	& 0.9428 \\
SBR &$\hat{\beta}_\mathrm{sfe}$  
& 1.0948 & 53.6611 & NR & NR & NR \\
& $\hat{\beta}_\mathrm{2s}$    
& 2.0388 & 237.6425 & NR & NR  & NR \\
\hline
& $\hat{\beta}_\mathrm{sat}$   
& 1.0145 & 48.7670	& 47.5906 & 47.1206	& 0.9366 \\
SRS & $\hat{\beta}_\mathrm{sfe}$ 
& 1.1114 & 64.5130 & NR & NR & NR \\
& $\hat{\beta}_\mathrm{2s}$     
& 2.0456 & 249.6926 & NR & NR & NR \\
\hline
& $\hat{\beta}_\mathrm{sat}$   
& 0.9977 & 46.8252	& 45.4723 & 47.1206	& 0.9422 \\
PSM & $\hat{\beta}_\mathrm{sfe}$ 
& 1.0580 & 52.2365 & U & U & U \\
& $\hat{\beta}_\mathrm{2s}$  
& 1.9421 & 206.1005 & U & U & U \\
\hline
& $\hat{\beta}_\mathrm{sat}$   
& 0.9964 & 47.5980	& 44.9222 & 47.1206	& 0.9372 \\
HHM & $\hat{\beta}_\mathrm{sfe}$ 
& 1.0428 & 51.6829 & NR & NR & NR \\
& $\hat{\beta}_\mathrm{2s}$    
& 1.9649 & 209.4983 & NR & NR & NR \\
\hline\hline
\end{tabular}
\caption{\small Simulation results over $5,000$ replications of Design 4 with sample size $n=200$. The columns labels are as follows: ``CAR'' denotes the CAR treatment assignment mechanism, which can be SBR, SRS, PSM or HHM, ``Estimator'' denotes the LATE estimator under consideration, which can be $\hat{\beta}_\mathrm{sat}$, $\hat{\beta}_\mathrm{sfe}$, or $\hat{\beta}_\mathrm{2s}$, ``Avg.\ est.'' denotes average LATE estimate over simulations, ``Avg.\ SE'' denotes the average squared error of estimation over simulations scaled by $n$, 
``Avg.\ AVar.\ est.'' denotes the average asymptotic variance estimate over simulations, ``AVar.'' denotes the asymptotic variance of the LATE estimator, and ``Coverage'' denotes the coverage rate of the true LATE over the simulations with desired coverage rate of $1-\alpha =95\%$. ``NR'' indicates that the corresponding asymptotic variance is not relevant as the conditions for the consistency of the corresponding SFE and 2S IV estimators are not satisfied. Finally, we use ``U'' to indicate that the asymptotic properties of PSM are unknown except in the case of the SAT regression.}
 \label{tab:Design4}
\end{table}

\section{Empirical Illustration}\label{sec:application}

In this section, we consider an empirical illustration based on \cite{dupas/karlan/robinson/ubfal:2018}. The authors use an RCT to investigate the economic impact of expanding access to basic bank accounts in several countries: Malawi, Uganda, and Chile.\footnote{The data are publicly available at \url{https://www.aeaweb.org/articles?id=10.1257/app.20160597}.} These countries differ significantly in their level of development and banking access, with Uganda being the intermediate country in both of these respects. For our illustration, we focus on their RCT conducted in Uganda.

We now briefly summarize the empirical setting; see \cite{dupas/karlan/robinson/ubfal:2018} for a more detailed description.
Bank accounts are important to daily economic life, but the rate of opening a bank account in developing countries is relatively low compared to developed countries. During the time of the RCT, the authors report that 74\% of households in Uganda were unbanked. Among the many benefits of having access to bank accounts, the authors focus on its arguably most primary function: safekeeping. 

\cite{dupas/karlan/robinson/ubfal:2018} selected a random sample of 2,159 Ugandan households who did not have a bank account in 2011. These households were assigned into a treatment or a control group. Treated households were given a voucher for a free savings account at their nearest bank branch, without any fees for the RCT duration, along with assistance to complete the necessary paperwork required to open this account. Households in the control group were not provided with these vouchers. We use the binary variable $A_i \in \{0,1\}$ to indicate if household $i$ was given the voucher for the free savings account. The treatment assignment was stratified by gender, occupation,
%\footnote{The occupation categories were classified as employee, self-employed: vendor, business owner, trader; or farmer: including animal rearing, and housewife or unemployed.} 
and bank branch,
%\footnote{There were three branches}
which generated 41 strata, i.e., $s\in \mathcal{S}=\{1,\cdots,41\}$.
Within each stratum, households were randomly assigned to treatment or control using SBR with $\pi_A(s)=1/2$ for all $s\in\mathcal{S}$. As a result, 1,080 households were selected to receive the treatment, while the remaining 1,079 were placed in the control group. The households in the sample were re-interviewed thrice during 2012-2013. In these follow-up surveys, the authors collected information about their saving behavior and several other related outcomes.

This RCT featured imperfect compliance. While none of the 1,079 households in the control group accessed a free savings account, not every one of the 1,080 households in the treatment group opened their free savings account. \citet[Table 1]{dupas/karlan/robinson/ubfal:2018} reveals that, out of the 1,080 treated households, only 54\% actually opened up the bank account, while only 42\% made at least one deposit during the RCT. As in \cite{dupas/karlan/robinson/ubfal:2018}, we are interested in the effect of opening and {\it using} the free savings account. We define the binary variable $D_i = D_i(A_i) \in \{0,1\}$ to indicate if household $i$ has opened and used the free savings account in this RCT.\footnote{Following the paper, we consider an account used if the owner has made at least one deposit during the RCT. We have also explored alternative definitions of usage, and we obtained qualitatively similar findings.}

Given this setup, we use our IV regressions to estimate and conduct inference on the LATE for several outcomes of interest.\footnote{In this RCT, $D_i(0)=0$ and so there are no always takers. As a consequence, the LATE coincides with the TOT, i.e., $E[Y(1)-Y(0)|D=1]$.} \cite{dupas/karlan/robinson/ubfal:2018} collect information on two types of outcomes: saving stocks and downstream outcomes. For brevity, we select three outcomes: savings in formal financial institutions, savings in cash at home or in a secret place, and expenditures in the last month. These variables were reported by the households in the final survey of the RCT and measured in 2010 US dollars.

\begin{table}[ht]
\centering
\begin{tabular}{c|ccc}
\hline
\hline
\multirow{2}{*}{Estimator} & {Savings in formal}
& Savings in cash at  & Expenditures in  \\
& fin.\ institutions & home or secret place & last month \\
\hline
$\hat{\beta}_\mathrm{sat}$  & 17.572*** & -7.323 & -2.427 \\
s.e.\ if $\tau(s)=0$ & 4.061 & 4.824 & 3.888 \\
s.e.\ if $\tau(s)=1$  & 4.061 & 4.824 & 3.888 \\
\hline
$\hat{\beta}_\mathrm{sfe}$   & 17.544*** & -7.371 & -2.439 \\
s.e.\ if $\tau(s)=0$ & 4.061 & 4.824 & 3.888 \\
s.e.\ if $\tau(s)=1$  & 4.063 & 4.825 & 3.898 \\
\hline
$\hat{\beta}_\mathrm{2s}$    & 18.000*** & -6.375 & -1.402 \\
s.e.\ if $\tau(s)=0$ & 4.061 & 4.824 & 3.888 \\
s.e.\ if $\tau(s)=1$  & 4.166 & 5.106 & 4.278 \\
\hline\hline
\multicolumn{4}{c}{Results copied from \cite{dupas/karlan/robinson/ubfal:2018}}\\
\hline
\hline
$\hat{\beta}$  & 20.117*** & -6.250 & 0.622 \\
s.e. & 2.795 & 3.503 & 2.72 \\
\hline
\hline
\end{tabular}
\caption{\small Results of the IV regressions based on data from \cite{dupas/karlan/robinson/ubfal:2018}. For $t\in\{0,1\}$, ``s.e.\ if $\tau(s)=t$'' denotes the estimated standard error of the LATE estimator (divided by $\sqrt{n}$) under the assumption that the RCT uses $\tau(s)=t$ for all $s \in \mathcal{S}$. Note that the RCT used SBR, and so $\tau(s)=0$ for all $s \in \mathcal{S}$. The significance level of LATE estimators is indicated with stars in the usual manner: ``***'' means significant at $\alpha=1\%$, ``**'' means significant at $\alpha=5\%$, and``*'' means significant at $\alpha=10\%$. The bottom of the table shows the corresponding results from \citet[Tables 4 and 5]{dupas/karlan/robinson/ubfal:2018}.}
\label{tab:application}
\end{table}

Table \ref{tab:application} presents the results of the IV regressions computed using our Stata package.\footnote{Our Stata package is publicly available at \href{https://sites.northwestern.edu/federicobugni/carlate/}{https://sites.northwestern.edu/federicobugni/carlate/}.} For each outcome variable, we estimate the LATE of using the savings account based on the SAT, SFE, and 2S IV regressions. Given the setup of the RCT, all of these estimators are consistent. We can also consistently estimate the standard errors of these LATE estimators. Since the RCT uses SBR (i.e., $\tau(s)=0$ for all $s\in \mathcal{S}$), the standard errors for all of these estimators coincide. To illustrate our results, it is also relevant to compare these with the estimates of the standard errors under the incorrect assumption that the RCT used SRS (i.e., $\tau(s)=1$ for all $s\in \mathcal{S}$). (It is relevant to note that the standard Stata command {\it ivregress} would presume that the data were collected using SRS.) First, by Theorem \ref{thm:AsyDist_SAT}, the standard error of the SAT IV estimator does not depend on the details of the CAR mechanism, so the estimated standard error does not depend on $(\tau(s):s\in \mathcal{S})$. Second, since this RCT uses $\pi_A(s)=1/2$ for all $s\in\mathcal{S}$, Theorem \ref{thm:AsyDist_SFE} implies that the standard error of the SFE IV estimator does not depend on $(\tau(s):s\in \mathcal{S})$. Finally, as predicted by Theorem \ref{thm:AsyDist_2SR}, the standard error of the 2S IV estimator under $\tau(s)=1$ for all $s\in \mathcal{S}$ is larger than necessary, resulting in a loss in statistical power. When measured in terms of effective sample size, using the larger standard errors in the 2S IV regression is analog to losing 2.52\%, 5.52\%, and 9.12\% of the sample in the SAT IV regressions of the first, second, and third outcome variables, respectively. 

We now briefly describe the quantitative findings. For brevity, we focus on the LATE estimators based on the SAT IV regression. For complier households, opening and using these savings accounts result in an average increase of savings in formal financial institutions of \$17.572, an average decrease in savings in cash at home or in a secret place of \$7.323, and an average decrease in total expenditures in the last month of \$2.427. The first estimator is significantly different from zero at all significance levels, while the other two are not statistically significantly different from zero. Finally, our estimates indicate that 44.15\% of the households are compliers, while the remaining 55.85\% are never takers.

For the sake of comparison, Table \ref{tab:application} also includes the corresponding results from \citet[Tables 4 and 5]{dupas/karlan/robinson/ubfal:2018}. Their regressions are analogous to our SFE IV specification but with several differences. First, they run regressions using panel data from three periods, while we only use data from the first period. Second, their sample is affected by attrition in the second and third periods. Third, their regressions include additional controls such as the baseline outcome value and period indicator variables. None of these controls are available for our regression based on the data from the first period. Despite these differences, we note that the two sets of results are qualitatively similar.

Following the arguments in Section \ref{sec:optim}, we could use the data to study the optimality of the RCT parameters. For example, we can use the results in Section \ref{sec:pi_opt} to estimate the optimal treatment assignment probabilities. We find that the optimal constant treatment assignment probabilities $\pi_A^*$ for the three outcome variables are 0.549, 0.491, and 0.496, respectively. These probabilities are all very close to the one used in the RCT and consequently result in very minor efficiency gains.
% 0.91%, 0.03%, 0.01%, respectively.

%\input{rest.tex}

\section{Conclusions}\label{sec:conclusions}

This paper studies inference in an RCT with CAR and imperfect compliance of a binary treatment. By CAR, we refer to randomization schemes that first stratify according to baseline covariates and then assign treatment status to achieve ``balance'' within each stratum. In this context, we allow the RCT participants to endogenously decide whether to comply or not with the assigned treatment status. Given the possibility of imperfect compliance, we study inference on the LATE.

We study the asymptotic properties of three LATE estimators derived from IV regression. The first one is the ``fully saturated'' or SAT IV regression, i.e., a linear regression of the outcome on all indicators for all strata and their interaction with the treatment decision, with the latter instrumented with the treatment assignment. We show that the proposed LATE estimator is asymptotically normal, and we characterize its asymptotic variance in terms of primitives of the problem. We provide consistent estimators of the standard errors and asymptotically exact hypothesis tests. This LATE estimator is consistent under weak conditions regarding the CAR method used to implement the RCT (i.e., Assumptions \ref{ass:1}-\ref{ass:2}).

Our second LATE estimator is based on the ``strata fixed effects'' or SFE IV linear regression, i.e., a linear regression of the outcome on indicators for all strata and the treatment decision, with the latter instrumented with the treatment assignment. Our last LATE estimator is based on the ``two-sample'' or 2S IV linear regression, i.e., a linear regression of the outcome on a constant and the treatment decision, with the latter instrumented with the treatment assignment. The consistency of both of these LATE estimators requires additional conditions relative to the one based on the SAT IV linear regression (i.e., Assumptions \ref{ass:1}-\ref{ass:3}). In particular, they require that the target proportion of RCT participants assigned to each treatment cannot vary by strata (see Assumption \ref{ass:3}). Under these conditions, we show that both LATE estimators are asymptotically normal, and we characterize their asymptotic variance in terms of primitives of the problem. We also provide consistent estimators of their standard errors and asymptotically exact hypothesis tests.

Our characterization of the asymptotic properties of the LATE estimators allows us to investigate the influence of the parameters of the RCT. We use this to propose strategies to minimize their asymptotic variance in a hypothetical RCT based on data from its pilot study. We also establish that the asymptotic variance of the proposed LATE estimators does not increase if the strata of the RCT becomes finer and all else remains equal. We determine the optimal treatment assignment probability vector and show how to estimate it consistently based on data from a pilot study.

We confirm our theoretical results in Monte Carlo simulations. We also illustrate the practical relevance of our findings by revisiting the RCT in \cite{dupas/karlan/robinson/ubfal:2018}.

% - -- - - - - - - - - - - 
% HERE STARTS THE APPENDIX 
% - -- - - - - - - - - - - 
\renewcommand{\theequation}{\Alph{section}-\arabic{equation}}
\begin{appendix} 

\begin{small}
\section{Appendix}

%%%%%%%% DIVIDER %%%%%%%%%%%%
\subsection{Additional notation}
%%%%%%%% DIVIDER %%%%%%%%%%%%

This appendix uses the following notation. We use LHS and RHS to denote ``left hand side'' and ``right hand side'', respectively. We also use LLN, CLT, CMT, and LIE to denote ``law of large numbers'', ``central limit theorem'', ``continuous mapping theorem'', ``law of iterated expectations'', respectively.

For any $i=1,\dots,n$ and $(d,a,s)\in \{0,1\}^2\times \mathcal{S}$, we also define
\begin{align}
\tilde{Y}_{i}(d) ~\equiv~ Y_{i}(d)-E[Y_{i}(d)|S_{i}] &~\overset{(1)}{=}~ Y_{i}(d)-E[Y(d)|S]\notag \\
E[ \tilde{Y}_{i}(d)|D_{i}(a)=d,S_{i}=s] &~\overset{(2)}{=}~E[ \tilde{Y}(d)|D(a)=d,S=s] \notag\\
V[ \tilde{Y}_{i}(d)|D_{i}(a)=d,S_{i}=s ]&~\overset{(3)}{=}~V[ \tilde{Y}(d)|D(a)=d,S=s ], \label{eq:defnY_pre}
\end{align}
where (1)-(3) follow from Assumption \ref{ass:1}. Based on \eqref{eq:defnY_pre}, we define
\begin{align}
\mu ( d,a,s) ~\equiv~ E[ \tilde{Y}(d)|D(a)=d,S=s]  ~~~\text{and}~~~
\sigma ^{2}( d,a,s) ~\equiv~ V[ \tilde{Y}(d)|D(a)=d,S=s ]. \label{eq:defnY}
\end{align}
Lemma \ref{lem:MeanTraslation} translates $\mu ( d,a,s) $ and $\sigma ^{2}( d,a,s) $ in terms of the conditional moments of $Y(d)$.
%%%%%%%% DIVIDER %%%%%%%%%%%%

For every $s \in \mathcal{S}$, Sections \ref{sec:A_SFE} and \ref{sec:A_2SR} will use the following notation:
\begin{align*}
n_{A} ~\equiv~ \sum_{s\in \mathcal{S}}n_{A}( s), ~~n_{D}~\equiv~ \sum_{s\in \mathcal{S}}n_{D}( s),~~\text{and}~~    n_{AD}~\equiv~ \sum_{s\in \mathcal{S}}n_{AD}( s).
\end{align*}

%%%%%%%% DIVIDER %%%%%%%%%%%%
\subsection{Auxiliary results}
%%%%%%%% DIVIDER %%%%%%%%%%%%

%%%%%%%% DIVIDER %%%%%%%%%%%%
\begin{lemma}\label{lem:P_type_representation}
Under Assumption \ref{ass:1}, and for any $i=1,\dots,n$ and $s\in \mathcal{S}$,
\begin{align}
&P(D_i( 1) =1,D_i( 0) =1|S_i=s) ~=~P( AT|S=s)~=~\pi _{D( 0) }( s)\in [0,1) \notag\\
&P(D_i( 1) =0,D_i( 0) =0|S_i=s)~=~P( NT|S=s)~=~1-\pi _{D( 1) }( s)\in [0,1) \notag \\
&P(D_i( 1) =1,D_i( 0) =0|S_i=s)~=~P( C|S=s)~=~\pi _{D( 1) }( s) -\pi _{D( 0) }( s)\in (0,1]. \label{eq:pi_defns2} 
\end{align}
\end{lemma}
\begin{proof}
%[Proof of Lemma \ref{lem:P_type_representation}]
We begin by showing the first line in \eqref{eq:pi_defns2}.
\begin{align*}
 P(D_i( 1) =1,D_i( 0) =1|S_i=s) &\overset{(1)}{=} P(D( 1) =1,D( 0) =1|S=s) 
%&=P(D( 0) =1|S=s)-P(D( 1) =0,D( 0) =1|S=s)\\
\overset{(2)}{=}P(D( 0) =1|S=s)=\pi _{D(0)}(s)\overset{(3)}{<}1,
\end{align*}
where (1) and (3) hold by the i.i.d.\ condition in Assumption \ref{ass:1}, and (2) holds by Assumption \ref{ass:1}(c). To complete the argument, note that $P(AT|S=s)\equiv P(D( 1) =1,D( 0) =1|S=s)$. The second line in \eqref{eq:pi_defns2} follows from a similar argument. The last line in \eqref{eq:pi_defns2} follows from the other two lines and Assumption \ref{ass:1}(c).
\end{proof}
%%%%%%%% DIVIDER %%%%%%%%%%%%

%\begin{lemma}\label{lem:aux_lemma}
%Assume Assumptions \ref{ass:1} and \ref{ass:2}. For any $i=1,\dots ,n$, $ (d,a,s,b)\in \{ 0,1\}^2 \times \mathcal{S} \times \{1,2\} $, $s_{-i} \in \mathcal{S}^{n-1}$,
%\begin{align}
%P( D_{i}( a) =1|A_{i}=a,S_{i}=s,S_{-i}=s_{-i}) &=P( D( a) =1|S=s) =\pi _{D( a) }( s). \label{eq:aux_lemma_1} 
%\end{align}
%Also, provided that the conditioning event has positive probability,
%\begin{align}
%E[ Y_{i}( d) ^{b}|D_{i}( a) =d,A_{i}=a,,S_{-i}=s_{-i} ] &=E[ Y( d) ^{b}|D( a) =d,S=s]. \label{eq:aux_lemma_3} 
%\end{align}
%\end{lemma}
%%%%%%%% DIVIDER %%%%%%%%%%%%
%\begin{proof}
%Fix $i=1,\dots ,n$ and $ (d,a,s,b)\in \{ 0,1\}^2 \times \mathcal{S} \times \{1,2\} $ arbitrarily, and set $s_i=s$ throughout this proof. To prove \eqref{eq:aux_lemma_1}, consider the following derivation.
%\begin{align*}
%P( D_{i}( a) =1|A_{i}=a,S_{i}=s,S_{-i}=s_{-i})
%\overset{(1)}{=}P( D_{i}( a) =1|S_{i}=s,S_{-i}=s_{-i})\overset{(2)}{=}P( D_{i}( a) =1|S_{i}=s) =\pi _{D( a) }( s) ,
%\end{align*}
%where (1) holds by Assumption \ref{ass:2}(a) and (2) holds by Assumption \ref{ass:1}.
%
%To derive \eqref{eq:aux_lemma_3}, consider the following argument.
%\begin{align*}
%E[ Y_{i}( d) ^{b}|A_{i}=a,D_{i}( a) =d,S_{i}=s,S_{-i}=s_{-i} ]
%&\overset{(1)}{=}E[ Y_{i}( d) ^{b}|D_{i}( a) =d,S_{i}=s ,S_{-i}=s_{-i}] \\
%&\overset{(2)}{=}E[ Y_{i}( d) ^{b}|D_{i}( a) =d,S_{i}=s ] \\
%&\overset{(3)}{=}E[ Y( d) ^{b}|D( a) =d,S=s] ,
%\end{align*}
%where (1) holds by Assumption \ref{ass:2}(a), and (2) and (3) hold by Assumption \ref{ass:1}.
%\end{proof}
%%%%%%%% DIVIDER %%%%%%%%%%%%

\begin{lemma}\label{lem:MeanTraslation}
Under Assumptions \ref{ass:1} and \ref{ass:2}, and provided that the conditioning event has positive probability,
\begin{align}
\mu (1,1,s) &~=~E[Y(1)|AT,S=s]\tfrac{\pi _{D(0)}(s)}{\pi _{D(1)}(s)} +E[Y(1)|C,S=s]\tfrac{\pi _{D( 1) }( s) -\pi _{D( 0) }( s) }{\pi _{D( 1) }( s) } -E[Y(1)|S=s] \notag \\
\mu (0,1,s) &~=~E[Y(0)|NT,S=s]-E[Y(0)|S=s]\notag  \\
\mu (1,0,s) &~=~E[Y(1)|AT,S=s]-E[Y(1)|S=s] \notag \\
\mu (0,0,s) &~=~E[Y(0)|NT,S=s]\tfrac{1-\pi _{D( 1) }( s) }{1-\pi _{D( 0) }( s) }+E[Y(0)|C,S=s]\tfrac{ \pi _{D( 1) }( s) -\pi _{D( 0) }( s) }{1-\pi _{D( 0) }( s) }-E[Y(0)|S=s]\label{eq:Translate}
\end{align}
and
\begin{align}
\sigma ^{2}(1,1,s) &~=~\left(
\begin{array}{c}
V[Y(1)|AT,S=s]\frac{\pi _{D(0)}(s)}{\pi _{D(1)}(s)}+V[Y(1)|C,S=s]\frac{\pi _{D(1)}(s)-\pi _{D(0)}(s)}{\pi _{D(1)}(s)}+ \\
(E[Y(1)|C,S=s]-E[Y(1)|AT,S=s])^{2}\frac{\pi _{D(0)}(s)}{\pi _{D(1)}(s)}\frac{ \pi _{D(1)}(s)-\pi _{D(0)}(s)}{\pi _{D(1)}(s)}
\end{array}
\right)   \notag  \\
\sigma ^{2}(0,1,s) &~=~V[Y(0)|NT,S=s] \notag \\
\sigma ^{2}(1,0,s) &~=~V[Y(1)|AT,S=s] \notag \\
\sigma ^{2}(0,0,s) &~=~\left(
\begin{array}{c}
V[Y(0)|NT,S=s]\frac{1-\pi _{D(1)}(s)}{1-\pi _{D(0)}(s)}+V[Y(0)|C,S=s]\frac{ \pi _{D(1)}(s)-\pi _{D(0)}(s)}{1-\pi _{D(0)}(s)}+ \\
(E[Y(0)|C,S=s]-E[Y(0)|NT,S=s])^{2}\frac{1-\pi _{D(1)}(s)}{1-\pi _{D(0)}(s)} \frac{\pi _{D(1)}(s)-\pi _{D(0)}(s)}{1-\pi _{D(0)}(s)}
\end{array}
\right) .\label{eq:Translate2}
\end{align}
\end{lemma}
%%%%%%%% DIVIDER %%%%%%%%%%%%
\begin{proof}
The subscript $i=1,\dots,n$ is absent from all expressions due to the i.i.d.\ condition in Assumption \ref{ass:1}. For any $( d,a,s,b) \in \{ 0,1\}^2 \times \mathcal{S}\times \{ 1,2\} $ such that $P(D(a)=d|S=s)=1[d=1]\pi _{D(a)}(s)+1[d=0](1-\pi _{D(a)}(s))>0$,
\begin{align}
E[ Y(d)^{b}|D(a)=d,S=s] 
=\frac{\left[ 
\begin{array}{c}
E[ Y(d)^{b}|D(a)=d,D(1-a)=0,S=s] P(D(1-a)=0,D(a)=d|S=s) \\
+E[ Y(d)^{b}|D(a)=d,D(1-a)=1,S=s] P(D(1-a)=1,D(a)=d|S=s)
\end{array}
\right] }{1[d=1]\pi _{D(a)}(s)+1[d=0](1-\pi _{D(a)}(s))}.\label{eq:Translate6}
\end{align}

We begin by showing \eqref{eq:Translate}. We only show the first line, as the others can be shown analogously.
\begin{align*}
\mu (1,1,s) &\overset{(1)}{=}E[ Y(1)|D(1)=1,S=s] -E[Y(1)|S=s] \\
&\overset{(2)}{=}\frac{\left[ 
\begin{array}{c}
E[ Y(1)|D(1)=1,D(0)=0,S=s] P(D(0)=0,D(1)=1|S=s) \\ 
+E[ Y(1)|D(1)=1,D(0)=1,S=s] P( D(0)=1,D(1)=1|S=s) 
\end{array}
\right] }{\pi _{D(1)}(s)}-E[Y(1)|S=s] \\
%&=E[ Y(1)|AT,S=s] \frac{P[ AT|S=s] }{\pi _{D(1)}(s)}+E [ Y(1)|C,S=s] \frac{P[ C|S=s] }{\pi _{D(1)}(s)} -E[Y(1)|S=s] \\
&\overset{(3)}{=}E[Y(1)|AT,S=s]\tfrac{\pi _{D(0)}(s)}{\pi _{D(1)}(s)}+E[Y(1)|C,S=s]
\tfrac{\pi _{D( 1) }( s) -\pi _{D( 0)
}( s) }{\pi _{D( 1) }( s) }-E[Y(1)|S=s],
\end{align*}
where (1) follows from \eqref{eq:defnY}, (2) follows from \eqref{eq:Translate6}, and (3) follows from Lemma \ref{lem:P_type_representation}.

To conclude, we show \eqref{eq:Translate2}. Again, we only show the first line, as the others can be shown analogously.
\begin{align*}
\sigma ^{2}(1,1,s) &\overset{(1)}{=} V[ Y(1)-E[Y(1)|S=s]|D(1)=1,S=s] \\
&=V[ Y(1)|D(1)=1,S=s] \\
%&=E[ Y(1)^{2}|D(1)=1,S=s] -( E[ Y(1)|D(1)=1,S=s ] ) ^{2} \\
&\overset{(2)}{=}\left(
\begin{array}{c}
\frac{\left[ 
\begin{array}{c}
E[ Y(1)^{2}|D(1)=1,D(0)=0,S=s] P( D(0)=0,D(1)=1|S=s)\\
+E[ Y(1)^{2}|D(1)=1,D(0)=1,S=s] P( D(0)=1,D(1)=1|S=s)
\end{array}
\right] }{\pi _{D(1)}(s)} \\ 
-\frac{\left[ 
\begin{array}{c}
E[ Y(1)|D(1)=1,D(0)=0,S=s] P(D(0)=0,D(1)=1|S=s)  \\ 
+E[ Y(1)|D(1)=1,D(0)=1,S=s] P( D(0)=1,D(1)=1|S=s) 
\end{array}
\right] ^{2}}{( \pi _{D(1)}(s)) ^{2}}
\end{array}
\right)  \\
%% FEDE: THE DERIVATION BELOW IS USEFUL TO HAVE BUT JUST ALGEBRA.
%&=\left\{ 
%\begin{array}{c}
%\frac{[ E[ Y(1)^{2}|C,S=s] P[ C|S=s] +E[ Y(1)^{2}|AT,S=s] P(AT|S=s) ] }{\pi_{D(1)}(s)} \\
%-\frac{[ E[ Y(1)|C,S=s] P[ C|S=s] +E[ Y(1)|AT,S=s] P(AT|S=s)] ^{2}}{( \pi_{D(1)}(s)) ^{2}}
%\end{array}
%\right\}  \\
&\overset{(3)}{=}\left( 
\begin{array}{c}
V[Y(1)|AT,S=s]\frac{\pi _{D(0)}(s)}{\pi _{D(1)}(s)}+V[Y(1)|C,S=s]\frac{\pi _{D(1)}(s)-\pi _{D(0)}(s)}{\pi _{D(1)}(s)}+ \\
(E[Y(1)|AT,S=s]-E[Y(1)|C,S=s])^{2}\frac{\pi _{D(0)}(s)}{\pi _{D(1)}(s)}\frac{ \pi _{D(1)}(s)-\pi _{D(0)}(s)}{\pi _{D(1)}(s)}
\end{array}
\right) ,
\end{align*}
where (1) follows from \eqref{eq:defnY}, (2) follows from \eqref{eq:Translate6}, and (3) follows  follows from Lemma \ref{lem:P_type_representation}.
\end{proof}

%%%%%%%% DIVIDER %%%%%%%%%%%%
\begin{lemma}\label{lem:A1and2_impliesold3}
Under Assumptions \ref{ass:1} and \ref{ass:2},
\begin{equation}
\left. \left\{ \left( \sqrt{n}\left( \frac{n_{AD}(s)}{n_{A}(s)}-\pi _{D(1)}(s),\frac{n_{D}(s)-n_{AD}(s)}{n(s)-n_{A}(s)}-\pi _{D(0)}(s)\right) ^{\prime }:s\in \mathcal{S}\right) \right\vert ((S_{i},A_{i}))_{i=1}^{n}\right\} 
\overset{d}{\to }N( \mathbf{0},\Sigma _{D})\text{ w.p.a.1,} \label{eq:A3A_statement}
\end{equation}
where
\begin{equation}
\Sigma _{D}~\equiv ~diag\left(\left[
\begin{array}{cc}
\frac{\pi _{D(1)}(s)(1-\pi _{D(1)}(s))}{\pi _{A}(s)} & 0 \\
0 & \frac{\pi _{D(0)}(s)(1-\pi _{D(0)}(s))}{1-\pi _{A}(s)}
\end{array}
\right] \frac{1}{p(s)}:s\in \mathcal{S}\right).
\end{equation}
In addition,
\begin{equation}
\left(\frac{n_{AD}(s)}{n_{A}( s) },\frac{ n_{D}(s)-n_{AD}(s) }{ n( s) -n_{A}(s) }\right) \overset{p}{\to}( \pi _{D( 1) }( s) ,\pi _{D( 0) }( s) ).
\label{eq:A3B_statement}
\end{equation} 
\end{lemma}
%%%%%%%% DIVIDER %%%%%%%%%%%%
\begin{proof}
We only show \eqref{eq:A3A_statement}, as \eqref{eq:A3B_statement} follows from \eqref{eq:A3A_statement} and elementary convergence arguments. We divide the proof of \eqref{eq:A3A_statement} in two steps. The first step shows that
\begin{align}
&\left. \left\{ \left( \left( \sqrt{n_{A}( s) }\left( \frac{n_{AD}(s)}{n_A(s)}-\pi _{D(1)}(s)\right) ,\sqrt{n( s) -n_{A}( s) }\left( \frac{n_{0D}(s)}{n(s)-n_A(s)}-\pi _{D(0)}(s)\right) \right) ^{\prime }:s\in \mathcal{S}\right) \right\vert ((S_{i},A_{i}))_{i=1}^{n}\right\} \notag \\
&\overset{d}{\to}N\left( \mathbf{0},diag\left( \left[
\begin{array}{cc}
\pi _{D(1)}(s)(1-\pi _{D(1)}(s)) & 0 \\
0 & \pi _{D(0)}(s)(1-\pi _{D(0)}(s))
\end{array}
\right] :s\in \mathcal{S}\right) \right) \text{ w.p.a.1,}\label{eq:pf_A3_6}
\end{align}
where $n_{0D}( s) \equiv \sum_{i=1}^{n}1[ A_{i}=0,D_{i}=1,S_{i}=s] $ for any $s\in \mathcal{S}$. The second step shows that
\begin{equation}
\left. \left\{ \left(\left( \frac{\sqrt{n}}{\sqrt{n_{A}( s) }}, \frac{\sqrt{n}}{\sqrt{n( s) -n_{A}( s) }}\right) ^{\prime }:s\in \mathcal{S}\right)\right\vert ((S_{i},A_{i}))_{i=1}^{n}\right\} \to \frac{1}{\sqrt{ p( s) }}\left( \frac{1}{\sqrt{\pi _{A}( s) }},\frac{1}{ \sqrt{1-\pi _{A}( s) }}\right) \text{ w.p.a.1.}
\label{eq:pf_A3_8}
\end{equation}
Then, \eqref{eq:A3A_statement} follows from \eqref{eq:pf_A3_6} and \eqref{eq:pf_A3_8} via elementary convergence arguments.

\underline{Step 1: Show \eqref{eq:pf_A3_6}.} Conditional on $(( S_{i},A_{i}))_{i=1}^{n}$, note that $(( n_{A}(s),n( s) ) :s\in \mathcal{S}) $ is non-stochastic, and so the only source of randomness in \eqref{eq:pf_A3_6} is $( ( n_{AD}(s),n_{0D}( s) ) :s\in \mathcal{S})$. Also, it is relevant to note that
\begin{align}
n_{AD}( s) &=\sum_{i=1}^{n}1[ A_{i}=1,D_{i}( 1) =1,S_{i}=s] =\sum_{i=1}^{n}1[ A_{i}=1,S_{i}=s] D_{i}( 1) \notag\\
n_{0D}( s) &=\sum_{i=1}^{n}1[ A_{i}=0,D_{i}( 0) =1,S_{i}=s] =\sum_{i=1}^{n}1[ A_{i}=0,S_{i}=s] D_{i}( 0) . \label{eq:pf_A3_2}
\end{align}
According to \eqref{eq:pf_A3_2}, each component of $(( n_{AD}(s),n_{0D}( s) ) ^{\prime }:s\in \mathcal{S})$ is determined by different subset of individuals in the random sample.

As a next step, consider the following derivation for any $(d_{0,i}) _{i=1}^{n}\times ( d_{1,i}) _{i=1}^{n}\times ( a_{i}) _{i=1}^{n}\times (s_{i})_{i=1}^{n}\in \{ 0,1\} ^{n}\times \{ 0,1\} ^{n}\times \{ 0,1\} ^{n}\times \mathcal{S}^{n}$.
\begin{align}
&P( ( ( D_{i}( 0) ,D_{i}( 1) ) ) _{i=1}^{n}=( ( d_{0,i},d_{1,i}) ) _{i=1}^{n}|( ( A_{i},S_{i}) ) _{i=1}^{n}=( ( a_{i},s_{i}) ) _{i=1}^{n}) \notag \\
&\overset{(1)}{=}P( ( ( D_{i}( 0) ,D_{i}( 1) ) ) _{i=1}^{n}=( ( d_{0,i},d_{1,i}) ) _{i=1}^{n}|( S_{i}) _{i=1}^{n}=( s_{i}) _{i=1}^{n}) \notag \\
&=\frac{P( ( ( D_{i}( 0) ,D_{i}( 1) ,S_{i}) ) _{i=1}^{n}=( ( d_{0,i},d_{1,i},s_{i}) ) _{i=1}^{n}) }{P( ( S_{i}) _{i=1}^{n}=( s_{i}) _{i=1}^{n}) } \notag \\ 
&\overset{(2)}{=}
\prod\nolimits_{i=1}^{n}P( ( D( 0) ,D( 1) ) =( d_{0,i},d_{1,i}) |S=s_{i}) , \label{eq:pf_A3_4}
\end{align}
where (1) follows from Assumption \ref{ass:2}(a) and (2) follows from Assumption \ref{ass:1}. Conditionally on $( ( A_{i},S_{i}) ) _{i=1}^{n}=( ( a_{i},s_{i}) ) _{i=1}^{n}$, \eqref{eq:pf_A3_4} reveals that $( ( D_{i}( 0) ,D_{i}( 1) ) ) _{i=1}^{n}$ is an independent sample with $( (  D_{i}( 0) ,D_{i}( 1) ) |( ( A_{i},S_{i}) ) _{i=1}^{n}=( ( a_{i},s_{i}) ) _{i=1}^{n}) \overset{d}{=}( ( D( 0) ,D( 1) ) |S=s_{i}) $.

By \eqref{eq:pf_A3_2}, $(( n_{AD}(s),n_{0D}( s) ) ^{\prime }:s\in \mathcal{S}) $ are the sum of binary observations from different individuals. If we condition on $( ( A_{i},S_{i})) _{i=1}^{n}$, \eqref{eq:pf_A3_2} and \eqref{eq:pf_A3_4} imply that
\begin{equation}
 \{ ( ( n_{AD}(s),n_{D}(s)-n_{AD}(s)) ^{\prime }:s\in \mathcal{S}) \vert (S_{i})_{i=1}^{n},(A_{i})_{i=1}^{n}) \} ~\overset{d}{=}~( ( B( 1,s) ,B( 0,s) ) ^{\prime }:s\in \mathcal{S} ) , \label{eq:pf_A3_5}
\end{equation}
where $( ( B( 1,s) ,B( 0,s) ) ^{\prime }:s\in \mathcal{S}) $ are independent random variables with $ B( 1,s) \sim Bi( n_{A}( s) ,\pi _{D(1)}(s)) $ and $B( 0,s) \sim Bi( n( s) -n_{A}( s) ,\pi _{D(0)}(s)) $. Provided that we condition on sequences of $(( S_{i},A_{i}) )_{i=1}^{n}$ with $n_{A}( s) \to \infty $ and $n( s) -n_{A}( s) \to \infty $ for all $s\in \mathcal{S}$, \eqref{eq:pf_A3_6} follows immediately from \eqref{eq:pf_A3_5} and the normal approximation to the binomial.

To conclude the step, it suffices to show that $n_{A}( s) \to \infty $ and $n( s) -n_{A}( s) \to \infty $ for all $s\in \mathcal{S}$ w.p.a.1. In turn, note that this is a consequence of $n( s) /n\overset{a.s}{\to }p( s) >0$ under Assumption \ref{ass:1} and $n_{A}( s) /n( s) \overset{p}{\to}\pi _{A}( s) \in ( 0,1) $ by Assumption \ref{ass:2}(b).

\underline{Step 2: Show \eqref{eq:pf_A3_8}.} Fix $s\in \mathcal{S}$ arbitrarily and notice that
\begin{align}
\left( \frac{\sqrt{n}}{\sqrt{n_{A}( s) }}, \frac{\sqrt{n}}{\sqrt{n( s) -n_{A}( s) }}\right) =\left[ \tfrac{n}{n( s) } \left( \tfrac{1}{n_{A}( s) /n( s) },\tfrac{1}{ 1-n_{A}( s) /n( s) }\right) \right] ^{1/2}
\overset{p}{\to}\frac{1}{\sqrt{p( s) }}\left( \frac{1}{\sqrt{ \pi _{A}( s) }},\frac{1}{\sqrt{1-\pi _{A}( s) }} \right) ,\label{eq:pf_A3_7}
\end{align}
where the convergence follows from Assumptions \ref{ass:1} and \ref{ass:2}(b). From \eqref{eq:pf_A3_7} and the fact that $( ( n_{A}(s),n( s) ) :s\in \mathcal{S}) $ is non-stochastic once we condition on $( ( A_{i},S_{i}) ) _{i=1}^{n}$, \eqref{eq:pf_A3_8} follows.
\end{proof}

\begin{lemma}\label{lem:AsyDist}
Assume Assumptions \ref{ass:1} and \ref{ass:3}, and define 
\begin{equation*}
R_{n}~\equiv~( R_{n,1}',R_{n,2}',R_{n,3}',R_{n,4}') ',
\end{equation*}
where
\begin{align}
R_{n,1} &~\equiv~\left( \frac{1}{\sqrt{n}}\sum_{i=1}^{n}1[ D_{i}=d,A_{i}=a,S_{i}=s] ( \tilde{Y}_{i}( d) -\mu (d,a,s) ) :( d,a,s) \in \{ 0,1\}^2 \times \mathcal{S}\right) \notag\\
R_{n,2} &~\equiv~\left( \left[ \sqrt{n}\left( \frac{n_{AD}(s)}{n_{A}( s) }-\pi _{D( 1) }( s) \right) ,\sqrt{n}\left( \frac{ n_{D}(s)-n_{AD}(s)}{n( s) -n_{A}(s)}-\pi _{D( 0) }( s) \right) \right]':s\in S\right)\notag \\ 
R_{n,3} &~\equiv~\left( \sqrt{n}\left( \frac{n_{A}( s) }{n( s) }-\pi _{A}( s) \right) :s\in S\right) \notag\\ 
R_{n,4} &~\equiv~\left( \sqrt{n}\left( \frac{n( s) }{n}-p( s) \right) :s\in S\right) .\label{eq:Rn_defn}
\end{align}
Then,
\begin{equation*}
R_{n}~\overset{d}{\to }~N\left( \left( 
\begin{array}{c}
\mathbf{0} \\ 
\mathbf{0} \\ 
\mathbf{0} \\ 
\mathbf{0}
\end{array}
\right) ,\left( 
\begin{array}{cccc}
\Sigma _{1} & \mathbf{0} & \mathbf{0} & \mathbf{0} \\ 
\mathbf{0} & \Sigma _{2} & \mathbf{0} & \mathbf{0} \\ 
\mathbf{0} & \mathbf{0} & \Sigma _{3} & \mathbf{0} \\ 
\mathbf{0} & \mathbf{0} & \mathbf{0} & \Sigma _{4}
\end{array}
\right) \right) ,
\end{equation*}
where
\begin{align*}
\Sigma _{1} &~\equiv~diag \left( \left[ 
\begin{array}{c}
1[ ( d,a) =( 0,0) ] ( 1-\pi _{D( 0) }( s) ) ( 1-\pi _{A}( s) ) \\
+1[ ( d,a) =( 1,0) ] \pi _{D( 0) }( s) ( 1-\pi _{A}( s) )\\
+1[ ( d,a) =( 0,1) ] ( 1-\pi _{D( 1) }( s) ) \pi _{A}( s) \\
+1[ ( d,a) =( 1,1) ] \pi _{D( 1) }( s) \pi _{A}( s)
\end{array}
\right]  p( s) \sigma^{2}( d,a,s)~:~( d,a,s) \in \{ 0,1\} ^2 \times \mathcal{S}\right) \\
\Sigma _{2} &~\equiv~diag \left(\left[
\begin{array}{cc}
{( 1-\pi _{D( 1) }( s) ) \pi _{D( 1) }( s) }/{\pi _{A}( s) } & 0 \\
0 & {( 1-\pi _{D( 0) }( s) ) \pi _{D( 0) }( s) }/{(1-\pi _{A}( s)) }
\end{array}
\right]/{p( s) }:s\in \mathcal{S}\right) \\
\Sigma _{3} &~\equiv~diag( \tau ( s) ( 1-\pi _{A}( s) ) \pi _{A}( s)/ p( s) :s\in \mathcal{S}) \\
\Sigma _{4} &~\equiv~diag( p( s) :s\in \mathcal{S}) -( p( s) :s\in \mathcal{S}) ( p( s) :s\in \mathcal{S}) '.
\end{align*}
\end{lemma}
%%%%%%%% DIVIDER %%%%%%%%%%%%
\begin{proof}
Throughout this proof, it is relevant to recall that Assumption \ref{ass:3} implies Assumption \ref{ass:2}. Also, let $\zeta _{j}\sim N( {\bf 0},\Sigma _{j}) $ for $j=1,2,3,4$, with $( \zeta _{1}',\zeta _{2}',\zeta _{3}',\zeta _{4}') $ are independent. Our goal is to show that $( R_{n,1}',R_{n,2}',R_{n,3}',R_{n,4}') \overset{d}{\to}( \zeta _{1}',\zeta _{2}',\zeta _{3}',\zeta _{4}') $. We divide the argument into 3 steps.

\underline{Step 1.} Under Assumptions \ref{ass:1} and \ref{ass:2}, we show that for random vectors $R_{n,1}^{C}$ and $R_{n,1}^{D}$,
\begin{align}
&( R_{n,1}',R_{n,2}',R_{n,3}',R_{n,4}') ~\overset{d}{=}~( {R_{n,1}^{C}}',R_{n,2}',R_{n,3}',R_{n,4}') \label{eq:ad_step1_eq1} \\
&R_{n,1}^{D} ~\perp~ ( R_{n,2}',R_{n,3}',R_{n,4}')' \label{eq:ad_step1_eq2} \\
&R_{n,1}^{D} ~\overset{d}{\to}~\zeta _{1} \label{eq:ad_step1_eq3} \\
&R_{n,1}^{C} ~=~R_{n,1}^{D}+o_{p}( 1) .\label{eq:ad_step1_eq4} 
\end{align}

For any arbitrary $( (y_{i})_{i=1}^{n},( d_{i}) _{i=1}^{n},( a_{i}) _{i=1}^{n},(s_{i})_{i=1}^{n}) \in \mathbb{R} ^{n}\times \{ 0,1\} ^{n}\times \{ 0,1\} ^{n}\times \mathcal{S}^{n}$, consider first the following derivation. Provided that the conditioning event has positive probability,
\begin{align}
&dP( ( \tilde{Y}_i( d_{i}) =y_{i}) _{i=1}^{n}|( ( D_{i},A_{i},S_{i}) =( d_{i},a_{i},s_{i}) ) _{i=1}^{n}) \notag \\
&\overset{(1)}{=}dP( ( Y_i( d_{i}) =y_{i}+E[ Y( d_{i}) |S=s_{i}] ) _{i=1}^{n}|( ( D_{i},A_{i},S_{i}) =( d_{i},a_{i},s_{i}) ) _{i=1}^{n}) \notag \\
&\overset{(2)}{=}\frac{dP( ( ( Y_i( d_{i}) ,D_i( a_{i}) ) =( y_{i}+E[ Y( d_{i}) |S=s_{i}] ,d_{i}) ) _{i=1}^{n}|( ( S_{i},A_{i}) =( s_{i},a_{i}) ) _{i=1}^{n}) }{ P( ( D_i( a_{i}) =d_{i}) _{i=1}^{n}|( ( S_{i},A_{i}) =( s_{i},a_{i}) ) _{i=1}^{n}) } \notag \\
&\overset{(3)}{=}\frac{dP( ( ( Y_i( d_{i}) ,D_i( a_{i}) ) =( y_{i}+E[ Y( d_{i}) |S=s_{i}] ,d_{i}) ) _{i=1}^{n}|( S_{i}=s_{i}) _{i=1}^{n}) }{P( ( D_i( a_{i}) =d_{i}) _{i=1}^{n}|( S_{i}=s_{i}) _{i=1}^{n}) } \notag \\
&\overset{(4)}{=}\frac{\int_{( z_{i}:S( z_{i}) =s_{i}) _{i=1}^{n}}dP( ( ( Y_i( d_{i}) ,D_i( a_{i}) ,Z_{i}) =( y_{i}+E[ Y( d_{i}) |S=s_{i}] ,d_{i},z_{i}) ) _{i=1}^{n}) }{\int_{( z_{i}:S( z_{i}) =s_{i}) _{i=1}^{n}}P( ( ( D_i( a_{i}) ,Z_{i}) =( d_{i},z_{i}) ) _{i=1}^{n}) } \notag \\
&\overset{(5)}{=}\frac{\prod_{i=1}^{n}\int_{z_{i}:S( z_{i}) =s_{i}}dP( ( Y_i( d_{i}) ,D_i( a_{i}) ,Z_{i}) =( y_{i}+E[ Y( d_{i}) |S=s_{i}] ,d_{i},z_{i}) ) }{\prod_{i=1}^{n}\int_{z_{i}:S( z_{i}) =s_{i}}P( ( D_i( a_{i}) ,Z_{i}) =( d_{i},z_{i}) ) } \notag \\
&=\prod_{i=1}^{n}dP( Y_i( d_{i}) =y_{i}+E[ Y_i( d_{i}) |S=s_{i}] |D_i( a_{i}) =d_{i},S=s_{i}) \notag \\
&\overset{(6)}{=}\prod_{i=1}^{n}dP( \tilde{Y}_i( d_{i}) =y_{i}|D_i( a_{i}) =d_{i},S=s_{i}) ,\label{eq:ad_step1_eq5}
\end{align}
where (1) and (6) hold by \eqref{eq:defnY_pre}, (2) holds by $D_{i}=D_i( A_{i}) $, (3) holds by  Assumption \ref{ass:2}(a), (4) holds by  $S_{i}=S( Z_{i}) $, and (5) holds by  Assumption \ref{ass:1}. A corollary of \eqref{eq:ad_step1_eq5} is that $( ( \tilde{Y}_i( d_{i}) ) _{i=1}^{n}|( ( D_{i},A_{i},S_{i}) =( d_{i},a_{i},s_{i}) ) _{i=1}^{n}) $ has the distribution of an independent sample with observation $i=1,\dots ,n$ distributed according to $( \tilde{Y}( d_{i}) |D( a_{i}) =d_{i},S=s_{i}) $. Then, conditionally on $( ( D_{i},A_{i},S_{i}) =( d_{i},a_{i},s_{i}) ) _{i=1}^{n} $, $( \tilde{Y}_{i}( d_{i}) -\mu (d_{i},a_{i},s_{i})) _{i=1}^{n}$ is an independent sample with $( \tilde{Y}_{i}( d_{i}) -\mu (d_{i},a_{i},s_{i})|( ( D_{i},A_{i},S_{i}) =( d_{i},a_{i},s_{i}) ) _{i=1}^{n}) \overset{d}{=}( \tilde{Y} ( d_{i}) -\mu (d_{i},a_{i},s_{i})|D( a_{i}) =d_{i},S=s_{i}) $.

Conditional on $( ( D_{i},A_{i},S_{i}) =( d_{i},a_{i},s_{i}) ) _{i=1}^{n}$, consider the following matrix
\begin{equation}
( ( ( 1[D_{i}=d,A_{i}=a,S_{i}=s](\tilde{Y}_{i}(d)-\mu (d,a,s))) :(d,a,s)\in \{0,1\}^2\times \mathcal{S}) ^{\prime }:i=1,\dots ,n) . \label{eq:ad_step1_eq6}
\end{equation}

Consider the following observation for each row $i=1,\dots ,n$ of \eqref{eq:ad_step1_eq6}. Row $i$ has one and only one indicator $( 1[D_{i}=d,A_{i}=a,S_{i}=s]:(d,a,s)\in \{0,1\}^2\times \mathcal{ S}) $ that is turned on, corresponding to $( d,a,s) =( d_{i},a_{i},s_{i}) $. For this entry, we have that $ 1[D_{i}=d,A_{i}=a,S_{i}=s](\tilde{Y}_{i}(d)-\mu (d,a,s))=(\tilde{Y} _{i}(d_{i})-\mu (d_{i},a_{i},s_{i}))$. The remaining observations in row $i$ are equal to zero and, thus, independent of $(\tilde{Y}_{i}(d_{i})-\mu (d_{i},a_{i},s_{i}))$. In this sense, the elements of the row are independent. By the derivation in \eqref{eq:ad_step1_eq5}, conditional on $ ( ( D_{i},A_{i},S_{i}) =( d_{i},a_{i},s_{i}) ) _{i=1}^{n}$, the rows are independent. As a consequence, conditional on $( ( D_{i},S_{i},A_{i}) ) _{i=1}^{n}$, \eqref{eq:ad_step1_eq6} has the same distribution as the following matrix
\begin{equation}
( ( ( 1(D_{i}=d,A_{i}=a,S_{i}=s)\breve{Y} _{i}(d,a,s)) :(d,a,s)\in \{0,1\}^2\times \mathcal{S} ) ^{\prime }:i=1,\dots ,n), \label{eq:ad_step1_eq7}
\end{equation}
where $(( \breve{Y}_{i}(d,a,s):(d,a,s)\in \{0,1\}^{2}\times \mathcal{S }) ^{\prime }:i=1,\dots ,n)$ denotes a matrix of $ 4|\mathcal{S}|\times n$ independent random variables, independent of $( ( D_{i},S_{i},A_{i}) ) _{i=1}^{n}$, with $\breve{Y}_{i}(d,a,s) \overset{d}{=}( \tilde{Y}(d)-\mu (d,a,s)|D( a) =d,S=s) $ for each $(d,a,s)\in \{0,1\}^{2}\times \mathcal{S}$. As a corollary,
\begin{equation}
( R_{n,1}|( ( D_{i},A_{i},S_{i}) ) _{i=1}^{n}) \overset{d}{=}( R_{n,1}^{B}|( ( D_{i},A_{i},S_{i}) ) _{i=1}^{n}), \label{eq:ad_step1_eq8}
\end{equation}
where
\begin{equation*}
R_{n,1}^{B}~\equiv ~\left( \frac{1}{\sqrt{n}}\sum_{i=1}^{n} 1[D_{i}=d,A_{i}=a,S_{i}=s]\breve{Y}_{i}(d,a,s):(d,a,s)\in \{0,1\}^2\times \mathcal{S}\right) .
\end{equation*}

Consider the following classification of observations. Let $g=1$ represent an observation with $( d,a,s) =( 0,0,1) $, $g=2$ represents $( d,a,s) =( 1,0,1) $, $g=3$ represents $ ( d,a,s) =( 0,1,1) $, $g=4$ represents $( d,a,s) =( 1,1,1) $, $g=5$ represents $( d,a,s) =( 0,0,2) $, and so on, until $g=4S$, which represents $( d,a,s) =( 1,1,|\mathcal{S}|) $. Let $G:( d,a,s) \to \mathcal{G}\equiv (1,\dots ,4|\mathcal{S}|)$ denote the function that maps each $( d,a,s) $ into a group $g\in \mathcal{G}$. For each $g\in \mathcal{G}$, let $N_{g}\equiv \sum_{i=1}^{n}1( G( D_{i},A_{i},S_{i}) <g) $. Also, let $N_{|\mathcal{G}|+1} =N_{4|\mathcal{S}|+1} = n$. Note that, conditional on $( ( D_{i},A_{i},S_{i}) ) _{i=1}^{n}$, $( N_{G( d,a,s) }:(d,a,s)\in \{0,1\}^{2}\times \mathcal{S}) $ is nonstochastic. Let's now consider a reordering of the units $i=1,\dots ,n$ in the following manner: first by strata $s\in \mathcal{S}$, then by treatment assignment $a\in \{ 0,1\} $, and then by decision $d\in \{0,1\}$. In other words, the units are reordered in increasing order of $ g\in \mathcal{G}$. Let $R_{n,1}^{C}$ denote the reordered sum. Let $ ( h ( i) :i=1,\dots ,n) $ denote the permutation of the units described by this reordering. Since $\breve{Y}_{i}(d,a,s) \overset{d}{=}( \tilde{Y}_{i}(d)-\mu (d,a,s)|D_{i}( a) =d,S_{i}=s) $, note that
\begin{align}
(( \breve{Y}_{h(i) }(d,a,s):(d,a,s)\in \{0,1\}^{2}\times \mathcal{S}) ^{\prime } :i=1,\dots ,n)
\overset{d}{=}(( \breve{Y}_{i}(d,a,s):(d,a,s)\in \{0,1\}^{2}\times \mathcal{S}) ^{\prime } :i=1,\dots ,n).\label{eq:ad_step1_eq10}
\end{align}
As a corollary of \eqref{eq:ad_step1_eq10},
\begin{equation}
(R_{n,1}^{B}|( ( D_{i},A_{i},S_{i})) _{i=1}^{n}) ~\overset{d}{=}~( R_{n,1}^{C}|( ( D_{i},A_{i},S_{i}) ) _{i=1}^{n}),
\label{eq:ad_step1_eq11}
\end{equation}
where
\begin{equation}
R_{n,1}^{C}~\equiv ~\left( \frac{1}{\sqrt{n}}\sum_{i=N_{G( d,a,s)} +1}^{N_{(G( d,a,s) +1)}}\check{Y} _{i}(d,a,s):(d,a,s)\in \{0,1\}^{2}\times \mathcal{S}\right)\label{eq:ad_step1_eq11B}
\end{equation}
and $(( \check{Y}_{i}(d,a,s):(d,a,s)\in \{0,1\}^{2}\times \mathcal{S} ) ^{\prime }:i=1,\dots ,n)$ denotes a matrix of $ 4|\mathcal{S}|\times n$ independent random variables, independent of $( ( D_{i},A_{i},S_{i}) ) _{i=1}^{n}$, with $\check{Y}_{i}(d,a,s) \overset{d}{=}( \tilde{Y}(d)-\mu (d,a,s)|D( a) =d,S=s) $.

By \eqref{eq:ad_step1_eq8} and \eqref{eq:ad_step1_eq11},
\begin{equation}
( R_{n,1}|( ( D_{i},A_{i},S_{i}) ) _{i=1}^{n}) ~\overset{d}{=}~( R_{n,1}^{C}|( ( D_{i},A_{i},S_{i}) ) _{i=1}^{n}) .
\label{eq:ad_step1_eq12}
\end{equation}

For any $( h_{1},h_{2},h_{3},h_{4}) \in \mathbb{R} ^{4|\mathcal{S}| }\times \mathbb{R} ^{|\mathcal{S}| }\times\mathbb{R} ^{|\mathcal{S}| }\times \mathbb{R} ^{|\mathcal{S}| }$, consider the following derivation.
\begin{align*}
&P( R_{n,1}\leq h_{1},R_{n,2}\leq h_{2},R_{n,3}\leq h_{3},R_{n,4}\leq h_{4})\\
&=E[ P( R_{n,1}\leq h_{1},R_{n,2}\leq h_{2},R_{n,3}\leq h_{3},R_{n,4}\leq h_{4}|( ( D_{i},A_{i},S_{i}) ) _{i=1}^{n}) ] \\
&\overset{(1)}{=}E[ P( R_{n,1}\leq h_{1}|( ( D_{i},A_{i},S_{i}) ) _{i=1}^{n}) 1[ R_{n,2}\leq h_{2},R_{n,3}\leq h_{3},R_{n,4}\leq h_{4}] ] \\
&\overset{(2)}{=}E[ P( R_{n,1}^{C}\leq h_{1}|( ( D_{i},A_{i},S_{i}) ) _{i=1}^{n}) 1[ R_{n,2}\leq h_{2},R_{n,3}\leq h_{3},R_{n,4}\leq h_{4}] ] \\
&\overset{(3)}{=}E[ P( R_{n,1}^{C}\leq h_{1},R_{n,2}\leq h_{2},R_{n,3}\leq h_{3},R_{n,4}\leq h_{4}|( ( D_{i},A_{i},S_{i}) ) _{i=1}^{n}) ] \\
&=P( R_{n,1}^{C}\leq h_{1},R_{n,2}\leq h_{2},R_{n,3}\leq h_{3},R_{n,4}\leq h_{4}) ,
\end{align*}
where (1) and (3) hold because $( R_{n,2},R_{n,3},R_{n,4}) $ is a nonstochastic function of $( ( D_{i},A_{i},S_{i}) ) _{i=1}^{n}$, and (2) holds by \eqref{eq:ad_step1_eq12}. Since the choice of $( h_{1},h_{2},h_{3},h_{4}) $ was arbitrary, \eqref{eq:ad_step1_eq1} follows.

For each $g\in \mathcal{G}$, let 
\begin{equation}
F_{g}~\equiv~ \sum_{s\in \mathcal{S}}[P(G(D(1),1,s)<g|S=s)\pi _{A}(s)p(s)+P(G(D(0),0,s)<g|S=s)(1-\pi _{A}(s))p(s)],
\label{eq:ad_step1_eq14}
\end{equation}
and also $F_{|\mathcal{G}|+1} =N_{4|\mathcal{S}|+1} = 1$. Also, define
\begin{equation}
R_{n,1}^{D}~\equiv ~\left( \frac{1}{\sqrt{n}}\sum_{i=\lfloor nF_{G(d,a,s)}\rfloor +1}^{\lfloor nF_{G(d,a,s)+1}\rfloor }\check{Y} _{i}(d,a,s):(d,a,s)\in \{0,1\}^{2}\times \mathcal{S}\right) . \label{eq:ad_step1_eq15}
\end{equation}
Since $R_{n,1}^{D}$ is a nonstochastic function of $((\check{Y} _{i}(d,a,s):(d,a,s)\in \{0,1\}^{2}\times \mathcal{S})^{\prime }:i=1,\dots ,n)$, $(R_{n,2},R_{n,3},R_{n,4})$ is a nonstochastic function of $ ((D_{i},A_{i},S_{i}))_{i=1}^{n}$, and $((\check{Y}_{i}(d,a,s):(d,a,s)\in \{0,1\}^{2}\times \mathcal{S})^{\prime }:i=1,\dots ,n)\perp ((D_{i},A_{i},S_{i}))_{i=1}^{n}$, we conclude that \eqref{eq:ad_step1_eq2} holds.

For each $(d,a,s,u)\in \{0,1\}^{2}\times \mathcal{S}\times (0,1]$, consider the following partial sum process:
\begin{equation*}
L_{n}(u)~=~\frac{1}{\sqrt{n}}\sum_{i=1}^{\lfloor nu\rfloor }\check{Y} _{i}(d,a,s).
\end{equation*}
Note that $\check{Y}_{i}(d,a,s)\overset{d}{=}(\tilde{Y}(d)-\mu (d,a,s)|D(a)=d,S=s)$ and so \eqref{eq:defnY} implies that $E[\check{Y} _{i}(d,a,s)]=0$ and $V[\check{Y}_{i}(d,a,s)]=\sigma ^{2}(d,a,s)$. By repeating arguments in the proof of \citet[Lemma B.2]{bugni/canay/shaikh:2018},
\begin{equation}
L_{n}(u)~\overset{d}{\to }~N(0,u\sigma ^{2}(d,a,s)).
\label{eq:ad_step1_eq16}
\end{equation}
By \eqref{eq:ad_step1_eq14}, \eqref{eq:ad_step1_eq16}, and the Brownian scaling relation,
%Since $((\check{Y}_{i}(d,a,s):(d,a,s)\in \{0,1\}^{2}\times \mathcal{S} )^{\prime }:i=1,\dots ,n)$, is i.i.d., $L_{n}(F_{G(d,a,s)})\perp L_{n}(F_{G(d,a,s)+1})-L_{n}(F_{G(d,a,s)})$. By this and that \eqref{eq:ad_step1_eq16} holds for $u\in \{F_{G(d,a,s)},F_{G(d,a,s)+1}\}$,
\begin{equation}
\frac{1}{\sqrt{n}}\sum_{i=\lfloor nF_{G(d,a,s)}\rfloor +1}^{\lfloor nF_{G(d,a,s)+1}\rfloor }\check{Y} _{i}(d,a,s)=L_{n}(F_{G(d,a,s)+1})-L_{n}(F_{G(d,a,s)})~\overset{d}{\to }~N(0,(F_{G(d,a,s)+1}-F_{G(d,a,s)})\sigma ^{2}(d,a,s)). \label{eq:ad_step1_eq17}
\end{equation}
Since $((\check{Y}_{i}(d,a,s):(d,a,s)\in \{0,1\}^{2}\times \mathcal{S} )^{\prime }:i=1,\dots ,n)$ are independent random variables, we conclude that for any $(d,a,s),(\tilde{d},\tilde{a},\tilde{s})\in \{0,1\}^{2}\times \mathcal{S}$ with $(d,a,s)\neq (\tilde{d},\tilde{a},\tilde{s})$,
\begin{equation}
\frac{1}{\sqrt{n}}\sum_{i=\lfloor nF_{G(d,a,s)}\rfloor +1}^{\lfloor nF_{G(d,a,s)+1}\rfloor }\check{Y}_{i}(d,a,s)~\perp ~\frac{1}{\sqrt{n}} \sum_{i=\lfloor nF_{G(\tilde{d},\tilde{a},\tilde{s})}\rfloor +1}^{\lfloor nF_{G(\tilde{d},\tilde{a},\tilde{s})+1}\rfloor }\check{Y}_{i}(\tilde{d}, \tilde{a},\tilde{s}). \label{eq:ad_step1_eq18}
\end{equation}
By \eqref{eq:ad_step1_eq17} and \eqref{eq:ad_step1_eq18},
\begin{equation}
R_{n,1}^{D}\overset{d}{\to }N(\mathbf{0},diag ((F_{G(d,a,s)+1}-F_{G(d,a,s)})\sigma ^{2}(d,a,s):(d,a,s)\in \{0,1\}^{2}\times \mathcal{S})) \label{eq:ad_step1_eq19}
\end{equation}
To show \eqref{eq:ad_step1_eq3} from \eqref{eq:ad_step1_eq19}, it then suffices to show that for all $(d,a,s)\in \{0,1\}^{2}\times \mathcal{S}$,
\begin{equation}
F_{G(d,a,s)+1}-F_{G(d,a,s)}~=~\left[
\begin{array}{c}
1[(a,d)=(0,0)](1-\pi _{D(0)}(s))(1-\pi _{A}(s)) \\ 
+1[(a,d)=(0,1)]\pi _{D(0)}(s)(1-\pi _{A}(s)) \\ 
+1[(a,d)=(1,0)](1-\pi _{D(1)}(s))\pi _{A}(s) \\ 
+1[(a,d)=(1,1)]\pi _{D(1)}(s)\pi _{A}(s)
\end{array}
\right] p(s).  \label{eq:ad_step1_eq20}
\end{equation}
We can show this from \eqref{eq:ad_step1_eq14} by using an inductive argument. As an initial step, note that \eqref{eq:ad_step1_eq14} implies that $F_{1}=0$, $F_{2}=(1-\pi _{D(0)}(1))(1-\pi _{A}(1))p(1)$, $F_{3}=(1-\pi _{A}(1))p(1)$, $F_{4}=(1-\pi _{A}(1))p(1)+(1-\pi _{D(1)}(1))\pi _{A}(1)p(1)$, and $F_{5} = p(1)$. As the inductive step, note that for $g=1,5,9,\dots ,4(|S|-1)$, \eqref{eq:ad_step1_eq14} implies that $F_{g+1}=F_{g} + (1-\pi _{D(0)}(s))(1-\pi _{A}(s))p(s)$, $F_{g+2}=F_{g} +(1-\pi _{A}(s))p(s)$, $F_{g+3}=F_{g} +(1-\pi _{A}(s))p(s)+(1-\pi _{D(1)}(s))\pi _{A}(s)p(s)$, and $F_{g+4}=F_{g} +p(s)$. By finite induction, \eqref{eq:ad_step1_eq20} follows.

By repeating arguments in the proof of \citet[Lemma B.2]{bugni/canay/shaikh:2018}, we can show \eqref{eq:ad_step1_eq4} follows from showing that $N_{g}/n\overset{p}{\to }F_{g}$ for all $g\in \mathcal{G}\cup\{|\mathcal{G}|+1\}$. To this end, consider the following argument for any $(d,a,s)\in \{0,1\}^{2}\times \mathcal{S}$.
\begin{align}
\frac{N_{G(d,a,s)+1}}{n}-\frac{N_{G(d,a,s)}}{n}& ~\overset{(1)}{=}~\left[ 
\begin{array}{c}
1[(a,d)=(0,0)](n(s)-n_{D}(s)-n_{A}(s)+n_{AD}(s)) \\ 
+1[(a,d)=(0,1)](n_{D}(s)-n_{AD}(s)) \\ 
+1[(a,d)=(1,0)](n_{A}(s)-n_{AD}(s)) \\ 
+1[(a,d)=(1,1)]n_{AD}(s)
\end{array}
\right] \frac{1}{n}  \notag \\
& =\left[ 
\begin{array}{c}
1[(a,d)=(0,0)](1-\frac{n_{D}(s)-n_{AD}(s)}{n(s)-n_{A}(s)})(1-\frac{n_{A}(s)}{
n(s)}) \\ 
+1[(a,d)=(0,1)](\frac{n_{D}(s)-n_{AD}(s)}{n(s)-n_{A}(s)})(1-\frac{n_{A}(s)}{
n(s)}) \\ 
+1[(a,d)=(1,0)](1-\frac{n_{AD}(s)}{n_{A}(s)})\frac{n_{A}(s)}{n(s)} \\ 
+1[(a,d)=(1,1)]\frac{n_{AD}(s)}{n_{A}(s)}\frac{n_{A}(s)}{n(s)}
\end{array}
\right] \frac{n(s)}{n}  \notag \\
& \overset{(2)}{=}F_{G(d,a,s)+1}-F_{G(d,a,s)}+o_{p}(1),
\label{eq:ad_step1_eq21}
\end{align}
where (1) follows from an induction argument similar to the one used to show \eqref{eq:ad_step1_eq20} and (2) follows from Assumptions \ref{ass:1}, \ref{ass:2}(b), Lemma \ref{lem:A1and2_impliesold3}, and the LLN, which implies that $\frac{n(s)}{n}=p(s)+o_{p}(1)$. By combining \eqref{eq:ad_step1_eq21} and $F_{1}=N_{1}/n=0$, the desired result follows.

\underline{Step 2.} Under Assumptions \ref{ass:1} and \ref{ass:3}, we show that  $( R_{n,2}',R_{n,3}',R_{n,4}') \overset{d}{\to}( \zeta _{2}',\zeta _{3}',\zeta _{4}') $.

By definition, $\zeta _{2}$ and $\zeta _{4}$ are continuously distributed, and $\zeta _{3}=( \zeta _{3,s}:s\in \mathcal{S}) $ is a vector of $|\mathcal{S}| $ independent coordinates, $\zeta _{3,s}$ is continuously distributed if $\tau ( s) >0$ and $\zeta _{3,s}=0$ if $\tau ( s) =0$. Then, $( h_{2}',h_{3}',h_{4}')' $ is continuity point of the CDF of $( \zeta _{2}',\zeta _{3}',\zeta _{4}')' $ if and only if $h_{3,s}\not=0$ for all $ s \in \mathcal{S}$ with $\tau ( s) =0$. Therefore, $ ( h_{2}',h_{3}',h_{4}')' $ is continuity point of the CDF of $( \zeta _{2}',\zeta _{3}',\zeta _{4}')' $ if and only if $h_{3}$ is continuity point of $\zeta _{3}$. For any such $( h_{2}',h_{3}',h_{4}')'$, consider the following argument.
\begin{align}
&\lim P( R_{n,2}\leq h_{2},R_{n,3}\leq h_{3},R_{n,4}\leq h_{4})\notag \\
&\overset{(1)}{=}\lim E[ E[ E[ 1( R_{n,2}\leq h_{2}) 1( R_{n,3}\leq h_{3}) 1( R_{n,4}\leq h_{4}) |( A_{i}) _{i=1}^{n},( S_{i}) _{i=1}^{n}] |( S_{i}) _{i=1}^{n}] ] \notag \\
&\overset{(2)}{=}\lim E[ E[ E[ 1( R_{n,2}\leq h_{2}) |( (A_{i},S_{i})) _{i=1}^{n}] 1( R_{n,3}\leq h_{3}) |( S_{i}) _{i=1}^{n}] 1( R_{n,4}\leq h_{4}) ] \notag \\
&=\lim \left[ 
\begin{array}{c}
E[ E[ ( P( R_{n,2}\leq h_{2}|( (A_{i},S_{i})) _{i=1}^{n}) -P( \zeta _{2}\leq h_{2}) ) 1( R_{n,3}^{c}\leq h_{3}) |( S_{i}) _{i=1}^{n}] 1( R_{n,4}\leq h_{4}) ] \\
+P( \zeta _{2}\leq h_{2}) E[ ( P( R_{n,3}\leq h_{3}|( S_{i}) _{i=1}^{n}) -P( \zeta _{3}\leq h_{3}) ) 1( R_{n,4}\leq h_{4}) ] \\
+P( \zeta _{2}\leq h_{2}) P( \zeta _{3}\leq h_{3}) ( P( R_{n,4}\leq h_{4}) -P( \zeta _{4}\leq h_{4}) ) +P( \zeta _{2}\leq h_{2}) P( \zeta _{3}\leq h_{3}) P( \zeta _{4}\leq h_{4})
\end{array}
\right], \label{eq:ad_step2_eq1}
\end{align}
where (1) follows from the LIE, and (2) follows from the fact that $R_{n,3}$ is nonstochastic conditional on $( (A_{i},S_{i})) _{i=1}^{n}$ and $ R_{n,4}$ is nonstochastic conditional on $( S_{i}) _{i=1}^{n}$. By \eqref{eq:ad_step2_eq1}, 
\begin{align}
&\vert \lim P(R_{n,2}\leq h_{2},R_{n,3}\leq h_{3},R_{n,4}\leq h_{4})-P(\zeta _{2}\leq h_{2})P(\zeta _{3}\leq h_{3})P(\zeta _{4}\leq h_{4})\vert \notag\\
&\leq \left[
\begin{array}{c}
\lim E[E[\vert P(R_{n,2}\leq h_{2}|((A_{i},S_{i}))_{i=1}^{n})-P(\zeta _{2}\leq h_{2})\vert |(S_{i})_{i=1}^{n}]] \\ +\lim E[\vert P(R_{n,3}\leq h_{3}|(S_{i})_{i=1}^{n})-P(\zeta _{3}\leq h_{3})\vert ] \\
+\lim \vert P(R_{n,4}\leq h_{4})-P(\zeta _{4}\leq h_{4})\vert
\end{array}
\right] \label{eq:ad_step2_eq2} 
\end{align}
The proof of this step is completed by showing that the three terms on the right hand side of \eqref{eq:ad_step2_eq2} are zero.

We begin with the first term. Fix $\varepsilon >0$ arbitrarily. It then suffices to find $N\in \mathbb{N} $ s.t.\ $\forall n\geq N$, $E[E[\vert P(R_{n,2}\leq h_{2}|((A_{i},S_{i}))_{i=1}^{n})-P(\zeta _{2}\leq h_{2})\vert |(S_{i})_{i=1}^{n}]]\leq \varepsilon$. By Assumption \ref{ass:2}(b) and Lemma \ref{lem:A1and2_impliesold3}, there exists a set of values of $((A_{i},S_{i}))_{i=1}^{n}$ denoted by $ M_{n}$ s.t.\ $P( ((A_{i},S_{i}))_{i=1}^{n}\in M_{n}) \to 1$ and for all $( ( a_{i},s_{i}) ) _{i=1}^{n}\in M_{n}$, $P(R_{n,2}\leq h_{2}|((A_{i},S_{i}))_{i=1}^{n}=( ( a_{i},s_{i}) ) _{i=1}^{n})\to P(\zeta _{2}\leq h_{2})$, where we are using that $\zeta _{2}$ is continuously distributed. This implies that $\exists N\in \mathbb{N} $ s.t.\ $\forall n\geq N$ and $\forall ( ( a_{i},s_{i})
) _{i=1}^{n}\in M_{n}$,
\begin{align}
&\vert P(R_{n,2}\leq h_{2}|((A_{i},S_{i}))_{i=1}^{n}=( ( a_{i},s_{i}) ) _{i=1}^{n})-P(\zeta _{2}\leq h_{2})\vert \leq \varepsilon /2 \label{eq:ad_step2_eq3} \\
&P( ((A_{i},S_{i}))_{i=1}^{n}\in M_{n}) \geq 1-\varepsilon /2.\label{eq:ad_step2_eq4}
\end{align}

Then,
\begin{align*}
&E[E[|P(R_{n,2} \leq h_{2}|((A_{i},S_{i}))_{i=1}^{n})-P(\zeta _{2}\leq h_{2})||(S_{i})_{i=1}^{n}]] \\
&=\left[ 
\begin{array}{c}
\int_{( ( a_{i},s_{i}) ) _{i=1}^{n}\in M_{n}}E[|P(R_{n,2}\leq h_{2}|((a_{i},s_{i}))_{i=1}^{n})-P(\zeta _{2}\leq h_{2})||(S_{i})_{i=1}^{n}=( s_{i}) _{i=1}^{n}]\times \\dP( ((A_{i},S_{i}))_{i=1}^{n}=( ( a_{i},s_{i}) ) _{i=1}^{n})+ \\
\int_{( ( a_{i},s_{i}) ) _{i=1}^{n}\in M_{n}^{c}}E[E[|P(R_{n,2}\leq h_{2}|((a_{i},s_{i}))_{i=1}^{n})-P(\zeta _{2}\leq h_{2})||(S_{i})_{i=1}^{n}=( s_{i}) _{i=1}^{n}]]\times\\
dP( ((A_{i},S_{i}))_{i=1}^{n}=( ( a_{i},s_{i}) ) _{i=1}^{n})
\end{array}
\right] \\
&\overset{(1)}{\leq}P( ((A_{i},S_{i}))_{i=1}^{n}\in M_{n}) \varepsilon /2+P( ((A_{i},S_{i}))_{i=1}^{n}\in M_{n}^{c}) \overset{(2)}{\leq}\varepsilon ,
\end{align*}
where (1) holds by \eqref{eq:ad_step2_eq3} and (2) holds by \eqref{eq:ad_step2_eq4}. This completes the proof for the first term on the right hand side of \eqref{eq:ad_step2_eq2}. The argument for the second term is similar, except that the argument that relies on Assumption \ref{ass:2}(b) would instead rely on Lemma \ref{lem:A1and2_impliesold3}. Finally, the argument for the third term holds by $ \zeta _{4}$ is continuously distributed and $R_{n,4}\overset{d}{\to }\zeta _{4}$, which holds by Assumption \ref{ass:1}, $S(Z_{i})=S_{i}$, and the CLT.

\underline{Step 3.} We now combine steps 1 and 2 to complete the proof. Let $ ( h_{1}',h_{2}',h_{3}',h_{4}') $ be a continuity point of the CDF of $( \zeta _{1}',\zeta _{2}',\zeta _{3}',\zeta _{4}') $. By the same argument as in step 2, this implies that $h_{3,s}\not=0$ for all $s \in \mathcal{S}$ with $\tau ( s) =0$. Under these conditions, consider the following derivation.
\begin{align*}
\lim P( R_{n,1}\leq h_{1},R_{n,2}\leq h_{2},R_{n,3}\leq h_{3},R_{n,4}\leq h_{4}) &\overset{(1)}{=}\lim P( R_{n,1}^{C}\leq h_{1},R_{n,2}\leq h_{2},R_{n,3}\leq h_{3},R_{n,4}\leq h_{4}) \\ 
&\overset{(2)}{=}\lim P( R_{n,1}^{D}\leq h_{1},R_{n,2}\leq h_{2},R_{n,3}\leq h_{3},R_{n,4}\leq h_{4}) \\
&\overset{(3)}{=}\lim P( R_{n,1}^{D}\leq h_{1}) \lim P( R_{n,2}\leq h_{2},R_{n,3}\leq h_{3},R_{n,4}\leq h_{4}) \\
&\overset{(4)}{=}P( \zeta _{1}\leq h_{1}) P( \zeta _{2}\leq h_{2}) P( \zeta _{3}\leq h_{3}) P( \zeta _{4}\leq h_{4}) ,
\end{align*}
as desired, where (1) holds by \eqref{eq:ad_step1_eq1} in step 1, (2)  holds by \eqref{eq:ad_step1_eq4} in step 1, (3) holds by \eqref{eq:ad_step1_eq2} in step 1, and (4) holds by \eqref{eq:ad_step1_eq4} in step 1 and \eqref{eq:ad_step2_eq2} in step 2.
\end{proof}
%%%%%%%% DIVIDER %%%%%%%%%%%%

%%%%%%%% DIVIDER %%%%%%%%%%%%
\begin{lemma}\label{lem:AsyDist2}
Assume Assumptions \ref{ass:1} and \ref{ass:2}. For any $( d,a,s) \in \{ 0,1\}^2 \times \mathcal{S}$,
\begin{align}
\frac{R_{n,1}(d,a,s)}{\sqrt{n}} &~=~ \frac{1}{n}\sum_{i=1}^{n}1[ D_{i}=d,A_{i}=a,S_{i}=s] ( \tilde{Y}_{i}( d) -\mu ( d,a,s) )~=~o_p(1) \label{eq:AsyLimit2_eq1}
\end{align}
and
\begin{align}
R_{n,5}(d,a,s) &~\equiv~ \frac{1}{n}\sum_{i=1}^{n}1[ D_{i}=d,A_{i}=a,S_{i}=s] ( \tilde{Y}_{i}(d)-\mu ( d,a,s) )^{2}\notag\\
& ~=~\left[
\begin{array}{c}
1[ ( a,d) =( 0,0) ] ( 1-\pi _{A}( s) ) ( 1-\pi _{D( 0) }( s) ) \\
+1[ ( a,d) =( 0,1) ] ( 1-\pi _{A}( s) ) \pi _{D( 0) }( s) \\
+1[ ( a,d) =( 1,0) ] \pi _{A}( s) ( 1-\pi _{D( 1) }( s) ) \\
+1[ ( a,d) =( 1,1) ] \pi _{A}( s) \pi _{D( 1) }( s)
\end{array}
\right] p( s) \sigma ^{2}( d,a,s)~+~o_p(1) ,\label{eq:AsyLimit2_eq2}
\end{align}
where $R_{n}$ is as in \eqref{eq:Rn_defn}.
\end{lemma}
%%%%%%%% DIVIDER %%%%%%%%%%%%
\begin{proof}
Fix $( d,a,s) \in \{0,1\}^{2}\times \mathcal{S}$ arbitrarily throughout this proof.  We begin by showing \eqref{eq:AsyLimit2_eq1}. Under our current assumptions, step 1 of the proof of Lemma \ref{lem:AsyDist} implies that
\begin{equation*}
\frac{R_{n,1}(d,a,s)}{\sqrt{n}}~\overset{d}{=}~\frac{R_{n,1}^{C}(d,a,s)}{\sqrt{ n}}~=~\frac{R_{n,1}^{D}(d,a,s)}{\sqrt{n}}+ \frac{o_{p}(1)}{\sqrt{n}}  ~=~\frac{R_{n,1}^{D}(d,a,s)}{\sqrt{n}}+o_{p}( 1) , 
\end{equation*}
where $R_{n,1}^{C}$ and $R_{n,1}^{D}$ are defined in \eqref{eq:ad_step1_eq11B} and \eqref{eq:ad_step1_eq15}, respectively. Therefore, \eqref{eq:AsyLimit2_eq1} follows from the following derivation.
\begin{align*}
\frac{R_{n,1}^{D}(d,a,s)}{\sqrt{n}} &~\overset{(1)}{=}~\frac{1}{n}\sum_{i=\lfloor n F_{G( d,a,s) }\rfloor +1}^{\lfloor n F_{G( d,a,s) +1}\rfloor }\check{Y} _{i}( d,a,s)\\
&~=~1[F_{G( d,a,s) +1}>F_{G( d,a,s)}]\frac{\lfloor nF_{G( d,a,s) +1}\rfloor -\lfloor nF_{G( d,a,s) }\rfloor }{n} \frac{\sum_{i=\lfloor nF_{G( d,a,s) }\rfloor +1}^{\lfloor nF_{G( d,a,s) +1}\rfloor }\check{Y}_{i}( d,a,s)}{ \lfloor nF_{G( d,a,s) +1}\rfloor -\lfloor nF_{G( d,a,s) }\rfloor }  \\
&~\overset{(2)}{=}~1[F_{G( d,a,s) +1}>F_{G( d,a,s)}]( F_{G( d,a,s) +1}-F_{G( d,a,s) }+o( 1) ) o_{p}( 1) =o_{p}( 1) ,
\end{align*}
as required, where (1) holds by \eqref{eq:ad_step1_eq15} with $ F_{g}$ as defined in \eqref{eq:ad_step1_eq14} in step 1 of the proof of Lemma \ref{lem:AsyDist}, and $( \check{Y}_{i}( d,a,s) :i=1,\ldots ,n) $ given by an i.i.d.\ sequence with $\check{Y}_{i}(d,a,s)\overset{d}{=}\{\tilde{Y} (d)-\mu (d,a,s)|D(a)=d,S=s\}$, (2) holds by $E[\check{ Y}_{i}(d,a,s)]=0$ (due to \eqref{eq:defnY}) and the LLN.

We now show \eqref{eq:AsyLimit2_eq2}. By repeating arguments used in step 1 of the proof of Lemma \ref{lem:AsyDist}, we can show that
\begin{equation*}
R_{n,5}(d,a,s)~\overset{d}{=}~R_{n,5}^{D}(d,a,s)+o_{p}( 1) ,
\end{equation*}
where
\begin{equation*}
R_{n,5}^{D}(d,a,s)~\equiv~\frac{1}{n}\sum_{i=\lfloor nF_{G( d,a,s) }\rfloor +1}^{\lfloor nF_{G( d,a,s) +1}\rfloor }U_{i}( d,a,s)
\end{equation*}
and $( U_{i}( d,a,s) :i=1,\ldots ,n) $ is an i.i.d.\ sequence with $U_{i}( d,a,s) \overset{d}{=}\{(\tilde{Y}(d)-\mu (d,a,s))^{2}|D(a)=d,S=s\}$. To show \eqref{eq:AsyLimit2_eq2}, consider the following argument.
\begin{align*}
R_{n,5}^{D}(d,a,s) 
&~=~\frac{\lfloor nF_{G( d,a,s) +1}\rfloor -\lfloor nF_{G( d,a,s) }\rfloor }{n} \frac{\sum_{i=\lfloor nF_{G( d,a,s) }\rfloor +1}^{\lfloor nF_{G( d,a,s) +1}\rfloor }U_{i}( d,a,s)}{ \lfloor nF_{G( d,a,s) +1}\rfloor -\lfloor nF_{G( d,a,s) }\rfloor } \\
&~\overset{(1)}{=}~1[F_{G( d,a,s) +1}>F_{G( d,a,s)}]( F_{G( d,a,s) +1}-F_{G( d,a,s) }+o( 1) ) (\sigma ^{2}(d,a,s)+o_{p}( 1) ) \\
&~\overset{(2)}{=}~\left[
\begin{array}{c}
1[(a,d)=(0,0)](1-\pi _{A}(s))(1-\pi _{D(0)}(s)) \\ 
+1[(a,d)=(0,1)](1-\pi _{A}(s))\pi _{D(0)}(s) \\ 
+1[(a,d)=(1,0)]\pi _{A}(s)(1-\pi _{D(1)}(s)) \\ 
+1[(a,d)=(1,1)]\pi _{A}(s)\pi _{D(1)}(s)
\end{array}
\right] p(s)\sigma ^{2}(d,a,s)+o_{p}( 1) ,
\end{align*}
where (1) holds by $ E[U_{i}( d,a,s) ]=\sigma ^{2}(d,a,s)$ (due to \eqref{eq:defnY}) and the LLN, and (2) follows from \eqref{eq:ad_step1_eq20}.
\end{proof}
%%%%%%%% DIVIDER %%%%%%%%%%%%

%%%%%%%%%%%%%%%%%%%%%%%%%%%%%%%%%%%%%%%%%%%%%%%%%%%%%%%%%%%%%
\subsection{Proofs of results related to Section \ref{sec:SAT}}

%%%%%%%% DIVIDER %%%%%%%%%%%%
\begin{lemma}[SAT matrices]\label{lem:MatrixSAT}
Assume Assumptions \ref{ass:1} and \ref{ass:2}. Then, 
\begin{align*}
{{\mathbf{Z}}^{\mathrm{sat}}_{n}}'\mathbf{X}^{\mathrm{sat}}_{n}/{n} &=\left[ 
\begin{array}{cc}
diag( n( s) /n:s\in \mathcal{S}) & diag( n_{D}(s)/n:s\in \mathcal{S}) \\
diag( n_{A}(s)/n:s\in \mathcal{S}) & diag( n_{AD}(s)/n:s\in \mathcal{S})
\end{array}
\right]  \\
&=\left[ 
\begin{array}{cc}
diag( p( s) :s\in \mathcal{S}) & diag( [ \pi _{D( 1) }( s) \pi _{A}( s) +\pi _{D( 0) }( s) ( 1-\pi _{A}( s) ) ] p( s) :s\in \mathcal{S}) \\
diag( \pi _{A}( s) p( s) :s\in \mathcal{S} ) & diag( \pi _{D( 1) }( s) \pi _{A}( s) p( s) :s\in \mathcal{S})
\end{array}
\right]  +o_{p}(1),
\end{align*}
and thus
\begin{align*}
&( {{{\mathbf{Z}}^{\mathrm{sat}}_{n}}'\mathbf{X}_{n}^{\mathrm{sat}}}/{n}) ^{-1} \\
&=
\left[ 
\begin{array}{cc}
diag( \frac{n_{AD}(s)n}{n( s) n_{AD}(s)-n_{A}(s)n_{D}(s)} :s\in \mathcal{S}) & diag( \frac{-n_{D}(s)n}{n( s) n_{AD}(s)-n_{A}(s)n_{D}(s)}:s\in \mathcal{S}) \\
diag( \frac{-n_{A}(s)n}{n( s) n_{AD}(s)-n_{A}(s)n_{D}(s)} :s\in \mathcal{S}) & diag( \frac{n(s)n}{ n(s)n_{AD}(s)-n_{A}(s)n_{D}(s)}:s\in \mathcal{S})
\end{array}
\right]  +o_{p}(1) \\
&=\left[ 
\begin{array}{cc}
diag( \frac{\pi _{D( 1) }( s) }{p(s)( 1-\pi _{A}( s) ) [ \pi _{D( 1) }( s) -\pi _{D( 0) }( s) ] }:s\in \mathcal{S})
& diag( \frac{-[ \pi _{D( 1) }( s) \pi _{A}( s) +\pi _{D( 0) }( s) ( 1-\pi _{A}( s) ) ] }{p(s)\pi _{A}( s) ( 1-\pi _{A}( s) ) [ \pi _{D( 1) }( s) -\pi _{D( 0) }( s) ] }:s\in \mathcal{S} ) \\
diag( \frac{-1}{p(s)( 1-\pi _{A}( s) ) [ \pi _{D( 1) }( s) -\pi _{D( 0) }( s) ] }:s\in \mathcal{S}) & diag( \frac{1}{p(s)\pi _{A}( s) ( 1-\pi _{A}( s) ) [ \pi _{D( 1) }( s) -\pi _{D( 0) }( s) ] }:s\in \mathcal{S})
\end{array}
\right]  +o_{p}(1).
\end{align*}
Also,
\begin{equation*}
{{\mathbf{Z}}^{\mathrm{sat}}_{n}}'  \mathbf{Y}_{n}/{n}=\left[ 
\begin{array}{c}
( \frac{1}{n}\sum_{i=1}^{n}1[ S_{i}=s] Y_{i} :s\in \mathcal{S}) , \\
( \frac{1}{n}\sum_{i=1}^{n} 1[A_{i}=1, S_{i}=s] Y_{i} :s\in \mathcal{S} )
\end{array}
\right] .
\end{equation*}
\end{lemma}
%%%%%%%% DIVIDER %%%%%%%%%%%%
\begin{proof}
This equalities follow from algebra and the convergences follow from the CMT. In particular, the first equality in the second display has an $o_p(1)$ to allow for the possibility that ${{{\mathbf{Z}}^{\mathrm{sat}}_{n}}'\mathbf{X}_{n}^{\mathrm{sat}}}/{n}$ is singular or $n(s)n_{AD}(s)=n_{A}(s)n_{D}(s)$ for some $s \in \mathcal{S}$. Both of these events occur with vanishing probability under our assumptions.
\end{proof}
%%%%%%%% DIVIDER %%%%%%%%%%%%

%%%%%%%% DIVIDER %%%%%%%%%%%%
\begin{theorem}[SAT limits]\label{thm:plim_SAT}
Assume Assumptions \ref{ass:1} and \ref{ass:2}. Then, for every $ s\in \mathcal{S}$,
\begin{align}
\hat{\beta}_{\mathrm{sat}}(s) &~\overset{p}{\to }~\beta ( s) ~\equiv~E[ Y( 1)-Y( 0) |C,S=s]\notag\\
\hat{\gamma}_{\mathrm{sat}}(s) &~\overset{p}{\to }~\gamma ( s)~\equiv~ \left[ 
\begin{array}{c}
\pi _{D( 1) }( s) E[ Y( 0) |C,S=s] -\pi _{D( 0) }( s) E[ Y( 1) |C,S=s ] + \\
\pi _{D( 0) }( s) E[ Y( 1) |AT,S=s ] +( 1-\pi _{D( 1) }( s) ) E[ Y( 0) |NT,S=s]
\end{array}
\right]\notag\\
\hat{P}(S=s,C)& ~\overset{p}{\to }~P(S=s,C)~\equiv ~p(s)(\pi _{D(1)}(s)-\pi _{D(0)}(s)) \notag\\\
\hat{P}(S=s|C)&~\overset{p}{\to } ~P(S=s|C)~\equiv ~\frac{p(s)(\pi _{D(1)}(s)-\pi _{D(0)}(s))}{\sum_{\tilde{s}\in \mathcal{S} }p(\tilde{s})(\pi _{D(1)}(\tilde{s})-\pi _{D(0)}(\tilde{s}))},\label{eq:hat_limits_1}
\end{align}
where $(\hat{\beta}_{\mathrm{sat}}(s),\hat{\gamma}_{\mathrm{sat}}(s))$ is as in \eqref{eq:IV_SAT_estimators} and $(\hat{P}(S=s,C),\hat{P}(S=s|C))$ is as is in \eqref{eq:P_C}. Also,
\begin{align}
\hat{\beta}_{\mathrm{sat}} &~\overset{p}{\to }~\beta  ~\equiv~E[ Y( 1)-Y( 0) |C] \notag\\
\hat{P}(C)& ~\overset{p}{\to }~P(C)~\equiv ~\sum_{s\in \mathcal{S} }p(s)(\pi _{D(1)}(s)-\pi _{D(0)}(s)),\label{eq:hat_limits_2}
\end{align}
where $\hat{\beta}_{\mathrm{sat}}$ is as in \eqref{eq:beta_SAT} and $\hat{P}(C)$ is as in \eqref{eq:P_C}.
\end{theorem}
%%%%%%%% DIVIDER %%%%%%%%%%%%
\begin{proof}
We focus on showing \eqref{eq:hat_limits_1}, as \eqref{eq:hat_limits_2} follows from \eqref{eq:hat_limits_1} and CMT. To show the first line of \eqref{eq:hat_limits_1}, consider the following derivation.
\begin{align}
\hat{\beta}_{\mathrm{sat}}(s)&\overset{(1)}{=}\tfrac{n(s)n}{n(s)n_{AD}(s)-n_{A}(s)n_{D}(s)} \frac{1}{n}\sum_{i=1}^{n}1[A_{i}=1,S_{i}=s]Y_{i}-\tfrac{n_{A}(s)n}{ n(s)n_{AD}(s)-n_{A}(s)n_{D}(s)}\frac{1}{n}\sum_{i=1}^{n}1[S_{i}=s]Y_{i} +o_{P}(1)\notag \\
& \overset{(2)}{=}\tfrac{\left[ 
\begin{array}{c}
(n(s)-n_{A}(s))\sum_{i=1}^{n}1[D_{i}=1,A_{i}=1,S_{i}=s][\tilde{Y} _{i}(1)+E[Y(1)|S=s]] \\
+(n(s)-n_{A}(s))\sum_{i=1}^{n}1[D_{i}=0,A_{i}=1,S_{i}=s][\tilde{Y} _{i}(0)+E[Y(0)|S=s]] \\
-n_{A}(s)\sum_{i=1}^{n}1[D_{i}=1,A_{i}=0,S_{i}=s][\tilde{Y} _{i}(1)+E[Y(1)|S=s]] \\
-n_{A}(s)\sum_{i=1}^{n}1[D_{i}=0,A_{i}=0,S_{i}=s][\tilde{Y} _{i}(0)+E[Y(0)|S=s]]
\end{array}
\right] }{n(s)n_{AD}(s)-n_{A}(s)n_{D}(s)}  +o_{P}(1)\notag \\
& \overset{(3)}{=}\tfrac{\left[ 
\begin{array}{c}
(n(s)-n_{A}(s))\sum_{i=1}^{n}1[D_{i}=1,A_{i}=1,S_{i}=s]\tilde{Y}_{i}(1) \\ 
+(n(s)-n_{A}(s))\sum_{i=1}^{n}1[D_{i}=0,A_{i}=1,S_{i}=s]\tilde{Y}_{i}(0) \\ 
-n_{A}(s)\sum_{i=1}^{n}1[D_{i}=1,A_{i}=0,S_{i}=s]\tilde{Y}_{i}(1) \\ 
-n_{A}(s)\sum_{i=1}^{n}1[D_{i}=0,A_{i}=0,S_{i}=s]\tilde{Y}_{i}(0)
\end{array}
\right] }{n(s)n_{AD}(s)-n_{A}(s)n_{D}(s)}+E[Y(1)-Y(0)|S=s]+o_{P}(1),
\label{eq:derivation1}
\end{align}
where (1) holds by \eqref{eq:IV_SAT_estimators} and Lemma \ref{lem:MatrixSAT}, (2) holds by the fact that, conditional on $ (D_{i},S_{i})=(d,s)$, $Y_{i}=Y_{i}(d)=\tilde{Y}_{i}(d)+E[Y(d)|S=s]$ (by \eqref{eq:defnY_pre}), and (3) holds by \eqref{eq:defnY} and the following algebraic derivation:
\begin{align*}
E[Y(1)-Y(0)|S=s] 
%&= 
%\tfrac{\left[ 
%\begin{array}{c}
%(n(s)-n_{A}(s))n_{AD}(s)E[Y(1)|S=s] \\ 
%+(n(s)-n_{A}(s))(n_{A}(s)-n_{AD}(s))E[Y(0)|S=s] \\ 
%-n_{A}(s)(n_{D}(s)-n_{AD}(s))E[Y(1)|S=s] \\ 
%-n_{A}(s)(n(s)-n_{A}(s)-n_{D}(s)+n_{AD}(s))E[Y(0)|S=s]
%\end{array}
%\right] }{n(s)n_{AD}(s)-n_{A}(s)n_{D}(s)} \\
& =\tfrac{\left[ 
\begin{array}{c}
(n(s)-n_{A}(s))\sum_{i=1}^{n}1[D_{i}=1,A_{i}=1,S_{i}=s]E[Y(1)|S=s] \\ 
+(n(s)-n_{A}(s))\sum_{i=1}^{n}1[D_{i}=0,A_{i}=1,S_{i}=s]E[Y(0)|S=s] \\ 
-n_{A}(s)\sum_{i=1}^{n}1[D_{i}=1,A_{i}=0,S_{i}=s]E[Y(1)|S=s] \\ 
-n_{A}(s)\sum_{i=1}^{n}1[D_{i}=0,A_{i}=0,S_{i}=s]E[Y(0)|S=s]
\end{array}
\right] }{n(s)n_{AD}(s)-n_{A}(s)n_{D}(s)}+o_{P}(1) .
\end{align*}

To complete the proof of the first line of \eqref{eq:hat_limits_1}, consider the following derivation.
\begin{align}
& \hat{\beta}_{\mathrm{sat}}(s)-\beta ( s)  \notag\\
& \overset{(1)}{=}\tfrac{\left[ 
\begin{array}{c}
\frac{n}{n( s) }(1-\frac{n_{A}(s)}{n( s) })\frac{1}{ \sqrt{n}}R_{n,1}( 1,1,s) +\frac{n_{A}( s) }{n( s) }(1-\frac{n_{A}(s)}{n( s) })\frac{n_{AD}( s) }{n_{A}( s) }\mu ( 1,1,s) \\
+\frac{n}{n( s) }(1-\frac{n_{A}(s)}{n( s) })\frac{1}{ \sqrt{n}}R_{n,1}( 0,1,s) +\frac{n_{A}( s) }{n( s) }(1-\frac{n_{A}(s)}{n( s) })( 1-\frac{n_{AD}( s) }{n_{A}( s) }) \mu ( 0,1,s) \\
-\frac{n}{n( s) }\frac{n_{A}(s)}{n( s) }\frac{1}{\sqrt{ n}}R_{n,1}( 1,0,s) -\frac{n_{A}(s)}{n( s) }(1-\frac{ n_{A}(s)}{n( s) })( \frac{n_{D}( s) -n_{AD}( s) }{n( s) -n_{A}( s) }) \mu ( 1,0,s) \\
-\frac{n}{n( s) }\frac{n_{A}(s)}{n( s) }\frac{1}{\sqrt{ n}}R_{n,1}( 0,0,s) -\frac{n_{A}(s)}{n( s) }( 1- \frac{n_{A}(s)}{n( s) }) ( 1-\frac{n_{D}( s) -n_{AD}( s) }{n( s) -n_{A}( s) } ) \mu ( 0,0,s) \\
+\frac{n_{A}( s) }{n( s) }( 1-\frac{n_{A}(s)}{ n( s) }) ( \frac{n_{AD}( s) }{n_{A}( s) }-\frac{n_{D}(s)-n_{AD}( s) }{n( s) -n_{A}( s) }) ( E[Y(1)-Y(0)|S=s]-\beta ( s) )
\end{array}
\right] }{\frac{n_{A}( s) }{n( s) }( 1-\frac{ n_{A}(s)}{n( s) }) ( \frac{n_{AD}( s) }{ n_{A}( s) }-\frac{n_{D}(s)-n_{AD}( s) }{n( s) -n_{A}( s) }) } +o_{P}(1)\notag \\
& \overset{(2)}{=}\tfrac{\left[ 
\begin{array}{c}
\frac{n_{AD}( s) }{n_{A}( s) }( E[Y(1)|AT,S=s] \tfrac{\pi _{D(0)}(s)}{\pi _{D(1)}(s)}+E[Y(1)|C,S=s]\tfrac{\pi _{D(1)}(s)-\pi _{D(0)}(s)}{\pi _{D(1)}(s)}-E[Y(1)|S=s]) \\
+( 1-\frac{n_{AD}( s) }{n_{A}( s) }) ( E[Y(0)|NT,S=s]-E[Y(0)|S=s]) \\
-( \frac{n_{D}( s) -n_{AD}( s) }{n( s) -n_{A}( s) }) ( E[Y(1)|AT,S=s]-E[Y(1)|S=s]) \\
-( 1-\frac{n_{D}( s) -n_{AD}( s) }{n( s) -n_{A}( s) })\times\\ (
E[Y(0)|NT,S=s]\tfrac{1-\pi _{D(1)}(s)}{1-\pi _{D(0)}(s)}+E[Y(0)|C,S=s]\tfrac{ \pi _{D(1)}(s)-\pi _{D(0)}(s)}{1-\pi _{D(0)}(s)}-E[Y(0)|S=s])  \\ 
+( \frac{n_{AD}( s) }{n_{A}( s) }-\frac{ n_{D}(s)-n_{AD}( s) }{n( s) -n_{A}( s) } ) ( E[Y(1)-Y(0)|S=s]-\beta ( s) )
\end{array}
\right] }{( \frac{n_{AD}( s) }{n_{A}( s) }-\frac{ n_{D}(s)-n_{AD}( s) }{n( s) -n_{A}( s) } ) }+o_{p}( 1)\notag \\
& \overset{(3)}{=}o_{p}( 1) ,\notag
\end{align}
as desired, where (1) holds by \eqref{eq:Rn_defn}, \eqref{eq:derivation1}, and $\beta ( s) =E[Y(1)-Y(0)|C, S=s]$, (2) holds by Lemma \ref{lem:MeanTraslation}, and (3) holds by Assumption \ref{ass:1} and Lemma \ref{lem:A1and2_impliesold3}.

To show the second line of \eqref{eq:hat_limits_1}, consider the following argument.
\begin{align*}
\hat{\gamma}_{\mathrm{sat}}(s) &\overset{(1)}{=}\tfrac{n_{AD}(s)n}{n( s) n_{AD}(s)-n_{A}(s)n_{D}(s)}\frac{1}{n}\sum_{i=1}^{n}Y_{i}1[ S_{i}=s] -\tfrac{n_{D}(s)n}{n( s) n_{AD}(s)-n_{A}(s)n_{D}(s) }\frac{1}{n}\sum_{i=1}^{n}Y_{i}A_{i}1[ S_{i}=s] +o_{P}(1)\\
&\overset{(2)}{=}
 \frac{\left[
 \begin{array}{c}
 -[ ( \frac{n_{D}(s)-n_{AD}(s)}{n( s) -n_{A}( s) }) ( 1-\frac{n_{A}(s)}{n( s) }) \frac{n }{n( s) }] [\frac{1}{ \sqrt{n}}(R_{n,1}( 1,1,s)+R_{n,1}( 0,1,s))] \\
 +[ \frac{n_{AD}(s)}{n_{A}( s) }\frac{n_{A}( s) }{ n( s) }\frac{n}{n( s) }]  [\frac{1}{ \sqrt{n}}(R_{n,1}( 1,0,s)+R_{n,1}( 0,0,s))] \\
 -[ \frac{n_{A}( s) }{n( s) }( 1-\frac{ n_{A}( s) }{n( s) }) \frac{n_{AD}( s) }{n_{A}( s) }\frac{n_{D}(s)-n_{AD}(s)}{n( s) -n_{A}( s) }] [ \mu ( 1,1,s) +E[ Y( 1) |S=s] ] \\
 -[ \frac{n_{A}( s) }{n( s) }( 1-\frac{ n_{A}( s) }{n( s) }) ( 1-\frac{n_{AD}( s) }{n_{A}( s) }) \frac{n_{D}(s)-n_{AD}(s)}{n( s) -n_{A}( s) }] [ \mu ( 0,1,s) +E [ Y( 0) |S=s] ] \\
 +[ \frac{n_{A}( s) }{n( s) }( 1-\frac{ n_{A}( s) }{n( s) }) \frac{n_{AD}(s)}{ n_{A}( s) }\frac{n_{D}(s)-n_{AD}(s)}{n( s) -n_{A}( s) }] [ \mu ( 1,0,s) +E[ Y( 1) |S=s] ] \\
 +[ \frac{n_{A}( s) }{n( s) }( 1-\frac{ n_{A}( s) }{n( s) }) \frac{n_{AD}(s)}{ n_{A}( s) }( 1-\frac{n_{D}(s)-n_{AD}(s)}{n( s) -n_{A}( s) }) ] [ \mu ( 0,0,s) +E [ Y( 0) |S=s] ]
 \end{array}
 \right] }{\frac{n_{A}( s) }{n( s) }( 1-\frac{ n_{A}(s)}{n( s) }) ( \frac{n_{AD}(s)}{n_{A}( s) }-\frac{n_{D}(s)-n_{AD}( s) }{n( s) -n_{A}( s) }) } +o_{P}(1)\\
&\overset{(3)}{=}\frac{\left[
\begin{array}{c}
-[ \pi _{D( 0) }( s) \pi _{D( 1) }( s) ] [ \mu ( 1,1,s) +E[ Y( 1) |S=s] ] \\
-[ \pi _{D( 0) }( s) ( 1-\pi _{D( 1) }( s) ) ] [ \mu ( 0,1,s) +E [ Y( 0) |S=s] ] \\
+[ \pi _{D( 1) }( s) \pi _{D( 0) }( s) ] [ \mu ( 1,0,s) +E[ Y( 1) |S=s] ] \\
+[ \pi _{D( 1) }( s) ( 1-\pi _{D( 0) }( s) ) ] [ \mu ( 0,0,s) +E [ Y( 0) |S=s] ]
\end{array}
\right] }{( \pi _{D( 1) }( s) -\pi _{D( 0) }( s) ) }+o_{p}( 1) \\
&\overset{(4)}{=}\left[ 
\begin{array}{c}
\pi _{D( 0) }( s) E[ Y( 1) |AT,S=s ] +( 1-\pi _{D( 1) }( s) ) E[ Y( 0) |NT,S=s]\\
-\pi _{D( 0) }( s) E[ Y( 1) |C,S=s ] +\pi _{D( 1) }( s) E[ Y( 0) |C,S=s]
\end{array}
\right] +o_{p}( 1) ,
\end{align*}
where (1) follows from \eqref{eq:IV_SAT_estimators} and Lemma \ref{lem:MatrixSAT}, (2) follows from the fact that, conditional on $(D_i,S_i)=(d,s)$, $Y_{i}=Y_{i}(d) =\tilde{Y}_{i}( d) +E[ Y( d) |S=s] $ (by \eqref{eq:defnY_pre}), (3) follows from Assumptions \ref{ass:1} and \ref{ass:2}, the LLN, and Lemmas \ref{lem:A1and2_impliesold3} and \ref{lem:AsyDist2}, and (4) follows from Lemma \ref{lem:MeanTraslation}.

To conclude the proof, note that the third line of \eqref{eq:hat_limits_1} holds by Lemma \ref{lem:AsyDist}. In turn, this and the CMT implies the fourth line of \eqref{eq:hat_limits_1}.
\end{proof}

%%%%%%%% DIVIDER %%%%%%%%%%%%
\begin{theorem}[SAT representation]\label{thm:Expression}
Assume Assumptions \ref{ass:1} and \ref{ass:2}. Then, for any $s\in \mathcal{S}$,
\begin{equation*}
\sqrt{n}(\hat{\beta}_{\mathrm{sat}}(s)-\beta (s))=\xi _{n,1}(s)+\xi _{n,2}(s)+\xi _{n,3}(s)+o_{p}(1),
\end{equation*}
where $\hat{\beta}_{\mathrm{sat}}(s)$ is as in \eqref{eq:IV_SAT_estimators}, $\beta(s)$ is as in \eqref{eq:SAT_limit}, and also
\begin{align}
%\beta (s) &\equiv E[Y(1)-Y(0)|C,S=s] \notag \\
\xi _{n,1}(s) &\equiv \frac{(1-\pi _{A}(s))(R_{n,1}(1,1,s)+R_{n,1}(0,1,s))-\pi _{A}(s)(R_{n,1}(1,0,s)+R_{n,1}(0,0,s))}{p(s)\pi _{A}(s)(1-\pi _{A}(s))(\pi _{D(1)}(s)-\pi _{D(0)}(s))} \notag \\
\xi _{n,2}(s) &\equiv \frac{R_{n,2}(1,s)}{\pi _{D( 1) }( s) -\pi _{D( 0) }( s) }\left[
\begin{array}{c}
( E[ Y( 0) |C,S=s] -E[ Y( 0) |NT,S=s] ) \\
-\frac{\pi _{D( 0) }( s) }{\pi _{D( 1) }( s) }( E[ Y( 1) |C,S=s] -E[ Y( 1) |AT,S=s] )
\end{array}
\right] \notag \\
\xi_{n,3}(s) &\equiv \frac{R_{n,2}(2,s)}{\pi _{D( 1) }( s) -\pi _{D( 0) }( s) }\left[
\begin{array}{c}
( E[ Y( 1) |C,S=s] -E[ Y( 1) |AT,S=s] ) \\
-\frac{1-\pi _{D( 1) }( s) }{1-\pi _{D( 0) }( s) }( E[ Y( 0) |C,S=s] -E[ Y( 0) |NT,S=s] )
\end{array}
\right], \label{eq:Expression1}
\end{align}
with $R_{n,1}$ and $R_{n,2}$ as in \eqref{eq:Rn_defn}.
\end{theorem}
% %%%%%%%% DIVIDER %%%%%%%%%%%%
\begin{proof}
Consider the following derivation.
\begin{align*}
\sqrt{n}( \hat{\beta}_{\mathrm{sat}}( s) -\beta ( s) ) &\overset{(1)}{=}\left[
\begin{array}{c}
\frac{n( n(s)-n_{A}(s)) }{n(s)n_{AD}(s)-n_{A}(s)n_{D}(s)}( R_{n,1}(1,1,s)+R_{n,1}(0,1,s)) \\
-\frac{n_{A}(s) n}{n(s)n_{AD}(s)-n_{A}(s)n_{D}(s)}( R_{n,1}(1,0,s)+R_{n,1}(0,0,s)) \\
+\frac{R_{n,2}(1,s)}{\frac{n_{AD}( s) }{n_{A}( s) }- \frac{n_{D}(s)-n_{AD}( s) }{n(s)-n_{A}(s)}}\left[
\begin{array}{c}
( E[ Y( 0) |C,S=s] -E[ Y( 0) |NT,S=s] ) \\
-\frac{\pi _{D( 0) }( s) }{\pi _{D( 1) }( s) }( E[ Y( 1) |C,S=s] -E[ Y( 1) |AT,S=s] )
\end{array}
\right] \\
+\frac{R_{n,2}(2,s)}{\frac{n_{AD}( s) }{n_{A}( s) }- \frac{n_{D}(s)-n_{AD}( s) }{n(s)-n_{A}(s)}}\left[
\begin{array}{c}
( E[ Y( 1) |C,S=s] -E[ Y( 1) |AT,S=s] ) \\
-\frac{1-\pi _{D( 1) }( s) }{1-\pi _{D( 0) }( s) }( E[ Y( 0) |C,S=s] -E[ Y( 0) |NT,S=s] )
\end{array}
\right]
\end{array}
\right) +o_p(1)\\
&\overset{(2)}{=}\xi _{n,1}(s)+\xi _{n,2}(s)+\xi _{n,3}(s)+o_{p}( 1) ,
\end{align*}
where (1) holds by \eqref{eq:Rn_defn}, \eqref{eq:derivation1}, and $\beta ( s) =E[ Y( 1) -Y( 0) |C,S=s] $,
% (the $o_p(1)$ captures probability of ${{{\mathbf{Z}}^{\mathrm{sat}}_{n}}'\mathbf{X}_{n}^{\mathrm{sat}}}/{n}$ being singular or any of the denominators being equal to zero), 
and (2) holds by the auxiliary derivations in \eqref{eq:derivation2} and \eqref{eq:derivation3} that appear below.

The first auxiliary derivation is 
\begin{align}
& \frac{n}{n(s)n_{AD}(s)-n_{A}(s)n_{D}(s)}\left[ 
\begin{array}{c}
(n(s)-n_{A}(s))\frac{1}{\sqrt{n}}\sum_{i=1}^{n}1[D_{i}=1,A_{i}=1,S_{i}=s]( \tilde{Y}_{i}(1)-\mu (1,1,s)) \\
+(n(s)-n_{A}(s))\frac{1}{\sqrt{n}}\sum_{i=1}^{n}1[D_{i}=0,A_{i}=1,S_{i}=s]( \tilde{Y}_{i}(0)-\mu (0,1,s)) \\
-n_{A}(s)\frac{1}{\sqrt{n}}\sum_{i=1}^{n}1[D_{i}=1,A_{i}=0,S_{i}=s](\tilde{Y} _{i}(1)-\mu (1,0,s)) \\
-n_{A}(s)\frac{1}{\sqrt{n}}\sum_{i=1}^{n}1[D_{i}=0,A_{i}=0,S_{i}=s](\tilde{Y} _{i}(0)-\mu (0,0,s))
\end{array}
\right]   \nonumber \\
& \overset{(1)}{=}\left[ 
\begin{array}{c}
\frac{\frac{n(s)}{n}(1-\frac{n_{A}(s)}{n(s)})}{(\frac{n(s)}{n})^{2}\frac{ n_{A}(s)}{n(s)}(1-\frac{n_{A}(s)}{n(s)})(\frac{n_{AD}(s)}{n_{A}(s)}-\frac{ n_{D}(s)-n_{AD}(s)}{n(s)-n_{A}(s)})}(R_{n,1}(1,1,s)+R_{n,1}(0,1,s)) \\
-\tfrac{\frac{n(s)}{n}\frac{n_{A}(s)}{n(s)}}{(\frac{n(s)}{n})^{2}\frac{ n_{A}(s)}{n(s)}(1-\frac{n_{A}(s)}{n(s)})(\frac{n_{AD}(s)}{n_{A}(s)}-\frac{ n_{D}(s)-n_{AD}(s)}{n(s)-n_{A}(s)})}(R_{n,1}(1,0,s)+R_{n,1}(0,0,s))
\end{array}
\right] + o_{p}(1)\notag  \\
& \overset{(2)}{=}\xi _{n,1}(s)+o_{p}(1), \label{eq:derivation2}
\end{align}
where (1) holds by \eqref{eq:Rn_defn}, and (2) holds by Lemma \ref{lem:AsyDist}, as this implies that $R_{n,1}=O_{p}(1)$ and also that $\frac{n(s)}{n} \overset{p}{\to}p( s) $, $\frac{n_{A}(s)}{n(s)}\overset{p}{\to}\pi _{A}( s) $, $\frac{n_{AD}(s)}{n_{A}(s)}\overset{p}{\to}\pi _{D( 1) }( s) $, and $\frac{n_{D}(s)-n_{AD}(s) }{n(s)-n_{A}(s)}\overset{p}{\to}\pi _{D( 0) }( s) $ for all $s\in \mathcal{S}$.

The second auxiliary derivation is
\begin{align}
&\sqrt{n}\left[ 
\begin{array}{c}
\frac{1}{n(s)n_{AD}(s)-n_{A}(s)n_{D}(s)}\left[ 
\begin{array}{c}
( n(s)-n_{A}(s)) n_{AD}( s) \mu ( 1,1,s) \\
( n(s)-n_{A}(s)) ( n_{A}( s) -n_{AD}( s) ) \mu ( 0,1,s) \\
-n_{A}(s)( n_{D}( s) -n_{AD}( s) ) \mu ( 1,0,s) \\
-n_{A}(s)( n( s) -n_{A}( s) -n_{D}( s) +n_{AD}( s) ) \mu ( 0,0,s)
\end{array}
\right] \\
+( E[ Y( 1) -Y( 0) |S=s] -E[ Y( 1) -Y( 0) |C,S=s] )
\end{array}
\right] \notag \\
&\overset{(1)}{=}\sqrt{n}\left[ 
\begin{array}{c}
( \frac{( n(s)-n_{A}(s)) n_{AD}( s) }{ n(s)n_{AD}(s)-n_{A}(s)n_{D}(s)}\frac{\pi _{D( 1) }( s) -\pi _{D( 0) }( s) }{\pi _{D( 1) }( s) }-1) E[ Y( 1) |C,S=s] + \\
( 1-\frac{n_{A}(s)( n( s) -n_{A}( s) -n_{D}( s) +n_{AD}( s) ) }{ n(s)n_{AD}(s)-n_{A}(s)n_{D}(s)}\frac{\pi _{D( 1) }( s) -\pi _{D( 0) }( s) }{1-\pi _{D( 0) }( s) }) E[ Y( 0) |C,S=s] + \\
( \frac{( n(s)-n_{A}(s)) n_{AD}( s) }{ n(s)n_{AD}(s)-n_{A}(s)n_{D}(s)}\frac{\pi _{D( 0) }( s) }{\pi _{D( 1) }( s) }-\frac{n_{A}(s)( n_{D}( s) -n_{AD}( s) ) }{n(s)n_{AD}(s)-n_{A}(s)n_{D}(s)} ) E[ Y( 1) |AT,S=s] + \\
( \frac{( n(s)-n_{A}(s)) ( n_{A}( s) -n_{AD}( s) ) }{n(s)n_{AD}(s)-n_{A}(s)n_{D}(s)}-\frac{ n_{A}(s)( n( s) -n_{A}( s) -n_{D}( s) +n_{AD}( s) ) }{n(s)n_{AD}(s)-n_{A}(s)n_{D}(s)}\frac{1-\pi _{D( 1) }( s) }{1-\pi _{D( 0) }( s) }) E[ Y( 0) |NT,S=s]
\end{array}
\right]  + o_{p}(1)\notag \\
&\overset{(2)}{=}\left[ 
\begin{array}{c}
\frac{( n( s) ) ^{2}\frac{n_{A}(s)}{n( s) } ( 1-\frac{n_{A}(s)}{n( s) }) ( R_{n,2}(2,s)-R_{n,2}(1,s)\frac{\pi _{D( 0) }( s) }{\pi _{D( 1) }( s) }) }{( n( s) ) ^{2}\frac{n_{A}( s) }{n( s) }( 1-\frac{ n_{A}(s)}{n( s) }) ( \frac{n_{AD}( s) }{ n_{A}( s) }-\frac{n_{D}(s)-n_{AD}( s) }{n(s)-n_{A}(s)} ) }( E[ Y( 1) |C,S=s] -E[ Y( 1) |AT,S=s] ) + \\
\frac{( n( s) ) ^{2}\frac{n_{A}( s) }{ n( s) }( 1-\frac{n_{A}(s)}{n( s) }) ( R_{n,2}(1,s)-R_{n,2}(2,s)\frac{1-\pi _{D( 1) }( s) }{ 1-\pi _{D( 0) }( s) }) }{( n( s) ) ^{2}\frac{n_{A}( s) }{n( s) }( 1-\frac{ n_{A}(s)}{n( s) }) ( \frac{n_{AD}( s) }{ n_{A}( s) }-\frac{n_{D}(s)-n_{AD}( s) }{n(s)-n_{A}(s)} ) }( E[ Y( 0) |C,S=s] -E[ Y( 0) |NT,S=s] )
\end{array}
\right] + o_{p}(1) \notag \\
&\overset{(3)}{=}\xi _{n,2}(s)+\xi _{n,3}(s)+o_{p}( 1) ,\label{eq:derivation3}
\end{align}
where (1) holds by Lemma \ref{lem:MeanTraslation}, (2) holds by \eqref{eq:Rn_defn}, and (3) by Lemma \ref{lem:AsyDist}, as this implies that $R_{n,2}=O_{p}(1)$ and also that $\frac{n(s)}{n} \overset{p}{\to}p( s) $ and $\frac{n_{A}(s)}{n(s)}\overset{p}{\to}\pi _{A}( s) $ for all $s\in \mathcal{S}$.
\end{proof}
%%%%%%%% DIVIDER %%%%%%%%%%%%

%%%%%%%% DIVIDER %%%%%%%%%%%%
\begin{theorem}[SAT strata specific asy.\ dist.]
\label{thm:CondLATE}
Under Assumptions \ref{ass:1} and \ref{ass:2}, and for
any $s\in \mathcal{S}$,
\begin{equation*}
( \sqrt{n}(\hat{\beta}_{\mathrm{sat}}(s)-\beta (s)):s\in \mathcal{S})' ~\overset {d}{\to }~N( {\bf 0},diag( ( V_{{Y} ,1}^{\mathrm{sat}}(s)+V_{{Y},0}^{\mathrm{sat}}(s)+V_{D,1}^{\mathrm{sat}}(s)+V_{D,0}^{\mathrm{sat}}(s):s\in \mathcal{S} ) ) ) ,
\end{equation*}
where $\hat{\beta}_{\mathrm{sat}}(s)$ is as \eqref{eq:IV_SAT_estimators}, $\beta(s)$ is as in \eqref{eq:SAT_limit}, and
\begin{align*}
V_{{Y} ,1}^{\mathrm{sat}}(s) &\equiv \frac{\left[
\begin{array}{c}
V[Y(1)|S=s,AT]\pi _{D(0)}(s)+V[Y(0)|S=s,NT](1-\pi _{D(1)}(s)) \\
+V[Y(1)|S=s,C](\pi _{D(1)}(s)-\pi _{D(0)}(s))+ \\
(E[Y(1)|S=s,C]-E[Y(1)|S=s,AT])^{2}
\pi _{D(0)}(s)( \pi _{D(1)}(s)-\pi _{D(0)}(s)) /( \pi _{D(1)}(s)) ^{2}
\end{array}
\right] }{p(s)(\pi _{D(1)}(s)-\pi _{D(0)}(s))^{2}\pi _{A}(s)} \\
V_{{Y} ,0}^{\mathrm{sat}}(s) &\equiv \frac{\left[
\begin{array}{c}
V[Y(1)|S=s,AT]\pi _{D(0)}(s)+V[Y(0)|S=s,NT](1-\pi _{D(1)}(s)) \\
+V[Y(0)|S=s,C](\pi _{D(1)}(s)-\pi _{D(0)}(s))+ \\
(E[Y(0)|S=s,C]-E[Y(0)|S=s,NT])^{2}
( 1-\pi _{D(1)}(s)) ( \pi _{D(1)}(s)-\pi _{D(0)}(s)) /( 1-\pi _{D(0)}(s)) ^{2}
\end{array}
\right] }{p(s)(\pi _{D(1)}(s)-\pi _{D(0)}(s))^{2}(1-\pi _{A}(s))} \\
V_{D,1}^{\mathrm{sat}}(s) &\equiv \frac{(1-\pi _{D(1)}(s))\left[
\begin{array}{c}
\pi _{D(0)}(s)(E[Y(1)|C,S=s]-E[Y(1)|AT,S=s]) \\
-\pi _{D(1)}(s)(E[Y(0)|C,S=s]-E[Y(0)|NT,S=s])
\end{array}
\right] ^{2}}{p(s)\pi _{A}(s)\pi _{D(1)}(s)(\pi _{D(1)}(s)-\pi _{D(0)}(s))^{2}} \\
V_{D,0}^{\mathrm{sat}}(s) &\equiv \frac{\pi _{D(0)}(s)\left[
\begin{array}{c}
(1-\pi _{D(0)}(s))(E[Y(1)|C,S=s]-E[Y(1)|AT,S=s]) \\
-(1-\pi _{D(1)}(s))(E[Y(0)|C,S=s]-E[Y(0)|NT,S=s])
\end{array}
\right] ^{2}}{p(s)(1-\pi _{A}(s))(1-\pi _{D(0)}(s))(\pi _{D(1)}(s)-\pi _{D(0)}(s))^{2}}.
\end{align*}
\end{theorem}
%%%%%%%% DIVIDER %%%%%%%%%%%%
\begin{proof}
This result follows from Lemma \ref{lem:AsyDist} and Theorem \ref{thm:Expression}.
\end{proof}
%%%%%%%% DIVIDER %%%%%%%%%%%%

%%%%%%%% DIVIDER %%%%%%%%%%%%
\begin{proof}[Proof of Theorem \ref{thm:AsyDist_SAT}] 
Before proving the result, we generate several auxiliary derivations for any arbitrary $s\in \mathcal{S}$.  The first auxiliary derivation is as follows. 
\begin{align}
\sqrt{n}[\hat{P}(S=s|C)-P(S=s|C)]& \overset{(1)}{=}\sqrt{n}\left[ \frac{ \frac{n(s)}{n}(\frac{n_{AD}(s)}{n_{A}(s)}-\frac{n_{D}(s)-n_{AD}(s)}{ n(s)-n_{A}(s)})}{\sum_{\tilde{s}\in \mathcal{S}}\frac{n(\tilde{s})}{n}(\frac{ n_{AD}(\tilde{s})}{n_{A}(\tilde{s})}-\frac{n_{D}(\tilde{s})-n_{AD}(\tilde{s}) }{n(\tilde{s})-n_{A}(\tilde{s})})}-\frac{P(S=s,C)}{P(C)}\right] \notag \\
& \overset{(2)}{=}\frac{1}{\hat{P}(C)}\left[
\begin{array}{c}
\frac{n(s)}{n}(R_{n,2}(1,s)-R_{n,2}(2,s)) +(\pi _{D(1)}(s)-\pi _{D(0)}(s))R_{n,4}(s) \\ 
-P(S=s|C)\sum_{\tilde{s}\in \mathcal{S}}\frac{n(\tilde{s})}{n}(R_{n,2}(1, \tilde{s})-R_{n,2}(2,\tilde{s})) \\
-P(S=s|C)\sum_{\tilde{s}\in \mathcal{S}}(\pi _{D(1)}(\tilde{s})-\pi _{D(0)}( \tilde{s}))R_{n,4}(\tilde{s})
\end{array}
\right]  +o_{p}(1)  \notag \\
& \overset{(3)}{=}\frac{1}{P(C)}\left[ 
\begin{array}{c}
p(s)(R_{n,2}(1,s)-R_{n,2}(2,s)) +(\pi _{D(1)}(s)-\pi _{D(0)}(s))R_{n,4}(s) \\ 
-P(S=s|C)\sum_{\tilde{s}\in \mathcal{S}}p(\tilde{s})(R_{n,2}(1,\tilde{s} )-R_{n,2}(2,\tilde{s})) \\
-P(S=s|C)\sum_{\tilde{s}\in \mathcal{S}}(\pi _{D(1)}(\tilde{s})-\pi _{D(0)}( \tilde{s}))R_{n,4}(\tilde{s})
\end{array}
\right] +o_{p}(1),  \label{eq:deriv4}
\end{align}
where (1) holds by \eqref{eq:hat_limits_2}, (2) holds by Lemma \ref{lem:AsyDist}, (3) follows from \eqref{eq:hat_limits_2} and Lemma \ref{lem:AsyDist}, as it implies that $R_{n,2}=O_{p}(1)$, $R_{n,4}=O_{p}(1)$, and Assumption \ref{ass:1}, as it implies that $\frac{n(s)}{n} \overset{p}{\to }p(s)$ for all $s\in \mathcal{S}$. 

The second auxiliary derivation is as follows.
\begin{align}
& \sum_{s\in \mathcal{S}}\hat{P}(S=s|C)\sqrt{n}(\hat{\beta}_{\mathrm{sat}}(s)-\beta (s)) \notag \\
& \overset{(1)}{=}\frac{1}{P(C)}\sum_{s\in \mathcal{S}}p(s)(\pi _{D(1)}(s)-\pi _{D(0)}(s))\sqrt{n}(\hat{\beta}_{\mathrm{sat}}(s)-\beta (s))+o_{p}(1) \notag \\
& \overset{(2)}{=}\frac{1}{P(C)}\left[
\begin{array}{c}
\sum_{s\in \mathcal{S}}\frac{1}{\pi _{A}(s)}(R_{n,1}(1,1,s)+R_{n,1}(0,1,s)) 
-\sum_{s\in \mathcal{S}}\frac{1}{(1-\pi _{A}(s))} (R_{n,1}(1,0,s)+R_{n,1}(0,0,s)) \\
+\sum_{s\in \mathcal{S} }p(s)(E[Y(1)|C,S=s]-E[Y(1)|AT,S=s])(R_{n,2}(2,s)-R_{n,2}(1,s)\frac{\pi _{D(0)}(s)}{\pi _{D(1)}(s)}) \\
+\sum_{s\in \mathcal{S} }p(s)(E[Y(0)|C,S=s]-E[Y(0)|NT,S=s])(R_{n,2}(1,s)-R_{n,2}(2,s)\frac{1-\pi _{D(1)}(s)}{1-\pi _{D(0)}(s)})
\end{array}
\right] +o_{p}(1), \label{eq:deriv5}
\end{align}
where (1) holds by \eqref{eq:hat_limits_1} and Theorem \ref{thm:CondLATE}, as this implies $\sqrt{n }(\hat{\beta}_{\mathrm{sat}}(s)-\beta (s))=O_{p}(1)$ for all $s\in \mathcal{S}$, and (2) follows from Theorem \ref{thm:Expression}.

We are now ready to complete the proof of the desired result. To this end, consider the following derivation.
\begin{align}
& \sqrt{n}(\hat{\beta}_{\mathrm{sat}}-\beta )  =\sum_{s\in \mathcal{S}}\hat{P}(S=s|C)\sqrt{n}(\hat{\beta}_{\mathrm{sat}}(s)-\beta (s))+\sum_{s\in \mathcal{S}}\beta (s)\sqrt{n}(\hat{P}(S=s|C)-P(S=s|C)) \notag \\
& \overset{(1)}{=}\frac{1}{P(C)}\left[ 
\begin{array}{c}
\sum_{s\in \mathcal{S}}\frac{1}{\pi _{A}(s)}(R_{n,1}(1,1,s)+R_{n,1}(0,1,s)) 
-\sum_{s\in \mathcal{S}}\frac{1}{(1-\pi _{A}(s))} (R_{n,1}(1,0,s)+R_{n,1}(0,0,s)) \\
+\sum_{s\in \mathcal{S} }p(s)(E[Y(1)|C,S=s]-E[Y(1)|AT,S=s])(R_{n,2}(2,s)-R_{n,2}(1,s)\frac{\pi _{D(0)}(s)}{\pi _{D(1)}(s)}) \\
+\sum_{s\in \mathcal{S} }p(s)(E[Y(0)|C,S=s]-E[Y(0)|NT,S=s])(R_{n,2}(1,s)-R_{n,2}(2,s)\frac{1-\pi _{D(1)}(s)}{1-\pi _{D(0)}(s)}) \\
+\sum_{s\in \mathcal{S}}\beta (s)p(s)(R_{n,2}(1,s)-R_{n,2}(2,s)) +\sum_{s\in \mathcal{S}}\beta (s)(\pi _{D(1)}(s)-\pi _{D(0)}(s))R_{n,4}(s) \\
-\sum_{s\in \mathcal{S}}\beta (s)P(S=s|C)[\sum_{\tilde{s}\in \mathcal{S}}p( \tilde{s})(R_{n,2}(1,\tilde{s})-R_{n,2}(2,\tilde{s}))] \\
-\sum_{s\in \mathcal{S}}\beta (s)P(S=s|C)[\sum_{\tilde{s}\in \mathcal{S} }(\pi _{D(1)}(\tilde{s})-\pi _{D(0)}(\tilde{s}))R_{n,4}(\tilde{s})]
\end{array}
\right] +o_{p}(1)  \notag \\
& \overset{(2)}{=}\frac{1}{P(C)}\left[ 
\begin{array}{c}
\sum_{s\in \mathcal{S}}\frac{1}{\pi _{A}(s)}(R_{n,1}(1,1,s)+R_{n,1}(0,1,s)) 
-\sum_{s\in \mathcal{S}}\frac{1}{(1-\pi _{A}(s))} (R_{n,1}(1,0,s)+R_{n,1}(0,0,s)) \\
+\sum_{s\in \mathcal{S}}p(s)\left[ 
\begin{array}{c}
(\beta (s)-\beta ) \\ 
+E[Y(0)|C,S=s]-E[Y(0)|NT,S=s] \\ 
-(E[Y(1)|C,S=s]-E[Y(1)|AT,S=s])\frac{\pi _{D(0)}(s)}{\pi _{D(1)}(s)}
\end{array}
\right] R_{n,2}(1,s) \\ 
+\sum_{s\in \mathcal{S}}p(s)\left[ 
\begin{array}{c}
E[Y(1)|C,S=s]-E[Y(1)|AT,S=s] \\ 
-(E[Y(0)|C,S=s]-E[Y(0)|NT,S=s])\frac{1-\pi _{D(1)}(s)}{1-\pi _{D(0)}(s)} \\ 
-(\beta (s)-\beta )
\end{array}
\right] R_{n,2}(2,s) \\ 
+\sum_{s\in \mathcal{S}}(\beta (s)-\beta )(\pi _{D(1)}(s)-\pi_{D(0)}(s))R_{n,4}(s)
\end{array}
\right] +o_{p}(1),\label{eq:deriv6}
\end{align}
where (1) holds by \eqref{eq:deriv4} and \eqref{eq:deriv5}, and (2) holds by $\beta =\sum_{s\in \mathcal{S}}\beta (s)P(S=s|C)$. The desired result follows from combining \eqref{eq:deriv6}, Lemmas \ref{lem:MeanTraslation} and \ref{lem:AsyDist}.
\end{proof}
%%%%%%%% DIVIDER %%%%%%%%%%%%

%%%%%%%% DIVIDER %%%%%%%%%%%%
\begin{lemma}[SAT residuals] \label{lem:U_asymptotics} 
Assume Assumptions \ref{ass:1} and \ref{ass:2}. Let $(u_{i})_{i=1}^{n}$ denote the population version of the SAT regression residuals, defined by
\begin{equation}
u_{i}~\equiv~ Y_{i}-\sum_{s\in \mathcal{S}}1[S_{i}=s]\gamma (s)-\sum_{s\in \mathcal{S}}D_{i}1[S_{i}=s]\beta (s),
\label{eq:resid_pop_SAT}
\end{equation}
where $(( \beta ( s) ,\gamma ( s) ) :s\in \mathcal{S}) $ is as in \eqref{eq:SAT_limit}. Then, the SAT residuals $(\hat{u}_i)_{i=1}^{n}$ defined in \eqref{eq:resid_SAT} are such that, for any $ (d,a,s)\in \{0,1\}^{2}\times S$,
\begin{align}
& \frac{1}{n}\sum_{i=1}^{n}1[D_{i}=d,A_{i}=a,S_{i}=s]\hat{u}_{i}=\frac{1}{n }\sum_{i=1}^{n}1[D_{i}=d,A_{i}=a,S_{i}=s]u_{i}+o_{p}(1)=o_{p}(1) + \notag \\
& p(s)\left[ 
\begin{array}{c}
-1[(d,a)=(1,1)]\pi _{A}(s)(1-\pi _{D(1)}(s))\left(
\begin{array}{c}
\pi _{D(0)}(s)(E[Y(1)|C,S=s]-E[Y(1)|AT,S=s]) \\ 
-\pi _{D(1)}(s)(E[Y(0)|C,S=s]-E[Y(0)|NT,S=s])
\end{array}
\right) \\ 
+1[(d,a)=(0,1)]\pi _{A}(s)(1-\pi _{D(1)}(s))\left( 
\begin{array}{c}
\pi _{D(0)}(s)(E[Y(1)|C,S=s]-E[Y(1)|AT,S=s]) \\ 
-\pi _{D(1)}(s)(E[Y(0)|C,S=s]-E[Y(0)|NT,S=s])
\end{array}
\right) \\ 
-1[(d,a)=(1,0)](1-\pi _{A}(s))\pi _{D(0)}(s)\left( 
\begin{array}{c}
(1-\pi _{D(0)}(s))(E[Y(1)|C,S=s]-E[Y(1)|AT,S=s]) \\ 
-(1-\pi _{D(1)}(s))(E[Y(0)|C,S=s]-E[Y(0)|NT,S=s])
\end{array}
\right) \\ 
+1[(d,a)=(0,0)](1-\pi _{A}(s))\pi _{D(0)}(s)\left( 
\begin{array}{c}
(1-\pi _{D(0)}(s))(E[Y(1)|C,S=s]-E[Y(1)|AT,S=s]) \\ 
-(1-\pi _{D(1)}(s))(E[Y(0)|C,S=s]-E[Y(0)|NT,S=s])
\end{array}
\right)
\end{array}
\right] \label{eq:lem_u1}
\end{align}
and
\begin{align}
& \frac{1}{n}\sum_{i=1}^{n}1[D_{i}=d,A_{i}=a,S_{i}=s]\hat{u}_{i}^{2}=\frac{
1}{n}\sum_{i=1}^{n}1[D_{i}=d,A_{i}=a,S_{i}=s]u_{i}^{2}+o_{p}(1)=o_{p}(1)+  \notag \\
& p(s)\left(
\begin{array}{l}
1[(d,a)=(1,1)]\pi _{A}(s)\pi _{D(1)}(s)\times \\ 
\left[\sigma ^{2}(1,1,s)+\left( 
\begin{array}{c}
\pi _{D(0)}(s)(E[Y(1)|C,S=s]-E[Y(1)|AT,S=s]) \\ 
-\pi _{D(1)}(s)(E[Y(0)|C,S=s]-E[Y(0)|NT,S=s])
\end{array}
\right) ^{2}(\frac{1-\pi _{D(1)}(s)}{\pi _{D(1)}(s)})^{2}\right] \\ 
+1[(d,a)=(0,1)]\pi _{A}(s)(1-\pi _{D(1)}(s))\times \\ 
\left[\sigma ^{2}(0,1,s)+\left( 
\begin{array}{c}
\pi _{D(0)}(s)(E[Y(1)|C,S=s]-E[Y(1)|AT,S=s]) \\ 
-\pi _{D(1)}(s)(E[Y(0)|C,S=s]-E[Y(0)|NT,S=s])
\end{array}
\right) ^{2}\right] \\ 
+1[(d,a)=(1,0)](1-\pi _{A}(s))\pi _{D(0)}(s)\times \\ 
\left[ \sigma ^{2}(1,0,s)+\left( 
\begin{array}{c}
(1-\pi _{D(0)}(s))(E[Y(1)|C,S=s]-E[Y(1)|AT,S=s]) \\ 
-(1-\pi _{D(1)}(s))(E[Y(0)|C,S=s]-E[Y(0)|NT,S=s])
\end{array}
\right) ^{2}\right] \\ 
+1[(d,a)=(0,0)](1-\pi _{A}(s))(1-\pi _{D(0)}(s))\times \\ 
\left[ \sigma ^{2}(0,0,s)+\left( 
\begin{array}{c}
(1-\pi _{D(0)}(s))(E[Y(1)|C,S=s]-E[Y(1)|AT,S=s]) \\ 
-(1-\pi _{D(1)}(s))(E[Y(0)|C,S=s]-E[Y(0)|NT,S=s])
\end{array}
\right) ^{2}(\frac{\pi _{D(0)}(s)}{1-\pi _{D(0)}(s)})^{2}\right]
\end{array}
\right).  \label{eq:lem_u2}
\end{align}
\end{lemma}
%%%%%%%% DIVIDER %%%%%%%%%%%%
\begin{proof}
We only show \eqref{eq:lem_u1}, as the proof of \eqref{eq:lem_u2} follows from analogous arguments. Fix $(d,a,s)\in \{0,1\}^{2}\times S$ arbitrarily. To show the first equality in \eqref{eq:lem_u1}, consider the following argument.
\begin{align*}
\frac{1}{n}\sum_{i=1}^{n}1[D_{i}=d,A_{i}=a,S_{i}=s]\hat{u}_{i}
& =\frac{1}{n }\sum_{i=1}^{n}1[D_{i}=d,A_{i}=a,S_{i}=s][u_{i}+\gamma (s)-\hat{\gamma}_{\mathrm{sat}} (s)+1[d=1](\beta (s)-\hat{\beta}_{\mathrm{sat}}(s))] \\
& \overset{(1)}{=}\frac{1}{n}\sum_{i=1}^{n}1[D_{i}=d,A_{i}=a,S_{i}=s ]u_{i}+o_{p}(1),
\end{align*}
where (1) holds by Theorem \ref{thm:plim_SAT} and $(\frac{1}{n} \sum_{i=1}^{n}1[D_{i}=d,A_{i}=a,S_{i}=s]u_{i})^{2}\leq \frac{1}{n} \sum_{i=1}^{n}u_{i}^{2}=O_{p}(1)$. To show the second equality in \eqref{eq:lem_u1}, consider the following argument.
\begin{align*}
& \frac{1}{n}\sum_{i=1}^{n}1[D_{i}=d,A_{i}=a,S_{i}=s]u_{i}=\frac{1}{n} \sum_{i=1}^{n}1[D_{i}=d,A_{i}=a,S_{i}=s](Y_{i}-\gamma (s)-d\beta (s)) \\
& \overset{(1)}{=}\frac{1}{n}\sum_{i=1}^{n}1[D_{i}=d,A_{i}=a,S_{i}=s]
\left[ 
\begin{array}{c}
(\tilde{Y}_{i}(d)-\mu (d,a,s))+(\mu (d,a,s)+E[Y(d)|S=s])+ \\ 
(\pi _{D(0)}(s)-d)E[Y(1)|C,S=s]+(d-\pi _{D(1)}(s))E[Y(0)|C,S=s] \\ 
-\pi _{D(0)}(s)E[Y(1)|AT,S=s]-(1-\pi _{D(1)}(s))E[Y(0)|NT,S=s]
\end{array}
\right]  \\
%% FEDE: The derivation below is very valuable but it is just algebra. I want to keep it in the paper but it probably will not make it to the published version.
%& =\left\{
%\begin{array}{c}
%\frac{1}{n}\sum_{i=1}^{n}I\{D_{i}=d,A_{i}=a,S_{i}=s\}(\tilde{Y}_{i}(d)-\mu (d,a,s)) \\
%+\left[ 
%\begin{array}{c}
%I\{(d,a)=(1,1)\}\frac{n_{AD}( s) }{n_{A}( s) }\frac{ n_{A}( s) }{n( s) }\frac{n( s) }{n}+ \\
%I\{(d,a)=(0,1)\}( 1-\frac{n_{AD}( s) }{n_{A}( s) } ) \frac{n_{A}( s) }{n( s) }\frac{n( s) }{n}+ \\
%I\{(d,a)=(1,0)\}\frac{n_{D}( s) -n_{AD}( s) }{n( s) -n_{A}( s) }( 1-\frac{n_{A}( s) }{ n( s) }) \frac{n( s) }{n}+ \\
%I\{(d,a)=(0,0)\}( 1-\frac{n_{D}( s) -n_{AD}( s) }{ n( s) -n_{A}( s) }) ( 1-\frac{n_{A}( s) }{n( s) }) \frac{n( s) }{n}
%\end{array}
%\right] \times  \\ 
%\left[ 
%\begin{array}{c}
%(\mu (d,a,s)+E[Y(d)|S=s])+ \\ 
%(\pi _{D(0)}(s)-d)E[Y(1)|C,S=s]+(d-\pi _{D(1)}(s))E[Y(0)|C,S=s] \\ 
%-\pi _{D(0)}(s)E[Y(1)|AT,S=s]-(1-\pi _{D(1)}(s))E[Y(0)|NT,S=s]
%\end{array}
%\right] 
%\end{array}
%\right\}  \\
& \overset{(2)}{=}p(s)\left[
\begin{array}{l}
[ 1[(d,a)=(0,1)]-1[(d,a)=(1,1)]]\pi _{A}(s)(1-\pi
_{D(1)}(s))\times  \\ 
\left( 
\begin{array}{c}
\pi _{D(0)}(s)(E[Y(1)|C,S=s]-E[Y(1)|AT,S=s]) \\ 
-\pi _{D(1)}(s)(E[Y(0)|C,S=s]-E[Y(0)|NT,S=s])
\end{array}
\right)  \\ 
+[1[(d,a)=(0,0)]-1[(d,a)=(1,0)]](1-\pi _{A}(s))\pi _{D(0)}(s)\times  \\ 
\left( 
\begin{array}{c}
(1-\pi _{D(0)}(s))(E[Y(1)|C,S=s]-E[Y(1)|AT,S=s]) \\ 
-(1-\pi _{D(1)}(s))(E[Y(0)|C,S=s]-E[Y(0)|NT,S=s])
\end{array}
\right) 
\end{array}
\right] +o_{p}(1),
\end{align*}
where (1) holds by Theorem \ref{thm:plim_SAT}, $\beta (s)=E[Y(1)-Y(0)|C,S=s]$, and the fact that, conditional on $(D_{i},S_{i})=(d,s)$, $Y_{i}=Y_{i}(d)= \tilde{Y}_{i}(d)+E[Y(d)|S=s]$ (by \eqref{eq:defnY_pre}), (2) holds by Assumption \ref{ass:2}(b) and Lemmas \ref{lem:MeanTraslation}, \ref{lem:A1and2_impliesold3}, and \ref{lem:AsyDist2}.
\end{proof}
%%%%%%%% DIVIDER %%%%%%%%%%%%

%%%%%%%% DIVIDER %%%%%%%%%%%%
\begin{proof}[Proof of Theorem \ref{thm:SAT_se}]
The desired result follows from showing that
\begin{align}
\hat{V}_{1}^{\textrm{sat}} &~\overset{p}{\to }~ V_{{Y},1}^{\textrm{sat}}+V_{D,1}^{\textrm{sat}},~~~~
\hat{V}_{0}^{\textrm{sat}}  ~\overset{p}{\to }~V_{{Y},0}^{\textrm{sat}}+V_{D,0}^{\textrm{sat}},~~~~\text{and}~~
\hat{V}_{H}^{\textrm{sat}} ~\overset{p}{\to }~ V_{{H}}^{\textrm{sat}}.\label{eq:sat_asyvar1}
\end{align}
We only show the result in \eqref{eq:sat_asyvar1}, as the others can be shown analogously. Consider the following derivation.
\begin{align*}
\hat{V}_{1}^{\textrm{sat}}& =\frac{1}{\hat{P}(C)^{2}}\sum_{s\in \mathcal{S}}( \frac{1}{n_{A}(s)/n(s)})^{2}\left[
\begin{array}{c}
\frac{1}{n}\sum_{i=1}^{n}1[D_{i}=1,A_{i}=1,S_{i}=s](\hat{u}_{i}+(1-\frac{ n_{AD}(s)}{n_{A}(s)})(\hat{\beta}_{\mathrm{sat}}(s)-\hat{\beta}_{\mathrm{sat}}))^{2} \\
+\frac{1}{n}\sum_{i=1}^{n}1[D_{i}=0,A_{i}=1,S_{i}=s](\hat{u}_{i}-\frac{ n_{AD}(s)}{n_{A}(s)}(\hat{\beta}_{\mathrm{sat}}(s)-\hat{\beta}_{\mathrm{sat}}))^{2}
\end{array}
\right]  \\
%& =\frac{1}{\hat{P}(C)^{2}}\sum_{s\in\mathcal{S}}(\frac{1}{n_{A}(s)/n(s)})^{2}\left[ 
%\begin{array}{c}
%\frac{1}{n}\sum_{i=1}^{n}(I\{D_{i}=1,A_{i}=1,S_{i}=s\}+I \{D_{i}=0,A_{i}=1,S_{i}=s\})\hat{u}_{i}^{2} \\
%+\frac{n(s)}{n}\frac{n_{A}(s)}{n(s)}\frac{n_{AD}(s)}{n_{A}(s)}(1-\frac{ n_{AD}(s)}{n_{A}(s)})(\hat{\beta}_{\mathrm{sat}}(s)-\hat{\beta}_{\mathrm{sat}})^{2} \\
%+2(1-\frac{n_{AD}(s)}{n_{A}(s)})(\hat{\beta}_{\mathrm{sat}}(s)-\hat{\beta}_{\mathrm{sat}})\frac{1}{n} \sum_{i=1}^{n}I\{D_{i}=1,A_{i}=1,S_{i}=s\}\hat{u}_{i} \\
%-2\frac{n_{AD}(s)}{n_{A}(s)}(\hat{\beta}_{\mathrm{sat}}(s)-\hat{\beta}_{\mathrm{sat}})\frac{1}{n} \sum_{i=1}^{n}I\{D_{i}=0,A_{i}=1,S_{i}=s\}\hat{u}_{i}
%\end{array}
%\right]  \\
& \overset{(1)}{=}\frac{1}{P(C)^{2}}\sum_{s\in \mathcal{S}}\frac{p(s)}{\pi _{A}(s)}\left[
\begin{array}{c}
\pi _{D(1)}(s)\sigma ^{2}(1,1,s)+(1-\pi _{D(1)}(s))\sigma ^{2}(0,1,s)+ \\
\frac{1-\pi _{D(1)}(s)}{\pi _{D(1)}(s)}\left[
\begin{array}{c}
\pi _{D(0)}(s)(E[Y(1)|C,S=s]-E[Y(1)|AT,S=s]) \\
-\pi _{D(1)}(s)(E[Y(0)|C,S=s]-E[Y(0)|NT,S=s]) \\
-\pi _{D(1)}(s)(\beta (s)-\beta )
\end{array}
\right] ^{2}
\end{array}
\right] +o_{p}(1) \\
& \overset{(2)}{=}V_{{Y},1}^{\textrm{sat}}+V_{D,1}^{\textrm{sat}}+o_{p}(1),
\end{align*}
where (1) holds by Theorem \ref{thm:plim_SAT} and Lemma \ref{lem:U_asymptotics}, and (2) holds by Lemma \ref{lem:MeanTraslation}.
\end{proof}
%%%%%%%% DIVIDER %%%%%%%%%%%%

\begin{theorem}[Estimation of primitive parameters]\label{thm:primitiveEstimation}
Under Assumptions \ref{ass:1} and \ref{ass:2}, the primitive parameters in \eqref{eq:pi_defns} can be consistently estimated. In particular, for any $s\in \mathcal{S}$,
\begin{align}
    &\left(\frac{n(s)}{n},\frac{n_{A}(s)}{n(s)},\frac{n_{AD}(s)}{n_{A}(s)},\frac{n_{D}(s)-n_{AD}(s)}{n(s)-n_{A}(s)}\right)~ \overset{p}{\to} ~(p(s), \pi_{A}(s), \pi _{D(1)}(s), \pi _{D(0)}(s)).\label{eq:primitiveEstimation1}
\end{align}
Also, provided that the conditioning event has positive probability,
\begin{align}
\hat{E}[Y(0)|NT,S =s]&\equiv\tfrac{n/2}{n_{A}( s) -n_{AD}( s) }\left[
\begin{array}{c}
\frac{1}{n}\sum_{i=1}^{n}1[D_{i}=0,A_{i}=1,S_{i}=s]\hat{u}_{i} \\
-\frac{1}{n}\sum_{i=1}^{n}1[D_{i}=1,A_{i}=1,S_{i}=s]\hat{u}_{i}
\end{array}
\right] +\hat{\gamma}_{\mathrm{sat}}(s)  \overset{p}{\to} E[Y(0)|NT,S =s]\notag\\%\label{eq:primitiveEstimation2a} \\
\hat{E}[Y(1)|AT,S =s]&\equiv\tfrac{n/2}{n_{D}( s) -n_{AD}( s) }\left[
\begin{array}{c}
\frac{1}{n}\sum_{i=1}^{n}1[D_{i}=1,A_{i}=0,S_{i}=s]\hat{u}_{i} \\ 
-\frac{1}{n}\sum_{i=1}^{n}1[D_{i}=0,A_{i}=0,S_{i}=s]\hat{u}_{i}
\end{array}
\right] +\hat{\beta}_{\mathrm{sat}}(s)+\hat{\gamma}_{\mathrm{sat}}(s)\notag\\
&\overset{p}{\to} E[Y(1)|AT,S =s]\notag\\%\label{eq:primitiveEstimation2b} \\
\hat{E}[Y(0)|C,S =s]&\equiv\tfrac{1}{\frac{n_{AD}( s) }{n_{A}( s) }-\frac{n_{D}( s) -n_{AD}( s) }{n( s) -n_{A}( s) }}\left[
\begin{array}{c}
\hat{\gamma}_{\mathrm{sat}}(s)+\frac{n_{D}( s) -n_{AD}( s) }{n( s) -n_{A}( s) }\hat{\beta}_{\mathrm{sat} }( s) \\
-( 1-\frac{n_{AD}( s) }{n_{A}( s) }) \hat{E}[Y(0)|NT,S=s]\\
-\frac{n_{D}( s) -n_{AD}( s) }{n( s) -n_{A}( s) }\hat{E}[Y(1)|AT,S=s]
\end{array}
\right] \overset{p}{\to} E[Y(0)|C,S =s]\notag\\%\label{eq:primitiveEstimation2c} \\
\hat{E}[Y(1)|C,S =s]&\equiv\hat{\beta}_{\mathrm{sat}}( s) +\hat{E}[Y(0)|C,S=s]\overset{p}{\to} E[Y(1)|C,S =s]\label{eq:primitiveEstimation2}
\end{align}
and
\begin{align}
\hat{V}[Y(1)|AT,S =s]&\equiv\left[ 
\begin{array}{c}
\frac{1}{\frac{n( s) }{n}( 1-\frac{n_{A}( s) }{ n( s) }) \frac{n_{D}( s) -n_{AD}( s) }{n( s) -n_{A}( s) }}\frac{1}{n}\sum_{i=1}^{n}1[D_{i}=1,A_{i}=0,S_{i}=s]\hat{u}_{i}^{2} \\
-\left( 
\begin{array}{c}
( 1-\frac{n_{D}( s) -n_{AD}( s) }{n( s) -n_{A}( s) }) (\hat{E}[Y(1)|C,S=s]-\hat{E} [Y(1)|AT,S=s]) \\
-( 1-\frac{n_{AD}( s) }{n_{A}( s) }) (\hat{E }[Y(0)|C,S=s]-\hat{E}[Y(0)|NT,S=s])
\end{array}
\right) ^{2}
\end{array}
\right]  \notag\\
&\overset{p}{\to} V[Y(1)|AT,S =s]\notag\\% \label{eq:primitiveEstimation3a} \\
\hat{V}[Y(0)|NT,S =s]&\equiv\left[ 
\begin{array}{c}
\frac{n}{n_{A}( s) -n_{AD}( s) }\frac{1}{n}
\sum_{i=1}^{n}1[D_{i}=0,A_{i}=1,S_{i}=s]\hat{u}_{i}^{2} \\ 
-\left( 
\begin{array}{c}
\frac{n_{D}( s) -n_{AD}( s) }{n( s) -n_{A}( s) }(\hat{E}[Y(1)|C,S=s]-\hat{E}[Y(1)|AT,S=s]) \\
-\frac{n_{AD}( s) }{n_{A}( s) }(\hat{E}[Y(0)|C,S=s]- \hat{E}[Y(0)|NT,S=s])
\end{array}
\right) ^{2}
\end{array}
\right] \overset{p}{\to} V[Y(0)|NT,S =s]\notag\\% \label{eq:primitiveEstimation3b}\\
\hat{V}[Y(1)|C,S =s]&\equiv\tfrac{1}{\frac{n_{AD}( s) }{n_{A}( s) }-\frac{n_{D}( s) -n_{AD}( s) }{n( s) -n_{A}( s) }}\times \notag\\
&\left[
\begin{array}{c}
\frac{n}{n_{A}( s) }\frac{1}{n}\sum_{i=1}^{n}1[D_{i}=1,A_{i}=1,S_{i}=s]\hat{u}_{i}^{2} \\
-\frac{( 1-\frac{n_{AD}( s) }{n_{A}( s) })^{2}}{\frac{n_{AD}( s) }{n_{A}( s) }}
\left( 
\begin{array}{c}
\frac{n_{D}( s) -n_{AD}( s) }{n( s) -n_{A}( s) }(\hat{E}[Y(1)|C,S=s]-\hat{E}[Y(1)|AT,S=s]) \\
-\frac{n_{AD}( s) }{n_{A}( s) }(\hat{E}[Y(0)|C,S=s]- \hat{E}[Y(0)|NT,S=s])
\end{array}
\right) ^{2} \\ 
-\tfrac{\frac{n_{D}( s) -n_{AD}( s) }{n( s) -n_{A}( s) }( \frac{n_{AD}( s) }{n_{A}( s) }-\frac{n_{D}( s) -n_{AD}( s) }{n( s) -n_{A}( s) }) }{\frac{n_{AD}( s) }{ n_{A}( s) }}(\hat{E}[Y(1)|C,S=s]-\hat{E}[Y(1)|AT,S=s])^{2} \\
-\tfrac{n_{D}( s) -n_{AD}( s) }{n( s) -n_{A}( s) }\hat{V}[Y(1)|AT,S=s]
\end{array}
\right]\notag\\
&\overset{p}{\to} V[Y(1)|C,S =s] \notag\\% \label{eq:primitiveEstimation3c} \\
\hat{V}[Y(0)|C,S =s]
&\equiv\tfrac{1}{\frac{n_{AD}( s) }{n_{A}( s) }-\frac{n_{D}( s) -n_{AD}( s) }{n( s) -n_{A}( s) }}\times \notag\\
&\left[
\begin{array}{c}
\frac{n}{( n( s) -n_{A}( s) ) }\frac{1}{n} \sum_{i=1}^{n}1[D_{i}=0,A_{i}=0,S_{i}=s]\hat{u}_{i}^{2} \\
-\frac{( \frac{n_{D}( s) -n_{AD}( s) }{n( s) -n_{A}( s) }) ^{2}}{( 1-\frac{n_{D}( s) -n_{AD}( s) }{n( s) -n_{A}( s) } ) }\left(
\begin{array}{c}
( 1-\frac{n_{D}( s) -n_{AD}( s) }{n( s) -n_{A}( s) }) (\hat{E}[Y(1)|C,S=s]-\hat{E} [Y(1)|AT,S=s]) \\
-( 1-\frac{n_{AD}( s) }{n_{A}( s) }) (\hat{E }[Y(0)|C,S=s]-\hat{E}[Y(0)|NT,S=s])
\end{array}
\right) ^{2} \\ 
-\frac{\frac{n_{A}( s) -n_{AD}( s) }{n_{A}( s) }( \frac{n_{AD}( s) }{n_{A}( s) }-\frac{ n_{D}( s) -n_{AD}( s) }{n( s) -n_{A}( s) }) }{( 1-\frac{n_{D}( s) -n_{AD}( s) }{n( s) -n_{A}( s) }) }(\hat{E} [Y(0)|C,S=s]-\hat{E}[Y(0)|NT,S=s])^{2} \\
-( 1-\frac{n_{AD}( s) }{n_{A}( s) }) \hat{V}[Y(0)|NT,S=s]
\end{array}
\right]\notag\\
&\overset{p}{\to}  V[Y(0)|C,S =s]. \label{eq:primitiveEstimation3}
\end{align}
\end{theorem}
%%%%%%%% DIVIDER %%%%%%%%%%%%
\begin{proof}
The first convergence in \eqref{eq:primitiveEstimation1} holds by Assumption \ref{ass:1} and the LLN. The second convergence of \eqref{eq:primitiveEstimation1} is imposed in Assumption \ref{ass:2}(b). The remaining results hold by Lemma \ref{lem:A1and2_impliesold3}.

We next show the first line of \eqref{eq:primitiveEstimation2}. By Lemma \ref{lem:U_asymptotics},
\begin{align}
&\frac{1}{2}\left[ 
\begin{array}{c}
\frac{1}{n}\sum_{i=1}^{n}1[D_{i}=0,A_{i}=1,S_{i}=s]\hat{u}_{i} \\ 
-\frac{1}{n}\sum_{i=1}^{n}1[D_{i}=1,A_{i}=1,S_{i}=s]\hat{u}_{i}
\end{array}
\right] \notag\\
&\overset{p}{\to}p(s)\pi _{A}(s)(1-\pi _{D(1)}(s))\left[
\begin{array}{c}
\pi _{D(0)}(s)(E[Y(1)|C,S=s]-E[Y(1)|AT,S=s]) \\ 
-\pi _{D(1)}(s)(E[Y(0)|C,S=s]-E[Y(0)|NT,S=s])
\end{array}
\right] .\label{eq:primitiveEstimation4} 
\end{align}
Then, consider the following derivation.
\begin{align}
\hat{E}[Y(0)|NT,S=s] 
&=\frac{1/2}{\frac{n( s) }{n}\frac{n_{A}( s) }{ n( s) }( 1-( \frac{n_{AD}( s) }{n_{A}( s) }) ) }\left[
\begin{array}{c}
\frac{1}{n}\sum_{i=1}^{n}1[D_{i}=0,A_{i}=1,S_{i}=s]\hat{u}_{i} \\ 
-\frac{1}{n}\sum_{i=1}^{n}1[D_{i}=1,A_{i}=1,S_{i}=s]\hat{u}_{i}
\end{array}
\right] +\hat{\gamma}_{\mathrm{sat}}(s) +o_{p}(1)\notag \\
&\overset{p}{\to} E[Y(0)|NT,S=s],\label{eq:primitiveEstimation5}
\end{align}
where the convergence follows from \eqref{eq:primitiveEstimation4}, Assumptions \ref{ass:1} and \ref{ass:2}(b), Lemma \ref{lem:A1and2_impliesold3}, and Theorem \ref{thm:plim_SAT}. An analogous argument can be used to show the second line of \eqref{eq:primitiveEstimation2}. Next, note that the third line of \eqref{eq:primitiveEstimation2} follows from Assumption \ref{ass:1}, \ref{ass:2}(b), Lemma \ref{lem:A1and2_impliesold3}, and the first and second lines of \eqref{eq:primitiveEstimation2}, and Theorem \ref{thm:plim_SAT}. Finally, note that the last line of \eqref{eq:primitiveEstimation2} follows from the third line of \eqref{eq:primitiveEstimation2} and Theorem \ref{thm:plim_SAT}.

To conclude, we note that \eqref{eq:primitiveEstimation3} follows from  \eqref{eq:lem_u2}, \eqref{eq:primitiveEstimation1}, and \eqref{eq:primitiveEstimation2}, and Assumptions \ref{ass:1} and \ref{ass:2}(b).
%To show the first line of \eqref{eq:primitiveEstimation3}, consider the following argument.
%\begin{align*}
%\hat{V}[Y(1)|AT,S=s]
%&=\left[ 
%\begin{array}{c}
%\frac{1}{\frac{n( s) }{n}( 1-\frac{n_{A}( s) }{ n( s) }) \frac{n_{D}( s) -n_{AD}( s) }{n( s) -n_{A}( s) }}\frac{1}{n}\sum_{i=1}^{n}I \{D_{i}=1,A_{i}=0,S_{i}=s\}\hat{u}_{i}^{2} \\
%-\left( 
%\begin{array}{c}
%( 1-\frac{n_{D}( s) -n_{AD}( s) }{n( s) -n_{A}( s) }) (\hat{E}[Y(1)|C,S=s]-\hat{E} [Y(1)|AT,S=s]) \\
%-( 1-\frac{n_{AD}( s) }{n_{A}( s) }) (\hat{E }[Y(0)|C,S=s]-\hat{E}[Y(0)|NT,S=s])
%\end{array}
%\right) ^{2}
%\end{array}
%\right] \\
%&\overset{p}{\to}V[Y(1)|AT,S=s],
%\end{align*}
%where the convergence holds by \eqref{eq:lem_u2} and \eqref{eq:primitiveEstimation2}, and Assumptions \ref{ass:1} and \ref{ass:2}(b). An analogous argument can be used to show the remaining lines of \eqref{eq:primitiveEstimation3}.
\end{proof}

\begin{proof}[Proof of Theorem \ref{thm:SAT_test}]
This result follows from elementary convergence arguments and Theorems \ref{thm:AsyDist_SAT} and \ref{thm:SAT_se}.
\end{proof}

%%%%%%%%%%%%%%%%%%%%%%%%%%%%%%%%%%%%%%%%%%%%%%%%%%%%%%%%%%%%%%%%%%%%%%%%%%%%%
\subsection{Proofs of results related to Section \ref{sec:SFE}}\label{sec:A_SFE}

%%%%%%%% DIVIDER %%%%%%%%%%%%
\begin{lemma}[SFE matrices]\label{lem:MatrixSFE}
Assume Assumptions \ref{ass:1} and \ref{ass:2}. Then,
\begin{align*}
{{{\mathbf{Z}_{n}^{\mathrm{sfe}}}'}\mathbf{X}_{n}^{\mathrm{sfe}}}/{n} &=\left[
\begin{array}{cc}
diag( n( s) /n:s\in \mathcal{S}) & ( n_{D}(s)/n:s\in \mathcal{S}) \\
( n_{A}(s)/n:s\in \mathcal{S}) ^{\prime } & n_{AD}/n 
\end{array} 
\right] \\
&=\left[ 
\begin{array}{cc}
diag( p( s) :s\in \mathcal{S})  & ( [
\pi _{D( 1) }( s) \pi _{A}( s) + 
\pi _{D( 0) }( s) ( 1-\pi _{A}( s)
] p( s) :s\in \mathcal{S})  \\ 
( \pi _{A}( s) p( s) :s\in \mathcal{S}) ^{\prime } & \sum_{s\in \mathcal{S}}\pi _{D( 1) }( s) \pi _{A}( s) p( s)
\end{array}
\right] +o_{p}(1) .
\end{align*}
Thus, 
\begin{align*}
&( {{{\mathbf{Z}_{n}^{\mathrm{sfe}}}'}\mathbf{X}_{n}^{\mathrm{sfe}}}/{n} )
^{-1} \\
&=\left[ 
\left[ 
\begin{array}{cc}
diag( \frac{n}{n( s) }:s\in \mathcal{S}) & \mathbf{0 }_{\vert \mathcal{S}\vert \times 1} \\
\mathbf{0}_{1\times \vert \mathcal{S}\vert } & 0 
\end{array}
\right] + 
\frac{\left[ 
\begin{array}{cc}
( \frac{n_{D}(s)}{n( s) }:s\in \mathcal{S}) \times ( \frac{n_{A}(s)}{n( s) }:s\in \mathcal{S}) ^{\prime } & ( -\frac{n_{D}(s)}{n( s) }:s\in \mathcal{S}) \\
( -\frac{n_{A}(s)}{n( s) }:s\in \mathcal{S}) ^{\prime } & 1
\end{array}
\right] }{\frac{n_{AD}}{n}-\sum_{s\in \mathcal{S}}\frac{n_{A}(s)}{n( s) }\frac{n_{D}(s)}{n( s) }\frac{n( s) }{n}}
\right] +o_{p}(1) \\
&=\left[ 
\begin{array}{c}
\left[ 
\begin{array}{cc}
diag( \frac{1}{p( s) }:s\in \mathcal{S}) & \mathbf{0 }_{\vert \mathcal{S}\vert \times 1} \\
\mathbf{0}_{1\times \vert \mathcal{S}\vert } & 0
\end{array}
\right] + \\ 
\frac{\left[ 
\begin{array}{cc}
\left[ 
\begin{array}{c}
\left( \left[ 
\begin{array}{c}
\pi _{D( 1) }( s) \pi_{A}( s) + \\ 
\pi _{D( 0) }( s) ( 1-\pi_{A}( s)) 
\end{array}
\right] :s\in \mathcal{S}\right)  \\ 
\times ( \pi _{A}( s) :s\in \mathcal{S}) ^{\prime }
\end{array}
\right]  & \left( - \left[ 
\begin{array}{c}
\pi _{D( 1) }( s) \pi _{A}( s) + \\ 
\pi _{D( 0) }( s) ( 1-\pi _{A}( s)
) 
\end{array}
\right] :s\in \mathcal{S}\right) +o_{p}(1) \\ 
( - \pi _{A}( s) :s\in \mathcal{S}) ^{\prime } & 1
\end{array}
\right] }{\sum_{s\in \mathcal{S}}p( s) \pi _{A}( s)
( 1-\pi _{A}( s) ) ( \pi _{D( 1)
}( s) -\pi _{D( 0) }( s) ) }
\end{array}
\right]  +o_{p}(1).
\end{align*}
Also,
\begin{equation*}
{{\mathbf{Z}_{n}^{\mathrm{sfe}}}'\mathbf{Y}_{n}}/{n}=\left[ 
\begin{array}{c}
( \frac{1}{n}\sum_{i=1}^{n}1[S_{i}=s] Y_{i}:s\in \mathcal{S}) , \\
( \frac{1}{n}\sum_{i=1}^{n}A_{i}Y_{i})
\end{array}
\right] .
\end{equation*}
\end{lemma}
%%%%%%%% DIVIDER %%%%%%%%%%%%
\begin{proof}
The equalities follow from algebra and the convergences follow from the CMT. In particular, the first equality in the second display has an $o_p(1)$ to allow for the possibility that ${{{\mathbf{Z}}^{\mathrm{sfe}}_{n}}'\mathbf{X}_{n}^{\mathrm{sfe}}}/{n}$ is singular or any of the denominators being equal to zero. These events occur with vanishing probability under our assumptions.
\end{proof}
%%%%%%%% DIVIDER %%%%%%%%%%%%

\begin{theorem}[SFE limits]\label{thm:plim_SFE}
Assume Assumptions \ref{ass:1} and \ref{ass:2}. For every $s\in S$,
\begin{align}
&\hat{\gamma}_{\mathrm{sfe}}( s) ~\overset{p}{\to}~\left[
\begin{array}{c}
\pi _{D(0)}(s)E[Y(1)|AT,S=s]+(1-\pi_{D(1)}( s) )E[Y(0)|NT,S=s] +\\
( \pi _{D(1)}(s)-\pi _{D(0)}(s)) [ \pi _{A}( s) E[Y(1)|C,S=s]+(1-\pi_{A}(s) )E[Y(0)|C,S=s]] \\
-[(1-\pi _{A}( s) )\pi _{D( 0) }( s) +\pi _{A}( s) \pi _{D( 1) }( s) ]\times\\
\frac{\sum_{\tilde{s}\in \mathcal{S}}p( \tilde{s}) \pi _{A}( \tilde{s}) (1-\pi _{A}( \tilde{s}) )( \pi _{D(1)}( \tilde{s})-\pi _{D(0)}(\tilde{s})) E[Y(1)-Y(0)|C,S=\tilde{s}]}{\sum_{ \tilde{s}\in \mathcal{S}}p( \tilde{s}) \pi _{A}( \tilde{s} ) (1-\pi _{A}( \tilde{s}) )( \pi _{D(1)}(\tilde{s} )-\pi _{D(0)}(\tilde{s})) }
\end{array}
\right] \notag\\
&\hat{\beta}_{\mathrm{sfe}}~\overset{p}{\to}~\frac{\sum_{s\in \mathcal{S}}p( s) \pi _{A}( s) ( 1-\pi _{A}( s) ) ( \pi _{D(1)}(s)-\pi _{D(0)}(s)) E[Y(1)-Y(0)|C,S=s]}{\sum_{\tilde{ s}\in \mathcal{S}}p( \tilde{s}) \pi _{A}( \tilde{s}) ( 1-\pi _{A}( \tilde{s}) ) ( \pi _{D(1)}(s)-\pi _{D(0)}(\tilde{s})) }. \label{eq:plim_SFE2}
\end{align}
If we add Assumption \ref{ass:3}(c),
\begin{equation}
\hat{\beta}_{\mathrm{sfe}}~\overset{p}{\to}~\beta.
\label{eq:plim_SFE3}
\end{equation}
\end{theorem}
%%%%%%%% DIVIDER %%%%%%%%%%%%
\begin{proof}
Throughout this proof, define
\begin{align}
\Lambda _{n}
& \equiv \frac{n_{AD}}{n}-\sum_{s\in \mathcal{S}}\frac{n_{A}(s) }{n( s) }\frac{n_{D}(s)}{n( s) }\frac{n( s) }{n}
=\sum_{s\in \mathcal{S}}\frac{n( s) }{n}\frac{n_{A}(s)}{n( s) }\left( 1-\frac{n_{A}( s) }{n( s) }\right) \left( \frac{n_{AD}( s) }{n_{A}( s) }-\frac{ n_{D}(s)-n_{AD}( s) }{n( s) -n_{A}( s) } \right)\notag \\
&\overset{p}{\to}\Lambda \equiv \sum_{s\in \mathcal{S}}p( s) \pi _{A}( s) ( 1-\pi _{A}( s) ) ( \pi _{D( 1) }( s) -\pi _{D( 0) }( s) ) , \label{eq:plim_SFE4}
\end{align}
where the convergence holds by Assumptions \ref{ass:1} and \ref{ass:2}, Lemma \ref{lem:A1and2_impliesold3}, and the LLN.

To show the first line of \eqref{eq:plim_SFE2}, consider the following derivation.
\begin{align*}
&\hat{\gamma}_{\mathrm{sfe}}( s) 
=\frac{n}{n( s) }\frac{1}{n} \sum_{i=1}^{n}1[ S_{i}=s] Y_{i}-\frac{n_{D}(s)}{n( s) }\frac{1}{\Lambda _{n}}\frac{1}{n}\sum_{i=1}^{n}1[ A_{i}=1] Y_{i}+\frac{n_{D}(s)}{n( s) }\frac{1}{\Lambda _{n} }\sum_{\tilde{s}\in \mathcal{S}}\frac{n_{A}(\tilde{s})}{n( \tilde{s} ) }\frac{1}{n}\sum_{i=1}^{n}1[ S_{i}=\tilde{s}] Y_{i}  +o_p(1)\\
&\overset{(1)}{=}\left[ 
\begin{array}{c}
\frac{n}{n( s) }[ \frac{1}{\sqrt{n}}R_{n,1}( 1,1,s) +\frac{1}{\sqrt{n}}R_{n,1}( 1,0,s) +\frac{1}{\sqrt{n}}R_{n,1}( 0,1,s) +\frac{1}{\sqrt{n}}R_{n,1}( 0,0,s) ] \\
+\left[ ( 1-\frac{n_{A}( s) }{n( s) }) \frac{ n_{D}(s)-n_{AD}( s) }{n( s) -n_{A}( s) }+ \frac{n_{A}( s) }{n( s) }\frac{n_{AD}( s) }{ n_{A}( s) }\right]\frac{1}{\Lambda _{n}}\times \\
\sum_{\tilde{s}\in \mathcal{S}}[
\frac{n_{A}(\tilde{s})}{n( \tilde{ s}) }\frac{1}{\sqrt{n}}(R_{n,1}( 1,0,\tilde{s}) +R_{n,1}( 0,0,\tilde{s}))
-( 1- \frac{n_{A}(\tilde{s})}{n( \tilde{s}) }) \frac{1}{\sqrt{n}} (R_{n,1}( 1,1,\tilde{s})+R_{n,1}( 0,1,\tilde{s}) )] \\
+\frac{n_{A}( s) }{n( s) }\frac{n_{AD}( s) }{n_{A}( s) }( \mu ( 1,1,s) +E[ Y( 1) |S=s] ) +( 1-\frac{n_{A}( s) }{n( s) }) ( \frac{n_{D}( s) -n_{AD}( s) }{ n( s) -n_{A}( s) }) ( \mu ( 1,0,s) +E[ Y( 1) |S=s] ) \\
+\frac{n_{A}( s) }{n( s) }( 1-\frac{n_{AD}( s) }{n_{A}( s) }) ( \mu ( 0,1,s) +E [ Y( 0) |S=s] ) \\+( 1-\frac{n_{A}( s) }{n( s) }) ( 1-\frac{n_{D}( s) -n_{AD}( s) }{n( s) -n_{A}( s) }) ( \mu ( 0,0,s) +E[ Y( 0) |S=s] ) \\
+( ( 1-\frac{n_{A}( s) }{n( s) }) \frac{n_{D}(s)-n_{AD}( s) }{n( s) -n_{A}( s) }+\frac{n_{A}( s) }{n( s) }\frac{n_{AD}( s) }{n_{A}( s) }) \frac{1}{\Lambda _{n}}\times 
\sum_{\tilde{s}\in \mathcal{S}}\frac{n( \tilde{s}) }{n}\frac{ n_{A}(\tilde{s})}{n( \tilde{s}) }( 1-\frac{n_{A}( \tilde{s}) }{n( \tilde{s}) }) \times\\
\left(
\begin{array}{c}
-\frac{n_{AD}( \tilde{s}) }{n_{A}( \tilde{s}) }( \mu ( 1,1,\tilde{s}) +E[ Y( 1) |S=\tilde{s}] ) -( 1-\frac{n_{AD}( \tilde{s}) }{n_{A}( \tilde{s} ) }) ( \mu ( 0,1,\tilde{s}) +E[ Y( 0) |S=\tilde{s}] ) \\
+( \frac{n_{D}( \tilde{s}) -n_{AD}( \tilde{s}) }{ n( \tilde{s}) -n_{A}( \tilde{s}) }) ( \mu ( 1,0,\tilde{s}) +E[ Y( 1) |S=\tilde{s}] ) +( 1-\frac{n_{D}( \tilde{s}) -n_{AD}( \tilde{s} ) }{n( \tilde{s}) -n_{A}( \tilde{s}) }) ( \mu ( 0,0,\tilde{s}) +E[ Y( 0) |S=\tilde{s}
] ) 
\end{array}
\right) 
\end{array}
\right] +o_p(1)  \\
&\overset{(2)}{=}\left[ 
\begin{array}{c}
\frac{n}{n( s) }[ \frac{1}{\sqrt{n}}R_{n,1}( 1,1,s) +\frac{1}{\sqrt{n}}R_{n,1}( 1,0,s) +\frac{1}{\sqrt{n}}R_{n,1}( 0,1,s) +\frac{1}{\sqrt{n}}R_{n,1}( 0,0,s)] \\
+\left[ ( 1-\frac{n_{A}( s) }{n( s) }) \frac{ n_{D}(s)-n_{AD}( s) }{n( s) -n_{A}( s) }+ \frac{n_{A}( s) }{n( s) }\frac{n_{AD}( s) }{ n_{A}( s) }\right] \frac{1}{\Lambda _{n}}\times\\
\sum_{\tilde{s}\in \mathcal{S}}[
\frac{n_{A}(\tilde{s})}{n( \tilde{ s}) }\frac{1}{\sqrt{n}}(R_{n,1}( 1,0,\tilde{s}) +R_{n,1}( 0,0,\tilde{s}))
-( 1- \frac{n_{A}(\tilde{s})}{n( \tilde{s}) }) \frac{1}{\sqrt{n}} (R_{n,1}( 1,1,\tilde{s})+R_{n,1}( 0,1,\tilde{s}) )]  \\
+\frac{n_{A}( s) }{n( s) }\frac{n_{AD}( s) }{n_{A}( s) }( E[Y(1)|AT,S=s]\tfrac{\pi _{D(0)}(s)}{\pi _{D(1)}(s)}+E[Y(1)|C,S=s]\tfrac{\pi _{D(1)}(s)-\pi _{D(0)}(s)}{\pi _{D(1)}(s) }) \\
+( 1-\frac{n_{A}( s) }{n( s) }) ( \frac{n_{D}( s) -n_{AD}( s) }{n( s) -n_{A}( s) }) E[Y(1)|AT,S=s]+\frac{n_{A}( s) }{ n( s) }( 1-\frac{n_{AD}( s) }{n_{A}( s) }) E[Y(0)|NT,S=s] \\
+( 1-\frac{n_{A}( s) }{n( s) }) ( 1- \frac{n_{D}( s) -n_{AD}( s) }{n( s) -n_{A}( s) }) ( E[Y(0)|NT,S=s]\tfrac{1-\pi _{D(1)}(s)}{ 1-\pi _{D(0)}(s)}+E[Y(0)|C,S=s]\tfrac{\pi _{D(1)}(s)-\pi _{D(0)}(s)}{1-\pi _{D(0)}(s)}) \\
+( ( 1-\frac{n_{A}( s) }{n( s) }) \frac{n_{D}(s)-n_{AD}( s) }{n( s) -n_{A}( s) }+\frac{n_{A}( s) }{n( s) }\frac{n_{AD}( s) }{n_{A}( s) }) \frac{1}{\Lambda _{n}}\times 
\sum_{\tilde{s}\in \mathcal{S}}\frac{n( \tilde{s}) }{n}\frac{ n_{A}(\tilde{s})}{n( \tilde{s}) }( 1-\frac{n_{A}( \tilde{s}) }{n( \tilde{s}) }) \times\\
\left(
\begin{array}{c}
-\frac{n_{AD}( \tilde{s}) }{n_{A}( \tilde{s}) }( E[Y(1)|AT,S=\tilde{s}]\tfrac{\pi _{D(0)}(\tilde{s})}{\pi _{D(1)}(\tilde{s})} +E[Y(1)|C,S=\tilde{s}]\tfrac{\pi _{D(1)}(\tilde{s})-\pi _{D(0)}(\tilde{s})}{ \pi _{D(1)}(\tilde{s})}) \\
-( 1-\frac{n_{AD}( \tilde{s}) }{n_{A}( \tilde{s}) }) E[Y(0)|NT,S=\tilde{s}]+( \frac{n_{D}( \tilde{s}) -n_{AD}( \tilde{s}) }{n( \tilde{s}) -n_{A}( \tilde{s}) }) E[Y(1)|AT,S=\tilde{s}]+ \\
( 1-\frac{n_{D}( \tilde{s}) -n_{AD}( \tilde{s}) }{ n( \tilde{s}) -n_{A}( \tilde{s}) }) ( E[Y(0)|NT,S=\tilde{s}]\tfrac{1-\pi _{D(1)}(\tilde{s})}{1-\pi _{D(0)}(\tilde{s })}+E[Y(0)|C,S=\tilde{s}]\tfrac{\pi _{D(1)}(\tilde{s})-\pi _{D(0)}(\tilde{s}) }{1-\pi _{D(0)}(\tilde{s})})
\end{array}
\right) 
\end{array}
\right]   +o_p(1)\\
&\overset{(3)}{=}\left[
\begin{array}{c}
\pi _{D(0)}(s)E[Y(1)|AT,S=s]+(1-\pi_{D(1)}( s) )E[Y(0)|NT,S=s] \\
+( \pi _{D(1)}(s)-\pi _{D(0)}(s)) [ \pi _{A}( s) E[Y(1)|C,S=s]+(1-\pi_{A}(s) )E[Y(0)|C,S=s]] \\
-((1-\pi _{A}( s) )\pi _{D( 0) }( s) +\pi _{A}( s) \pi _{D( 1) }( s) )\frac{1}{\Lambda} \times\\
{\sum_{\tilde{s}\in \mathcal{S}}p( \tilde{s}) \pi _{A}( \tilde{s}) (1-\pi _{A}( \tilde{s}) )( \pi _{D(1)}( \tilde{s})-\pi _{D(0)}(\tilde{s})) E[Y(1)-Y(0)|C,S=\tilde{s}]}
\end{array}
\right] +o_{p}(1),
\end{align*}
where (1) holds by $Y_{i}=Y_{i}( D_{i}) $, \eqref{eq:defnY_pre}, \eqref{eq:defnY}, and \eqref{eq:Rn_defn}, (2) holds by Lemma \ref{lem:MeanTraslation}, and (3) holds by Assumptions \ref{ass:1} and \ref{ass:2}(b), Lemmas \ref{lem:A1and2_impliesold3} and \ref{lem:AsyDist2}, and \eqref{eq:plim_SFE4}.

To show the second line of \eqref{eq:plim_SFE2}, consider the following derivation.
\begin{align}
\hat{\beta}_{\mathrm{sfe}} &=\frac{1}{\Lambda _{n}}\frac{1}{n}\sum_{i=1}^{n}1[ A_{i}=1] Y_{i}-\frac{1}{\Lambda _{n}}\sum_{s\in \mathcal{S}}\frac{1}{n }\sum_{i=1}^{n}Y_{i}1[ S_{i}=s] \frac{n_{A}(s)}{n( s) } +o_p(1)\notag \\
&\overset{(1)}{=}\frac{1}{\Lambda _{n}}\sum_{\tilde{s}\in \mathcal{S}}\left[
\begin{array}{c}
( 1-\frac{n_{A}(\tilde{s})}{n( \tilde{s}) }) \frac{1}{ \sqrt{n}}(R_{n,1}( 1,1,s) +R_{n,1}( 0,1,s)) -\frac{ n_{A}(\tilde{s})}{n( \tilde{s}) }\frac{1}{\sqrt{n}}(R_{n,1}( 1,0,s)+R_{n,1}( 0,0,s)) \\
+\frac{n( \tilde{s}) }{n}\frac{n_{A}( \tilde{s}) }{ n( \tilde{s}) }( 1-\frac{n_{A}(\tilde{s})}{n( \tilde{s} ) }) \frac{n_{AD}( \tilde{s}) }{n_{A}( \tilde{s} ) }( \mu ( 1,1,\tilde{s}) +E[ Y( 1) |S= \tilde{s}] ) \\
+\frac{n( \tilde{s}) }{n}\frac{n_{A}( \tilde{s}) }{ n( \tilde{s}) }( 1-\frac{n_{A}(\tilde{s})}{n( \tilde{s} ) }) ( 1-\frac{n_{AD}( \tilde{s}) }{n_{A}( \tilde{s}) }) ( \mu ( 0,1,\tilde{s}) +E[ Y( 0) |S=\tilde{s}] ) \\
-\frac{n( \tilde{s}) }{n}\frac{n_{A}(\tilde{s})}{n( \tilde{s} ) }( 1-\frac{n_{A}(\tilde{s})}{n( \tilde{s}) }) ( \frac{n_{D}( \tilde{s}) -n_{AD}( \tilde{s}) }{ n( \tilde{s}) -n_{A}( \tilde{s}) }) ( \mu ( 1,0,\tilde{s}) +E[ Y( 1) |S=\tilde{s}] ) \\
-\frac{n( \tilde{s}) }{n}\frac{n_{A}(\tilde{s})}{n( \tilde{s} ) }( 1-\frac{n_{A}(\tilde{s})}{n( \tilde{s}) }) ( 1-\frac{n_{D}( \tilde{s}) -n_{AD}( \tilde{s}) }{ n( \tilde{s}) -n_{A}( \tilde{s}) }) ( \mu ( 0,0,\tilde{s}) +E[ Y( 0) |S=\tilde{s}] )
\end{array}
\right]  +o_p(1) \notag \\
&\overset{(2)}{=}\frac{1}{\Lambda _{n}}\sum_{s\in \mathcal{S}}\left[
\begin{array}{c}
( 1-\frac{n_{A}(s)}{n( s) }) \frac{1}{\sqrt{n}} R_{n,1}( 1,1,s) +( 1-\frac{n_{A}(s)}{n( s) }) \frac{1}{\sqrt{n}}R_{n,1}( 0,1,s) \\
-\frac{n_{A}(s)}{n( s) }\frac{1}{\sqrt{n}}R_{n,1}( 1,0,s) -\frac{n_{A}(s)}{n( s) }\frac{1}{\sqrt{n}} R_{n,1}( 0,0,s) \\
+\frac{n( s) }{n}\frac{n_{A}( s) }{n( s) } ( 1-\frac{n_{A}(s)}{n( s) }) [ \frac{n_{AD}( s) }{n_{A}( s) }\tfrac{\pi _{D(0)}(s)}{\pi _{D(1)}(s)} -( \frac{n_{D}( s) -n_{AD}( s) }{n( s) -n_{A}( s) }) ] E[Y(1)|AT,S=s] \\
+\frac{n( s) }{n}\frac{n_{A}( s) }{n( s) } ( 1-\frac{n_{A}(s)}{n( s) }) [ ( 1-\frac{ n_{AD}( s) }{n_{A}( s) }) -( 1-\frac{ n_{D}( s) -n_{AD}( s) }{n( s) -n_{A}( s) }) \tfrac{1-\pi _{D(1)}(s)}{1-\pi _{D(0)}(s)}] E[Y(0)|NT,S=s] \\
+\frac{n( s) }{n}\frac{n_{A}( s) }{n( s) } ( 1-\frac{n_{A}(s)}{n( s) }) \frac{n_{AD}( s) }{n_{A}( s) }\tfrac{\pi _{D(1)}(s)-\pi _{D(0)}(s)}{\pi _{D(1)}(s)}E[Y(1)|C,S=s] \\
-\frac{n( s) }{n}\frac{n_{A}(s)}{n( s) }( 1-\frac{ n_{A}( s) }{n( s) }) ( 1-\frac{n_{D}( s) -n_{AD}( s) }{n( s) -n_{A}( s) } ) \tfrac{\pi _{D(1)}(s)-\pi _{D(0)}(s)}{1-\pi _{D(0)}(s)}E[Y(0)|C,S=s]
\end{array}
\right]   +o_p(1) \notag \\
&\overset{(3)}{=}\frac{1}{\Lambda }\sum_{s\in \mathcal{S}}p( s) \pi _{A}( s) ( 1-\pi _{A}( s) ) ( \pi _{D(1)}(s)-\pi _{D(0)}(s)) E[Y(1)-Y(0)|C,S=s]+o_{p}( 1) ,\label{eq:plim_SFE5}
\end{align}
where (1) holds by $Y_{i}=Y_{i}( D_{i}) $, \eqref{eq:defnY_pre}, \eqref{eq:defnY}, and \eqref{eq:Rn_defn}, (2) holds by Lemma \ref{lem:MeanTraslation}, and (3) holds by Assumptions \ref{ass:1} and \ref{ass:2}(b), Lemmas \ref{lem:A1and2_impliesold3} and \ref{lem:AsyDist2}, and \eqref{eq:plim_SFE4}.

Finally, \eqref{eq:plim_SFE3} holds by the following derivation.
\begin{align*}
\hat{\beta}_{\mathrm{sfe}}&~\overset{(1)}{=}~\frac{\sum_{s\in \mathcal{S}}p( s) ( \pi _{D(1)}(s)-\pi _{D(0)}(s)) E[Y(1)-Y(0)|C,S=s]}{\sum_{\tilde{s}\in \mathcal{S}}p( \tilde{s}) ( \pi _{D(1)}(\tilde{s})-\pi _{D(0)}( \tilde{s})) } + o_p(1)\\ &~\overset{(2)}{=}~\sum_{s\in \mathcal{S}}P( S=s|C) E[Y(1)-Y(0)|C,S=s] + o_p(1)~=~\beta + o_p(1) ,
\end{align*}
where (1) holds by Assumption \ref{ass:3}(c), and \eqref{eq:plim_SFE4} and \eqref{eq:plim_SFE5}, and (2) holds by \eqref{eq:pi_defns}.
\end{proof}
%%%%%%%% DIVIDER %%%%%%%%%%%%

%%%%%%%% DIVIDER %%%%%%%%%%%%
\begin{proof}[Proof of Theorem \ref{thm:AsyDist_SFE}.]
As a preliminary result, note that \eqref{eq:plim_SFE4} and Assumption \ref{ass:3}(c) imply that
\begin{align}
\Lambda _{n} \equiv \frac{n_{AD}}{n}-\sum_{s\in \mathcal{S}}\frac{n_{A}(s)}{ n(s)}\frac{n_{D}(s)}{n(s)}\frac{n(s)}{n}  \overset{p}{\to }\Lambda \equiv \pi _{A}(1-\pi _{A})P( C) .\label{eq:SFE_av2}
\end{align}
From here, consider the following argument.
\begin{align}
\xi _{n,1} &\equiv \sqrt{n}( \Lambda _{n}-\Lambda ) \overset{(1)}{=}\sum_{s\in \mathcal{S}}\left[
\begin{array}{c}
\frac{n_{A}(s)}{n(s)}(1-\frac{n_{A}(s)}{n(s)})(\frac{n_{AD}(s)}{n_{A}(s)}- \frac{n_{D}(s)-n_{AD}(s)}{n(s)-n_{A}(s)})R_{n,4}( s) \\
+p(s)(1-\frac{n_{A}(s)}{n(s)}-\pi _{A})(\frac{n_{AD}(s)}{n_{A}(s)}-\frac{ n_{D}(s)-n_{AD}(s)}{n(s)-n_{A}(s)})R_{n,3}( s) \\
+p(s)\pi _{A}(1-\pi _{A})( R_{n,2}( 1,s) -R_{n,2}( 2,s) )
\end{array}
\right] +o_p(1) \notag \\
&\overset{(2)}{=}\sum_{s\in \mathcal{S}}\left[
\begin{array}{c}
p(s)\pi _{A}(1-\pi _{A})( R_{n,2}( 1,s) -R_{n,2}( 2,s) ) \\
+p(s)(1-2\pi _{A})(\pi _{D(1)}(s)-\pi _{D(0)}(s))R_{n,3}( s) \\
+\pi _{A}(1-\pi _{A})(\pi _{D(1)}(s)-\pi _{D(0)}(s))R_{n,4}( s)
\end{array}
\right] +o_{p}( 1) ,\label{eq:SFE_av3}
\end{align}
where (1) holds by the definitions in \eqref{eq:SFE_av2} and (2) holds by Assumption \ref{ass:2}(b) and Lemma \ref{lem:A1and2_impliesold3}.

Next, consider the following derivation.
\begin{align}
&\sqrt{n}( \hat{\beta}_{\mathrm{sfe}}-\beta ) \times \Lambda _{n}\notag \\
&\overset{(1)}{=} \sum_{s\in \mathcal{S}}\left[
\begin{array}{c}
(1-\frac{n_{A}(s)}{n(s)})R_{n,1}(1,1,s)+(1-\frac{n_{A}(s)}{n(s)} )R_{n,1}(0,1,s)\\
-\frac{n_{A}(s)}{n(s)}R_{n,1}(1,0,s)-\frac{n_{A}(s)}{n(s)} R_{n,1}(0,0,s) \\
+\sqrt{n}\frac{n(s)}{n}\frac{n_{A}(s)}{n(s)}(1-\frac{n_{A}(s)}{n(s)})[\frac{ n_{AD}(s)}{n_{A}(s)}\tfrac{\pi _{D(0)}(s)}{\pi _{D(1)}(s)}-(\frac{ n_{D}(s)-n_{AD}(s)}{n(s)-n_{A}(s)})]E[Y(1)|AT,S=s] \\
+\sqrt{n}\frac{n(s)}{n}\frac{n_{A}(s)}{n(s)}(1-\frac{n_{A}(s)}{n(s)})[(1- \frac{n_{AD}(s)}{n_{A}(s)})-(1-\frac{n_{D}(s)-n_{AD}(s)}{n(s)-n_{A}(s)}) \tfrac{1-\pi _{D(1)}(s)}{1-\pi _{D(0)}(s)}]\times\\
E[Y(0)|NT,S=s] \\
+\sqrt{n}\frac{n(s)}{n}\frac{n_{A}(s)}{n(s)}(1-\frac{n_{A}(s)}{n(s)})\frac{ n_{AD}(s)}{n_{A}(s)}\tfrac{\pi _{D(1)}(s)-\pi _{D(0)}(s)}{\pi _{D(1)}(s)} E[Y(1)|C,S=s] \\
-\sqrt{n}\frac{n(s)}{n}\frac{n_{A}(s)}{n(s)}(1-\frac{n_{A}(s)}{n(s)})(1- \frac{n_{D}(s)-n_{AD}(s)}{n(s)-n_{A}(s)})\tfrac{\pi _{D(1)}(s)-\pi _{D(0)}(s) }{1-\pi _{D(0)}(s)}E[Y(0)|C,S=s] \\
-\sqrt{n}p(s)\pi _{A}(1-\pi _{A})(\pi _{D(1)}(s)-\pi _{D(0)}(s))E[Y(1)-Y(0)|C,S=s]
\end{array}
\right] -\beta \xi _{n,1} +o_p(1)\notag \\
&\overset{(2)}{=} \sum_{s\in \mathcal{S}}\left[
\begin{array}{c}
(1-\frac{n_{A}(s)}{n(s)})R_{n,1}(1,1,s)+(1-\frac{n_{A}(s)}{n(s)} )R_{n,1}(0,1,s) \\
-\frac{n_{A}(s)}{n(s)}R_{n,1}(1,0,s)-\frac{n_{A}(s)}{n(s)}R_{n,1}(0,0,s) \\ +\xi _{n,2}( s) +\xi _{n,3}( s) +\xi _{n,4}( s)
\end{array}
\right] -\beta \xi _{n,1}  +o_p(1),\label{eq:SFE_av4}
\end{align}
where (1) follows from \eqref{eq:plim_SFE5} and \eqref{eq:SFE_av3} and (2) follows from defining $\xi _{n,2}( s) $, $\xi _{n,3}( s) $, and $ \xi _{n,4}( s) $ as in \eqref{eq:SFE_av5}, \eqref{eq:SFE_av6}, and \eqref{eq:SFE_av7}, respectively. To complete the argument in \eqref{eq:SFE_av4}, consider the following definitions. First,
\begin{align}
\xi _{n,2}( s) &\equiv \sqrt{n}\frac{n(s)}{n}\frac{n_{A}(s)}{ n(s)}(1-\frac{n_{A}(s)}{n(s)})\left[\frac{n_{AD}(s)}{n_{A}(s)}\tfrac{\pi _{D(0)}(s)}{\pi _{D(1)}(s)}-(\frac{n_{D}(s)-n_{AD}(s)}{n(s)-n_{A}(s)} )\right]E[Y(1)|AT,S=s]\notag \\
%&=\frac{n(s)}{n}\frac{n_{A}(s)}{n(s)}(1-\frac{n_{A}(s)}{n(s)})E[Y(1)|AT,S=s] [ \frac{\pi _{D(0)}(s)}{\pi _{D(1)}(s)}R_{n,2}( 1,s) -R_{n,2}( 2,s) ] \notag \\
&\overset{(1)}{=}p( s) \pi _{A}(1-\pi _{A})E[Y(1)|AT,S=s]\left[ \frac{\pi _{D(0)}(s)}{\pi _{D(1)}(s)}R_{n,2}( 1,s) -R_{n,2}( 2,s) \right] +o_{p}( 1) , \label{eq:SFE_av5}
\end{align}
where (1) uses Assumptions \ref{ass:1} and \ref{ass:2}(b). Second,
\begin{align}
\xi _{n,3}( s) &\equiv \sqrt{n}\frac{n(s)}{n}\frac{n_{A}(s)}{ n(s)}(1-\frac{n_{A}(s)}{n(s)})\left[(1-\frac{n_{AD}(s)}{n_{A}(s)})-(1-\frac{ n_{D}(s)-n_{AD}(s)}{n(s)-n_{A}(s)})\tfrac{1-\pi _{D(1)}(s)}{1-\pi _{D(0)}(s)} \right]E[Y(0)|NT,S=s] \notag\\
%&=\frac{n(s)}{n}\frac{n_{A}(s)}{n(s)}(1-\frac{n_{A}(s)}{n(s)})E[Y(0)|NT,S=s] \left[ \frac{1-\pi _{D(1)}(s)}{1-\pi _{D(0)}(s)} R_{n,2}( 2,s) -R_{n,2}( 1,s)\right] \notag \\
&\overset{(1)}{=}p( s) \pi _{A}(1-\pi _{A})E[Y(0)|NT,S=s]\left[ \frac{1-\pi _{D(1)}(s)}{1-\pi _{D(0)}(s)} R_{n,2}( 2,s) -R_{n,2}( 1,s)\right] +o_{p}( 1), \label{eq:SFE_av6} 
\end{align}
where (1) uses Assumptions \ref{ass:1} and \ref{ass:2}(b). Third,
\begin{align}
\xi _{n,4}( s)  &\equiv\left[ 
\begin{array}{c}
+\sqrt{n}\frac{n(s)}{n}\frac{n_{A}(s)}{n(s)}(1-\frac{n_{A}(s)}{n(s)})\frac{ n_{AD}(s)}{n_{A}(s)}\tfrac{\pi _{D(1)}(s)-\pi _{D(0)}(s)}{\pi _{D(1)}(s)} E[Y(1)|C,S=s] \\
-\sqrt{n}\frac{n(s)}{n}\frac{n_{A}(s)}{n(s)}(1-\frac{n_{A}(s)}{n(s)})(1- \frac{n_{D}(s)-n_{AD}(s)}{n(s)-n_{A}(s)})\tfrac{\pi _{D(1)}(s)-\pi _{D(0)}(s) }{1-\pi _{D(0)}(s)}E[Y(0)|C,S=s] \\
-\sqrt{n}p(s)\pi _{A}(1-\pi _{A})(\pi _{D(1)}(s)-\pi _{D(0)}(s))E[Y(1)-Y(0)|C,S=s]
\end{array}
\right] \notag \\
%&=\left\{ 
%\begin{array}{c}
%+p( s) \pi _{A}(1-\pi _{A})\tfrac{\pi _{D(1)}(s)-\pi _{D(0)}(s)}{ \pi _{D(1)}(s)}E[Y(1)|C,S=s]R_{n,2}( 1,s) \\
%+p(s)\pi _{A}(1-\pi _{A})\tfrac{\pi _{D(1)}(s)-\pi _{D(0)}(s)}{1-\pi _{D(0)}(s)}E[Y(0)|C,S=s]R_{n,2}( 2,s) \\
%+p( s) ( \pi _{D(1)}(s)-\pi _{D(0)}(s)) \left[
%\begin{array}{c}
%\frac{n_{AD}(s)}{n_{A}(s)}\tfrac{1}{\pi _{D(1)}(s)}[ (1-\frac{n_{A}(s)}{ n(s)})-\pi _{A}] E[Y(1)|C,S=s]+ \\
%\tfrac{1}{1-\pi _{D(0)}(s)}(1-\frac{n_{D}(s)-n_{AD}(s)}{n(s)-n_{A}(s)})[ \pi _{A}-(1-\frac{n_{A}(s)}{n(s)})] E[Y(0)|C,S=s]
%\end{array}
%\right] R_{n,3}( s) \\
%+\frac{n_{A}(s)}{n(s)}(1-\frac{n_{A}(s)}{n(s)})( \pi _{D(1)}(s)-\pi _{D(0)}(s)) \left[
%\begin{array}{c}
%\frac{n_{AD}(s)}{n_{A}(s)}\tfrac{1}{\pi _{D(1)}(s)}E[Y(1)|C,S=s] \\
%-(1-\frac{n_{D}(s)-n_{AD}(s)}{n(s)-n_{A}(s)})\tfrac{1}{1-\pi _{D(0)}(s)} E[Y(0)|C,S=s]
%\end{array}
%\right] R_{n,4}( s)
%\end{array}
%\right\} \notag \\
&\overset{(1)}{=}\left[
\begin{array}{c}
p( s) \pi _{A}(1-\pi _{A})( \pi _{D(1)}(s)-\pi _{D(0)}(s)) \tfrac{1}{\pi _{D(1)}(s)}E[Y(1)|C,S=s]R_{n,2}( 1,s) \\
+p( s) \pi _{A}(1-\pi _{A})( \pi _{D(1)}(s)-\pi _{D(0)}(s)) \tfrac{1}{1-\pi _{D(0)}(s)}E[Y(0)|C,S=s]R_{n,2}( 2,s) \\
+p( s) (1-2\pi _{A})( \pi _{D(1)}(s)-\pi _{D(0)}(s)) E[Y(1)-Y(0)|C,S=s]R_{n,3}( s) \\
+\pi _{A}(1-\pi _{A})( \pi _{D(1)}(s)-\pi _{D(0)}(s)) E[Y(1)-Y(0)|C,S=s]R_{n,4}( s)
\end{array}
\right] +o_{p}( 1) ,\label{eq:SFE_av7}
\end{align}
where (1) uses Assumptions \ref{ass:1} and \ref{ass:2}(b), and Lemma \ref{lem:A1and2_impliesold3}.

From there results, the next result follows.
\begin{equation}
\sqrt{n}( \hat{\beta}_{\mathrm{sfe}}-\beta ) 
\overset{(1)}{=}
\frac{1}{P( C) }\sum_{s\in \mathcal{S}}\left[
\begin{array}{c}
\frac{R_{n,1}(1,1,s)}{\pi _{A}}+\frac{R_{n,1}(0,1,s)}{\pi _{A}}-\frac{ R_{n,1}(1,0,s)}{( 1-\pi _{A}) }-\frac{R_{n,1}(0,0,s)}{( 1-\pi _{A}) } \\
+p( s) \left(
\begin{array}{c}
(\beta ( s) -\beta ) \\ 
+E[Y(0)|C,S=s]-E[Y(0)|NT,S=s] \\
+( E[Y(1)|AT,S=s]-E[Y(1)|C,S=s]) \frac{\pi _{D(0)}(s)}{\pi _{D(1)}(s)}
\end{array}
\right) R_{n,2}( 1,s) \\
+p( s) \left(
\begin{array}{c}
E[Y(1)|C,S=s]-E[Y(1)|AT,S=s] \\
-( E[Y(0)|C,S=s]-E[Y(0)|NT,S=s]) \frac{1-\pi _{D(1)}(s)}{1-\pi _{D(0)}(s)} \\
-( \beta ( s) -\beta )
\end{array}
\right) R_{n,2}( 2,s) \\
+p( s) \frac{(1-2\pi _{A})}{\pi _{A}(1-\pi _{A})}( \pi _{D(1)}(s)-\pi _{D(0)}(s)) [ E[Y(1)-Y(0)|C,S=s]-\beta ] R_{n,3}( s) \\
+( \pi _{D(1)}(s)-\pi _{D(0)}(s)) [ E[Y(1)-Y(0)|C,S=s]-\beta ] R_{n,4}( s)
\end{array}
\right] +o_{p}( 1) , \label{eq:SFE_av8}
\end{equation}
where (1) uses \eqref{eq:SFE_av2}, \eqref{eq:SFE_av4}, \eqref{eq:SFE_av5}, \eqref{eq:SFE_av6}, and \eqref{eq:SFE_av7}, and Lemma \ref{lem:AsyDist}, as it implies $ R_{n}=O_{p}( 1) $. The desired result then follows from \eqref{eq:SFE_av8}, Lemmas \ref{lem:MeanTraslation} and \ref{lem:AsyDist}, and $\sum_{s\in \mathcal{S}}( \pi _{D(1)}(s)-\pi _{D(0)}(s)) ( E[Y(1)-Y(0)|C,S=s]-\beta ) p( s) =0,$
which in turn follows from $\beta ( s) =E[Y(1)-Y(0)|C,S=s]$ and \eqref{eq:pi_defns}.
\end{proof}
%%%%%%%% DIVIDER %%%%%%%%%%%%

\begin{proof}[Proof of Theorem \ref{thm:SFE_se}]
Note that Assumption \ref{ass:1}, \ref{ass:2}(b), and \ref{ass:3}(c), Lemma \ref{lem:A1and2_impliesold3}, and \eqref{eq:hat_limits_1} and \eqref{eq:hat_limits_2}, imply that $\hat{V}_{A}^{\mathrm{sfe}} ~\overset{p}{\to }~V_{{A}}^{\mathrm{sfe}}$. The desired result follows from this and Theorem \ref{thm:SAT_se}.
\end{proof}
%%%%%%%% DIVIDER %%%%%%%%%%%%

\begin{proof}[Proof of Theorem \ref{thm:SFE_test}]
This result follows from elementary convergence arguments and Theorems \ref{thm:AsyDist_SFE} and \ref{thm:SFE_se}.
\end{proof}
%%%%%%%% DIVIDER %%%%%%%%%%%%

%%%%%%%%%%%%%%%%%%%%%%%%%%%%%%%%%%%%%%%%%%%%%%%%%%%%%%%%%%%%%%%%%%%%%%%%%%%%%
\subsection{Proofs of results in  Section \ref{sec:2STT}}\label{sec:A_2SR}

\begin{lemma}[2S matrices]\label{lem:Matrix2S} 
Assume Assumptions \ref{ass:1} and \ref{ass:2}. Then,
\begin{align*}
{{\mathbf{Z}_{n}^{\mathrm{2s}}}^{\prime }\mathbf{X}_{n}^{\mathrm{2s}}}/n
&=\left[ 
\begin{array}{cc}
1 & n_{D}/n \\ 
n_{A}/n & n_{AD}/n
\end{array}
\right] \\
&=\left[ 
\begin{array}{cc}
1 & \sum_{s\in \mathcal{S}}p(s)[ \pi _{D(1)}( s) \pi _{A}( s) +\pi _{D( 0) }( s) ( 1-\pi _{A}( s) ) ] \\
\sum_{s\in \mathcal{S}}p(s)\pi _{A}( s) & \sum_{s\in \mathcal{S} }p(s)\pi _{D(1)}( s) \pi _{A}( s)
\end{array}
\right] +o_{p}(1) .
\end{align*}
Thus, 
\begin{align*}
( {{\mathbf{Z}_{n}^{\mathrm{2s}}}^{\prime }\mathbf{X}_{n}^{\mathrm{2s}}}/{n}) ^{-1} &=\frac{1}{n_{AD}/n-( n_{A}/n) ( n_{D}/n) }\left[
\begin{array}{cc}
n_{AD}/n & -n_{D}/n \\ 
-n_{A}/n & 1
\end{array}
\right] +o_{p}(1)  \\
&=\frac{\left[ 
\begin{array}{cc}
\sum_{s\in \mathcal{S} }p(s)\pi _{D(1)}( s) \pi _{A}( s) & -\sum_{s\in \mathcal{S}}p(s)[ \pi _{D(1)}( s) \pi _{A}( s) +\pi _{D( 0) }( s) ( 1-\pi _{A}( s) ) ] \\
-\sum_{s\in \mathcal{S}}p(s)\pi _{A}( s) & 1
\end{array}
\right] }{\left( 
\begin{array}{c}
\sum_{s\in \mathcal{S}}p(s)\pi _{D(1)}( s) \pi _{A}( s) \\
-( \sum_{s\in \mathcal{S}}p(s)[ \pi _{D( 1) }( s) \pi _{A}( s) +\pi _{D( 0) }( s) ( 1-\pi _{A}( s) ) ] ) ( \sum_{s\in \mathcal{S}}p(s)\pi _{A}( s) )
\end{array}
\right) } +o_{p}(1).
\end{align*}
Also,
\begin{equation*}
{{\mathbf{Z}_{n}^{\mathrm{2s}}}'\mathbf{Y}_{n}}/{n}=\left[ 
\begin{array}{c}
 \frac{1}{n}\sum_{i=1}^{n}Y_{i} , \\
 \frac{1}{n}\sum_{i=1}^{n}1[A_{i}=1]Y_{i}
\end{array}
\right] ,
\end{equation*}
\end{lemma}
%%%%%%%% DIVIDER %%%%%%%%%%%%
\begin{proof}
The equalities follow from algebra and the convergences follow from the CMT and Lemma \ref{lem:A1and2_impliesold3}. In particular, the first equality in the second display has an $o_p(1)$ to allow for the possibility that ${{\mathbf{Z}_{n}^{\mathrm{2s}}}^{\prime }\mathbf{X}_{n}^{\mathrm{2s}}}/{n}$ is singular or $n_{AD}/n=( n_{A}/n) ( n_{D}/n) $. Both of these events occur with vanishing probability under our assumptions.
\end{proof}

%%%%%%%% DIVIDER %%%%%%%%%%%%
\begin{theorem}[2S limits]\label{thm:plim_2SR} 
Assume Assumptions \ref{ass:1} and \ref{ass:2}. Then,
\begin{align}
&\hat{\gamma}_{\mathrm{2s}}~\overset{p}{\to }~\frac{\sum_{s \in \mathcal{S}}\left[
\begin{array}{c}
\left(
\begin{array}{c}
( \sum_{\tilde{s}\in \mathcal{S}}p(\tilde{s}) \pi _{A}(\tilde{s}) \pi _{D( 1) }(\tilde{s}) ) ( 1-\pi _{A}( s ) ) \\
-( \sum_{\tilde{s}\in \mathcal{S}}p(\tilde{s}) \pi _{D( 0) }(\tilde{s}) ( 1-\pi _{A}(\tilde{s}) ) ) \pi _{A}(s)
\end{array}
\right) p(s) \pi _{D(0)}(s)E[Y(1)|AT,S=s] \\
-( \sum_{\tilde{s}\in \mathcal{S}}p(\tilde{s}) \pi _{D( 0) }(\tilde{s}) ( 1-\pi _{A}(\tilde{s}) ) ) p(s) \pi _{A}(s) ( \pi _{D(1)}(s )-\pi _{D(0)}(s)) E[Y(1)|C,S=s] \\
+( \sum_{\tilde{s}\in \mathcal{S}}p(\tilde{s}) \pi _{A}(\tilde{s}) \pi _{D( 1) }(\tilde{s}) ) p(s) ( 1-\pi _{A}(s) ) ( \pi _{D(1)}(s)-\pi _{D(0)}(s)) E[Y(0)|C,S=s] \\
+\left( 
\begin{array}{c}
( \sum_{\tilde{s}\in \mathcal{S}}p(\tilde{s}) \pi _{A}(\tilde{s}) \pi _{D( 1) }(\tilde{s}) ) ( 1-\pi _{A}( s ) ) \\
-( \sum_{\tilde{s}\in \mathcal{S}}p(\tilde{s}) \pi _{D( 0) }(\tilde{s}) ( 1-\pi _{A}(\tilde{s}) ) ) \pi _{A}(s)
\end{array}
\right) p(s) ( 1-\pi _{D(1)}(s)) E[Y(0)|NT,S=s]
\end{array}
\right] }{( 1-\sum_{\tilde{s}\in \mathcal{S}}p(\tilde{s})\pi _{A}(\tilde{s})) ( \sum_{\tilde{s}\in \mathcal{S}}p(\tilde{s})\pi _{A}(\tilde{s})\pi _{D(1)}(\tilde{s})) -( \sum_{\tilde{s}\in \mathcal{S}}p(\tilde{s})\pi _{A}(\tilde{s})) ( \sum_{\tilde{s}\in \mathcal{S} }p(\tilde{s})( 1-\pi _{A}(\tilde{s})) \pi _{D(0)}(\tilde{s})) }\notag\\
&\hat{\beta}_{\mathrm{2s}}~\overset{p}{\to }~\frac{\sum_{s\in \mathcal{S}}p(s) \left[ 
\begin{array}{c}
+[ \pi _{A}(s) -\sum_{\tilde{s}\in \mathcal{S}}p(\tilde{s}) \pi _{A}(\tilde{s}) ]  \pi _{D(0)}(s)E[Y(1)|AT,S=s]+ \\
+[ \pi _{A}(s) -\sum_{\tilde{s}\in \mathcal{S}}p(\tilde{s}) \pi _{A}(\tilde{s}) ] ( 1-\pi _{D( 1) }(s) ) E[Y(0)|NT,S=s] \\
+( 1-\sum_{\tilde{s}\in \mathcal{S}}p(\tilde{s}) \pi _{A}(\tilde{s}) ) \pi _{A}(s) ( \pi _{D(1)}(s)-\pi _{D(0)}(s)) E[Y(1)|C,S=s] \\
-( \sum_{\tilde{s}\in \mathcal{S}}p(\tilde{s}) \pi _{A}(\tilde{s}) ) ( 1-\pi _{A}(s) ) ( \pi _{D(1)}(s)-\pi _{D(0)}(s)) E[Y(0)|C,S=s]
\end{array}
\right] }{( 1-\sum_{\tilde{s}\in \mathcal{S}}p(\tilde{s})\pi _{A}(\tilde{s})) ( \sum_{\tilde{s}\in \mathcal{S}}p(\tilde{s})\pi _{A}(\tilde{s})\pi _{D(1)}(\tilde{s})) -( \sum_{\tilde{s}\in \mathcal{S}}p(\tilde{s})\pi _{A}(\tilde{s})) ( \sum_{\tilde{s}\in \mathcal{S} }p(\tilde{s})( 1-\pi _{A}(\tilde{s})) \pi _{D(0)}(\tilde{s})) }.
\label{eq:plim_2SR2}
\end{align}
If we add Assumption \ref{ass:3}(c),
\begin{equation}
\hat{\beta}_{\mathrm{2s}}~\overset{p}{\to }~\beta . \label{eq:plim_2SR3}
\end{equation}
\end{theorem}
%%%%%%%% DIVIDER %%%%%%%%%%%%
\begin{proof}
Throughout this proof, define
\begin{align}
&\Xi _{n} \equiv \frac{n_{AD}}{n}-\frac{n_{A}}{n}\frac{n_{D}}{n} \notag \\
&\overset{p}{\to }\Xi \equiv \left( 1-\sum_{\tilde{s}\in \mathcal{S}}p(\tilde{s})\pi _{A}(\tilde{s})\right) \left( \sum_{\tilde{s}\in \mathcal{S}}p(\tilde{s})\pi _{A}(\tilde{s})\pi _{D(1)}(\tilde{s})\right) -\left( \sum_{\tilde{s}\in \mathcal{S}}p(\tilde{s})\pi _{A}(\tilde{s})\right) \left( \sum_{\tilde{s}\in \mathcal{S} }p(\tilde{s})( 1-\pi _{A}(\tilde{s})) \pi _{D(0)}(\tilde{s})\right) , \label{eq:plim_2SR4}
\end{align}
where the convergence holds by Assumptions \ref{ass:1} and \ref{ass:2}(b), and Lemma \ref{lem:A1and2_impliesold3}.

To show the first line of \eqref{eq:plim_2SR2}, consider the following derivation.
\begin{align}
& \hat{\gamma}_{\mathrm{2s}}=\frac{1}{n_{AD}/n-(n_{A}/n)(n_{D}/n)}\left[ \frac{n_{AD}}{n}\frac{1}{n}\sum_{i=1}^{n}Y_{i}-\frac{n_{D}}{n}\frac{1}{n} \sum_{i=1}^{n}1[A_{i}=1]Y_{i}\right]+o_{p}(1)\notag \\
& \overset{(1)}{=}\frac{1}{\Xi _{n}}\sum_{s\in \mathcal{S}}\left[
\begin{array}{c}
\left( 
\begin{array}{c}
\sum_{\tilde{s}\in \mathcal{S}}\frac{n(\tilde{s})}{n}\frac{n_{A}(\tilde{s})}{n(\tilde{s})}\frac{n_{AD}(\tilde{s})}{n_{A}(\tilde{s})}- \\ 
\sum_{\tilde{s}\in   \mathcal{S}}\frac{n(\tilde{s})}{n}[(\frac{n_{D}(\tilde{s})-n_{AD}(\tilde{s})}{n(\tilde{s})-n_{A}(\tilde{s})})(1-\frac{n_{A}(\tilde{s})}{n(\tilde{s})})+\frac{n_{AD}(\tilde{s})}{n_{A}(\tilde{s})}\frac{n_{A}(\tilde{s})}{n(\tilde{s})}]
\end{array}
\right) \frac{1}{\sqrt{n}}R_{n,1}(1,1,s) \\ 
+\left( 
\begin{array}{c}
\sum_{\tilde{s}\in  \mathcal{S}}\frac{n(\tilde{s})}{n}\frac{n_{A}(\tilde{s})}{n(\tilde{s})}\frac{n_{AD}(\tilde{s})}{n_{A}(\tilde{s})}- \\ 
\sum_{\tilde{s}\in  \mathcal{S}}\frac{n(\tilde{s})}{n}[(\frac{n_{D}(\tilde{s})-n_{AD}(\tilde{s})}{n(\tilde{s})-n_{A}(\tilde{s})})(1-\frac{n_{A}(\tilde{s})}{n(\tilde{s})})+\frac{n_{AD}(\tilde{s})}{n_{A}(\tilde{s})}\frac{n_{A}(\tilde{s})}{n(\tilde{s})}]
\end{array}
\right) \frac{1}{\sqrt{n}}R_{n,1}(0,1,s) \\ 
+( \sum_{\tilde{s}\in \mathcal{S}}\frac{n(\tilde{s}) }{n}\frac{n_{A}(\tilde{s}) }{n(\tilde{s}) }\frac{n_{AD}(\tilde{s}) }{n_{A}(\tilde{s}) }) \frac{1}{\sqrt{n}}R_{n,1}(1,0,s)+( \sum_{\tilde{s}\in \mathcal{S}}\frac{n(\tilde{s}) }{n}\frac{n_{A}(\tilde{s}) }{n(\tilde{s}) }\frac{n_{AD}(\tilde{s}) }{n_{A}(\tilde{s}) }) \frac{1 }{\sqrt{n}}R_{n,1}(0,0,s) \\
+\left[
\begin{array}{c}
\left( 
\begin{array}{c}
\sum_{\tilde{s}\in \mathcal{S}}\frac{n(\tilde{s}) }{n}\frac{n_{A}(\tilde{s}) }{n(\tilde{s}) }\frac{n_{AD}(\tilde{s}) }{n_{A}(\tilde{s}) } \\
-\sum_{\tilde{s}\in \mathcal{S}}\frac{n(\tilde{s}) }{n}[ ( \frac{ n_{D}(\tilde{s}) -n_{AD}(\tilde{s})}{n(\tilde{s})-n_{A}(\tilde{s})}) ( 1-\frac{ n_{A}(\tilde{s})}{n(\tilde{s}) }) +\frac{n_{AD}(\tilde{s})}{n_{A}(\tilde{s})}\frac{ n_{A}(\tilde{s}) }{n(\tilde{s}) }]
\end{array}
\right) 
\\
\times\frac{n_{A}(s) }{n(s)}\frac{ n_{AD}(s) }{n_{A}(s) }\tfrac{\pi _{D(0)}(s)}{\pi _{D(1)}(s)} \\
+( \sum_{\tilde{s}\in \mathcal{S}}\frac{n(\tilde{s}) }{n}\frac{n_{A}(\tilde{s}) }{n(\tilde{s}) }\frac{n_{AD}(\tilde{s}) }{n_{A}(\tilde{s}) }) ( 1-\frac{n_{A}(s) }{n(s)}) \frac{n_{D}(s) -n_{AD}(s) }{n(s) -n_{A}(s) }
\end{array}
\right] \times \\
\frac{n(s)}{n}E[Y(1)|AT,S=s] \\
+\left( 
\begin{array}{c}
\sum_{\tilde{s}\in \mathcal{S}}\frac{n(\tilde{s}) }{n}\frac{n_{A}(\tilde{s}) }{n(\tilde{s}) }\frac{n_{AD}(\tilde{s}) }{n_{A}(\tilde{s}) } \\
-\sum_{\tilde{s}\in \mathcal{S}}\frac{n(\tilde{s}) }{n}[ ( \frac{ n_{D}(\tilde{s}) -n_{AD}(\tilde{s})}{n(\tilde{s})-n_{A}(\tilde{s})}) ( 1-\frac{ n_{A}(\tilde{s})}{n(\tilde{s}) }) +\frac{n_{AD}(\tilde{s})}{n_{A}(\tilde{s})}\frac{ n_{A}(\tilde{s}) }{n(\tilde{s}) }]
\end{array}
\right)\times\\
\frac{n(s)}{n}\frac{n_{A}( s ) }{n(s)}\frac{n_{AD}(s) }{ n_{A}(s) }\tfrac{\pi _{D(1)}(s)-\pi _{D(0)}(s)}{\pi _{D(1)}(s)}E[Y(1)|C,S=s]+ \\
( \sum_{\tilde{s}\in \mathcal{S}}\frac{n(\tilde{s}) }{n}\frac{n_{A}(\tilde{s}) }{n(\tilde{s}) }\frac{n_{AD}(\tilde{s}) }{n_{A}(\tilde{s}) }) \frac{n(s)}{n}( 1-\frac{ n_{A}(s) }{n(s)}) ( 1- \frac{( n_{D}(s) -n_{AD}(s) ) }{n(s) -n_{A}(s) })
\times\\
\tfrac{\pi _{D(1)}(s)-\pi _{D(0)}(s)}{1-\pi _{D(0)}(s )}E[Y(0)|C,S=s] \\
+\left[
\begin{array}{c}
\left( 
\begin{array}{c}
\sum_{\tilde{s}\in \mathcal{S}}\frac{n(\tilde{s}) }{n}\frac{n_{A}(\tilde{s}) }{n(\tilde{s}) }\frac{n_{AD}(\tilde{s}) }{n_{A}(\tilde{s}) } \\
-\sum_{\tilde{s}\in \mathcal{S}}\frac{n(\tilde{s}) }{n}[ ( \frac{ n_{D}(\tilde{s}) -n_{AD}(\tilde{s})}{n(\tilde{s})-n_{A}(\tilde{s})}) ( 1-\frac{ n_{A}(\tilde{s})}{n(\tilde{s}) }) +\frac{n_{AD}(\tilde{s})}{n_{A}(\tilde{s})}\frac{ n_{A}(\tilde{s}) }{n(\tilde{s}) }]
\end{array}
\right) \frac{n_{A}(s) }{n(s)} ( 1-\frac{n_{AD}(s) }{n_{A}(s) } ) \\
+( \sum_{\tilde{s}\in \mathcal{S}}\frac{n(\tilde{s}) }{n}\frac{n_{A}(\tilde{s}) }{n(\tilde{s}) }\frac{n_{AD}(\tilde{s}) }{n_{A}(\tilde{s}) }) ( 1-\frac{n_{A}(s) }{n(s)}) ( 1-\frac{( n_{D}(s) -n_{AD}(s) ) }{n(s) -n_{A}(s) }) \tfrac{1-\pi _{D(1)}(s)}{ 1-\pi _{D(0)}(s)}
\end{array}
\right]\\
\times\frac{n(s)}{n}E[Y(0)|NT,S=s]
\end{array}
\right]+o_{p}(1)  \notag\\
& \overset{(2)}{=}\frac{1}{\Xi }\sum_{s\in \mathcal{S}}\left[
\begin{array}{c}
[ ( \sum_{\tilde{s}\in \mathcal{S}}p(\tilde{s}) \pi _{A}(\tilde{s}) \pi _{D( 1) }(\tilde{s}) ) ( 1-\pi _{A}(s) ) -( \sum_{\tilde{s}\in \mathcal{S}}p(\tilde{s}) \pi _{D( 0) }(\tilde{s}) ( 1-\pi _{A}(\tilde{s}) ) ) \pi _{A}(s) ] \times\\
p(s) \pi _{D(0)}(s)E[Y(1)|AT,S=s] \\
-( \sum_{\tilde{s}\in \mathcal{S}}p(\tilde{s}) \pi _{D( 0) }(\tilde{s}) ( 1-\pi _{A}(\tilde{s}) ) ) p(s) \pi _{A}(s) ( \pi _{D(1)}(s )-\pi _{D(0)}(s)) E[Y(1)|C,S=s] \\
+( \sum_{\tilde{s}\in \mathcal{S}}p(\tilde{s}) \pi _{A}(\tilde{s}) \pi _{D( 1) }(\tilde{s}) ) p(s) ( 1-\pi _{A}(s) ) ( \pi _{D(1)}(s)-\pi _{D(0)}(s))E[Y(0)|C,S=s] \\
\left[ ( \sum_{\tilde{s}\in \mathcal{S}}p(\tilde{s}) \pi _{A}(\tilde{s}) \pi _{D( 1) }(\tilde{s}) ) ( 1-\pi _{A}(s) ) -( \sum_{\tilde{s}\in \mathcal{S}}p(\tilde{s}) \pi _{D( 0) }(\tilde{s}) ( 1-\pi _{A}(\tilde{s}) ) ) \pi _{A}(s) \right] \times\\
p(s) ( 1-\pi _{D(1)}(s)) E[Y(0)|NT,S=s]
\end{array}
\right] +o_{p}(1), \label{eq:plim_2SR5}
\end{align}
where (1) holds by $Y_{i}=Y_{i}(D_{i})$, \eqref{eq:defnY_pre}, \eqref{eq:defnY}, \eqref{eq:Rn_defn}, and (2) holds by Assumptions \ref{ass:1} and \ref{ass:2}(b), Lemma \ref{lem:A1and2_impliesold3} and \ref{lem:AsyDist2}, and \eqref{eq:plim_2SR4}.

To show the second line of \eqref{eq:plim_2SR2}, consider the following derivation.
\begin{align}
\hat{\beta}_{\mathrm{2s}} &=\frac{1}{n_{AD}/n-(n_{A}/n)(n_{D}/n)}\left[ - \frac{n_{A}}{n}\frac{1}{n}\sum_{i=1}^{n}Y_{i}+\frac{1}{n}\sum_{i=1}^{n}1 [A_{i}=1]Y_{i}\right] +o_{p}(1) \notag\\
&\overset{(1)}{=}\frac{1}{\Xi _{n}}\sum_{s\in \mathcal{S}}\left[
\begin{array}{c}
( 1-\sum_{\tilde{s}\in \mathcal{S}}\frac{n(\tilde{s}) }{n} \frac{n_{A}(\tilde{s}) }{n(\tilde{s}) }) \frac{ 1}{\sqrt{n}}R_{n,1}(1,1,s)+( 1-\sum_{\tilde{s}\in \mathcal{S}}\frac{ n(\tilde{s}) }{n}\frac{n_{A}(\tilde{s}) }{n(\tilde{s}) }) \frac{1}{\sqrt{n}}R_{n,1}(0,1,s) \\
-( \sum_{\tilde{s}\in \mathcal{S}}\frac{n(\tilde{s}) }{n} \frac{n_{A}(\tilde{s}) }{n(\tilde{s}) }) \frac{ 1}{\sqrt{n}}R_{n,1}(1,0,s)-( \sum_{\tilde{s}\in \mathcal{S}}\frac{n(\tilde{s}) }{n}\frac{n_{A}(\tilde{s}) }{n( \tilde{s} ) }) \frac{1}{\sqrt{n}}R_{n,1}(0,0,s) \\
+\left[ 
\begin{array}{c}
( 1-\sum_{\tilde{s}\in \mathcal{S}}\frac{n(\tilde{s}) }{n} \frac{n_{A}(\tilde{s}) }{n(\tilde{s}) }) \frac{ n(s) }{n}\frac{n_{A}(s) }{n(s)}\frac{ n_{AD}(s) }{n_{A}(s) }\tfrac{\pi _{D(0)}(s)}{\pi _{D(1)}(s)} \\
-( \sum_{\tilde{s}\in \mathcal{S}}\frac{n(\tilde{s}) }{n} \frac{n_{A}(\tilde{s}) }{n(\tilde{s}) }) \frac{ n(s) }{n}( 1-\frac{n_{A}(s) }{n(s)} ) \frac{n_{D}(s) -n_{AD}(s) }{n(s) -n_{A}(s) }
\end{array}
\right] E[Y(1)|AT,S=s] +\\
\left[ 
\begin{array}{c}
( 1-\sum_{\tilde{s}\in \mathcal{S}}\frac{n(\tilde{s}) }{n} \frac{n_{A}(\tilde{s}) }{n(\tilde{s}) }) \frac{ n(s) }{n}\frac{n_{A}(s) }{n(s)}( 1- \frac{n_{AD}(s) }{n_{A}(s) }) \\
-( \sum_{\tilde{s}\in \mathcal{S}}\frac{n(\tilde{s}) }{n} \frac{n_{A}(\tilde{s}) }{n(\tilde{s}) }) \frac{ n(s) }{n}( 1-\frac{n_{A}(s) }{n(s)} ) ( 1-\frac{n_{D}(s) -n_{AD}(s) }{n(s) -n_{A}(s) }) \tfrac{1-\pi _{D(1)}(s)}{1-\pi _{D(0)}(s)}
\end{array}
\right] E[Y(0)|NT,S=s] \\
+( 1-\sum_{\tilde{s}\in \mathcal{S}}\frac{n(\tilde{s}) }{n} \frac{n_{A}(\tilde{s}) }{n(\tilde{s}) }) \frac{ n(s) }{n}\frac{n_{A}(s) }{n(s)}\frac{ n_{AD}(s) }{n_{A}(s) }\tfrac{\pi _{D(1)}(s)-\pi _{D(0)}(s)}{\pi _{D(1)}(s)}E[Y(1)|C,S=s] \\
-( \sum_{\tilde{s}\in \mathcal{S}}\frac{n(\tilde{s}) }{n} \frac{n_{A}(\tilde{s}) }{n(\tilde{s}) }) \frac{ n(s) }{n}( 1-\frac{n_{A}(s) }{n(s)} ) ( 1-\frac{n_{D}(s) -n_{AD}(s) }{n(s) -n_{A}(s) }) \tfrac{\pi _{D(1)}(s)-\pi _{D(0)}(s) }{1-\pi _{D(0)}(s)}E[Y(0)|C,S=s]
\end{array}
\right]+o_{p}(1)   \notag\\
&\overset{(2)}{=}\frac{1}{\Xi }\sum_{s\in \mathcal{S}}p(s)\left[
\begin{array}{c}
+[ \pi _{A}(s) -\sum_{\tilde{s}\in \mathcal{S}}p(\tilde{s}) \pi _{A}(\tilde{s}) ]  \pi _{D(0)}(s)E[Y(1)|AT,S=s]+ \\
+[ \pi _{A}(s) -\sum_{\tilde{s}\in \mathcal{S}}p(\tilde{s}) \pi _{A}(\tilde{s}) ]  ( 1-\pi _{D( 1) }(s) ) E[Y(0)|NT,S=s] \\
+( 1-\sum_{\tilde{s}\in \mathcal{S}}p(\tilde{s}) \pi _{A}(\tilde{s}) )  \pi _{A}(s) ( \pi _{D(1)}(s)-\pi _{D(0)}(s)) E[Y(1)|C,S=s] \\
-( \sum_{\tilde{s}\in \mathcal{S}}p(\tilde{s}) \pi _{A}(\tilde{s}) )  ( 1-\pi _{A}(s) ) ( \pi _{D(1)}(s)-\pi _{D(0)}(s)) E[Y(0)|C,S=s]
\end{array}
\right] +o_{p}( 1) , \label{eq:plim_2SR6}
\end{align}
where (1) holds by $Y_{i}=Y_{i}(D_{i})$, \eqref{eq:defnY_pre}, \eqref{eq:defnY}, \eqref{eq:Rn_defn}, and (2) holds by Assumption \ref{ass:1}, \ref{ass:2}(b), Lemma \ref{lem:A1and2_impliesold3} and \ref{lem:AsyDist2}, and \eqref{eq:plim_2SR4}.

Finally, \eqref{eq:plim_2SR3} holds by the following derivation. 
\begin{align*}
\hat{\beta}_{\mathrm{2s}}~
&~\overset{(1)}{=}~\frac{\sum_{s\in \mathcal{S}}p( s) ( \pi _{D(1)}(s)-\pi _{D(0)}(s)) E[Y(1)-Y(0)|C,S=s]}{\sum_{\tilde{s}\in \mathcal{S}}p( \tilde{s}) ( \pi _{D(1)}(s)-\pi _{D(0)}( \tilde{s})) } +o_p(1)\\
&~\overset{(2)}{=}~\sum_{s\in \mathcal{S}}P(S=s|C)E[Y(1)-Y(0)|C,S=s]+o_p(1)~=~\beta +o_p(1),
\end{align*}
where (1) holds by \eqref{eq:plim_2SR4}, \eqref{eq:plim_2SR6}, and Assumption \ref{ass:3}(c), and (2) holds by \eqref{eq:pi_defns}.
\end{proof}
%%%%%%%% DIVIDER %%%%%%%%%%%%

\begin{proof}[Proof of Theorem \ref{thm:AsyDist_2SR}.]
As a preliminary result, consider the following derivation.
\begin{align*}
\sqrt{n}( \Xi _{n}-\Xi ) &\overset{(1)}{=}\sqrt{n}\left[ 
\begin{array}{c}
\frac{n_{AD}}{n}-\frac{n_{A}}{n}\frac{n_{D}}{n}-\sum_{s\in \mathcal{S} }p(s)\pi _{A}\pi _{D(1)}(s)+ \\
\pi _{A}\sum_{s\in \mathcal{S}}p(s)[ ( 1-\pi _{A}) \pi _{D(0)}(s)+\pi _{A}\pi _{D(1)}(s)]
\end{array}
\right]  \\
%&=\left\{ 
%\begin{array}{c}
%\sum_{s\in \mathcal{S}}\frac{n_{A}( s) }{n( s) }\frac{ n_{AD}( s) }{n_{A}( s) }R_{n,4}( s) +\sum_{s\in \mathcal{S}}p( s) \frac{n_{AD}( s) }{ n_{A}( s) }R_{n,3}( s) \\
%-( \sum_{\tilde{s}\in \mathcal{S}}\frac{n( \tilde{s}) }{n} [ ( \frac{n_{D}( \tilde{s}) -n_{AD}(\tilde{s})}{n( \tilde{s})-n_{A}(\tilde{s})}) ( 1-\frac{n_{A}(\tilde{s})}{n( \tilde{s}) }) +\frac{n_{AD}(\tilde{s})}{n_{A}(\tilde{s})}\frac{ n_{A}( \tilde{s}) }{n( \tilde{s}) }] ) \sum_{s\in \mathcal{S}}\frac{n_{A}( s) }{n( s) } R_{n,4}( s) \\
%-\pi _{A}\sum_{s\in \mathcal{S}}[ ( \frac{n_{D}( s) -n_{AD}(s)}{n(s)-n_{A}(s)}) ( 1-\frac{n_{A}(s)}{n( s) } ) +\frac{n_{AD}(s)}{n_{A}(s)}\frac{n_{A}( s) }{n( s) }] R_{n,4}( s) \\
%-( \sum_{\tilde{s}\in \mathcal{S}}\frac{n( \tilde{s}) }{n} [ ( \frac{n_{D}( \tilde{s}) -n_{AD}(\tilde{s})}{n( \tilde{s})-n_{A}(\tilde{s})}) ( 1-\frac{n_{A}(\tilde{s})}{n( \tilde{s}) }) +\frac{n_{AD}(\tilde{s})}{n_{A}(\tilde{s})}\frac{ n_{A}( \tilde{s}) }{n( \tilde{s}) }] ) \sum_{s\in \mathcal{S}}p( s) R_{n,3}( s) \\
%-\pi _{A}\sum_{s\in \mathcal{S}}p( s) [ \pi _{D(1)}(s)-\pi _{D(0)}(s)] R_{n,3}( s) \\
%\pi _{A}\sum_{s\in \mathcal{S}}p( s) ( 1-\frac{n_{A}( s) }{n( s) }) R_{n,2}( 1,s) -\pi _{A}\sum_{s\in \mathcal{S}}p( s) ( 1-\frac{n_{A}(s)}{n( s) }) R_{n,2}( 2,s)
%\end{array}
%\right\} \\
&\overset{(2)}{=}\left[ 
\begin{array}{c}
\sum_{s\in \mathcal{S}}\pi _{A}( 1-\pi _{A}) [ \pi _{D( 1) }( s) -\pi _{D(0)}( s) ] R_{n,4}( s) \\
+\sum_{s\in \mathcal{S}}p( s) \left[
\begin{array}{c}
( 1-\pi _{A}) \pi _{D(1)}(s)-\pi _{A}\sum_{\tilde{s}\in \mathcal{S }}p( \tilde{s}) \pi _{D(1)}(\tilde{s}) \\
+\pi _{A}\pi _{D(0)}(s)-( 1-\pi _{A}) \sum_{\tilde{s}\in \mathcal{ S}}p( \tilde{s}) \pi _{D(0)}(\tilde{s})
\end{array}
\right] R_{n,3}( s) \\
+\pi _{A}( 1-\pi _{A}) \sum_{s\in \mathcal{S}}p( s) R_{n,2}( 1,s) -\pi _{A}( 1-\pi _{A}) \sum_{s\in \mathcal{S}}p( s) R_{n,2}( 2,s)
\end{array}
\right] +o_{p}( 1) , 
\end{align*}
where (1) holds by definitions in \eqref{eq:plim_2SR4} and Assumption \ref{ass:3}(c), and (2) holds by Assumptions \ref{ass:1} and \ref{ass:2}(b), and Lemmas \ref{lem:A1and2_impliesold3} and \ref{lem:AsyDist}. In particular, Lemma \ref{lem:AsyDist} implies that $R_{n}=O_{p}( 1) $ and $\sum_{s\in \mathcal{S} }R_{n,4}( s) =0$.

Consider the following derivation.
\begin{align}
&\sqrt{n}( \hat{\beta}_{\mathrm{2s}}-\beta ) \overset{(1)}{=}\left[ 
\begin{array}{c}
\frac{1}{\Xi _{n}}\sum_{s\in \mathcal{S}}( 1-\sum_{\tilde{s}\in \mathcal{S}}\frac{n(\tilde{s})}{n}\frac{n_{A}(\tilde{s})}{n(\tilde{s})} ) R_{n,1}(1,1,s)+\frac{1}{\Xi _{n}}\sum_{s\in \mathcal{S}}( 1-\sum_{\tilde{s}\in \mathcal{S}}\frac{n(\tilde{s})}{n}\frac{n_{A}(\tilde{s}) }{n(\tilde{s})}) R_{n,1}(0,1,s) \\
-\frac{1}{\Xi _{n}}\sum_{s\in \mathcal{S}}( \sum_{\tilde{s}\in \mathcal{ S}}\frac{n(\tilde{s})}{n}\frac{n_{A}(\tilde{s})}{n(\tilde{s})}) R_{n,1}(1,0,s)-\frac{1}{\Xi _{n}}\sum_{s\in \mathcal{S}}( \sum_{\tilde{s }\in \mathcal{S}}\frac{n(\tilde{s})}{n}\frac{n_{A}(\tilde{s})}{n(\tilde{s})} ) R_{n,1}(0,0,s) \\
+\frac{1}{\Xi _{n}}\sum_{s\in \mathcal{S}}\zeta _{n,2}(s)+\frac{1}{\Xi _{n}} \sum_{s\in \mathcal{S}}\zeta _{n,3}( s) +\zeta _{n,4}
\end{array}
\right] + o_{p}(1), \label{eq:beta_2sr}
\end{align}
where (1) holds by the definitions of $\zeta _{n,2}(s)$, $\zeta _{n,3}( s) $, and $\zeta _{n,4}$ that appear below in \eqref{eq:eta_2s}, \eqref{eq:eta_3s}, and \eqref{eq:eta_4}.
First,
\begin{align}
\zeta _{n,2}(s) &\equiv \sqrt{n}\left[
\begin{array}{c}
(1-\sum_{\tilde{s}\in \mathcal{S}}\frac{n(\tilde{s})}{n}\frac{n_{A}(\tilde{s} )}{n(\tilde{s})})\frac{n(s)}{n}\frac{n_{A}(s)}{n(s)}\frac{n_{AD}(s)}{n_{A}(s) }\tfrac{\pi _{D(0)}(s)}{\pi _{D(1)}(s)} \\
-(\sum_{\tilde{s}\in \mathcal{S}}\frac{n(\tilde{s})}{n}\frac{n_{A}(\tilde{s}) }{n(\tilde{s})})\frac{n(s)}{n}(1-\frac{n_{A}(s)}{n(s)})\frac{ n_{D}(s)-n_{AD}(s)}{n(s)-n_{A}(s)}
\end{array}
\right] E[Y(1)|AT,S=s] \notag \\
&\overset{(1)}{=}\left[ 
\begin{array}{c}
\frac{n(s)}{n}[ -\frac{n_{A}(s)}{n(s)}\frac{n_{AD}(s)}{n_{A}(s)}\tfrac{ \pi _{D(0)}(s)}{\pi _{D(1)}(s)}-(1-\frac{n_{A}(s)}{n(s)})\frac{ n_{D}(s)-n_{AD}(s)}{n(s)-n_{A}(s)}] (\sum_{\tilde{s}\in \mathcal{S}} \frac{n_{A}(\tilde{s})}{n(\tilde{s})}R_{n,4}( \tilde{s}) ) \\
+[ (1-\pi _{A})\frac{n_{A}(s)}{n(s)}\frac{n_{AD}(s)}{n_{A}(s)}\tfrac{ \pi _{D(0)}(s)}{\pi _{D(1)}(s)}-\pi _{A}(1-\frac{n_{A}(s)}{n(s)})\frac{ n_{D}(s)-n_{AD}(s)}{n(s)-n_{A}(s)}] R_{n,4}( s) \\
-\frac{n(s)}{n}[ \frac{n_{A}(s)}{n(s)}\frac{n_{AD}(s)}{n_{A}(s)}\tfrac{ \pi _{D(0)}(s)}{\pi _{D(1)}(s)}+(1-\frac{n_{A}(s)}{n(s)})\frac{ n_{D}(s)-n_{AD}(s)}{n(s)-n_{A}(s)}] (\sum_{\tilde{s}\in \mathcal{S} }p( \tilde{s}) R_{n,3}( \tilde{s}) ) \\
+p( s) [ (1-\pi _{A})\frac{n_{AD}(s)}{n_{A}(s)}\tfrac{\pi _{D(0)}(s)}{\pi _{D(1)}(s)}+\pi _{A}\frac{n_{D}(s)-n_{AD}(s)}{n(s)-n_{A}(s)} ] R_{n,3}( s) \\
+\pi _{A}(1-\pi _{A})p( s) \tfrac{\pi _{D(0)}(s)}{\pi _{D(1)}(s)} R_{n,2}( 1,s) -\pi _{A}(1-\pi _{A})p( s) R_{n,2}( 2,s)
\end{array}
\right] E[Y(1)|AT,S=s]  \notag \\
&\overset{(2)}{=}p( s) \left[ 
\begin{array}{c}
-\pi _{D( 0) }( s) E[Y(1)|AT,S=s](\sum_{\tilde{s}\in \mathcal{S}}p( \tilde{s}) R_{n,3}( \tilde{s}) ) \\
+\pi _{D(0)}(s)E[Y(1)|AT,S=s]R_{n,3}( s) + \\
\pi _{A}(1-\pi _{A})\tfrac{\pi _{D(0)}(s)}{\pi _{D(1)}(s)} E[Y(1)|AT,S=s]R_{n,2}( 1,s) \\
-\pi _{A}(1-\pi _{A})E[Y(1)|AT,S=s]R_{n,2}( 2,s)
\end{array}
\right] +o_{p}( 1) ,\label{eq:eta_2s}
\end{align}
where (1) holds by the definitions in \eqref{eq:plim_2SR4}, and (2) holds by Assumptions \ref{ass:1}, \ref{ass:2}(b), and Lemmas \ref{lem:A1and2_impliesold3} and \ref{lem:AsyDist}. This implies that
\begin{align}
\frac{1}{\Xi _{n}}\sum_{s\in \mathcal{S}}\zeta _{n,2}(s) &=\frac{1}{\Xi }\sum_{s\in \mathcal{S}}p( s) \left[
\begin{array}{c}
\left[ 
\begin{array}{c}
\pi _{D(0)}(s)E[Y(1)|AT,S=s] \\ 
-( \sum_{\tilde{s}\in \mathcal{S}}p( \tilde{s}) \pi _{D( 0) }( \tilde{s}) E[Y(1)|AT,S=\tilde{s}])
\end{array}
\right] R_{n,3}( s) \\
+\pi _{A}(1-\pi _{A})\tfrac{\pi _{D(0)}(s)}{\pi _{D(1)}(s)} E[Y(1)|AT,S=s]R_{n,2}( 1,s) \\
-\pi _{A}(1-\pi _{A})E[Y(1)|AT,S=s]R_{n,2}( 2,s)
\end{array}
\right] +o_{p}( 1).\label{eq:eta_2}
\end{align}
Second,
\begin{align}
\zeta _{n,3}( s) &\equiv \sqrt{n}\left[
\begin{array}{c}
(1-\sum_{\tilde{s}\in \mathcal{S}}\frac{n(\tilde{s})}{n}\frac{n_{A}(\tilde{s} )}{n(\tilde{s})})\frac{n(s)}{n}\frac{n_{A}(s)}{n(s)}(1-\frac{n_{AD}(s)}{ n_{A}(s)}) \\
-(\sum_{\tilde{s}\in \mathcal{S}}\frac{n(\tilde{s})}{n}\frac{n_{A}(\tilde{s}) }{n(\tilde{s})})\frac{n(s)}{n}(1-\frac{n_{A}(s)}{n(s)})(1-\frac{ n_{D}(s)-n_{AD}(s)}{n(s)-n_{A}(s)})\tfrac{1-\pi _{D(1)}(s)}{1-\pi _{D(0)}(s)}
\end{array}
\right]E[Y(0)|NT,S=s] \notag \\
&\overset{(1)}{=}\left[ 
\begin{array}{c}
-\frac{n(s)}{n}[ \frac{n_{A}(s)}{n(s)}(1-\frac{n_{AD}(s)}{n_{A}(s)})+(1- \frac{n_{A}(s)}{n(s)})(1-\frac{n_{D}(s)-n_{AD}(s)}{n(s)-n_{A}(s)})\tfrac{ 1-\pi _{D(1)}(s)}{1-\pi _{D(0)}(s)}] (\sum_{\tilde{s}\in \mathcal{S} }R_{n,4}( \tilde{s}) \frac{n_{A}(\tilde{s})}{n(\tilde{s})}) \\
+[ (1-\pi _{A})\frac{n_{A}(s)}{n(s)}(1-\frac{n_{AD}(s)}{n_{A}(s)})-\pi _{A}(1-\frac{n_{A}(s)}{n(s)})(1-\frac{n_{D}(s)-n_{AD}(s)}{n(s)-n_{A}(s)}) \tfrac{1-\pi _{D(1)}(s)}{1-\pi _{D(0)}(s)}] R_{n,4}( s) \\
-\frac{n(s)}{n}[ \frac{n_{A}(s)}{n(s)}(1-\frac{n_{AD}(s)}{n_{A}(s)})+(1- \frac{n_{A}(s)}{n(s)})(1-\frac{n_{D}(s)-n_{AD}(s)}{n(s)-n_{A}(s)})\tfrac{ 1-\pi _{D(1)}(s)}{1-\pi _{D(0)}(s)}] (\sum_{\tilde{s}\in \mathcal{S} }p( \tilde{s}) R_{n,3}( \tilde{s}) ) \\
p( s) [ (1-\pi _{A})(1-\frac{n_{AD}(s)}{n_{A}(s)})+\pi _{A}(1-\frac{n_{D}(s)-n_{AD}(s)}{n(s)-n_{A}(s)})\tfrac{1-\pi _{D(1)}(s)}{ 1-\pi _{D(0)}(s)}] R_{n,3}( s) \\
-p( s) (1-\pi _{A})\pi _{A}R_{n,2}( 1,s) +p( s) \pi _{A}(1-\pi _{A})\tfrac{1-\pi _{D(1)}(s)}{1-\pi _{D(0)}(s)} R_{n,2}( 2,s)
\end{array}
\right] \times\notag\\
&E[Y(0)|NT,S=s] \notag \\
&\overset{(2)}{=}p( s) \left[
\begin{array}{c}
+(1-\pi _{D( 1) }( s) )E[Y(0)|NT,S=s]R_{n,3}( s) \\
-( 1-\pi _{D(1)}(s)) E[Y(0)|NT,S=s](\sum_{\tilde{s}\in \mathcal{S} }p( \tilde{s}) R_{n,3}( \tilde{s}) ) \\
-\pi _{A}(1-\pi _{A})E[Y(0)|NT,S=s]R_{n,2}( 1,s) \\
+\pi _{A}(1-\pi _{A})\tfrac{1-\pi _{D(1)}(s)}{1-\pi _{D(0)}(s)} E[Y(0)|NT,S=s]R_{n,2}( 2,s)
\end{array}
\right] +o_{p}( 1) ,\label{eq:eta_3s}
\end{align}
where (1) holds by the definitions in \eqref{eq:plim_2SR4}, and (2) holds by Assumptions \ref{ass:1}, \ref{ass:2}(b), and Lemmas \ref{lem:A1and2_impliesold3} and \ref{lem:AsyDist}. This implies that
\begin{align}
 \frac{1}{\Xi _{n}}\sum_{s\in \mathcal{S}}\zeta _{n,3}( s) &=\frac{1}{\Xi }\sum_{s\in \mathcal{S}}p( s) \left[
\begin{array}{c}
\left[ 
\begin{array}{c}
(1-\pi _{D( 1) }( s) )E[Y(0)|NT,S=s] \\
-( \sum_{\tilde{s}\in \mathcal{S}}p( \tilde{s}) ( 1-\pi _{D(1)}(\tilde{s})) E[Y(0)|NT,S=\tilde{s}])
\end{array}
\right] R_{n,3}( s) \\
-\pi _{A}(1-\pi _{A})E[Y(0)|NT,S=s]R_{n,2}( 1,s) \\
+\pi _{A}(1-\pi _{A})\tfrac{1-\pi _{D(1)}(s)}{1-\pi _{D(0)}(s)} E[Y(0)|NT,S=s]R_{n,2}( 2,s)
\end{array}
\right] +o_{p}( 1).\label{eq:eta_3} 
\end{align}
Third,
\begin{align}
&\zeta _{n,4} \equiv \frac{1}{\Xi _{n}}\sqrt{n}\left[
\begin{array}{c}
\sum_{s\in \mathcal{S}}(1-\sum_{\tilde{s}\in \mathcal{S}}\frac{n(\tilde{s})}{ n}\frac{n_{A}(\tilde{s})}{n(\tilde{s})})\frac{n(s)}{n}\frac{n_{A}(s)}{n(s)} \frac{n_{AD}(s)}{n_{A}(s)}\tfrac{\pi _{D(1)}(s)-\pi _{D(0)}(s)}{\pi _{D(1)}(s)}E[Y(1)|C,S=s] \\
-\sum_{s\in \mathcal{S}}\pi _{A}(1-\pi _{A}) p( s) ( \pi _{D(1)}(s)-\pi _{D(0)}(s)) E[Y(1)|C,S=s] \\
-\sum_{s\in \mathcal{S}}(\sum_{\tilde{s}\in \mathcal{S}}\frac{n(\tilde{s})}{n }\frac{n_{A}(\tilde{s})}{n(\tilde{s})})\frac{n(s)}{n}(1-\frac{n_{A}(s)}{n(s)} )(1-\frac{n_{D}(s)-n_{AD}(s)}{n(s)-n_{A}(s)})\tfrac{\pi _{D(1)}(s)-\pi _{D(0)}(s)}{1-\pi _{D(0)}(s)}E[Y(0)|C,S=s] \\
+\sum_{s\in \mathcal{S}}\pi _{A}(1-\pi _{A}) p( s) ( \pi _{D(1)}(s)-\pi _{D(0)}(s)) E[Y(0)|C,S=s]\\
+\sum_{s\in \mathcal{S}}\pi _{A}(1-\pi _{A})p( s) ( \pi _{D(1)}(s)-\pi _{D(0)}(s)) \beta ( s) -\Xi _{n}\beta
\end{array}
\right] \notag\\
&\overset{(1)}{=}\frac{1}{\Xi _{n}}\left[
\begin{array}{c}
-\sum_{s\in \mathcal{S}}\frac{n(s)}{n}\frac{n_{A}(s)}{n(s)}\frac{n_{AD}(s)}{ n_{A}(s)}\tfrac{\pi _{D(1)}(s)-\pi _{D(0)}(s)}{\pi _{D(1)}(s)} E[Y(1)|C,S=s](\sum_{\tilde{s}\in \mathcal{S}}\frac{n_{A}(\tilde{s})}{n( \tilde{s})}R_{n,4}( \tilde{s}) ) \\
-\sum_{s\in \mathcal{S}}\frac{n(s)}{n}\frac{n_{A}(s)}{n(s)}\frac{n_{AD}(s)}{ n_{A}(s)}\tfrac{\pi _{D(1)}(s)-\pi _{D(0)}(s)}{\pi _{D(1)}(s)} E[Y(1)|C,S=s](\sum_{\tilde{s}\in \mathcal{S}}p( \tilde{s}) R_{n,3}( \tilde{s}) ) \\
+(1-\pi _{A})\sum_{s\in \mathcal{S}}\frac{n_{A}(s)}{n(s)}\frac{n_{AD}(s)}{ n_{A}(s)}\tfrac{\pi _{D(1)}(s)-\pi _{D(0)}(s)}{\pi _{D(1)}(s)} E[Y(1)|C,S=s]R_{n,4}( s) \\
+(1-\pi _{A})\sum_{s\in \mathcal{S}}p( s) \frac{n_{AD}(s)}{ n_{A}(s)}\tfrac{\pi _{D(1)}(s)-\pi _{D(0)}(s)}{\pi _{D(1)}(s)} E[Y(1)|C,S=s]R_{n,3}( s) \\
+(1-\pi _{A})\sum_{s\in \mathcal{S}}p( s) \pi _{A}\tfrac{\pi _{D(1)}(s)-\pi _{D(0)}(s)}{\pi _{D(1)}(s)}E[Y(1)|C,S=s]R_{n,2}( 1,s) \\
-\sum_{s\in \mathcal{S}}\frac{n(s)}{n}(1-\frac{n_{A}(s)}{n(s)})(1-\frac{ n_{D}(s)-n_{AD}(s)}{n(s)-n_{A}(s)})\tfrac{\pi _{D(1)}(s)-\pi _{D(0)}(s)}{ 1-\pi _{D(0)}(s)}E[Y(0)|C,S=s]\sum_{\tilde{s}\in \mathcal{S}}\frac{n_{A}( \tilde{s})}{n(\tilde{s})}R_{n,4}( \tilde{s})  \\
-\sum_{s\in \mathcal{S}}\frac{n(s)}{n}(1-\frac{n_{A}(s)}{n(s)})(1-\frac{ n_{D}(s)-n_{AD}(s)}{n(s)-n_{A}(s)})\tfrac{\pi _{D(1)}(s)-\pi _{D(0)}(s)}{ 1-\pi _{D(0)}(s)}E[Y(0)|C,S=s]\sum_{\tilde{s}\in \mathcal{S}}p( \tilde{ s}) R_{n,3}( \tilde{s})  \\
-\pi _{A}\sum_{s\in \mathcal{S}}(1-\frac{n_{A}(s)}{n(s)})(1-\frac{ n_{D}(s)-n_{AD}(s)}{n(s)-n_{A}(s)})\tfrac{\pi _{D(1)}(s)-\pi _{D(0)}(s)}{ 1-\pi _{D(0)}(s)}E[Y(0)|C,S=s]R_{n,4}( s) \\
+\pi _{A}\sum_{s\in \mathcal{S}}p( s) (1-\frac{n_{D}(s)-n_{AD}(s) }{n(s)-n_{A}(s)})\tfrac{\pi _{D(1)}(s)-\pi _{D(0)}(s)}{1-\pi _{D(0)}(s)} E[Y(0)|C,S=s]R_{n,3}( s) \\
+\pi _{A}\sum_{s\in \mathcal{S}}p( s) (1-\pi _{A})\tfrac{\pi _{D(1)}(s)-\pi _{D(0)}(s)}{1-\pi _{D(0)}(s)}E[Y(0)|C,S=s]R_{n,2}( 2,s) -\beta \sqrt{n}( \Xi _{n}-\Xi )
\end{array}
\right] +o_{p}(1) \notag\\
&\overset{(2)}{=}\frac{1}{\Xi }\left[ 
\begin{array}{c}
\pi _{A}(1-\pi _{A})\sum_{s\in \mathcal{S}}( \pi _{D(1)}(s)-\pi _{D(0)}(s)) ( \beta ( s) -\beta ) R_{n,4}( s) \\
+(1-\pi _{A})\sum_{s\in \mathcal{S}}p( s) ( \pi _{D(1)}(s)-\pi _{D(0)}(s)) E[Y(1)|C,S=s]R_{n,3}( s) \\
+\pi _{A}\sum_{s\in \mathcal{S}}p( s) ( \pi _{D(1)}(s)-\pi _{D(0)}(s)) E[Y(0)|C,S=s]R_{n,3}( s) \\
-\beta \sum_{s\in \mathcal{S}}p( s) (
\begin{array}{c}
( 1-\pi _{A}) \pi _{D(1)}(s)-\pi _{A}\sum_{\tilde{s}\in \mathcal{S }}p( \tilde{s}) \pi _{D(1)}(\tilde{s}) \\
+\pi _{A}\pi _{D(0)}(s)-( 1-\pi _{A}) \sum_{\tilde{s}\in \mathcal{ S}}p( \tilde{s}) \pi _{D(0)}(\tilde{s})
\end{array}
) R_{n,3}( s) \\
-\pi _{A}\sum_{s\in \mathcal{S}}p( s) ( \pi _{D(1)}(s)-\pi _{D(0)}(s)) E[Y(1)|C,S=s](\sum_{\tilde{s}\in \mathcal{S}}p( \tilde{s}) R_{n,3}( \tilde{s}) ) \\
-(1-\pi _{A})\sum_{s\in \mathcal{S}}p( s) ( \pi _{D(1)}(s)-\pi _{D(0)}(s)) E[Y(0)|C,S=s](\sum_{\tilde{s}\in \mathcal{S} }p( \tilde{s}) R_{n,3}( \tilde{s}) ) \\
+\pi _{A}(1-\pi _{A})\sum_{s\in \mathcal{S}}p( s) \tfrac{\pi _{D(1)}(s)-\pi _{D(0)}(s)}{\pi _{D(1)}(s)}E[Y(1)|C,S=s]R_{n,2}( 1,s) \\
-\beta \pi _{A}( 1-\pi _{A}) \sum_{s\in \mathcal{S}}p( s) R_{n,2}( 1,s) +\beta \pi _{A}( 1-\pi _{A}) \sum_{s\in \mathcal{S}}p( s) R_{n,2}( 2,s) \\
+\pi _{A}(1-\pi _{A})\sum_{s\in \mathcal{S}}p( s) \tfrac{\pi _{D(1)}(s)-\pi _{D(0)}(s)}{1-\pi _{D(0)}(s)}E[Y(0)|C,S=s]R_{n,2}( 2,s)
\end{array}
\right] +o_{p}( 1) ,
\label{eq:eta_4}
\end{align}
where (1) holds by the definitions in \eqref{eq:plim_2SR4}, and (2) holds by Assumptions \ref{ass:1} and \ref{ass:2}(b), and Lemmas \ref{lem:A1and2_impliesold3} and \ref{lem:AsyDist}.

From these results, the following derivation follows.
\begin{align}
&\sqrt{n}( \hat{\beta}_{\mathrm{2s}}-\beta ) = \notag \\
&\frac{1}{P( C) }\left[ 
\begin{array}{c}
\sum_{s\in \mathcal{S}}\frac{R_{n,1}(1,1,s)}{\pi _{A}}+\sum_{s\in \mathcal{S} }\frac{R_{n,1}(0,1,s)}{\pi _{A}}-\sum_{s\in \mathcal{S}}\frac{R_{n,1}(1,0,s) }{1-\pi _{A}}-\sum_{s\in \mathcal{S}}\frac{R_{n,1}(0,0,s)}{1-\pi _{A}} \\
+\sum_{s\in \mathcal{S}}p( s) \left[ 
\begin{array}{c}
\tfrac{\pi _{D(0)}(s)}{\pi _{D(1)}(s)}E[Y(1)|AT,S=s]-E[Y(0)|NT,S=s] \\ 
+\tfrac{\pi _{D(1)}(s)-\pi _{D(0)}(s)}{\pi _{D(1)}(s)}E[Y(1)|C,S=s]-\beta
\end{array}
\right] R_{n,2}( 1,s)  \\ 
+\sum_{s\in \mathcal{S}}p( s) \left[
\begin{array}{c}
\tfrac{1-\pi _{D(1)}(s)}{1-\pi _{D(0)}(s)}E[Y(0)|NT,S=s]-E[Y(1)|AT,S=s] \\
+\tfrac{\pi _{D(1)}(s)-\pi _{D(0)}(s)}{1-\pi _{D(0)}(s)}E[Y(0)|C,S=s]+\beta
\end{array}
\right] R_{n,2}( 2,s) + \\
\sum_{s\in \mathcal{S}}\frac{p( s) }{\pi _{A}(1-\pi _{A})}\left[
\begin{array}{c}
[ \pi _{A}\pi _{D( 0) }( s) +( 1-\pi_{A}) \pi _{D( 1) }( s) ] ( \beta( s) -\beta )  \\ 
-\sum_{\tilde{s}\in \mathcal{S}}p( \tilde{s}) ( \pi _{A}\pi_{D( 1) }( \tilde{s}) +( 1-\pi _{A}) \pi
_{D( 0) }( \tilde{s}) ) ( \beta ( \tilde{s}) -\beta )  \\ 
+\gamma ( s) -\sum_{\tilde{s}\in \mathcal{S}}p(\tilde{s})\gamma( \tilde{s}) 
\end{array}
\right] 
R_{n,3}( s) \\
+\sum_{s\in \mathcal{S}}( \pi _{D(1)}(s)-\pi _{D(0)}(s)) ( \beta ( s) -\beta ) R_{n,4}( s)
\end{array}
\right] + o_p(1) ,\label{eq:beta_2sr_2}
\end{align}
where the equality holds by \eqref{eq:plim_2SR4}, \eqref{eq:beta_2sr}, \eqref{eq:eta_2}, \eqref{eq:eta_3}, and \eqref{eq:eta_4}, and where $(\beta(s),\gamma(s))$ are defined as in \eqref{eq:hat_limits_1}. 
%In particular, it relies on the following derivation.
%\begin{align}
%&\left[
%\begin{array}{c}
%\pi _{D(0)}(s)E[Y(1)|AT,S=s] \\
%-( \sum_{\tilde{s}\in \mathcal{S}}p( \tilde{s}) \pi _{D( 0) }( \tilde{s}) E[Y(1)|AT,S=\tilde{s}]) \\
%+(1-\pi _{D( 1) }( s) )E[Y(0)|NT,S=s] \\
%-( \sum_{\tilde{s}\in \mathcal{S}}p( \tilde{s}) ( 1-\pi _{D(1)}(\tilde{s})) E[Y(0)|NT,S=\tilde{s}]) \\
%+(1-\pi _{A})( \pi _{D(1)}(s)-\pi _{D(0)}(s)) E[Y(1)|C,S=s] \\
%-\pi _{A}( \sum_{\tilde{s}\in \mathcal{S}}p( \tilde{s}) ( \pi _{D(1)}(\tilde{s})-\pi _{D(0)}(\tilde{s})) E[Y(1)|C,S= \tilde{s}]) \\
%+\pi _{A}( \pi _{D(1)}(s)-\pi _{D(0)}(s)) E[Y(0)|C,S=s] \\ -(1-\pi _{A})( \sum_{\tilde{s}\in \mathcal{S}}p( \tilde{s}) ( \pi _{D(1)}(\tilde{s})-\pi _{D(0)}(\tilde{s})) E[Y(0)|C,S= \tilde{s}]) \\
%-\beta ( 1-\pi _{A}) \pi _{D(1)}(s)+\beta \pi _{A}\sum_{\tilde{s} \in \mathcal{S}}p( \tilde{s}) \pi _{D(1)}(\tilde{s}) \\
%-\beta \pi _{A}\pi _{D(0)}(s)+\beta ( 1-\pi _{A}) \sum_{\tilde{s} \in \mathcal{S}}p( \tilde{s}) \pi _{D(0)}(\tilde{s})
%\end{array} \right]\notag\\
%&=\left[
%\begin{array}{c}
%[ \pi _{A}\pi _{D( 0) }( s) +( 1-\pi_{A}) \pi _{D( 1) }( s) ] ( \beta( s) -\beta )  \\ 
%-\sum_{\tilde{s}\in \mathcal{S}}p( \tilde{s}) ( \pi _{A}\pi_{D( 1) }( \tilde{s}) +( 1-\pi _{A}) \pi_{D( 0) }( \tilde{s}) ) ( \beta ( \tilde{s}) -\beta )  \\ 
%+\gamma ( s) -\sum_{\tilde{s}\in \mathcal{S}}p(\tilde{s})\gamma( \tilde{s})
%\end{array}
%\right], \label{eq:beta_2sr_3}
%\end{align}
The desired result then follows from \eqref{eq:beta_2sr_2} and Lemma \ref{lem:AsyDist}.
\end{proof}
%%%%%%%% DIVIDER %%%%%%%%%%%%

%%%%%%%% DIVIDER %%%%%%%%%%%%
\begin{proof}[Proof of Theorem \ref{thm:2SR_se}.]
Note that Assumption \ref{ass:1}, \ref{ass:2}(b), and \ref{ass:3}(c), Lemma \ref{lem:A1and2_impliesold3}, and \eqref{eq:hat_limits_1} and \eqref{eq:hat_limits_2}, imply that $\hat{V}_{A}^{\mathrm{2s}} ~\overset{p}{\to }~V_{{A}}^{\mathrm{2s}}$. The desired result follows from this and Theorem \ref{thm:SAT_se}.
\end{proof}
%%%%%%%% DIVIDER %%%%%%%%%%%%

\begin{proof}[Proof of Theorem \ref{thm:2SR_test}.]
This result follows from elementary convergence arguments and Theorems \ref{thm:AsyDist_2SR} and \ref{thm:2SR_se}.
\end{proof}

%%%%%%%%%%%%%%%%%%%%%%%%%%%%%%%%%%%%%%%%%%%%%%%%%%%%%%%%%%%%%%%%%%%%%%%%%%%%%
\subsection{Proofs of results in Section \ref{sec:optim}}

%%%%%%%% DIVIDER %%%%%%%%%%%%

\begin{proof}[Proof of Theorem \ref{thm:coarser_moreVar}]
Note that \eqref{eq:coarser_2} is a consequence of Theorems \ref{thm:AsyDist_SAT}, \ref{thm:AsyDist_SFE}, and \ref{thm:AsyDist_2SR}. To complete the proof, it suffices to show \eqref{eq:coarser_3}.

Fix $s\in \mathcal{S}_{2}$ arbitrarily. By \eqref{eq:coarser_1}, there is a set $(s_{j}( s) \in \mathcal{S}_{1})_{j=1}^{J( s) }$ (dependent on $s$) s.t.\ $(S_{2}=s)=\cup _{j=1}^{J(s)}(S_{1}=s_{j}( s) )$. Then, consider the following derivation.
\begin{align}
& V[Y(1)|AT,S_{2}=s]{\pi }_{D(0)}(s)  \notag \\
& \overset{(1)}{=}\sum_{j=1}^{J( s) }[ V[Y(1)|AT,S_{1}={s_{j} }( s) ]+(E[Y(1)|AT,S_{1}={s_{j}}( s) ]-{E} [Y(1)|AT,S_{2}=s])^{2}] {\pi }_{D(0)}({s_{j}}( s) )\frac{{p} ({s_{j}}( s) )}{{p}(s)} \notag \\
& =\left( 
\begin{array}{c}
\sum_{j=1}^{J( s) }V[Y(1)|AT,S_{1}={s_{j}}( s) ]{\pi } _{D(0)}(s_{j})\frac{{p}({s_{j}}( s) )}{{p}(s)}+\\
\sum_{j=1}^{J ( s) }E[Y(1)|AT,S_{1}={s_{j}}( s) ]^{2}{\pi }_{D(0)}({ s_{j}}( s) )\frac{{p}({s_{j}}( s) )}{{p}(s)} -E[Y(1)|AT,S_{2}=s]^{2}{\pi }_{D(0)}(s)
\end{array}
\right) ,\label{eq:coarse1}
\end{align}
where (1) follows from the law of total variance. By a similar argument,
\begin{align}
&V[Y(0)|NT,S_{2}=s](1-{\pi }_{D(1)}(s))  \notag\\
&=\left( 
\begin{array}{c}
\sum_{j=1}^{J( s) }V[Y(0)|NT,S_{1}={s_{j}}( s) ](1-{ \pi }_{D(1)}({s_{j}}( s) ))\frac{{p}({s_{j}}( s) )}{{p} (s)}+\\
\sum_{j=1}^{J( s) }E[Y(0)|NT,S_{1}={s_{j}}( s) ]^{2}(1-{\pi }_{D(1)}(s))\frac{{p}({s_{j}}( s) )}{{p}(s)} -E[Y(0)|NT,S_{2}=s]^{2}(1-{\pi }_{D(1)}(s))
\end{array}
\right)  \label{eq:coarse2}
\end{align}
and, for $d\in\{ 0,1\} $,
\begin{align}
V[ Y(d)|C,S_{2}=s] ({\pi }_{D(1)}(s)-{\pi }_{D(0)}(s))=\left( 
\begin{array}{c}
\sum_{j=1}^{J( s) }V[Y(d)|C,S_{1}={s_{j}}( s) ]({\pi } _{D(1)}({s_{j}}( s) )-{\pi }_{D(0)}({s_{j}}( s) )) \frac{{p}({s_{j}}( s) )}{{p}(s)}+\\
\sum_{j=1}^{J( s) }E[Y(d)|C,S_{1}={s_{j}}( s) ]^{2}({\pi }_{D(1)}({s_{j}}( s) )-{\pi }_{D(0)}({s_{j}}( s) ))\frac{{p}({s_{j}}( s) )}{{p}(s)} \\
-E[Y(d)|C,S_{2}=s]^{2}({\pi }_{D(1)}(s)-{\pi }_{D(0)}(s))
\end{array}
\right) .\label{eq:coarse3}
\end{align}

For $a=1,2,$ let $V_{1,a}$ be equal to $V_{Y,1}^{\mathrm{sat}}+V_{D,1}^{ \mathrm{sat}}$ for RCT with strata $\mathcal{S}_{a}$. Next, consider the following derivation.
\begin{align}
V_{1,2} &\overset{(1)}{=}\frac{1}{\pi _{A}P(C)^{2}}\sum_{s\in \mathcal{S}_{2}}p(s)
\left[ 
\begin{array}{c}
{V}[Y(1)|AT,S_{2}=s]{\pi }_{D(0)}(s)+{V}[Y(0)|NT,S_{2}=s](1-{\pi }_{D(1)}(s)) \\
+{V}[Y(1)|C,S_{2}=s]({\pi }_{D(1)}(s)-{\pi }_{D(0)}(s)) \\
+({E}[Y(1)|C,S_{2}=s]-{E}[Y(1)|AT,S_{2}=s])^{2}\frac{{\pi }_{D(0)}(s)({\pi } _{D(1)}(s)-{\pi }_{D(0)}(s))}{{\pi }_{D(1)}(s)} \\
+\left[ 
\begin{array}{c}
-\pi _{D(0)}(s)(E[Y(1)|C,S_{2}=s]-E[Y(1)|AT,S_{2}=s]) \\ 
+\pi _{D(1)}(s)(E[Y(0)|C,S_{2}=s]-E[Y(0)|NT,S_{2}=s]) \\ 
+\pi _{D(1)}(s)(E[Y(1)-Y(0)|C,S_{2}=s]-\beta )
\end{array}
\right] ^{2}\frac{(1-{\pi }_{D(1)}(s))}{{\pi }_{D(1)}(s)}
\end{array}
\right]   \notag \\
&\overset{(2)}{=}
\frac{1}{\pi _{A}P(C)^{2}}\sum_{s\in \mathcal{S}_{2}}p(s)\left[
\begin{array}{c}
\sum_{j=1}^{J( s) }V[Y(1)|AT,S_{1}={s_{j}}( s) ]{\pi } _{D(0)}(s_{j})\frac{{p}({s_{j}}( s) )}{{p}(s)}\\
+\sum_{j=1}^{J ( s) }V[Y(0)|NT,S_{1}={s_{j}}( s) ](1-{\pi }_{D(1)}({ s_{j}}( s) ))\frac{{p}({s_{j}}( s) )}{{p}(s)} \\
+\sum_{j=1}^{J( s) }V[Y(1)|C,S_{1}={s_{j}}( s) ]({\pi } _{D(1)}({s_{j}}( s) )-{\pi }_{D(0)}({s_{j}}( s) )) \frac{{p}({s_{j}}( s) )}{{p}(s)} \\
+\sum_{j=1}^{J( s) }E[Y(1)|AT,S_{1}={s_{j}}( s) ]^{2}{ \pi }_{D(0)}({s_{j}}( s) )\frac{{p}({s_{j}}( s) )}{{p} (s)}\\
+\sum_{j=1}^{J( s) }E[Y(0)|NT,S_{1}={s_{j}}( s) ]^{2}(1-{\pi }_{D(1)}(s))\frac{{p}({s_{j}}( s) )}{{p}(s)} \\
+\sum_{j=1}^{J( s) }E[Y(1)|C,S_{1}={s_{j}}( s) ]^{2}({ \pi }_{D(1)}({s_{j}}( s) )-{\pi }_{D(0)}({s_{j}}( s) )) \frac{{p}({s_{j}}( s) )}{{p}(s)} \\
-\left[ 
\begin{array}{c}
( \pi _{D(1)}(s)-\pi _{D(0)}(s)) E[Y(1)|C,S_{2}=s]+\pi _{D(0)}(s)E[Y(1)|AT,S_{2}=s] \\
+( 1-\pi _{D(1)}(s)) ( E[Y(0)|NT,S_{2}=s]+\beta )
\end{array}
\right] ^{2} \\
+\beta ^{2}\sum_{j=1}^{J( s) }( {1}-{\pi }_{D(1)}( s_{j}( s) ) ) \frac{p( s_{j}( s) ) }{p( s) }\\
+2\beta \sum_{j=1}^{J( s) }E[Y(0)|NT,S_{1}=s_{j}( s) ]( {1}-{\pi }_{D(1)}(s_{j}( s) )) \frac{p( s_{j}( s) ) }{p(s)}
\end{array}\label{eq:coarse4}
\right],
\end{align}
where (1) follows from Theorem \ref{thm:AsyDist_SAT}, and (2) follows from \eqref{eq:coarse1}, \eqref{eq:coarse2}, and \eqref{eq:coarse3}. In turn,
\begin{align}
&V_{1,1} \overset{(1)}{=}
\frac{1}{\pi_{A}P(C)^{2}}
\sum_{s\in \mathcal{S}_{2}}p( s)
\left[
\begin{array}{c}
\sum_{j=1}^{J( s) }{V}[Y(1)|AT,S_{1}=s_{j}( s) ]{\pi } _{D(0)}(s_{j}( s) )\frac{p(s_{j}( s) )}{p( s) }\\
+\sum_{j=1}^{J( s) }{V}[Y(0)|NT,S_{1}=s_{j}( s) ](1-{\pi }_{D(1)}(s_{j}( s) ))\frac{p(s_{j}( s) )}{p( s) } \\
+\sum_{j=1}^{J( s) }{V}[Y(1)|C,S_{1}=s_{j}( s) ]({\pi } _{D(1)}(s_{j}( s) )-{\pi }_{D(0)}(s_{j}( s) ))\frac{ p(s_{j}( s) )}{p( s) } \\
+\sum_{j=1}^{J( s) }({E}[Y(1)|AT,S_{1}=s_{j}( s) ])^{2} {\pi }_{D(0)}(s_{j}( s) )\frac{p(s_{j}( s) )}{p( s) }\\
+\sum_{j=1}^{J( s) }E[Y(0)|NT,S_{1}=s_{j}( s) ]^{2}(1-{\pi }_{D(1)}(s_{j}( s) ))\frac{p(s_{j}( s) )}{p( s) } \\
+\sum_{j=1}^{J( s) }({E}[Y(1)|C,S_{1}=s_{j}( s) ])^{2}( {\pi }_{D(1)}(s_{j}( s) )-{\pi }_{D(0)}(s_{j}( s) )) \frac{p(s_{j}( s) )}{p( s) } \\
-\sum_{j=1}^{J( s) }
\left[
\begin{array}{c}
( \pi _{D(1)}(s_{j}( s) )-\pi _{D(0)}(s_{j}( s) )) E[Y(1)|C,S_{1}=s_{j}( s) ]\\
+\pi _{D(0)}(s_{j}( s) )E[Y(1)|AT,S_{1}=s_{j}( s) ]) \\
+( 1-\pi _{D(1)}(s_{j}( s) )) ( E[Y(0)|NT,S_{1}=s_{j}( s) ]+\beta )
\end{array}
\right] ^{2}\frac{p(s_{j}( s) )}{p( s) } \\
+\beta ^{2}\sum_{j=1}^{J( s) }( 1-{\pi }_{D(1)}(s_{j}( s) )) \frac{p(s_{j}( s) )}{p( s) }\\
+2\beta \sum_{j=1}^{J( s) }E[Y(0)|NT,S_{1}=s_{j}( s) ]( 1- {\pi }_{D(1)}(s_{j}( s) )) \frac{p(s_{j}( s) )}{ p( s) }
\end{array}
\right],\label{eq:coarse5}
\end{align} 
where (1) holds by Theorem \ref{thm:AsyDist_SAT}. By \eqref{eq:coarse4} and \eqref{eq:coarse5},
\begin{align}
V_{1,1}-V_{1,2}
%\notag\\&=-\frac{1}{\pi _{A}P(C)^{2}}\sum_{s\in \mathcal{S} _{2}}p(s)\left[
%\begin{array}{c}
%\sum_{j=1}^{J( s) }\left[ 
%\begin{array}{c}
%( \pi _{D(1)}(s_{j}( s) )-\pi _{D(0)}(s_{j}( s) )) E[Y(1)|C,S_{1}=s_{j}( s) ]\\
%+\pi _{D(0)}(s_{j}( s) )E[Y(1)|AT,S_{1}=s_{j}( s) ] \\
%+( 1-\pi _{D(1)}(s_{j}( s) )) ( E[Y(0)|NT,S_{1}=s_{j}( s) ]+\beta )
%\end{array}
%\right] ^{2}\frac{p(s_{j}( s) )}{p( s) } \\
%-\left( \sum_{\tilde{j}=1}^{J( s) }\left[ 
%\begin{array}{c}
%( \pi _{D(1)}(s_{\tilde{j}}( s) )-\pi _{D(0)}(s_{\tilde{j} }( s) )) E[Y(1)|C,S_{1}=s_{\tilde{j}}( s) ]\\
%+\pi_{D(0)}(s_{\tilde{j}}( s) )E[Y(1)|AT,S_{1}=s_{\tilde{j}}( s) ]) \\
%+( 1-\pi _{D(1)}(s_{\tilde{j}}( s) )) ( E[Y(0)|NT,S_{1}=s_{\tilde{j}}( s) ]+\beta )
%\end{array}
%\right] \frac{p(s_{\tilde{j}}( s) )}{p( s) }\right) ^{2}
%\end{array}
%\right]  \notag \\
\overset{(1)}{=}-\frac{1}{\pi _{A}P(C)^{2}}\sum_{s\in \mathcal{S} _{2}}p(s)\sum_{j=1}^{J( s) }\frac{p(s_{j}( s) )}{ p( s) }\left( M_{1}( s_{j}( s) ) -\sum_{ \tilde{j}=1}^{J( s) }M_{1}( s_{\tilde{j}}( s) ) \frac{p(s_{\tilde{j}}( s) )}{p( s) }\right) ^{2}\leq 0, \label{eq:coarse6}
\end{align}
where (1) uses the following definition:
\begin{equation}
M_{1}( s_{j}( s) ) \equiv \left[ 
\begin{array}{c}
( \pi _{D(1)}(s_{j}( s) )-\pi _{D(0)}(s_{j}( s) )) E[Y(1)|C,S_{1}=s_{j}( s) ]+\pi _{D(0)}(s_{j}( s))E[Y(1)|AT,S_{1}=s_{j}( s) ] \\
+( 1-\pi _{D(1)}(s_{j}( s) )) ( E[Y(0)|NT,S_{1}=s_{j}( s) ]+\beta )
\end{array}
\right] . \label{eq:coarse7}
\end{equation}

For $a=1,2,$ let $V_{0,a}$ be equal to $V_{Y,0}^{\mathrm{sat} }+V_{D,0}^{\mathrm{sat}}$ for RCT with strata $\mathcal{S}_{a}$. By a similar argument,
\begin{equation}
V_{0,1}-V_{0,2}
=-\frac{1}{( 1-\pi _{A}) P(C)^{2}}\sum_{s\in \mathcal{S}_{2}}p(s)\sum_{j=1}^{J( s) }\frac{p(s_{j}( s))}{p( s) }\left( M_{0}( s_{j}( s) ) -\sum_{\tilde{j}=1}^{J( s) }M_{0}( s_{\tilde{j} }( s) ) \frac{p(s_{\tilde{j}}( s) )}{p( s) }\right) ^{2}\leq 0, \label{eq:coarse8}
\end{equation}
where we use the following definition:
\begin{equation}
M_{0}( s_{j}( s) ) \equiv \left[
\begin{array}{c}
( \pi _{D(1)}(s_{j}( s) )-\pi _{D(0)}(s_{j}( s) )) E[Y(0)|C,S_{1}=s_{j}( s) ]+\\
( 1-\pi _{D(1)}(s_{j}( s) )) E[Y(0)|NT,S_{1}=s_{j}( s) ] 
+\pi _{D(0)}(s_{j}( s) )( E[Y(1)|AT,S_{1}=s_{j}( s) ]-\beta )
\end{array}
\right] . \label{eq:coarse9}
\end{equation}

For $a=1,2,$ let $V_{H,a}$ be equal to $V_{H}^{\mathrm{sat}}$ for RCT with strata $\mathcal{S}_{a}$. By Theorem \ref{thm:AsyDist_SAT},
\begin{align}
V_{H,2} &=\frac{1}{P(C)^{2}}\sum_{s\in \mathcal{S}_{2}}p(s)(\pi _{D(1)}(s)-\pi_{D(0)}(s))^{2}(E[Y(1)-Y(0)|C,S_{2}=s]-\beta )^{2} \notag \\
&\overset{(1)}{=}\frac{1}{P(C)^{2}}\sum_{s\in \mathcal{S}_{2}}p(s)\left( \sum_{j=1}^{J( s) }{({\pi }_{D(1)}(s_{j}}( s) {)-{\pi } _{D(0)}{(s_{j}}( s) {)})({E}[Y(1)-Y(0)|C,S_{1}={s_{j}}( s) ]-\beta )}\frac{{p}({s_{j}}( s) )}{p( s) } \right) ^{2} \notag \\
&\overset{(2)}{=}\frac{1}{P(C)^{2}}\sum_{s\in \mathcal{S}_{2}}p(s)\left( \sum_{j=1}^{J( s) }\frac{{p}({s_{j}}( s) )}{p( s) }M_{H}( s_{j}( s) ) \right) ^{2},\label{eq:coarse10}
\end{align}
where (1) follows from the LIE, and (2) uses the following definition:
\begin{equation}
M_{H}( s_{j}( s) ) \equiv (\pi _{D(1)}( s_{j}( s) ) -\pi _{D(0)}( s_{j}( s) ) )(E[Y(1)-Y(0)|C,S_{1}=s_{j}( s) ]-\beta ). \label{eq:coarse11}
\end{equation}
In turn,
\begin{align}
V_{H,1} 
%&=\frac{1}{P(C)^{2}}\sum_{s\in \mathcal{S}_{2}}\sum_{j=1}^{J ( s) }{p}({s_{j}}( s) ){({\pi }_{D(1)}(s_{j}}( s) {)-{\pi }_{D(0)}{(s_{j}}( s) {)})}^{2}{({E} [Y(1)-Y(0)|C,S_{1}={s_{j}}( s) ]-\beta )}^{2} \notag \\
&=\frac{1}{P(C)^{2}}\sum_{s\in \mathcal{S}_{2}}p(s)\sum_{j=1}^{J( s) }\frac{{p}({s_{j}}( s) )}{p( s) }( M_{H}{ (s_{j}}( s) {)}) ^{2}.\label{eq:coarse12}
\end{align}
By \eqref{eq:coarse11} and \eqref{eq:coarse12},
\begin{equation}
V_{H,1}-V_{H,2}=\frac{1}{P(C)^{2}}\sum_{s\in \mathcal{S}_{2}}p( s) \sum_{j=1}^{J( s) }\frac{{p}({s_{j}}( s) )}{ p( s) }( M_{H}{(s_{j}}( s) {)-}\sum_{\tilde{j} =1}^{J( s) }\frac{{p}({s_{\tilde{j}}}( s) )}{p( s) }M_{H}( s_{\tilde{j}}( s) ) ) ^{2}\geq 0, \label{eq:coarse13}
\end{equation}

To conclude the proof, consider the following argument.
\begin{align*}
V_{\mathrm{sat},2}-V_{\mathrm{sat},1}
 &\overset{(1)}{=}V_{1,2}+V_{0,2}+V_{H,2}-V_{1,1}-V_{0,1}-V_{H,1}\\
&\overset{(2)}{=} \frac{1}{P(C)^{2}}\sum_{s\in \mathcal{S}_{2}}\sum_{j=1}^{J ( s) }p(s)\frac{p(s_{j}( s) )}{p( s) }\left[
\begin{array}{c}
\frac{1}{\pi _{A}}( M_{1}( s_{j}( s) ) -\sum_{ \tilde{j}=1}^{J( s) }M_{1}( s_{\tilde{j}}( s) ) \frac{p(s_{\tilde{j}}( s) )}{p( s) }) ^{2} \\
\frac{1}{( 1-\pi _{A}) }( M_{0}( s_{j}( s) ) -\sum_{\tilde{j}=1}^{J( s) }M_{0}( s_{\tilde{j} }( s) ) \frac{p(s_{\tilde{j}}( s) )}{p( s) }) ^{2} \\
-( M_{H}{(s_{j}}( s) {)-}\sum_{\tilde{j}=1}^{J( s) }M_{H}( s_{\tilde{j}}( s) ) \frac{{p}({s_{ \tilde{j}}}( s) )}{p( s) }) ^{2}
\end{array}
\right] \\
&\overset{(3)}{=}\frac{1}{P(C)^{2}}\sum_{s\in \mathcal{S}_{2}}p(s)\left[
\begin{array}{c}
\frac{1}{\pi _{A}}\sum_{j=1}^{J( s) }\frac{p(s_{j}( s) )}{p( s) }\left( 
\begin{array}{c}
M_{0}( s_{j}( s) ) +M_{H}( s_{j}( s) ) +\beta\\
-\sum_{\tilde{j} =1}^{J( s) }( M_{0}( s_{\tilde{j}}( s) ) +M_{H}( s_{\tilde{j}}( s) ) +\beta ) \frac{p(s_{\tilde{j}}( s) )}{p( s) }
\end{array}
\right) ^{2} \\
+\frac{1}{( 1-\pi _{A}) }\sum_{j=1}^{J( s) }\frac{ p(s_{j}( s) )}{p( s) }( M_{0}( s_{j}( s) ) -\sum_{\tilde{j}=1}^{J( s) }M_{0}( s_{ \tilde{j}}( s) ) \frac{p(s_{\tilde{j}}( s) )}{ p( s) }) ^{2} \\
-\sum_{j=1}^{J( s) }\frac{p(s_{j}( s) )}{p( s) }( M_{H}{(s_{j}}( s) {)-}\sum_{\tilde{j} =1}^{J( s) }M_{H}( s_{\tilde{j}}( s) ) \frac{{p}({s_{\tilde{j}}}( s) )}{p( s) }) ^{2}
\end{array}
\right]  \\
&=\frac{1}{\pi _{A}P(C)^{2}}\sum_{s\in \mathcal{S}_{2}}p(s)\left[
\begin{array}{c}
\sum_{j=1}^{J( s) }\frac{p(s_{j}( s) )}{p( s) }\left(
\begin{array}{c}
\frac{M_{0}( s_{j}( s) ) }{\sqrt{1-\pi _{A}}}+\sqrt{ 1-\pi _{A}}M_{H}( s_{j}( s) ) \\
-\sum_{\tilde{j}=1}^{J( s) }( \frac{M_{0}( s_{\tilde{j} }( s) ) }{\sqrt{1-\pi _{A}}}+\sqrt{1-\pi _{A}}M_{H}( s_{\tilde{j}}( s) ) ) \frac{p(s_{\tilde{j}}( s) )}{p( s) }
\end{array}
\right) ^{2}
\end{array}
\right] \geq 0,
\end{align*}
as required by \eqref{eq:coarser_3}, where (1) follows from Theorem \ref{thm:AsyDist_SAT}, (2) follows from \eqref{eq:coarse6}, \eqref{eq:coarse8}, and \eqref{eq:coarse13}, and (3) follows from \eqref{eq:coarse7}, \eqref{eq:coarse9}, and \eqref{eq:coarse11}, as it implies that $M_{1}( s_{j}( s) ) =M_{0}( s_{j}( s) ) +M_{H}( s_{j}( s) ) +\beta $.
\end{proof}

%%%%%%%% DIVIDER %%%%%%%%%%%%
\begin{proof}[Proof of Theorem \ref{thm:finer_strata}]
Note that \eqref{eq:newS_1} is a consequence of Theorems \ref{thm:AsyDist_SAT}, \ref{thm:AsyDist_SFE}, and \ref{thm:AsyDist_2SR}. It then suffices to show that $V_{ \mathrm{sat}}$ and the primitive parameters can be consistently estimated as in Theorems \ref{thm:SAT_se} and \ref{thm:primitiveEstimation} based on $( ( Y_{i},Z_{i},S_{i},A_{i}) ) _{i=1}^{n_{P}}$ with $S_{i}=S( Z_{i}) $ for all $ i=1,\dots ,n_{P}$. To this end, it suffices to show that $ ((Y_{i}(1),Y_{i}(0),Z_{i},D_{i}(1),D_{i}(0)))_{i=1}^{n_{P}}$, $ (A_{i})_{i=1}^{n_{P}}$, and $( S_{i}) _{i=1}^{n_{P}}$ satisfy Assumptions \ref{ass:1} and \ref{ass:2} with strata function $S: \mathcal{Z}\to \mathcal{S}$.

Note that Assumption \ref{ass:1} in the hypothetical RCT follows from the fact that the data in the pilot RCT is an i.i.d.\ sample from the same underlying population. Second, note that Assumption \ref{ass:2}(b) in the hypothetical RCT is directly imposed in the statement.

To conclude, we show Assumption \ref{ass:2}(a) in the hypothetical RCT. We use $W_{i}\equiv ( Y_{i}( 1) ,Y_{i}( 0) ,Z_{i},D_{i}( 1) ,D_{i}( 0) ) $ and $S^{P}_i= S^{P}(Z_i)$ for all $i=1,\dots ,n_{P}$. By Assumption \ref{ass:2}(a) in the pilot RCT,
\begin{equation}
P( ( ( A_{i},W_{i}) ) _{i=1}^{n_{P}}|( S_{i}^{P}) _{i=1}^{n_{P}}) ~=~P( ( A_{i}) _{i=1}^{n_{P}}|( S_{i}^{P}) _{i=1}^{n_{P}}) P( ( W_{i}) _{i=1}^{n_{P}}|( S_{i}^{P}) _{i=1}^{n_{P}}) . \label{eq:newS_2}
\end{equation}
Since $( S_{i}^{P}) _{i=1}^{n_{P}}$ is determined by $( Z_{i}) _{i=1}^{n_{P}}$ via $S_{i}^{P}=S^{P}( Z_{i}) $, \eqref{eq:newS_2} implies that
\begin{equation*}
\frac{P( ( ( A_{i},W_{i}) ) _{i=1}^{n_{P}}) }{P( ( S_{i}^{P}) _{i=1}^{n_{P}}) }~=~P( ( A_{i}) _{i=1}^{n}|( S_{i}^{P}) _{i=1}^{n}) \frac{P( ( W_{i}) _{i=1}^{n_{P}}) }{P( ( S_{i}^{P}) _{i=1}^{n_{P}}) },
\end{equation*}
and so
\begin{equation}
P( ( A_{i}) _{i=1}^{n_{P}}|( W_{i}) _{i=1}^{n_{P}}) ~=~P( ( A_{i}) _{i=1}^{n_{P}}|( S_{i}^{P}) _{i=1}^{n_{P}}) . \label{eq:newS_3}
\end{equation}

Next, consider the following derivation for any arbitrary $( s_{i}) _{i=1}^{n_{P}}\in \mathcal{S}^{n_{P}}$.
\begin{align}
&P( ( A_{i}) _{i=1}^{n_{P}}|( S_{i}) _{i=1}^{n_{P}}=( s_{i}) _{i=1}^{n_{P}}) \notag\\
&=\int_{( w_{i}) _{i=1}^{n_{P}}:( S_{i}) _{i=1}^{n_{P}}=( s_{i}) _{i=1}^{n_{P}}}
P( ( A_{i}) _{i=1}^{n}|( W_{i}) _{i=1}^{n_{P}}=( w_{i}) _{i=1}^{n_{P}}) dP( ( U_{i}) _{i=1}^{n_{P}}=( w_{i}) _{i=1}^{n_{P}}|( S_{i}) _{i=1}^{n_{P}}=( s_{i}) _{i=1}^{n_{P}})  \notag \\
&=\int_{( w_{i}) _{i=1}^{n_{P}}:( ( S_{i},S_{i}^{P}) ) _{i=1}^{n_{P}}=( ( s_{i},s_{i}^{P}) ) _{i=1}^{n_{P}}}\left[
\begin{array}{c}
P( ( A_{i}) _{i=1}^{n}|( U_{i}) _{i=1}^{n_{P}}=( w_{i}) _{i=1}^{n_{P}}) \\
dP( ( U_{i}) _{i=1}^{n_{P}}=( w_{i}) _{i=1}^{n_{P}}|( S_{i}) _{i=1}^{n_{P}}=( s_{i}) _{i=1}^{n_{P}},( S_{i}^{P}) _{i=1}^{n_{P}}=( s_{i}^{P}) _{i=1}^{n_{P}})
\end{array}
\right]   \notag \\
&\overset{(1)}{=}\int_{( w_{i}) _{i=1}^{n_{P}}:( ( S_{i},S_{i}^{P}) ) _{i=1}^{n_{P}}=( ( s_{i},s_{i}^{P}) ) _{i=1}^{n_{P}}}\left[
\begin{array}{c}
P( ( A_{i}) _{i=1}^{n}|( U_{i}) _{i=1}^{n_{P}}=( w_{i}) _{i=1}^{n_{P}}) \\
dP( ( U_{i}) _{i=1}^{n_{P}}=( w_{i}) _{i=1}^{n_{P}}|( ( S_{i},S_{i}^{P}) ) _{i=1}^{n_{P}}=( ( s_{i},s_{i}^{P}) ) _{i=1}^{n_{P}})
\end{array}
\right]   \notag \\
&\overset{(2)}{=}\int_{( w_{i}) _{i=1}^{n_{P}}:( ( S_{i},S_{i}^{P}) ) _{i=1}^{n_{P}}=( ( s_{i},s_{i}^{P}) ) _{i=1}^{n_{P}}}\left[
\begin{array}{c}
P( ( A_{i}) _{i=1}^{n_{P}}|( S_{i}^{P}) _{i=1}^{n_{P}}=( s_{i}^{P}) _{i=1}^{n_{P}}) \\
dP( ( U_{i}) _{i=1}^{n_{P}}=( w_{i}) _{i=1}^{n_{P}}|( ( S_{i},S_{i}^{P}) ) _{i=1}^{n_{P}}=( ( s_{i},s_{i}^{P}) ) _{i=1}^{n_{P}})
\end{array}
\right] \notag \\
&=P( ( A_{i}) _{i=1}^{n_{P}}|( S_{i}^{P}) _{i=1}^{n_{P}}=( s_{i}^{P}) _{i=1}^{n_{P}}) ,\label{eq:newS_4} 
\end{align}
where (1) holds by \eqref{eq:newS_0}, and so the values of $( s_{i}^{P}) _{i=1}^{n_{P}}$ are determined by $( s_{i}) _{i=1}^{n_{P}}$, and (2) holds by \eqref{eq:newS_3}. By combining \eqref{eq:newS_3} and \eqref{eq:newS_4}, the desired result follows.
\end{proof}
%%%%%%%% DIVIDER %%%%%%%%%%%%

%%%%%%%% DIVIDER %%%%%%%%%%%%

\begin{proof}[Proof of Theorem \ref{thm:opt_pi}]
We only prove both results for the SAT IV estimator. The proof of the second result for the other IV estimators is identical under $\tau(s)=0$ for all $s \in \mathcal{S}$. Using the notation in the statement, the asymptotic variance of the SAT IV estimator is
\begin{equation*}
    V_{\mathrm{sat}}~=~\frac{1}{P(C)^{2}}\sum_{s\in \mathcal{S}}p( s) \left( \frac{\Pi _{1}( s) }{\pi _{A}(s)}+\frac{\Pi _{2}( s) }{1-\pi _{A}(s)}\right) +V_{H}^{\mathrm{sat}}.
\end{equation*}
Then, \eqref{eq:Optimal_Pi_s} follows from minimizing the expression with respect to $( \pi _{A}( s) :s\in \mathcal{S}) $. In turn, \eqref{eq:Optimal_Pi} follows from a similar minimization under the restriction imposed by Assumption \ref{ass:3}(c). 

To conclude, note that the CMT implies the consistency of the plug-in estimators based on \eqref{eq:Optimal_Pi_s} and \eqref{eq:Optimal_Pi}.
\end{proof}
%%%%%%%% DIVIDER %%%%%%%%%%%%

%\subsection{Additional Monte Carlo results}\label{sec:MoreMC}
%TBA

\end{small}

\end{appendix}

% - -- - - - - - - - - - -
% HERE STARTS THE BIBLIOGRAPHY 
% - -- - - - - - - - - - - 
%\newpage
\bibliography{BIBLIOGRAPHY}

\end{document}